\documentclass[english]{article}
\usepackage[latin9]{inputenc}
\usepackage{mathrsfs}
\usepackage{bm}
\usepackage{amsmath}
\usepackage{amsthm}
\usepackage{amssymb}
\usepackage{graphicx}
\usepackage{esint}

\makeatletter
\theoremstyle{plain}

\theoremstyle{plain}

\ifx\proof\undefined
\newenvironment{proof}[1][\protect\proofname]{\par
	\normalfont\topsep6\p@\@plus6\p@\relax
	\trivlist
	\itemindent\parindent
	\item[\hskip\labelsep\scshape #1]\ignorespaces
}{%
	\endtrivlist\@endpefalse
}
\providecommand{\proofname}{Proof}
\fi

\usepackage{jhepmod}
\pdfoutput=1

\usepackage[table ]{ xcolor}

\usepackage{amssymb}
\usepackage{graphicx,color}

\usepackage{amsmath,amsthm}

\usepackage{mathrsfs}

\usepackage{bm}

\def\beq{\begin{equation}}
\def\eeq{\end{equation}}
\def\bi{\begin{itemize}}
\def\ei{\end{itemize}}
	\def\ba{\begin{array}}
	\def\ea{\end{array}}
	\def\bfig{\begin{figure}}
	\def\efig{\end{figure}}

	\def\C{\mathbb{C}}
	\def\R{\mathbb{R}}
	\def\Z{\mathbb{Z}}

	\newtheorem{theorem}{Theorem}[section]

	\newtheorem{lemma}[theorem]{Lemma}

	\newcommand{\Slc}{\mathrm{SL}(2,\mathbb{C})}

	\def\be{\begin{eqnarray}}
	\def\ee{\end{eqnarray}}

	\newcommand{\ca}{\mathcal A}

	\newcommand{\cc}{\mathcal C}
	\newcommand{\cd}{\mathcal D}
	\newcommand{\ce}{\mathcal E}
	\newcommand{\cf}{\mathcal F}
	\newcommand{\cg}{\mathcal G}
	\newcommand{\ch}{\mathcal H}
	\newcommand{\ci}{\mathcal I}
	
	\newcommand{\ck}{\mathcal K}
	\newcommand{\cl}{\mathcal L}
	\newcommand{\cm}{\mathcal M}
	
	\newcommand{\co}{\mathcal O}
	\newcommand{\calp}{\mathcal P}
	
	\newcommand{\calr}{\mathcal R}
	\newcommand{\cs}{\mathcal S}
	
	\newcommand{\cu}{\mathcal U}
	
	\newcommand{\cw}{\mathcal W}
	\newcommand{\cx}{\mathcal X}
	\newcommand{\cy}{\mathcal Y}

	\newcommand{\sm}{\mathscr{M}}

	  \newcommand{\Fa}{\mathfrak{A}}

	  \newcommand{\Fd}{\mathfrak{D}}
	  
	  \newcommand{\Ff}{\mathfrak{F}}

	\newcommand{\fp}{\mathfrak{p}}

	  \newcommand{\Fs}{\mathfrak{S}}

	  \newcommand{\Fw}{\mathfrak{W}}

	\renewcommand{\a}{\alpha}
	\renewcommand{\b}{\beta}
	\newcommand{\g}{\gamma}
	\newcommand{\G}{\Gamma}
	
	\newcommand{\eps}{\varepsilon}

	\newcommand{\sig}{\sigma}
	\newcommand{\Sig}{\Sigma}
	\renewcommand{\l}{\lambda}
	
	\renewcommand{\o}{\omega}
	
	\renewcommand{\t}{\tau}

	\newcommand{\rmd}{\mathrm d}

	\newcommand{\lt}{\left}
	\newcommand{\rt}{\right}

	\newcommand{\lag}{\left\langle}
	\newcommand{\rag}{\right\rangle}

	\newcommand{\act}{\rhd}

	\newcommand{\sn}{\mathscr{N}}
	
	\newcommand{\sk}{\mathscr{K}}
	\newcommand{\sw}{\mathscr{W}}
	\newcommand{\css}{\mathscr{S}}
	\newcommand{\re}{\mathrm{Re}}
	\newcommand{\im}{\mathrm{Im}}
	
	\newcommand{\half}{\frac{1}{2}}
	\newcommand{\tr}{\mathrm{Tr}}

	\newcommand{\bmy}{\bm{y}}
	\newcommand{\bmx}{\bm{x}}
	\newcommand{\bmL}{\bm{L}}
	
\newcommand{\bfq}{\bm{q}}

\newcommand{\uquqt}{U_{h}(sl_2)\otimes U_{\widetilde h}(sl_2)}
\newcommand{\suquqt}{\mathscr{U}_{\bm q}(sl_2)\otimes \mathscr{U}_{\widetilde{\bm {q}}}(sl_2)}
\newcommand{\uq}{U_{h}(sl_2)}
\newcommand{\uqt}{U_{\widetilde h}(sl_2)}
\newcommand{\suq}{\mathscr{U}_{\bfq}(sl_2)}
\newcommand{\suqt}{\mathscr{U}_{\widetilde \bfq}(sl_2)}
	
\newcommand{\wt}{\widetilde}

\title{Hamiltonian quantization of complex Chern-Simons theory at level-$\bm{k}$}

\author[1,2]{Muxin Han}

\affiliation[1]{Department of Physics, Florida Atlantic University, 777 Glades Road, Boca Raton, FL 33431-0991, USA}

\affiliation[2]{Department Physik, Institut f\"ur Quantengravitation, Theoretische Physik III, Friedrich-Alexander Universit\"at Erlangen-N\"urnberg, Staudtstr. 7/B2, 91058 Erlangen, Germany}

\emailAdd{hanm(At)fau.edu}

\abstract{
This paper develops a framework for the Hamiltonian quantization of complex Chern-Simons theory with gauge group $\mathrm{SL}(2,\mathbb{C})$ at an even level $k\in\mathbb{Z}_+$. Our approach follows the procedure of combinatorial quantization to construct the operator algebras of quantum holonomies on 2-surfaces and develop the representation theory. The $*$-representation of the operator algebra is carried by the infinite dimensional Hilbert space $\mathcal{H}_{\vec \l}$ and closely connects to the infinite-dimensional $*$-representation of the quantum deformed Lorentz group $\mathscr{U}_{\bm q}(sl_2)\otimes \mathscr{U}_{\widetilde{\bm {q}}}(sl_2)$. The quantum group $\mathscr{U}_{\bm q}(sl_2)\otimes \mathscr{U}_{\widetilde{\bm {q}}}(sl_2)$ also emerges from the quantum gauge transformations of the complex Chern-Simons theory. Focusing on a $m$-holed sphere $\Sigma_{0,m}$, the physical Hilbert space $\mathcal{H}_{phys}$ is identified by imposing the gauge invariance and the flatness constraint. The states in $\mathcal{H}_{phys}$ are the $\mathscr{U}_{\bm q}(sl_2)\otimes \mathscr{U}_{\widetilde{\bm {q}}}(sl_2)$-invariant linear functionals on a dense domain in $\ch_{\vec \l}$. Finally, we demonstrate that the physical Hilbert space carries a Fenchel-Nielsen representation, where a set of Wilson loop operators associated with a pants decomposition of $\Sigma_{0,m}$ are diagonalized. 
	}


\makeatother

\usepackage{babel}
\providecommand{\lemmaname}{Lemma}
\providecommand{\theoremname}{Theorem}

\begin{document}

\maketitle

\section{Introduction}

Chern-Simons theory, a topological quantum field theory (TQFT) defined by the Chern-Simons action \cite{ChernSimons1974, Witten1989a}, has profoundly impacted both mathematics and physics over the past decades. Its connections range from knot theory and 3-manifold invariants \cite{Witten1989a, ReshetikhinTuraev1991} to conformal field theory \cite{Elitzur:1989nr,Moore:1989yh} and quantum gravity \cite{Witten1988,thiemann2008modern,Engle2011}. The quantization of Chern-Simons theory with a compact gauge group is well-understood, leading to a rich and powerful toolbox for applications \cite{Witten:2015aoa,1995AdPhy..44..405W,Freedman:2001eqc}.

The quantization of Chern-Simons theory with non-compact or complex gauge groups, often referred to as complex Chern-Simons theory, is one of important open issues in TQFT. The classical Chern-Simons theory of complex gauge group on an oriented 3-manifold $M_3$ is defined by the action 
\be
S(\ca,\overline{\ca})=\frac{k+is}{8 \pi} \int_M \operatorname{Tr}\left(\ca \wedge d \ca+\frac{2}{3} \ca \wedge \ca \wedge \ca\right)+\frac{k-is}{8 \pi} \int_M \operatorname{Tr}\left(\overline{\ca} \wedge d \overline{\ca}+\frac{2}{3} \overline{\ca} \wedge \overline{\ca} \wedge \overline{\ca}\right),
\ee
where $\ca,\overline{\ca}$ are the complex connections. The theory depends on the integer level $k$. The other parameter $s\in\R$ or $i\R$ correspond to one of two different unitary branches. The theory is introduced by Witten in \cite{Witten1991} and arises naturally in various different contexts: It emerges from the analytic continuation of quantum group invariants and connections to hyperbolic geometry, known as the volume conjecture \cite{analcs,QFTvolume,1997LMaPh..39..269K,1999math......5075M}. It links to quantum gravity in (2+1) dimensions \cite{Witten1988} and spinfoam quantum gravity in (3+1) dimensions via the non-compact groups such as $\Slc$ \cite{HHKR,Han:2021tzw}. The complex Chern-Simons theory also closely relates to (2+1)-dimensional supersymmetric gauge theories known as the 3d-3d correspondence \cite{DGG11,Cordova:2013cea}.  However, the non-compactness of the gauge group leads to challenges in defining the theory rigorously, because the Hilbert spaces are typically infinite-dimensional \cite{analcs, Dimofte2011,levelk}.

Various approaches have been developed to tackle the quantization of complex Chern-Simons theory, including geometric quantization \cite{Witten1991, 2020arXiv201215630E}, state-integral models \cite{Dimofte2011,levelk,andersen2016level,Andersen2014}, and perturbative quantization including resurgence \cite{Bar-Natan1991,DGLZ,GukovMarinoPutrov2017}. The combinatorial quantization, successfully applied to Chern-Simons theory of compact groups \cite{Fock:1998nu, Alekseev:1994pa,Alekseev:1994au,Alekseev:1995rn}, offer a beautiful scheme of Hamiltonian quantization and relating the theory to quantum groups, which arise naturally from quantum deformation of gauge symmetries. This approach aims to construct Hilbert spaces and operators directly from combinatorial data associated with graph discretizations of the underlying manifold. The combinatorial quantization of Chern-Simons theory is based on the canonical formalism with constraints and uses holononomies as basic operators, closely resemble the scheme of Loop Quantum Gravity \cite{thiemann2008modern}. The combinatorial quantization has been applied to $\Slc$ Chern-Simons theory in the case of the vanishing level $k=0$ \cite{BNR}, and the theory relates to the quantum Lorentz group $U_q(\sl(2,\C)_{\R})$ with real $q$ \cite{Podles1990,BR}. A different quantization of the theory at $k=0$ known as Schur quantization recently arises from the connection with supersymmetric gauge theory with eight supercharges and results in a different quantum deformation of Lorentz group \cite{Gaiotto:2024osr}.

In this paper, we develop the combinatorial quantization of $\Slc$ Chern-Simons theory at even level $k=2N> 0$ and with $s\in\R$. The quantization connects the Chern-Simons theory to the quantum group $\suquqt$, which is a product of two Hopf algebras with 
\be
\bfq=\exp\lt(\frac{4\pi i}{k+is}\rt),\qquad \widetilde{\bfq}=\exp\lt(\frac{4\pi i}{k-is}\rt).\label{qandk000}
\ee
and the $*$-structure interchange between these two Hopf algebras. This quantum group emerges as the $\bfq$-deformed gauge symmetry and reduces to the classical Lorentz group when $\bfq\to1$. The combinatorial quantization gives the operator algebra of quantum holonomies on 2-surfaces, a.k.a. the graph algebra, It turns out that the representation of the graph algebra closely relates to the infinite-dimensional $*$-representation of $\suquqt$ \cite{Han:2024nkf}. 

The motivation of choosing the even level $k=2N$ is two folds: First, the even level leads to an interesting simplification: Although the classical theory depends on $k$, the quantum theory only sees the integer $N$ instead of $k$. As discussed in Section \ref{Quantization of discrete connections}, the $*$-representation of both the graph algebra and quantum gauge symmetry $\suquqt$ are constructed by the elementary operators $\bm{u}$ and $\bmy$ satisfying the $\bfq^2$-Weyl algebra $\bm{u}\bmy=\bfq^2\bmy\bm{u}$ represented irreducibly on the Hilbert space $\ch\simeq L^2(\R)\otimes\C^N$. Although the quantization procedure ``derives'' the quantum theory from classical theory, the quantum theory is in fact fundamental, whereas the classical theory is only an approximation from the quantum theory in semiclassical regime. Therefore, $\bfq^2=\exp\lt(\frac{4\pi i}{N+is/2}\rt)$ depending on $N$ plays a more fundamental role than $\bfq$. The level $k$ is only identified in the semiclassical limit, by relating the operator algebra to classical poisson bracket.

Second, $k=2N$ connects the combinatorial quantization to the existing approaches of quantum complex Chern-Simons theory, including the state-integral model \cite{levelk} and the level-$N$ quantum Teichm\"uller theory \cite{andersen2016level,Andersen2014}. The quantizations in these approaches are based on the $q$-Weyl algebra represented on $L^2(\R)\otimes\C^N$, and their $q$ corresponds to our $\bfq^2$. When $N=1$, the combinatorial quantization relates to the standard quantum Teichm\"uller theory and the representation of modular double of $\mathrm{SL}(2,\R)$ \cite{Derkachov:2013cqa,Kashaev2001,Nidaiev:2013bda}. The quantum group $\suquqt$ may be viewed as a generalization of the modular double by including an integer level. 


When generalizing to odd level $k$, one may consider the $\bfq$-Weyl algebra $\bmx\bmy=\bfq\bmx\bmy$ and the irreducible representation on $\ch_0\simeq L^2(\R)\otimes\C^k$. But the graph algebra from the combinatorial quantization will only relate to a subalgebra generated by $\bm{u}=\bmx^2$ and $\bmy$. When $k=2N$, one can find a subspace $\ch\subset \ch_0$ where the $(\bm{u},\bmy)$ subalgebra is irreducible, but $\ch$ does not exist for odd $k$. Therefore, one has to take $\bmx,\bmy$ as the basic operator and representation the operator algebra on $\ch_0$. As a result, the theory at odd $k$ is very different from the level-$N$ quantum Teichm\"uller theory, since the operator algebra always involves the square of the basic operator $\bmx$. The detailed discussion for the case of odd $k$ is postponed to a future publication.

The following gives the architecture of this paper and summarizes the main results: Section \ref{Quantum Lorentz group at level-k} introduces the quantum Lorentz group $\suquqt$ and the quasi-triangular extension $\uquqt$. Some useful properties about the $R$-matrices and the finite-dimensional representations are discussed. Section \ref{Infinite-dimensional representations} discusses the infinite-dimensional $*$-representation $\pi^\l$ of $\suquqt$ based on the Weyl algebra and the quantum torus algebra. A theorem on the existence and uniqueness of invariant bilinear form of $\pi^\l\otimes\pi^\l$ is proven in this section, and this result is useful for constructing physical states for complex Chern-Simons theory. Section \ref{Complex Chern-Simons theory and discrete connections} briefly discuss the classical complex Chern-Simons theory and the Fock-Rosly Poisson structure on discrete connections and the gauge group. In Section \ref{Quantization of discrete connections}, we perform the combinatorial quantization and define the graph algebra of quantum holonomies. The quantization also deforms the gauge group from $\Slc$ to the quantum Lorentz group $\suquqt$. Section \ref{Representation of the graph algebra} defines the infinite-dimensional $*$-representation of the graph algebra for the $m$-holed sphere. The representation is carried by $m$ copies of $\ch\simeq L^2(\R)\otimes\C^N$. Starting from this section, our discussion only focuses on the quantization on $m$-holed sphere. Section \ref{Gauge invariance and flatness constraint} defines the $*$-representation of gauge transformations on the same Hilbert space. The flatness constraint is automatically satisfied by the gauge invariant states, so the physical states of the theory are given by the gauge invariant states, which are the $\suquqt$ invariant linear functionals on certain dense domain in the Hilbert space. Section \ref{Physical Hilbert space} derives the Hilbert space of physical states using Clebsch-Gordan decompositions for the infinite-dimensional representations of $\suquqt$ \cite{Han:2024nkf}. We introduce the "3j-symbol" notation for the physical states and relate the construction to refined algebraic quantization. The operators on the physical Hilbert space gives a $*$-representation to gauge invariant observables. Section \ref{Wilson loop operators} connects the formalism to the pants decomposition of $m$-holed sphere. The Wilson loop operators along the cuts of the pants decomposition are gauge invariant observables, and we find the representation of the physical Hilbert space where the Wilson loop operator are diagonalized. We call this representation the Fenchel-Nielsen representation of the complex Chern-Simons theory, as a generalization of quantizing Fenchel-Nielsen coordinates in quantum Teichm\"uller theory.

\section{A quantum Lorentz group}\label{Quantum Lorentz group at level-k}

\subsection{Quantum Lorentz group $\mathscr{U}_{\bf q}(sl_2)\otimes \mathscr{U}_{\widetilde{\bf {q}}}(sl_2)$}

The deformation parameters in this paper are $\bfq,\wt{\bfq}$ given by
\be
\bfq=e^h=\exp\lt[\frac{2\pi i}{k}(1+b^2)\rt],\qquad \widetilde{\bfq}=e^{\widetilde{h}}=\exp\lt[\frac{2\pi i}{k}(1+b^{-2})\rt]=\overline{\bm q}^{-1},
\ee
and $k,b$ satisfies 
\be
k\in\mathbb{Z}_+,\qquad|b|=1,\qquad\mathrm{Re}(b)>0, \qquad \mathrm{Im}(b)>0.
\ee
$b$ and $s\in\R$ are related by 
\be
is=k\frac{1-b^2}{1+b^2}\ \ \text{or}\ \ b^2=\frac{k-is}{k+is}.
\ee

We define the Hopf $*$-algebra $\suquqt$ and refer to this algebra as the quantum Lorentz group at level-$k$.  $\suquqt$ is the polynomial algebra generated by $\bm{1}$ and $E,F,K,K^{-1}$, $\widetilde{E},\widetilde{F},\widetilde K,\widetilde K^{-1}$. The generators satisfy the following Hopf-algebra relations
\be 
KE&=&\bm{q}^{2}EK,\qquad KF=\bm{q}^{-2}FK,\qquad\left[E,F\right]=\frac{K-K^{-1}}{\bm{q}-\bm{q}^{-1}},\label{EFKalg1}\\
\widetilde{K}\widetilde{E}	&=&\widetilde{\bm{q}}^{2}\widetilde{E}\widetilde{K},\qquad\widetilde{K}\widetilde{F}=\widetilde{\bm{q}}^{-2}\widetilde{F}\widetilde{K},\qquad\left[\widetilde{E},\widetilde{F}\right]=\frac{\widetilde{K}-\widetilde{K}^{-1}}{\widetilde{\bm{q}}-\widetilde{\bm{q}}^{-1}},\label{EFKalg2}
\ee
The set of generators $\{E,F,K\}$ commutes with $\{\widetilde{E},\widetilde{F},\widetilde{K}\}$. The coproduct, antipode, and counit are given by 
\be 
&&\Delta E=E\otimes K+1\otimes E,\qquad\Delta\widetilde{E}=\widetilde{E}\otimes\widetilde{K}+1\otimes\widetilde{E},\\
&&\Delta F=F\otimes1+K^{-1}\otimes F,\qquad\Delta\widetilde{F}=\widetilde{F}\otimes1+\widetilde{K}^{-1}\otimes\widetilde{F}\\
&&\Delta K^{\pm1}=K^{\pm1}\otimes K^{\pm1},\qquad\Delta\widetilde{K}^{\pm1}=\widetilde{K}^{\pm1}\otimes\widetilde{K}^{\pm1}\\
&&S\left(K^{\pm1}\right)=K^{\mp1},\qquad S\left(E\right)=-EK^{-1},\qquad S\left(F\right)=-KF,\\
&&S\left(\widetilde{K}^{\pm1}\right)=\widetilde{K}^{\mp1},\qquad S\left(\widetilde{E}\right)=-\widetilde{E}\widetilde{K}^{-1},\qquad S\left(\widetilde{F}\right)=-\widetilde{K}\widetilde{F}\\
&&\varepsilon\left({K}^{\pm1}\right)=\varepsilon\left(\widetilde{{K}}^{\pm1}\right)=1,\qquad \varepsilon(E)=\varepsilon(F)=\varepsilon(\widetilde{E})=\varepsilon(\widetilde{F})=0.
\ee
The $*$-action is given by
\be 
E^*=\widetilde{E},\qquad F^*=\widetilde{F},\qquad K^{\pm1}{}^*=\widetilde{K}^{\pm1}.
\ee
obeying $*^2=1$, $(\xi \zeta )^*=\zeta ^*\xi^*$, and $(\xi \otimes \zeta )^*=\xi^*\otimes\zeta^*$ for $\xi,\zeta\in \suquqt$. Moreover, we verify that
\be
\Delta \xi^*=\lt(\Delta\xi\rt)^*,\qquad \eps(\xi^*)=\overline{\eps(\xi)},\qquad (S\circ *)^2=\mathrm{id}.
\ee
In our notation, we identify $\xi\in \suq$ and its embedding $\xi\otimes 1\in \suquqt$, and similarly, $\widetilde{\xi}\in\suqt$ is identified with $1\otimes \widetilde{\xi}\in \suquqt$. Moreover, for $\xi\otimes \widetilde{\zeta}\in \suquqt$, we ignore $\otimes$ and denote by $\xi \widetilde{\zeta}\in \suquqt$, while keeping in mind that $[\xi, \widetilde{\zeta}]=0$.

The properties of $\suquqt$ share some similarities with the modular double of $U_{\bfq}(sl(2,\R))$ (see e.g. \cite{Kashaev2001,Derkachov:2013cqa,Nidaiev:2013bda}), especially in the following discussion of infinite-dimensional irreducible representations. 
For $k=2$, the deformation parameters reduces to $\bfq\to-e^{\pi i b^2}$ and $\wt{\bfq}\to -e^{\pi i b^{-2}}$, and $\suquqt$ reproduces the modular double of $U_{\bfq}(sl(2,\R))$ up to flipping signs of $\bfq,\wt{\bfq}$.

\subsection{Dual quantum group}

The $\star$-Hopf algebra ${SL}_{\bfq}(2)\otimes {SL}_{\widetilde \bfq}(2)$ dual to $\suquqt$ is the matrix quantum group generated by 1 and the elements $g^i_{\ j}$, $\widetilde{g}^i_{\ j}$ of the $2\times 2$ matrices $g$, $\widetilde{g}$, subject to the relations $[g^i_{\ j}$, $\widetilde{g}^k_{\ l}]=0$ and 
\be
R_{12}{g}_1{g}_2={g}_2{g}_1 R_{12},\qquad \det{\!}_{\bm q} (g)=g^1_{\ 1}g^2_{\ 2}-\bfq^{-1} g^1_{\ 2}g^2_{\ 1}=1,\\
\widetilde{R}_{12}\widetilde{{g}}_1\widetilde{{g}}_2=\widetilde{{g}}_2\widetilde{{g}}_1 \widetilde{R}_{12},\qquad \det{\!}_{\widetilde{\bm q}} (\widetilde{g})=\widetilde{g}^1_{\ 1}\widetilde{g}^2_{\ 2}-\widetilde{\bfq}^{-1} \widetilde{g}^1_{\ 2}\widetilde{g}^2_{\ 1}=1.
\ee
where $R$ and $\widetilde{R}$ matrices are
\be
R=\left(
	\begin{array}{cccc}
	 R^1_{\ 1}{}^1_{\ 1} & R^1_{\ 1}{}^1_{\ 2} & R^1_{\ 2}{}^1_{\ 1} & R^1_{\ 2}{}^1_{\ 2} \\
	 R^1_{\ 1}{}^2_{\ 1} & R^1_{\ 1}{}^2_{\ 2} & R^1_{\ 2}{}^2_{\ 1} & R^1_{\ 2}{}^2_{\ 2} \\
	 R^2_{\ 1}{}^1_{\ 1} & R^2_{\ 1}{}^1_{\ 2} & R^2_{\ 2}{}^1_{\ 1} & R^2_{\ 2}{}^1_{\ 2}\\
	 R^2_{\ 1}{}^2_{\ 1} & R^2_{\ 1}{}^2_{\ 2} & R^2_{\ 2}{}^2_{\ 1} & R^2_{\ 2}{}^2_{\ 2}\\
	\end{array}
	\right)=q^{-1}\left(
	\begin{array}{cccc}
	 \bfq & 0 & 0 & 0 \\
	 0 & 1 & \bfq-\bfq^{-1} & 0 \\
	 0 & 0 & 1 & 0\\
	 0 & 0 & 0 & \bfq^{-1}\\
	\end{array}
	\right)\ .\label{Rmatrixhalfrep}
\ee
They are representations of the $R$-matrices that is discussed in the next section.

Generally, the matrix elements $(g^I)^i_{\ j}$ and $(\widetilde{g}^J)^i_{\ j}$ of the spin-$I$ and spin-$J$ irreducible representations of $\suq$ and $ \suqt$ belong to ${SL}_{\bfq}(2)\otimes {SL}_{\widetilde \bfq}(2)$. The duality between ${SL}_{\bfq}(2)\otimes {SL}_{\widetilde \bfq}(2)$ and $\suquqt$ is given by the pairing 
\be
\langle (g^I)^i_{\ j}(\widetilde{g}^J)^k_{\ l},\xi\widetilde{\zeta}\rangle = \rho^{I}(\xi)^i_{\ j}\wt{\rho}^{J}(\widetilde{\zeta})^k_{\ l}.\label{quantumduality}
\ee
where $\xi\in \suq$, $\widetilde{\zeta}\in  \suqt$, $(g^I)^i_{\ j}\in {SL}_{\bfq}(2)$, and $(\widetilde{g}^J)^k_{\ l}\in {SL}_{\widetilde \bfq}(2)$. $\rho^I,\wt{\rho}^{J}$ are the finite-dimensional irreducible representations of $\suq$ and $\suqt$.
 
The comultiplication $\delta$, counit $\epsilon$, and antipode $\cs$ are given by
\be
&&\delta {g}^i_{\ j}={g}^i_{\ k}\otimes{g}^k_{\ j},\qquad \delta \widetilde{g}^i_{\ j}=\widetilde{g}^i_{\ k}\otimes\widetilde{g}^k_{\ j},\\
&&\epsilon ( {g}^i_{\ j})=\epsilon ( \widetilde{g}^i_{\ j})=\delta^i_j, \qquad \cs(g)=g^{-1},\qquad \cs(\widetilde{g})=\widetilde{g}^{-1}.
\ee
where $g^{-1},\widetilde{g}^{-1}$ are the inverses respect to the above non-commutative multiplication rules.

The $\star$-structure of ${SL}_{\bfq}(2)\otimes {SL}_{\widetilde \bfq}(2)$ is given by
\be
(g^I)^i_{\ j}{}^\star=\cs \lt((\widetilde{g}^I)^j_{\ i}\rt),\qquad (\widetilde{g}^J)^i_{\ j}{}^\star=\cs \lt(({g}^J)^j_{\ i}\rt).
\ee
The $\star$-Hopf algebra structure of ${SL}_{\bfq}(2)\otimes {SL}_{\widetilde \bfq}(2)$ can be derived from the duality \eqref{quantumduality}. The detailed discussion is given in Appendix \ref{duality and star-Hopf algebra strucutre}.

\subsection{Quasi-triangular Hopf $*$-algebra}

The combinatorial quantization of Chern-Simons theory uses $R$-matrix in finite-dimensional representations in the formulation of operator algebra, so we need to introduce a quasi-triangular Hopf $*$-algebra $\uquqt$, which may be seen as an extension from $\suquqt$.

The quantum group $\uquqt$ are the algebra of formal power series in $h,\wt{h}$ generated by ${1}$, ${\cx},{\cy},H$ and $\widetilde{ \cx},\widetilde{\cy},\widetilde{H}$, subject to the following commutation relations
\be 
&&{ \ck}{ \ck}^{-1}={ \ck}^{-1}{ \ck}=1,\qquad { \ck}{\cx}=\bfq {\cx}{\ck} ,\qquad { \ck}{ \cy}=\bfq^{-1} { \cy} { \ck},\qquad\left[{ \cx},{\cy}\right]=\frac{{ \ck}^{2}-{ \ck}^{-2}}{\bfq-\bm{q}^{-1}},\label{commutation00}\\
&&\widetilde{ \ck}\widetilde{ \ck}^{-1}=\widetilde{ \ck}^{-1}\widetilde{ \ck}=1,\qquad \wt{ \ck}\widetilde{{\cx}}=\widetilde{\bfq} \widetilde{{\cx}}\widetilde{\ck} ,\qquad \widetilde{\ck}\widetilde{\cy}=\widetilde{\bfq}^{-1} \widetilde{\cy} \widetilde{ \ck},\qquad\left[\widetilde{ \cx},\widetilde{ \cy}\right]=\frac{\widetilde{ \ck}^{2}-\widetilde{ \ck}^{-2}}{\widetilde{\bfq}-\widetilde{\bm{q}}^{-1}}.
\ee
where 
\be
{\ck}=\bm{q}^{\frac{H}{2}}=e^{\frac{h}{2}H},\qquad {\ck}^{-1}=\bm{q}^{-\frac{H}{2}}=e^{-\frac{h}{2}H},\qquad 
\widetilde{\ck}=\widetilde{\bm{q}}^{\frac{\widetilde H}{2}}=e^{\frac{\widetilde h}{2} \widetilde{H}},\qquad \widetilde{\ck}^{-1}=\widetilde{\bm{q}}^{-\frac{\widetilde H}{2}}=e^{-\frac{\widetilde h}{2}\widetilde{H}}.
\ee

The comultiplication $\Delta$, counit $\eps$, and antipode $S$ are defined by
\be
&&\Delta \ck^{\pm1}=\ck^{\pm1}\otimes \ck^{\pm1},\qquad \Delta \cx=\cx\otimes\ck+\ck^{-1}\otimes \cx,\quad \Delta \cy=\cy\otimes\ck+\ck^{-1}\otimes \cy,\\
&&\Delta\widetilde{\ck}^{\pm1}=\widetilde{\ck}^{\pm1}\otimes\widetilde{\ck}^{\pm1},\qquad \Delta \widetilde{\cx}=\widetilde{\cx}\otimes\widetilde{\ck}+\widetilde{\ck}^{-1}\otimes \widetilde{\cx},\quad \Delta \widetilde{\cy}=\widetilde{\cy}\otimes\widetilde{\ck}+\widetilde{\ck}^{-1}\otimes \widetilde{\cy},\\
&&\eps(\ck^{\pm1})=\eps(\widetilde{\ck}^{\pm1})=1,\qquad \eps(\cx)=\eps(\cy)=\eps(\widetilde{\cx})=\eps(\widetilde{\cy})=0,\\
&&S(\ck^{\pm1})=\ck^{\mp1},\qquad S(\cx)=-\bfq \cx,\qquad S(\cy)=-\bfq^{-1} \cy,\\
&&S(\widetilde{\ck}^{\pm1})=\widetilde{\ck}^{\mp1},\qquad S(\widetilde{\cx})=-\widetilde{\bfq }\widetilde{\cx},\qquad S(\widetilde{\cy})=-\widetilde{\bfq}^{-1} \widetilde{\cy}.
\ee
The $*$-structure is given by
\be
\cx^*=\widetilde{\cx},\qquad \cy^*=\widetilde{\cy},\qquad H^*=-\widetilde{H}.\label{starstr1}
\ee

The quantum Lorentz group $\suquqt$ is embedded into $\uquqt$ by the following relations
\be 
E&=&\mathcal{X}{\cal K},\qquad F={\cal K}^{-1}{\cal Y},\qquad K={\cal K}^{2},\qquad K^{-1}={\cal K}^{-2},\label{embeduq1}\\
\widetilde{E}&=&\widetilde{{\cal K}}\widetilde{\mathcal{X}},\qquad\widetilde{F}=\widetilde{{\cal Y}}\widetilde{{\cal K}}^{-1},\qquad\widetilde{K}=\widetilde{{\cal K}}^{2},\qquad\widetilde{K}^{-1}=\widetilde{{\cal K}}^{-2}.\label{embeduq2}
\ee

We have chosen the deformation parameters to be $(h,\widetilde{h})$. An equivalent choice is $(h,\overline{h})$, and the corresponding quantum group is $U_{h}(sl_2)\otimes U_{\overline{h}}(sl_2)$. As Hopf $*$-algebras, $\uquqt$ and $\uq\otimes U_{\overline{h}}(sl_2)$ are isomorphic due to the mapping 
\be 
\widetilde{\cx}\mapsto \cy',\qquad \widetilde{\cy}\mapsto \cx',\qquad \widetilde{H}\mapsto H',\label{ismorphictildeprime}
\ee
where $\cx',\cy',H'$ denote generators of $U_{\overline{\bm q}}(sl_2)$.

${h}$ becomes real when $k\to 0$ and $b^2\to -1$ such that $s$ is finite. We compare $U_{{h}}(sl_{2})\otimes U_{{\overline h}}(sl_{2})$ to the quantum Lorentz group $U_q(sl(2,\C)_{\R})$ in \cite{MR1059324,BR} in this limit: The algebra structure of $U_{{h}}(sl_{2})\otimes U_{{\overline h}}(sl_{2})$ reduces to the quantum Lorentz group studied earlier in \cite{Podles1990,BR}. But the co-algebraic structure is different from there by a twist. The coproduct of $U_q(sl(2,\C)_{\R})$ are twisted so that it preserves the features from the Iwasawa decomposition of $\Slc$. The quantum group $\uquqt$ preserves the features from the self-dual and anti-self-dual decompostion of the Lorentz algebra. In this sense, it is similar to the quantum Lorentz group in \cite{Gaiotto:2024osr}.

Both the Hopf algebras $\suquqt$ and $\uquqt$ are the quantum deformation of the enveloping algebra of the Lorentz Lie algebra. Their $*$-structures are consistent with the $*$-structure of the Lorentz Lie algebra \cite{Han:2024nkf}. In contrast to $\uquqt$, $\suquqt$ does not admit any $R$-element, and the $R$-element can only be defined in $\uquqt$, since it is a power series of $h,\wt{h}$. $\uquqt$ has the following $R$-element
\be
\calr&=&R\widetilde{R},\\
R&=&\bm{q}^{\frac{1}{2}\left(H\otimes H\right)}\, \ce_{\bm{q}}^{\left(1-\bfq^{-2}\right)\left(\bfq^{\frac{H}{2}}{\cal X}\right)\otimes\left(\bfq^{-\frac{H}{2}}{\cal Y}\right)},\qquad
\widetilde R=\widetilde{\bm{q}}^{\frac{1}{2}\left(\widetilde H\otimes\widetilde H\right)}\, \ce_{\widetilde{\bm{q}}}^{\left(1-\widetilde{\bfq}^{-2}\right)\left(\widetilde{\bfq}^{\frac{\widetilde H}{2}}\widetilde{\cal X}\right)\otimes\left(\widetilde{\bfq}^{-\frac{\widetilde H}{2}}\widetilde{\cal Y}\right)},\label{RandtildeR}
\ee
where $R$ and $\widetilde{R}$ are respectively the R-element of $\uq $ and $ \uqt$, and 
\be
\ce_{\bfq}^{x}&=&\sum_{n=0}^{\infty}\frac{x^{n}}{\left[n\right]_{\bfq}!}\bfq^{n(n-1)/2}\\
\left[n\right]_{\bfq}&=&\frac{\bfq^n-\bfq^{-n}}{\bfq-\bfq^{-1}},\qquad \left[n\right]_{\bfq}!=\left[n\right]_{\bfq}\left[n-1\right]_{\bfq}\cdots\left[1\right]_{\bfq}.
\ee
The properties of $R$ and $\widetilde{R}$ and the commutativity between $\{\cx,\cy,H\}$ and $\{\widetilde{\cx},\widetilde{\cy},\widetilde{H}\}$ imply that 
\be
&&{\cal R}\Delta\left(\xi\right){\cal R}^{-1}=\Delta^{\prime}\left(\xi\right),\qquad \forall\ \xi\in \uq \otimes  \uqt,\\
&&\left(\Delta\otimes1\right){\cal R}={\cal R}_{13}{\cal R}_{23},\qquad \left(1\otimes\Delta\right){\cal R}={\cal R}_{13}{\cal R}_{12}.
\ee
In other words, the quasi-triangularity of $\uq $ and $ \uqt$ implies the quasi-triangularity of $\uq \otimes  \uqt$. 

\begin{lemma}\label{RstarR}
$R^{*\otimes *}=\widetilde{R}^{-1}$ and $\calr^{*\otimes *}=\calr^{-1}$.
\end{lemma}

\begin{proof} We apply $*$ to $R$ in \eqref{RandtildeR} and obtain
\be 
R^{*\otimes*}&=&\ce_{\widetilde{\bfq}^{-1}}^{\left(1-\widetilde{\bm{q}}^{2}\right)\left(\widetilde{{\cal X}}\widetilde{\bm{q}}^{\frac{\widetilde{H}}{2}}\right)\otimes\left(\widetilde{{\cal Y}}\widetilde{\bm{q}}^{-\frac{\widetilde{H}}{2}}\right)}\widetilde{\bm{q}}^{-\frac{1}{2}\left(\widetilde{H}\otimes\widetilde{H}\right)}
=\ce_{\widetilde{\bfq}^{-1}}^{-(1-\widetilde{\bm{q}}^{-2})\left(\widetilde{\bm{q}}^{\frac{\widetilde{H}}{2}}\widetilde{{\cal X}}\right)\otimes\left(\widetilde{\bm{q}}^{-\frac{\widetilde{H}}{2}}\widetilde{{\cal Y}}\right)}\widetilde{\bm{q}}^{-\frac{1}{2}\left(\widetilde{H}\otimes\widetilde{H}\right)}\nonumber\\
&=&\left[\ce_{\widetilde{\bm{q}}}^{(1-\widetilde{\bm{q}}^{-2})\left(\widetilde{\bm{q}}^{\frac{\widetilde{H}}{2}}\widetilde{{\cal X}}\right)\otimes\left(\widetilde{\bm{q}}^{-\frac{\widetilde{H}}{2}}\widetilde{{\cal Y}}\right)}\right]^{-1}\widetilde{\bm{q}}^{-\frac{1}{2}\left(\widetilde{H}\otimes\widetilde{H}\right)}
=\widetilde{R}^{-1}
\ee
We use the relation $\ce_{\bfq}^{x}\ce_{\bfq^{-1}}^{-x}=1$ \cite{majid2000foundations} in the third step.


\end{proof}

\subsection{Finite-dimensional irreducible representations}\label{Finite-dimensional irreducible representations}

The finite dimensional irreducible representation of $\uq$ is completely classified by a couple $(\o,J)\in \{1,-1,i,-i\}\times \mathbb{N}_0/2$. The appearance of $\o$ is due to the automorphism $\t_\o$ defined by $\t_\o(\bfq^{H/2})=\o \bfq^{H/2},\  \t_\o(\cx)=\o^2 \cx, \t_\o(\cy)=\cy$. In the following, we only consider the representations with $\o=1$, while all other representations can be induced by the automorphism $\t_\o$.

The quantum group $\uquqt$ has infinitely many finite-dimensional irreducible representations, each of which is labelled by a pair of spins $(I,J)$, $I,J\in \mathbb{N}_0/2$ and carried by the tensor product of vector spaces $V^I\otimes \widetilde{V}^J$. The representation is the tensor product of a pair of irreducible representations $\rho^I$ and $\widetilde{\rho}^J$ of $\uq$ and $ \uqt$ respectively. We denote by $e^I_m\otimes \widetilde{e}^J_n$, $m=I,I-1,\cdots,-I$, $n=J,J-1,\cdots,-J$, a basis of $V^I\otimes\widetilde{V}^J$, such that $\rho^I$, $\widetilde{\rho}^J$ are given by 
\be 
&\rho^{I}\left(\bfq^{\frac{H}{2}}\right)e_{m}^{I} =\bfq^{m}\, e_{m}^{I},
&\quad\qquad\qquad\widetilde{\rho}^{J}\left(\widetilde{\bm{q}}^{\frac{\widetilde{H}}{2}}\right)\widetilde{e}_{m}^{J}=\widetilde{\bm{q}}^{-m}\widetilde{e}_{m}^{J}, \label{rhorep1}\\
&\rho^{I}\lt(\cx\rt)e_{m}^{I} =\sqrt{[I-m]_{\bfq}[I+m+1]_{\bfq}}\, e_{m+1}^{I},&\quad \widetilde{\rho}^{J}\lt(\widetilde{{\cal Y}}\rt)\widetilde{e}_{m}^{J}=\sqrt{[J-m]_{\widetilde{\bm{q}}}[J+m+1]_{\widetilde{\bm{q}}}}\, \widetilde{e}_{m+1}^{J}, \label{rhorep2}\\
&\rho^{I}\lt(\cy\rt)e_{m}^{I} =\sqrt{[I+m]_{\bfq}[I-m+1]_{\bfq}}\, e_{m-1}^{I}, &\quad \widetilde{\rho}^{J}\lt(\widetilde{{\cal X}}\rt)\widetilde{e}_{m}^{J}=\sqrt{[J+m]_{\widetilde{\bm{q}}}[J-m+1]_{\widetilde{\bm{q}}}}\, \widetilde{e}_{m-1}^{J}. \label{rhorep3}
\ee 
Given a generic element $\xi\widetilde{\zeta}\in \uquqt$, the irreducible representation $(\varrho^{IJ},V^I\otimes\widetilde{V}^J)$ is given by 
\be 
\varrho^{IJ}(\xi\widetilde{\zeta})=\rho^I(\xi)\otimes \widetilde{\rho}^J(\widetilde{\zeta}).
\ee 

\begin{lemma}\label{rhodagger0}

Given $\xi\in  \uq$ and $\xi^*\in \uqt$, their representation matrices w.r.t the above basis are related by 
\be
\rho^J(\xi)^\dagger=\widetilde{\rho}^J(\xi^*).\label{rhodagger}
\ee
Similar relation holds for $\widetilde{\xi}\in \uqt$ and $\widetilde{\xi}^*\in  \uq$:
\be
\widetilde{\rho}^J(\widetilde{\xi})^\dagger={\rho}^J(\widetilde{\xi}^*).\label{rhodagger1}
\ee

\end{lemma}

\begin{proof}

For the generators, the representation matrices of $\rho^J$ w.r.t the basis $e^J_m$ are given by
\be 
&\rho^{J}(\bfq^{\frac{H}{2}})^n_{\ m}=\bfq^{m}\, \delta_{nm},&\quad \qquad \qquad
\widetilde{\rho}^{J}(\widetilde{\bm{q}}^{\frac{\widetilde{H}}{2}})^n_{\ m}=\widetilde{\bm{q}}^{-m}\delta_{nm}, \nonumber\\
&\rho^{J}(\cx)^n_{\ m}=\sqrt{[J-m]_{\bfq}[J+m+1]_{\bfq}}\, \delta_{n,m+1},&\quad \widetilde{\rho}^{J}(\widetilde{{\cal Y}})^n_{\ m}=\sqrt{[J-m]_{\widetilde{\bm{q}}}[J+m+1]_{\widetilde{\bm{q}}}}\, \delta_{n,m+1},\nonumber\\
&\rho^{J}(\cy)^n_{\ m}=\sqrt{[J+m]_{\bfq}[J-m+1]_{\bfq}}\, \delta_{n,m-1},&\quad \widetilde{\rho}^{J}(\widetilde{{\cal X}})^n_{\ m}=\sqrt{[J+m]_{\widetilde{\bm{q}}}[J-m+1]_{\widetilde{\bm{q}}}}\, \delta_{n,m-1}^{J}. \nonumber
\ee
All matrix elements are understood as a power series of $h$, so the complex conjugate commutes with the square-root, and $\overline{[n]_{\bfq}}=[n]_{\overline{\bfq}}=[n]_{\wt{\bfq}}$:
\be 
\overline{\rho^{J}(\bfq^{\frac{H}{2}})^n_{\ m}}=\widetilde{\rho}^{J}(\widetilde{\bfq}^{\frac{\widetilde H}{2}})^m_{\ n},\qquad \overline{\rho^{J}(\cx)^n_{\ m}}=\widetilde{\rho}^{J}(\widetilde{\cx})^m_{\ n}\qquad \overline{\rho^{J}(\cy)^n_{\ m}}=\widetilde{\rho}^{J}(\widetilde{\cy})^m_{\ n}
\ee 
where $\widetilde{\cy}=\cy^*$ and $\widetilde{\cx}=\cx^*$. For monomials of generators, given $\xi,\zeta$ satisfying $\rho^J(\xi)^\dagger=\widetilde{\rho}^J(\xi^*)$ and $\rho^J(\zeta)^\dagger=\widetilde{\rho}^J(\zeta^*)$, we have $\rho^J(\xi\zeta)^\dagger=\rho^J(\zeta)^\dagger\rho^J(\xi)^\dagger=\widetilde{\rho}^J(\zeta^*)\widetilde{\rho}^J(\xi^*)=\widetilde{\rho}^J((\xi\zeta)^*)$.

Eq.\eqref{rhodagger1} can be shown by a similar derivation.

\end{proof}

If we introduce a Hermitian inner product $\lag\ ,\ \rag$ on $V^J$ such that $\langle e^J_m,e^J_n\rangle =\delta_{mn}$, Eq.\eqref{rhodagger} can be understood as an equation of operators on $V^J$, where $\widetilde{\rho}^J$ is given by \eqref{rhorep1} - \eqref{rhorep3} with $\widetilde{e}^J_n$ replaced by ${e}^J_n$. Similar for \eqref{rhodagger1} if we introduce a Hermitian inner product $\lag\ ,\ \rag$ on $\widetilde{V}^J$ such that $\langle \widetilde{e}^J_m,\widetilde{e}^J_n\rangle =\delta_{mn}$. We emphasize that \eqref{rhodagger} does not imply $(\varrho^{JJ},V^J\otimes \widetilde{V}^J)$ to be a $*$-representation, because $\rho^J(\xi)^\dagger$ still acts on $V^J$ rather than $\widetilde{V}^J$. Indeed, $\varrho^{IJ}(\xi\widetilde{\zeta})^\dagger =\rho^{I}(\xi)^\dagger\otimes \widetilde{\rho}^{J}(\widetilde{\zeta})^\dagger=\widetilde{\rho}^{I}(\xi^*)\otimes {\rho}^{J}(\widetilde{\zeta}^*)$ does not equal to $\varrho^{IJ}(\widetilde{\zeta}^*\xi^*)$ even for $I=J$ but becomes a representation of a different quantum group $\uqt\otimes \uq$.

We apply Lemma \ref{rhodagger0} to $R^{*\otimes *}=\widetilde{R}^{-1}$ in Lemma \ref{RstarR}, we denote by $R^{IJ}=(\rho^I\otimes\rho^J)(R)$ and $\widetilde{R}^{IJ}=(\widetilde{\rho}^I\otimes\widetilde{\rho}^J)(\widetilde{R})$, then we obtain 
\be 
(R^{IJ})^{\dagger\otimes \dagger}=(\widetilde{R}^{-1})^{IJ},\qquad (\widetilde{R}^{IJ})^{\dagger\otimes \dagger}=({R}^{-1})^{IJ}.\label{Rdagger}
\ee 
Again, these are relations of representation matrices. 

We define the quantum traces $\tr_{\bfq}(X^I)$ and $\tr_{\widetilde \bfq}(\widetilde X^I)$ for $X^I\in \mathrm{End}(V^I)$ and $\widetilde X^I\in \mathrm{End}(\widetilde V^I)$:
\be    
\tr_{\bfq}(X^I)=\tr\lt(\rho^I(\theta)X^I\rt),\qquad \tr_{\wt \bfq}(\wt X^I)=\tr\lt(\wt\rho^I(\wt\theta)\wt X^I\rt),\qquad \theta=\bfq^{-H},\qquad \wt\theta=\wt{\bfq}^{-\wt H}.
\ee
$\theta$ is a group-like element satisfying
\begin{eqnarray}
\Delta \theta =\theta\otimes\theta,\qquad S(\theta)=\theta^{-1},\qquad \theta S(\xi)=S^{-1}(\xi)\theta,\quad \xi\in\uq,
\end{eqnarray}
and similar for $\wt\theta$. The following relation will be useful in Section \ref{Wilson loop operators}.

\begin{lemma}\label{quantumtracelemma}
For $R=\sum_{\alpha}R_{\alpha}^{(1)}\otimes R_{\alpha}^{(2)}$ and $R'=\sum_{\alpha}R_{\alpha}^{(2)}\otimes R_{\alpha}^{(1)}$, we denote by $\left(R^{\prime}\right)^{I}=[(\rho^I\otimes\mathrm{id})R']\otimes 1$. Given $Y^I=\sum_\alpha  y^{I}_\alpha\otimes1\otimes\xi_\a$ where $y^I_\a\in \mathrm{End}(V^I)$ and $\xi_\a\in\uq$, we have
\begin{eqnarray}
	\mathrm{Tr}_{\bm{q}}\left[\left(R^{\prime}\right)^{I}Y^{I}\left(R^{\prime-1}\right)^{I}\right]=\mathrm{Tr}_{\bm{q}}\left[Y^{I}\right].
\end{eqnarray}
The similar relation holds for the tilded sector.
\end{lemma}

\begin{proof}
Applying the relation $R'{}^{-1}=(S^{-1}\otimes \mathrm{id})R'$,
\begin{eqnarray}
&&\mathrm{Tr}_{\bm{q}}\left[\left(R^{\prime}\right)^{I}Y^{I}\left(R^{\prime-1}\right)^{I}\right]	=\mathrm{Tr}\left[R_{\alpha}^{(2)}y_\mu S^{-1}\left(R_{\beta}^{(2)}\right)\theta\right]\otimes R_{\alpha}^{(1)}R_{\beta}^{(1)}\otimes\xi_\mu\nonumber\\
&=&\mathrm{Tr}\left[y_\mu\theta S\left(R_{\beta}^{(2)}\right)R_{\alpha}^{(2)}\right]\otimes R_{\alpha}^{(1)}R_{\beta}^{(1)}\otimes\xi_\mu
=\mathrm{Tr}\left[y_\mu\theta S\left(S^{-1}\left(R_{\alpha}^{(2)}\right)R_{\beta}^{(2)}\right)\right]\otimes R_{\alpha}^{(1)}R_{\beta}^{(1)}\otimes\xi_\mu\nonumber\\
&=&\mathrm{Tr}_{\bm{q}}\left[Y^{I}\right],
\end{eqnarray}
where summation is implied over the repeated index.
\end{proof}

\subsection{Clebsch-Gordan maps}

The tensor product representations $(\rho^I\otimes\rho^J,V^I\otimes V^J)$ and $(\widetilde\rho^I\otimes\widetilde\rho^J,\widetilde V^I\otimes\widetilde V^J)$ have the Clebsch-Gordan (CG) decompositions:
\be 
V^I\otimes V^J\simeq \bigoplus_{K=|I-J|}^{I+J}V^K,\qquad \
\widetilde V^I\otimes\widetilde V^J\simeq \bigoplus_{K=|I-J|}^{I+J}\widetilde V^K
\ee
We define the quantum CG map to be the linear transformation relating the basis $\{e^K_l\}_{K=|I-J|, |l|\leq K}^{I+J}$ and $\{e^I_m\otimes e^J_n\}_{|m|\leq I,|n|\leq J}$ (or relating $\{\widetilde e^K_l\}_{K=|I-J|, |l|\leq K}^{I+J}$ and $\{\widetilde e^I_m\otimes \widetilde e^J_n\}_{|m|\leq I,|n|\leq J}$)
\be 
e^K_l=e^I_m\otimes e^J_n(C_1^{IJ})^{mn}_{\ \ Kl},\qquad \widetilde e^K_l=\widetilde e^I_m\otimes \widetilde e^J_n(\widetilde C_1^{IJ})^{mn}_{\ \ Kl}
\ee
We denote the inverse transformation by $(C_2^{IJ})^{Kl}_{\ \ mn}$ and $(\widetilde C_2^{IJ})^{Kl}_{\ \ mn}$
\be
(C_2^{IJ})^{Kl}_{\ \ mn}( C_1^{IJ})^{mn}_{\ \ K'l'}=\delta^K_{K'}\delta^l_{l'},\qquad (\widetilde C_2^{IJ})^{Kl}_{\ \ mn}(\widetilde C_1^{IJ})^{mn}_{\ \ K'l'}=\delta^K_{K'}\delta^l_{l'},\label{CCKK}\\
( C_1^{IJ})^{mn}_{\ \ Kl}(C_2^{IJ})^{Kl}_{\ \ m'n'}=\delta^m_{m'}\delta^n_{n'},\qquad (\widetilde C_1^{IJ})^{mn}_{\ \ Kl}(\widetilde C_2^{IJ})^{Kl}_{\ \ m'n'}=\delta^m_{m'}\delta^n_{n'}.\label{CCmn}
\ee
The CG maps intertwine the representations
\be 
(C_1^{IJ})^{mn}_{\ \ Kl'}\rho^K(\xi)^{l'}_{\ l}&=&[\lt(\rho^I\otimes\rho^J\rt)(\Delta\xi)]^{mn}_{\ m'n'}(C_1^{IJ})^{m'n'}_{\ \ Kl},\label{Crho1}\\ 
(\widetilde C_1^{IJ})^{mn}_{\ \ Kl'}\widetilde\rho^K(\widetilde\xi)^{l'}_{\ l}&=&[\lt(\widetilde\rho^I\otimes\widetilde\rho^J\rt)(\Delta\widetilde\xi)]^{mn}_{\ m'n'}(\widetilde C_1^{IJ})^{m'n'}_{\ \ Kl}\\
\rho^K(\xi)_{\ l'}^{l}(C_2^{IJ})_{\ \ mn}^{ Kl'}&=&(C_2^{IJ})_{\ \ m'n'}^{Kl}[\lt(\rho^I\otimes\rho^J\rt)(\Delta\xi)]_{\ \ mn}^{m'n'},\label{Crho3}\\ 
\widetilde\rho^K(\widetilde\xi)_{\ l'}^{l}(\widetilde C_2^{IJ})_{\ \ mn}^{ Kl'}&=&(\widetilde C_2^{IJ})_{\ \ m'n'}^{Kl}[\lt(\widetilde\rho^I\otimes\widetilde\rho^J\rt)(\Delta\widetilde\xi)]_{\ \ mn}^{m'n'}.\label{Crho4}
\ee
These relations may also be rewritten as 
\be
\sum_{K}(C_1^{IJ})_{K}\rho^K(\xi)(C_2^{IJ})^{ K}=(\rho^I\otimes\rho^J)(\Delta\xi),\qquad  \sum_{K}(\widetilde C_1^{IJ})_{K}\widetilde\rho^K(\widetilde\xi)(\widetilde C_2^{IJ})^{ K}=(\widetilde\rho^I\otimes\widetilde\rho^J)(\Delta\widetilde\xi).
\ee


By the complex conjugate of \eqref{Crho1} and Lemma \ref{rhodagger0}, we can choose
\be 
(\widetilde C_2^{IJ})^{Kl}_{\ \ mn}
=\overline{(C_1^{IJ})^{mn}_{\ \ Kl}}\ .\label{C1conjugate}
\ee

\section{Infinite-dimensional $*$-representations}\label{Infinite-dimensional representations}

The purpose of this section is to construct a family of Infinite-dimensional irreducible $*$-representations of $\suquqt$. The construction of these representations are based on the quantum torus algebra at level-$k$ and its $*$-representation \cite{Han:2024nkf}. In the following, we assume $k=2N$ where $N$ is a positive integer.

\subsection{Quantum torus algebra at level-$k$ and $*$-representation}

The holomorphic quantum torus algebra $\mathcal{O}_{q} $ is spanned
by Laurent polynomials of the symbols $\bm{x}_{\alpha,\beta}$ with
$\alpha,\beta\in\mathbb{Z}$, satisfying the following relation
\begin{equation}
\bm{x}_{\alpha,\beta}\bm{x}_{\gamma,\delta}=q^{\alpha\delta-\beta\gamma}\bm{x}_{\alpha+\gamma,\beta+\delta},\qquad q=e^{h/2},\qquad h=\frac{2\pi i}{k}\left(1+b^{2}\right).\label{eq:quantumtorus1}
\end{equation}
We associated to $\mathcal{O}_{q} $ the anti-holomorphic
counterpart $\mathcal{O}_{\widetilde{q}} $ generated
by $\widetilde{\bm{x}}_{\alpha,\beta}$ with $\alpha,\beta\in\mathbb{Z}$,
satisfying
\begin{equation}
\widetilde{\bm{x}}_{\alpha,\beta}\widetilde{\bm{x}}_{\gamma,\delta}=\widetilde{q}^{\alpha\delta-\beta\gamma}\widetilde{\bm{x}}_{\alpha+\gamma,\beta+\delta},\qquad\widetilde{q}=e^{\widetilde{h}/2},\qquad\widetilde{h}=\frac{2\pi i}{k}\left(1+b^{-2}\right),\label{eq:quantumtorus2}
\end{equation}
and $\mathcal{O}_{\widetilde{q}} $ commutes with $\mathcal{O}_{q} $.
The entire quantum torus algebra is $\mathcal{A}_{h} =\mathcal{O}_{q} \otimes\mathcal{O}_{\widetilde{q}} $. The algebra is endowed with a $*$-structure which interchanges the holomorphic and antiholomorphic copies:
\be
*\left(\bm{x}_{\alpha,\beta}\right)=\widetilde{\bm{x}}_{\alpha,\beta},\qquad *\left(\widetilde{\bm{x}}_{\alpha,\beta}\right)=\bm{x}_{\alpha,\beta}.
\ee
The deformation parameter $q$ and $\widetilde{q}$ are square-roots of $\bfq $ and $\widetilde{\bfq}$ respectively.


An infinite-dimensional irreducible $*$-representation of $\mathcal{A}_{h} $ is proposed in the quantization of complex Chern-Simons theory \cite{levelk,andersen2016level}. The Hilbert space carrying the representation
is $\mathcal{H}_0\simeq L^{2}(\mathbb{R})\otimes\mathbb{C}^{k}$. A state in $\ch_0$ is a function $f(\mu,m),$ $\mu\in\mathbb{R}$, $m\in\mathbb{Z}/k\mathbb{Z}$. The inner product is given by
\be 
\langle f\mid f'\rangle_0=\sum_{m\in\Z/k\Z}\int\rmd\mu\, \overline{f(\mu,m)}f'(\mu,m).
\ee
The following elementary operators are densely defined on $\mathcal{H}$ 
\begin{align}
\bm{\mu}f(\mu,m) & =\mu f(\mu,m),\qquad\bm{\nu}f(\mu,m)=-\frac{k}{2\pi i}\frac{\partial}{\partial\mu}f(\mu,m)\\
e^{\frac{2\pi i}{k}\bm{m}}f(\mu,m) & =e^{\frac{2\pi i}{k}m}f(\mu,m),\qquad e^{\frac{2\pi i}{k}\bm{n}}f(\mu,m)=f(\mu,m+1).
\end{align}
They satisfy
\be 
[\bm{\mu},\bm{\nu}]=\frac{k}{2\pi i},\qquad e^{\frac{2\pi i}{k}\bm{n}}e^{\frac{2\pi i}{k}\bm{m}}=e^{\frac{2\pi i}{k}}e^{\frac{2\pi i}{k}\bm{m}}e^{\frac{2\pi i}{k}\bm{n}},\qquad [\bm{\nu},e^{\frac{2\pi i}{k}\bm{m}}]=[\bm{\mu},e^{\frac{2\pi i}{k}\bm{n}}]=0
\ee 
We define the exponential operators $\bm{x},\bm{y},\widetilde{\bm{x}},\widetilde{\bm{y}}$ by
\be
\bm{y} & =&\exp\left[\frac{2\pi i}{k}(-ib\boldsymbol{\mu}-\bm{m})\right],\qquad\widetilde{\bm{y}}=\exp\left[\frac{2\pi i}{k}\left(-ib^{-1}\boldsymbol{\mu}+\bm{m}\right)\right],\\
\bm{x} & =&\exp\left[\frac{2\pi i}{k}(-ib\boldsymbol{\nu}-\bm{n})\right],\qquad\widetilde{\bm{x}}=\exp\left[\frac{2\pi i}{k}\left(-ib^{-1}\boldsymbol{\nu}+\bm{n}\right)\right].
\ee
For any $f(\mu,m)$ analytic in $\mu$ in the strip $\im(\mu)\in[0,\re(b)]$ and decaying sufficiently fast at $\re(\mu)\to\infty$, the actions of the above operators are given by
\be
\bm{y}f(\mu,m)&=e^{\frac{2\pi i}{k}(-ib {\mu}-{m})}f(\mu,m),\qquad& \bm{x}f(\mu,m)=f(\mu+ib,m-1),\label{xyrepresentation1}\\
\widetilde{\bm{y}}f(\mu,m)&=e^{\frac{2\pi i}{k}(-ib^{-1} {\mu}+{m})}f(\mu,m),\qquad & \widetilde{\bm{x}}f(\mu,m)=f(\mu+ib^{-1},m+1).\label{xyrepresentation2}
\ee
The operators $\bm{x},\widetilde{\bm{x}},\bm{y},\widetilde{\bm{y}}$ are unbounded operators. We denote by $\overline{\mathfrak{D}}_0$ the maximal common domain of $\bm{x},\bm{y},\widetilde{\bm{x}},\widetilde{\bm{y}}$ and their Laurent polynomials. $\overline{\mathfrak{D}}_0$ contains $f(\mu,m)$ that can be analytic continued to entire functions in $\mu$ and satisfy
\be
e^{\alpha\frac{2\pi}{k}\re(b)\mu}f(\mu+i\b \re(b),m)\in L^{2}(\mathbb{R}),
\qquad\forall\,m\in\mathbb{Z}/k\mathbb{Z},\quad\alpha,\b\in\mathbb{Z},\quad \mu\in\R.\label{domainD}
\ee
The Hermite functions $e^{-\mu^{2}/2}H_{n}(\mu)$ , $n=1,\cdots,\infty$
satisfy all the requirements and span a dense domain in $L^{2}(\mathbb{R})$, so $\overline{\mathfrak{D}}_0$ is dense in $\mathcal{H}_0$. On $\overline{\Fd}_0$, the operators $\bm{x},\widetilde{\bm{x}},\bm{y},\widetilde{\bm{y}}$ form the $\bfq,\widetilde{\bfq}$-Weyl algebra with $\bfq=q^{2}=e^h$,
and $\widetilde{\bfq}=\widetilde{q}^{2}=e^{\widetilde{h}}$: 
\be
\bm{x}\bm{y}=\bfq\bm{y}\bm{x},\qquad\widetilde{\bm{x}}\widetilde{\bm{y}}=\widetilde{\bfq}\widetilde{\bm{y}}\widetilde{\bm{x}},\qquad \bm{x}\widetilde{\bm{y}}=\widetilde{\bm{y}}\bm{x},\qquad \widetilde{\bm{x}}{\bm{y}}={\bm{y}}\widetilde{\bm{x}}.
\ee
We define the Fourier transformation and the inverse by
\be 
\tilde{f}(\nu,n)&=&\frac{1}{k}\sum_{m\in\mathbb{Z}/k\mathbb{Z}}\int d\mu\,e^{-\frac{2\pi i}{k}\left(\mu\nu-mn\right)}f(\mu,m),\\
{f}(\mu,m)&=&\frac{1}{k}\sum_{n\in\mathbb{Z}/k\mathbb{Z}}\int d\nu\,e^{\frac{2\pi i}{k}\left(\mu\nu-mn\right)}\tilde{f}(\nu,n).\label{invFourier}
\ee
The actions of $\bmx,\bmy,\widetilde{\bmx},\widetilde{\bmy}$ on $\tilde{f}(\nu,n)$ are given by
\be
\bm{y}\tilde f(\nu,n)=&\tilde f(\nu+ib,n-1),\qquad& \bm{x}\tilde f(\nu,n)=e^{-\frac{2\pi i}{k}(-ib {\nu}-{n})}\tilde f(\nu,n),\\
\widetilde{\bm{y}}\tilde f(\nu,n)=&\tilde f(\nu+ib^{-1},n+1),\qquad & \widetilde{\bm{x}}\tilde f(\nu,n)=e^{-\frac{2\pi i}{k}(-ib^{-1} {\nu}+{n})}\tilde f(\nu,n).
\ee

The tilded and untilded operators are related by the Hermitian conjugate
\be
\bm{x}^{\dagger}=\widetilde{\bm{x}},\qquad\bm{y}^{\dagger}=\widetilde{\bm{y}}.
\ee
$\bm{x},\bm{y},\widetilde{\bm{x}},\widetilde{\bm{y}}$ are all normal operators.

We denote by $\cl(\overline{\mathfrak{D}}_0)$ the space of linear operators leaving $\overline{\mathfrak{D}}_0$ invariant. The representation $\rho$: $\mathcal{A}_{h} \to\mathcal{L}(\overline{\mathfrak{D}}_0)$
is given by
\be
\rho:\  \bm{x}_{\alpha,\beta}\mapsto 
q^{-\alpha\beta}\bm{x}^{\alpha}\bm{y}^{\beta},\qquad 
\widetilde{\bm{x}}_{\alpha,\beta}\mapsto 
 \widetilde{q}^{-\alpha\beta}\widetilde{\bm{x}}^{\alpha}\widetilde{\bm{y}}^{\beta}.\label{eq:reptor}
\ee
The relations (\ref{eq:quantumtorus1}) and (\ref{eq:quantumtorus2})
are obtained by applying the $(\bfq,\widetilde{\bfq})$-Weyl
algebra. In the following, we often denote $\rho(\bm{x}_{\alpha,\beta}$)
by $\bm{x}_{\alpha,\beta}$ for simplifying notations. The $*$-stucture is represented by the Hermitian conjugate
\be
\bm{x}_{\alpha,\beta}^{\dagger}=\widetilde{q}^{\alpha\beta}\widetilde{\bm{y}}^{\beta}\widetilde{\bm{x}}^{\alpha}=\widetilde{q}^{-\alpha\beta}\widetilde{\bm{x}}^{\alpha}\widetilde{\bm{y}}^{\beta}=\widetilde{\bm{x}}_{\alpha,\beta},
\ee
which holds on $\overline{\mathfrak{D}}_0$. All $\bm{x}_{\alpha,\beta}$ and $\widetilde{\bm{x}}_{\alpha,\beta}$ are closable operators since their adjoints are densely defined. Their closure are still denoted by $\bm{x}_{\alpha,\beta}$ and $\widetilde{\bm{x}}_{\alpha,\beta}$. For any $\a,\b\in\Z$, $\bm{x}_{\alpha,\beta}$ and $\widetilde{\bm{x}}_{\alpha,\beta}$ leaves $\overline{\mathfrak{D}}_0$ invariant.

We define countably many semi-norms $\left\Vert \cdot \right\Vert _{\alpha,\beta}$ on $\overline{\Fd}_0$ by 
\begin{eqnarray}
	\left\Vert f\right\Vert _{\alpha,\beta}=\left\Vert \bm{x}_{\alpha,\beta}f\right\Vert_0,\qquad \a,\b\in\Z,\quad f\in\overline{\Fd}_0,
\end{eqnarray}
where $\Vert\cdot\Vert_0$ is the Hilbert-space norm of $\ch_0$. The natural topology on $\overline{\mathfrak{D}}_0$ endowed by the semi-norms is the weakest topology in which all semi-norms are continuous and in which the operation of addition is continuous. 

\begin{lemma}\label{Frechet}
$\overline{\mathfrak{D}}_0$ is a Fr\'echet space.
\end{lemma}

\begin{proof} To prove $\overline{\mathfrak{D}}_0$ is Fr\'echet, we need to show that every sequence $f_n\in \overline{\Fd}_0$ $n=1,2,\cdots$ that is Cauchy with respect to every semi-norm converges to $f\in \overline{\mathfrak{D}}_0$. Suppose a sequence $\{f_n\}$ is Cauchy with respect to each semi-norms $\Vert\cdot\Vert _{\alpha,\beta}$. Then $\bm{x}_{\alpha,\beta}f_n$ converges to $g_{\a,\b}\in\ch_0$ for all $\a,\b\in\Z$ with respect to the Hilbert space norm. Each $\bm{x}_{\alpha,\beta}$ is a closed operator, so $f_n\to g_{0,0}\equiv f$ and $\bm{x}_{\alpha,\beta}f_n\to g_{\a,\b}$ imply $f$ belongs to the domain of $\bm{x}_{\alpha,\beta}$ and $g_{\a,\b}=\bm{x}_{\alpha,\beta}f$. Since it applies to all $\a,\b\in\Z$, $f$ belongs to the common domain $\overline{\mathfrak{D}}_0$.
\end{proof}

The topological dual $\overline{\Fd}_0'$ is the space of continuous linear functionals of $\overline{\Fd}$. $T\in\overline{\Fd}'_0$ if and only if there are $C>0$ and a finite set of semi-norms $\Vert\cdot\Vert_{\a,\b}$, $(\a,\b)\in \ci$ where $\ci$ is a finite set, such that
\be 
\lt\vert T[f]\rt\vert\leq C\sum_{(\a,\b)\in \ci} \Vert f \Vert_{\a,\b},\qquad \forall f\in\Fd_0.
\ee
Given any $T\in\overline{\Fd}'_0$ and any $\bm{a}\in \ca_h$, the linear functional $T_{\bm{a}}[f]\equiv T[\rho(\bm{a})f]$, $\forall f\in\overline{\Fd}_0$,  belongs to $\overline{\Fd}'_0$.


We denote by $\Fd_0$ the space of functions $f(\mu,m)$ satisfying additionally that for any $\eta\in\R$, 
\be
|f(\mu+i\eta,m)|\leq C_\eta e^{-a|\mu|} ,\qquad \text{for any $a>0$}.
\ee
We denote by $\Fw\subset L^2(\R)$ the space of functions
\be
e^{-\alpha \mu^2+\b\mu}\, \mathrm{Pol}(\mu),\quad\text{where},\quad \re(\alpha)>0,\ \b\in\C,\ \text{and $\mathrm{Pol}(\mu)$ is a polynomial in $\mu$}.
\ee
we define $\sw_0\simeq \Fw\otimes\C^k\subset \Fd_0$. Both $\Fd_0$ and $\sw_0$ are dense in $\ch_0$ because all Hermite functions are inside them, and both of them are invariant by the action of all Laurent polynomials of $\bm{x},\bmy$ and their tilded partners. Moreover, $\sw_0$ and $\Fd_0$ are dense in $\overline{\Fd}_0$ by the Fr\'echet topology \cite{2008InMat.175..223F}.




\begin{lemma}\label{distributiondelta}

(1) Assume $F$ to be a linear functional on $\Fd_0$ and satisfy 
\begin{eqnarray}
F\lt[\left(\bm{y}-1\right)f\rt]=0,\qquad \text{and}\qquad F\lt[\left(\widetilde{\bm{y}}-1\right)f\rt]=0,\label{FyFy}
\end{eqnarray}
for all $f\in \Fd_0$. Then $F$ is given by \footnote
{
$F$ in \eqref{deltadistri} belongs to $\overline{\Fd}'$: $F[\psi]=C\sum_{n\in\Z/k\Z}\int d\nu\, \tilde{\psi}(\nu,n)=C\sum_{n\in\Z/k\Z}\int d\nu\,\overline{g(\nu,n)} \frac{\tilde{\psi}(\nu,n)}{\overline{g(\nu,n)}}$ where $\tilde{\psi}$ is the Fourier transfromation of $\psi$. We choose $g(\nu, n)=[e^{\frac{2\pi i}{k/l}(-ib\nu-n)}+e^{-\frac{2\pi i}{k/l}(-ib\nu-n)}]^{-1}\in\ch_0$ with $k/l$ odd. We obtain $|F[\psi]|\leq \Vert g\Vert \Vert (\widetilde{\bm x}^l+\widetilde{\bm x}^{-l})\psi \Vert \leq \Vert g\Vert(\Vert {\bm x}^l\psi \Vert+\Vert{\bm x}^{-l}\psi \Vert)$.
}
\begin{eqnarray}
F\lt[\psi\rt]=C\sum_{m\in\Z/k\Z}\int_\R\rmd\mu\,\delta(\mu)\delta_{\exp(\frac{2\pi i}{k}m),1}\psi(\mu,m)\qquad \forall \psi\in\Fd_0.\label{deltadistri}
\end{eqnarray}
where $C$ is a constant.

(2) Assume $S$ to be a linear functional on $\Fd_0$ and satisfy 
\begin{eqnarray}
S\lt[\left(\bm{x}-1\right)f\rt]=0,\qquad \text{and}\qquad S\lt[\left(\widetilde{\bm{x}}-1\right)f\rt]=0,
\end{eqnarray}
for all $f\in \Fd_0$. Then $S$ is given by
\begin{eqnarray}
S\lt[\psi\rt]=C\sum_{m\in\Z/k\Z}\int_\R\rmd\mu\,\psi(\mu,m)\qquad \forall \psi\in\Fd_0.\label{deltadistri1}
\end{eqnarray}
where $C$ is a constant.

\end{lemma}

\begin{proof} Given $m\in\Z/k\Z$, $y-1$ has zeros at $\mu=\mu_{\a,m}=ib^{-1}(m+k\alpha)$, $\a\in\Z$. Let us first consider any $f(\mu,m)\in\Fd_0$ with $f(0,0)=0$, we define $h(\mu,m)\in\Fd_0$ with prescribed values at $\mu=\mu_{\a,m}$ by
	\be
	f(\mu_{\a,m},m)=\lt[\widetilde{y}(\mu_{\a,m},m)-1\rt]h(\mu_{\a,m},m),\qquad \forall \a\in\Z,\ m\in\Z/k\Z,\label{constraintmualpham}
	\ee
where $\widetilde{y}(\mu_{\a,m},m)=e^{\frac{2\pi i}{k}\left(1+b^{-2}\right)\left(m+k\alpha\right)}$. $f\left(\mu,m\right)-\left(\widetilde{{y}}-1\right)h\left(\mu,m\right)$ has zeros at $\mu_{\a,m}$, which are the zeros of $y-1$. Then there exists $g(\mu,m)\in\Fd_0$ satisfying $f\left(\mu,m\right)=\left({y}-1\right)g\left(\mu,m\right)+\left(\widetilde{{y}}-1\right)h\left(\mu,m\right)$. 

For any $f(\mu,m)\in\Fd_0$ whose $f(0,0)$ is not necessarily zero, and for any $\phi_0(\mu,m)\in \Fd_0$ with $\phi_0(0,0)=1$, $f(\mu,m)-f(0,0)\phi_0(\mu,m)$ vanishes at $(0,0)$, there exist $g,h\in\Fd_0$ such that
\be
f\left(\mu,m\right)=\left({y}-1\right)g\left(\mu,m\right)+\left(\widetilde{{y}}-1\right)h\left(\mu,m\right)+f(0,0)\phi_0(\mu,m).
\ee
As a result, we obtain $F[f]=f(0,0)F[\phi_0]$. $F[\phi_0]$ is a constant, since $\phi_0$ is arbitrary and independent of $f$, and $F[f]$ does not depend on the choice of $\phi_0$.


The above argument relies on the existence of $h(\mu,m)\in\Fd_0$ satisfying \eqref{constraintmualpham}. For any $f(\mu,m)\in\Fd_0$ with $f(0,0)=0$, we can construct $h(\mu,m)\in\Fd_0$ explicitly by
\be
h\left(\mu,m\right)=f\left(\mu,m\right)s\left(\mu,m\right)\qquad s\left(\mu,m\right)=\sum_{\alpha\in\mathbb{Z}}\frac{1}{e^{\frac{2\pi i}{k}\left(1+b^{-2}\right)\left(m+k\alpha\right)}-1}\left[\frac{\sin\left(\frac{\pi}{k}\left(-ib\mu-m\right)\right)}{\frac{\pi}{k}\left(-ib\mu-m\right)-\pi\alpha}\right]^{2}
\ee
for $m\neq 0$ and remove the $\a=0$ term in the sum for the case of $m=0$. We have restrict $m=0,\cdots,k-1$ in a single period. $h\left(\mu,m\right)$ is entire in $\mu$: For any compact domain $K\subset \C$, there exists $\alpha_0>0$ such that $\left|\frac{\pi}{k}\left(-ib\mu-m\right)-\pi\alpha\right|\geq\left|\left|\pi\alpha\right|-\left|\frac{\pi}{k}\left(-ib\mu-m\right)\right|\right|\geq\left|\pi\alpha\right|-\sup_{\mu\in K}\left|\frac{\pi}{k}\left(-ib\mu-m\right)\right|$ for any $|\a |>\a_0$, and $|e^{\frac{2\pi i}{k}\left(1+b^{-2}\right)\left(m+k\alpha\right)}-1|^{-1}$ is bounded, then
\be 
\lt|{s(\mu,m)}\rt|\leq\text{finite sum}+C\sum_{|\alpha|>\a_0}\frac{|\sin\left(\frac{\pi}{k}\left(-ib\mu-m\right)\right)|^2}{\left(\left|\pi\alpha\right|-\sup_{\mu\in K}\left|\frac{\pi}{k}\left(-ib\mu-m\right)\right|\right)^{2}}.
\ee
which shows the uniform convergence within any compact domain $K$ (the poles can only affect a finite number of terms in the finite sum and cancels with the zeros of the sine function). Additionally, for large $|\re(\mu)|$, $|s(\mu ,m)|$ grows at most exponentially \footnote{Use the boundedness of $[e^{\frac{2\pi i}{k}\left(1+b^{-2}\right)\left(m+k\alpha\right)}-1]^{-1}$ and the formula $\sum_{\alpha\in\mathbb{Z}}\left|\frac{1}{x+iy-\pi\alpha}\right|^{2}=\frac{\sinh(2y)}{y\left[\cosh(2y)-\cos(2x)\right]}$ ($x,y\in\R$).}, while $f(\mu,m)$ decays faster than any exponential. Therefore $h(\mu,m)\in\Fd_0$.

By the Fourier transformation $\bmx=\cf \bmy\cf^{-1}$, we have $0=S\lt[\left(\bm{x}-1\right)f\rt]=S\lt[\cf\left(\bm{y}-1\right)\cf^{-1}f\rt]$, and thus $S\lt[\cf\,\cdot\,\rt]=F[\,\cdot\,]$, so $S$ is is given by \eqref{deltadistri1}. 

\end{proof}

\begin{lemma}\label{irreducible1}
The representation $\rho$ is irreducible in the sense that any bounded operator $\bm\co\in \cl(\ch_0)$ commuting with all elements in $\ca_h$ (i.e. $\bm\co\bm{a}\psi=\bm{a}\bm\co\psi$ \footnote{The commutativity may be written as $\bm\co\bm{a}\subset\bm{a}\bm\co$ (the graph of $\bm{a}\bm\co$ contains the graph of $\bm\co\bm{a}$), namely, $\bm{a}\bm\co$ is an extension of $\bm\co\bm{a}$.} for all $\psi\in\mathfrak{D}_0$ and $\bm{a}\in\ca_h$) is a scalar multiple of identity operator.

\end{lemma}

\begin{proof} See \cite{Han:2024nkf}.

\end{proof}


It is useful to introduce an auxiliary Hopf $*$-algebra $U_{\bfq}(sl_2)\otimes U_{\wt \bfq}(sl_2)$ defined by the polynomials of the generators ${1}$, ${\cx},{\cy},\ck$ and $\widetilde{ \cx},\widetilde{\cy},\widetilde{\ck}$ and the same Hopf $*$-structure as in \eqref{commutation00} - \eqref{starstr1} without relating $\ck,\wt{\ck}$ to $H,\wt{H}$ by the power series (i.e. $H,\wt{H}$ are not elements in $U_{\bfq}(sl_2)\otimes U_{\wt \bfq}(sl_2)$). The Hopf $*$-algebra $U_{\bfq}(sl_2)\otimes U_{\wt \bfq}(sl_2)$ is not anymore quasi-triangular since the R-element is a power series of $h$ and is not a polynomial of the generators of $U_{\bfq}(sl_2)\otimes U_{\wt \bfq}(sl_2)$.  

Based in the representation $(\mathcal{H}_0,\rho_0)$ of the quantum torus
algebra, we construct a representation of $U_{{\bfq}}(sl_{2})$ on the dense domain $\overline{\Fd}_0$: We define a family of operators $\cx_{\lambda},\cy_{\lambda},\ck_{\lambda},\ck_{\lambda}^{-1}$
\begin{align}
\ck_{\lambda} & =\bm{x}_{-1,0},\qquad \ck_{\lambda}^{-1}=\bm{x}_{1,0},\qquad \cy_{\lambda}=-\frac{i}{{\bfq}-{\bfq}^{-1}}\bm{x}_{-1,1},\label{KKEF1}\\
\cx_{\lambda} & =-\frac{i}{{\bfq}-{\bfq}^{-1}}\left[\left(\lambda+\lambda^{-1}\right)\bm{x}_{1,-1}+\bm{x}_{3,-1}+\bm{x}_{-1,-1}\right],\label{KKEF2}
\end{align}
parametrized by $\lambda\in\mathbb{C^{\times}}$. It is straight-forward to check that $\cx_{\lambda},\cy_{\lambda},\ck_{\lambda},\ck_{\lambda}^{-1}$
satisfy the commutation relation of $U_{{\bfq}}(sl_{2})$: 
\begin{eqnarray}
{\cal K}_{\lambda}{\cal X}_{\lambda}=\bm{q}{\cal X}_{\lambda}{\cal K}_{\lambda},\qquad{\cal K}_{\lambda}{\cal Y}_{\lambda}=\bm{q}^{-1}{\cal Y}_{\lambda}{\cal K}_{\lambda},\qquad\left[{\cal X}_{\lambda},{\cal Y}_{\lambda}\right]=\frac{{\cal K}_{\lambda}^{2}-{\cal K}_{\lambda}^{-2}}{\bm{q}-\bm{q}^{-1}}
\end{eqnarray}

Similarly, we define the tilded operators 
\be
\widetilde{\ck}_{\lambda} & =&\widetilde{\bm{x}}_{-1,0},\qquad\widetilde{\ck}_{\lambda}^{-1}=\widetilde{\bm{x}}_{1,0},\qquad\widetilde{\cy}_{\lambda}=-\frac{i}{\widetilde{{\bfq}}-\widetilde{{\bfq}}^{-1}}\widetilde{\bm{x}}_{-1,1},\\
\widetilde{\cx}_{\lambda} & =&-\frac{i}{\widetilde{{\bfq}}-\widetilde{{\bfq}}^{-1}}\left[\left(\overline{\lambda}+\overline{\lambda}^{-1}\right)\widetilde{\bm{x}}_{1,-1}+\widetilde{\bm{x}}_{3,-1}+\widetilde{\bm{x}}_{-1,-1}\right],
\ee
They satisfy the commutation relation of $U_{\widetilde{{\bfq}}}(sl_{2})$:
\be
\widetilde{\ck}_{\lambda}\widetilde{\cx}_{\lambda}=\widetilde{{\bfq}}\widetilde{\cx}_{\lambda}\widetilde{\ck}_{\lambda},\qquad\widetilde{\ck}_{\lambda}\widetilde{\cy}_{\lambda}=\widetilde{{\bfq}}^{-1}\widetilde{\cy}_{\lambda}\widetilde{\ck}_{\lambda},\qquad\left[\widetilde{\cx}_{\lambda},\widetilde{\cy}_{\lambda}\right]=\frac{\widetilde{\ck}_{\lambda}^{2}-\widetilde{\ck}_{\lambda}^{-2}}{\widetilde{{\bfq}}-\widetilde{{\bfq}}^{-1}}.
\ee
The tilded generators and untilded
generators are related by the Hermitian conjugate
\be
\widetilde{\ck}_{\lambda}=\ck_{\lambda}^{\dagger},\qquad\widetilde{\cx}_{\lambda}=\cx_{\lambda}^{\dagger},\qquad\widetilde{\cy}_{\lambda}=\cy_{\lambda}^{\dagger},
\ee
as the representation of the $*$-structure \eqref{starstr1}. $\cx_{\lambda},\cy_{\lambda},\ck_{\lambda},\ck_{\lambda}^{-1}, \widetilde\cx_{\lambda},\widetilde\cy_{\lambda},\widetilde\ck_{\lambda},\widetilde\ck_{\lambda}^{-1}$ are all normal operators.

A useful observation is that any even-degree monomial of $\cx_{\lambda},\cy_{\lambda},\ck_{\lambda},\ck_{\lambda}^{-1}$ is a sum of $\bmx_{\a,\b}$ with even $\a$. A similar statement holds for $\widetilde\cx_{\lambda},\widetilde\cy_{\lambda},\widetilde\ck_{\lambda},\widetilde\ck_{\lambda}^{-1}$.

The above shows that $\overline{\Fd}_0$ carries a family of $*$-representation of
$U_{{\bfq}}(sl_{2})\otimes U_{\widetilde{{\bfq}}}(sl_{2})$. The family of representations is parametrized
by the parameter $\lambda\in\mathbb{C^{\times}}$. We denote this representation by $\pi^{\lambda}$. It is manifest that $\pi^\l$ and $\pi^{\l^{-1}}$ are equivalent.

\begin{lemma}\label{irred0}
The representation $(\overline{\Fd}_0,\pi^{\lambda})$ is irreducible.
\end{lemma}

\begin{proof} The generators of the $\bfq,\widetilde{\bfq}$-Weyl algebra can be recovered by $U_{{\bfq}}(sl_{2})\otimes U_{\widetilde{{\bfq}}}(sl_{2})$ generators in the following way:
\be
&&\bm{x}=\bm{x}_{1,0}=\ck_\l^{-1},\qquad \bm{y}=\bm{x}_{0,1}=q^{-1}\bm{x}_{1,0}\bm{x}_{-1,1}=i(\bfq-\bfq^{-1})q^{-1}\ck_\l^{-1}\cy_\l,\\
&&\widetilde{\bm{x}}=\widetilde{\bm{x}}_{1,0}=\widetilde\ck_\l^{-1},\qquad \widetilde{\bm{y}}=\widetilde{\bm{x}}_{0,1}=\widetilde{q}^{-1}\widetilde{\bm{x}}_{1,0}\widetilde{\bm{x}}_{-1,1}=i(\widetilde{\bfq}-\widetilde{\bfq}^{-1})\widetilde{q}^{-1}\widetilde{\ck}_\l^{-1}\widetilde{\cy}_\l.
\ee
Since any bounded operator commuting with $\bm{x},\bm{y},\widetilde{\bm{x}},\widetilde{\bm{y}}$ is a scalar multiple of identity operator, the bounded operator commuting with $\pi^{\lambda}[U_{{\bfq}}(sl_{2})\otimes U_{\widetilde{{\bfq}}}(sl_{2})]$ also has to be a scalar multiple of identity operator. Therefore, the representation $(\overline{\Fd}_0,\pi^{\lambda})$ is irreducible.

\end{proof}


$\pi^\lambda$ is not the representation of $\uquqt$, because $\uquqt$ has generators $H,\wt{H}$, whereas $\pi^\l$ does not give the representation to $H,\wt{H}$. $\pi^\lambda(\ck)=\pi^\lambda(\bfq^{H/2})=\bm{x}^{-1}$ is well-defined on $\overline{\Fd}_0$, but $\log(\bm{x})$ is not well-define, since $\bm{n}$ is not well-defined on $\ch_0$ ($e^{\frac{2\pi i}{k}\bm{n}}$ is well-defined, since it respects the periodicity $n\sim n+k$). Since $\pi^\l(H)$ and $\pi^\l(\wt H)$ do not exist it is not clear how to define the representation of $R$ or $\wt{R}$ on $\ch_0\otimes \ch_0$. However, by using $\pi^\l$ for $U_{\bfq}(sl_2)\otimes U_{\wt \bfq}(sl_2)$, it is possible to define $R^{I\l}=(\rho^I\otimes\pi^\l)(R)$ and $\wt R^{I\l}=(\wt\rho^I\otimes\pi^\l)(\wt R)$ as the matrices of operators on $\ch_0$: the finite-dimensional representation $\rho^I$,$\wt\rho^I$ truncates $R^{I\l}$,$\wt R^{I\l}$ to finite polynomials
\be
R^{I\lambda}=\left(\rho^{I}\otimes\pi^{\lambda}\right)\left(\bm{q}^{\frac{1}{2}(H\otimes H)}\right)\sum_{n=0}^{\dim(\rho^{I})}\frac{\bm{q}^{n(n-1)/2}}{[n]_{\bm{q}}!}\left(1-\bm{q}^{-2}\right)^{n}\left[\rho^{I}\left({\cal K}{\cal X}\right)\otimes\pi^{\lambda}\left({\cal K}^{-1}{\cal Y}\right)\right]^{n}
\ee
where $\left(\rho^{I}\otimes\pi^{\lambda}\right)\left(\bm{q}^{\frac{1}{2}(H\otimes H)}\right)$ can be written as a diagonal matrix
\be
\left(\rho^{I}\otimes\pi^{\lambda}\right)\left(\bm{q}^{\frac{1}{2}(H\otimes H)}\right)=\mathrm{diag}\left({\cal K}_\l^{2I},{\cal K}_\l^{2(I-1)},\cdots{\cal K}_\l^{-2(I-1)},{\cal K}_\l^{2I}\right).
\ee
Therefore, $R^{I\lambda}$ is well-defined as an operator in $\mathrm{End}(V^I)\otimes\cl(\overline{\Fd}_0)$. The same argument also applies to $\wt R^{I\l}$, $R'{}^{I\l}$, and $\wt R'{}^{I\l}$. Appendix \ref{rep of Rs} gives explicitly the expressions of these $R$-operators when $I=1/2$.

\subsection{Infinite-dimensional $*$-representation of $\mathscr{U}_{\bf q}(sl_2)\otimes \mathscr{U}_{\widetilde{\bf {q}}}(sl_2)$}\label{Infinite-dimensional representation of suquqt}

The relations \eqref{embeduq1} and \eqref{embeduq2} embed $\suquqt$ as a subalgebra in $U_{\bfq}(sl_2)\otimes U_{\bfq}(sl_2)$. $\pi^\l$ defines the $*$-representation for $\suquqt$ carried by $\ch_0$:
\be 
K_{\lambda}&=&\bm{x}^{-2},\qquad K_{\lambda}^{-1}=\bm{x}^{2},\qquad F_{\lambda}=-\frac{i}{\bm{q}-\bm{q}^{-1}}q\bm{y},\\
E_{\lambda}&=&-\frac{i\bm{q}^{-1}}{\bm{q}-\bm{q}^{-1}}\left(1+\bm{q}\lambda^{-1}\bm{x}^{2}\right)\left(1+\bm{q}\lambda\bm{x}^{2}\right)q^{-1}\bm{x}^{-2}\bm{y}^{-1},\\
\widetilde{K}_{\lambda}&=&\widetilde{\bm{x}}^{-2},\qquad\widetilde{K}_{\lambda}^{-1}=\widetilde{\bm{x}}^{2},\qquad \widetilde{F}_{\lambda}=-\frac{i}{\widetilde{\bm{q}}-\widetilde{\bm{q}}^{-1}}\widetilde{q}^{-1}\widetilde{\bm{y}},\\
\widetilde{E}_{\lambda}&=&-\frac{i}{\widetilde{\bm{q}}-\widetilde{\bm{q}}^{-1}}\left(1+\widetilde{\bm{q}}\overline{\lambda}^{-1}\widetilde{\bm{x}}^{2}\right)\left(1+\widetilde{\bm{q}}\overline{\lambda}\widetilde{\bm{x}}^{2}\right)\widetilde{q}^{-1}\widetilde{\bm{x}}^{-2}\widetilde{\bm{y}}^{-1}.
\ee
These operators relate to $\bm{x},\widetilde{\bmx}$ only through $\bm{x}^2,\widetilde{\bmx}^2$. When $k=2N$ is an even integer, $(\pi^\l,\ch_0)$ is not irreducible for $\suquqt$, because $e^{-i\pi\bm{m}}\tilde{f}(\nu,n)=\tilde{f}(\nu,n+N)$ commutes with $\bm{x}^2,\widetilde{\bmx}^2$ and $\bmy,\widetilde{\bmy}$. Note that $e^{-i\pi\bm{m}}=(e^{-\frac{2\pi i}{k}\bm{m}})^N$ is well-defined when $k=2N$. The eigenspaces of $e^{-i\pi\bm{m}}$ are denoted by $\ch_\pm$ on which $e^{-i\pi\bm{m}}=\pm 1$:
\be 
\tilde{f}_\pm(\nu,n+N)=\pm\tilde{f}_\pm(\nu,n+N),\qquad \forall f_\pm\in\ch_{\pm}.
\ee
In terms of $f(\mu,m)$, $f_+$ is nonzero only at even $m$, whereas $f_-$ is nonzero only at odd $m$. The full Hilbert space $\ch_0$ is a direct sum: $\ch_0=\ch_+\oplus\ch_-$. 

In the following, we denote $\ch\equiv \ch_+$ for simplicity. The Hilbert space $\ch$ has the isomorphism $\ch\simeq L^2(\R)\otimes \C^N$. If we denote by $\bm{u}=\bmx^2$ and $\widetilde{\bm{u}}=\widetilde{\bmx}^2$ and represent states in $\ch$ by $f(\mu,m)$ with $\mu\in \R$ and $m\in\Z/N\Z$, 
\be
\bm{y}\psi(\mu,m)&=e^{\frac{2\pi i}{N}(-ib {\mu}-{m})}\psi(\mu,m),\qquad& \bm{u}\psi(\mu,m)=\psi(\mu+ib,m-1),\label{repuandy}\\
\tilde{\bm{y}}\psi(\mu,m)&=e^{\frac{2\pi i}{N}(-ib^{-1} {\mu}+{m})}\psi(\mu,m),\qquad & \tilde{\bm{u}}\psi(\mu,m)=\psi(\mu+ib^{-1},m+1).\label{repuandy1}
\ee
It is important to note that $(\mu,m)$ for $\ch$ in \eqref{repuandy} and \eqref{repuandy1} is obtained by $(\mu,m)\to (2\mu,2m)$ from $(\mu,m)$ for $\ch_0$, so that $m\in\Z/N\Z$ for $\ch$. The inner product on $\ch$ is defined by
\be
\langle f\mid f'\rangle=\sum_{m\in\Z/N\Z}\int_{\R}\rmd \mu\, \overline{f(\mu,m)}f'(\mu,m).
\ee

The operators in \eqref{repuandy} and \eqref{repuandy1} are exactly the same as in complex Chern-Simons theory at level-$N$ in \cite{levelk,andersen2016level,Andersen2014}, and it turns out that $\ch$ indeed carries the irreducible representation of the operator algebra from the combinatorial quantization, which is developed in Section \ref{Quantization of discrete connections}. Therefore, we mainly focus on the irreducible representation of $\suquqt$ on $\ch$ in this paper.

In the same way as in Lemma \ref{irred0}, the operators in \eqref{repuandy} and \eqref{repuandy1} define an irreducible representation of the $(\bfq^2,\widetilde{\bfq}^2)$-Weyl algebra
\be
\bm{u}\bm{y}=\bfq^2\bm{y}\bm{u},\qquad\widetilde{\bm{u}}\widetilde{\bm{y}}=\widetilde{\bfq}^2\widetilde{\bm{y}}\widetilde{\bm{u}},\qquad \bm{u}\widetilde{\bm{y}}=\widetilde{\bm{y}}\bm{u},\qquad \widetilde{\bm{u}}{\bm{y}}={\bm{y}}\widetilde{\bm{u}},\qquad \bfq^2=\exp\lt[\frac{2\pi i}{N}(b^2+1)\rt],
\ee
and the quantum torus algebra at level-$N$ generated by 
\be
\bm{u}_{\alpha,\beta}= 
\bfq^{-\alpha\beta}\bm{u}^{\alpha}\bm{y}^{\beta},\qquad 
\widetilde{\bm{u}}_{\alpha,\beta}=
 \widetilde{\bfq}^{-\alpha\beta}\widetilde{\bm{u}}^{\alpha}\widetilde{\bm{y}}^{\beta}.\label{eq:reptor11}
\ee

The representation $(\pi^\l,\ch)$ can be written in terms of $\bm{u}_{\a,\b}$ and $\widetilde{\bm{u}}_{\a,\b}$:
\begin{eqnarray}
K_{\lambda}&=&\bm{u}_{-1,0},\qquad K_{\lambda}^{-1}=\bm{u}_{1,0},\qquad F_\l=-\frac{iq}{\bm{q}-\bm{q}^{-1}}\bm{u}_{0,1}\\
E_\l&=&-\frac{iq^{-1}}{\bm{q}-\bm{q}^{-1}}\left[\left(\lambda+\lambda^{-1}\right)\bm{u}_{0,-1}+\bm{u}_{1,-1}+\bm{u}_{-1,-1}\right],\label{E354}\\
\widetilde{K}_{\lambda}&=&\widetilde{\bm{u}}_{-1,0},\qquad\widetilde{K}_{\lambda}^{-1}=\widetilde{\bm{u}}_{1,0},\qquad \widetilde{F}_{\lambda}=-\frac{i\widetilde{q}^{-1}}{\widetilde{\bm{q}}-\widetilde{\bm{q}}^{-1}}\widetilde{\bm{u}}_{0,1},\\
\widetilde{E}_\l&=&-\frac{i\widetilde{q}}{\widetilde{\bm{q}}-\widetilde{\bm{q}}^{-1}}\left[\left(\overline{\lambda}+\overline{\lambda}^{-1}\right)\widetilde{\bm{u}}_{0,-1}+\widetilde{\bm{u}}_{1,-1}+\widetilde{\bm{u}}_{-1,-1}\right].\label{Et359}
\end{eqnarray}
Only the representation of $\suquqt$ will be taken into account in our following discussion. We could have started the construction on $\ch$ with the representation of the $(\bfq^2,\widetilde{\bfq}^2)$-Weyl algebra with $\bm{u},\bmy,\widetilde{\bm u},\widetilde{\bmy}$, without involving the $(\bfq,\widetilde{\bfq})$-Weyl algebra with $\bm{x},\bmy,\widetilde{\bm x},\widetilde{\bmy}$. Indeed, the results from combinatorial quantization developed in this paper only involve the representation of $\suquqt$. We define the auxiliary $U_{\bfq}(sl_2)\otimes U_{\wt \bfq}(sl_2)$ only for some convenience of using the $R$-operators in various intermediate steps of the combinatorial quantization.

For all $\xi\in\suquqt$, $\pi^\l(\xi)$ are defined on the common dense domain $\overline{\Fd}$ of $\bm{u}_{\a,\b},\widetilde{\bm{u}}_{\a,\b}$. $\overline{\Fd}$ is the set of $f(\mu,m)$ that can be analytic continued to entire functions in $\mu$ and satisfy
\be
e^{\alpha\frac{2\pi}{N}\re(b)\mu}f(\mu+i \beta\re(b),m)\in L^{2}(\mathbb{R}),
\qquad\forall\,m\in\mathbb{Z}/N\mathbb{Z},\quad\alpha,\b\in\mathbb{Z},\quad \mu\in\R.\label{domainD11}
\ee
In the same way as in Lemma \ref{Frechet}, $\overline{\Fd}$ is a Fr\'echet space with the seminorms $\Vert f\Vert_{\a,\b}=\Vert \bm{u}_{\a,\b}f\Vert$. The subspace $\Fd\subset\overline{\Fd}$ is defined by requiring additionally that $|f(\mu+i\eta,m)|\leq C_\eta e^{-a|\mu|}$ for any $a>0$ and $\eta\in\R$. We also define $\sw\simeq \Fw\otimes\C^N\subset \Fd$. $\overline{\Fd}$, ${\Fd}$ and $\sw$ are all dense in $\ch$ and  are invariant by the action of all Laurent polynomials of $\bm{u},\bmy$ and their tilded partners.


\subsection{Conjugate representation}

The representation $(\pi^\l,\ch)$ is associated with the conjugate representation $(\overline{\pi}^\l,\ch)$ defined by 
\be     
\overline{\pi}^{\l}\lt(\xi\rt)f=\overline{\pi^{\l}\left(S\left(\xi\right)^*\right)}f,\qquad \xi\in\suquqt,\quad \forall f\in \Fd.
\ee
Using the representation \eqref{repuandy} - \eqref{repuandy1} and $\overline{\bm{\co}\psi(\mu,m)}=\overline{\bm{\co}}\,\overline{\psi}(\mu,m)$, we obtain the complex conjugates of $\bm{u},\bm{y},\widetilde{\bm u},\widetilde{\bm y}$:
\be  
\overline{\bm{u}}=\widetilde{\bm{u}}^{-1},\qquad \overline{\widetilde{\bm{u}}}=\bm{u}^{-1},\qquad \overline{\bm{y}}=\widetilde{\bm{y}},\qquad \overline{\widetilde{\bm{y}}}={\bm{y}}
\ee
The conjugate representation of the generators can be written in terms of $\bm{u}_{\a,\b}$ and $\widetilde{\bm{u}}_{\a,\b}$:
\be
\overline{\pi}^\l(K)&=&\bm{u}_{-1,0},\qquad \overline{\pi}^\l(K^{-1})=\bm{u}_{1,0},\qquad \overline{\pi}^\l(F)=\frac{iq^{-1}}{\bm{q}-\bm{q}^{-1}}\bm{u}_{1,1}\\
\overline{\pi}^\l(E)&=&\frac{iq}{\bm{q}-\bm{q}^{-1}}\left[\left(\lambda+\lambda^{-1}\right)\bm{u}_{-1,-1}+\bm{u}_{0,-1}+\bm{u}_{-2,-1}\right],\\
\overline{\pi}^\l(\widetilde{K})&=&\widetilde{\bm{u}}_{-1,0},\qquad\overline{\pi}^\l(\widetilde{K}^{-1})=\widetilde{\bm{u}}_{1,0},\qquad \overline{\pi}^\l(\widetilde{F})=\frac{i\widetilde{q}^{-1}\widetilde{\bm{q}}^{-1}}{\widetilde{\bm{q}}-\widetilde{\bm{q}}^{-1}}\widetilde{\bm{u}}_{1,1},\\
\overline{\pi}^\l(\widetilde{E})&=&\frac{i\widetilde{q}\widetilde{\bm{q}}}{\widetilde{\bm{q}}-\widetilde{\bm{q}}^{-1}}\left[\left(\overline{\lambda}+\overline{\lambda}^{-1}\right)\widetilde{\bm{u}}_{-1,-1}+\widetilde{\bm{u}}_{0,-1}+\widetilde{\bm{u}}_{-2,-1}\right].
\ee
It is straight-forward to check that $\overline{\pi}^\l$ gives a representation of \eqref{EFKalg1} and \eqref{EFKalg2}. It is also clear that $\overline{\pi}^\l$ is not a $*$-representation. 

As a representation, $\overline{\pi}^\l$ is equivalent to $\pi^\l$ in the sense that there exist densely defined operator $\bm{V}:\ \overline{\Fd}\to\overline{\Fd}$, such that 
\be    
\bm{V}\overline{\pi}^\l(\xi)\bm{V}^{-1}={\pi}^\l(\xi),\qquad \forall\xi\in\suquqt. \label{equivRepw}
\ee     
Indeed, we introduce the following unitary transformation 
\be
\mathscr{T}f\left(\mu,m\right)=\frac{1}{N^{2}}\sum_{n,m_1\in\mathbb{Z}/N\mathbb{Z}}\int d\nu d\mu_{1}\,e^{\frac{2\pi i}{N}\left(\mu\nu-mn\right)}(-1)^{n^{2}}e^{\frac{\pi i}{N}\left(-\nu^{2}+n^{2}\right)}e^{-\frac{2\pi i}{N}\left(\mu_{1}\nu-m_{1}n\right)}f\left(\mu_{1},m_{1}\right)
\ee
satisfying 
\be
\mathscr{T}^{-1}\bm{u}_{\alpha,\beta}\mathscr{T}=\left(-1\right)^{\beta}\bm{u}_{\alpha-\beta,\beta},\qquad\mathscr{T}^{-1}\widetilde{\bm{u}}_{\alpha,\beta}\mathscr{T}=\left(-1\right)^{\beta}\widetilde{\bm{u}}_{\alpha-\beta,\beta},
\ee
and the affine shift
\be 
\sigma_{1}f\left(\mu,m\right)=f\left(\mu+\frac{i}{2}\left(b+b^{-1}\right),m\right)
\ee
satisfying 
\be
\sigma_{1}\bm{u}_{\alpha,\beta}\sigma_{1}^{-1}=\bm{q}^{\beta}\bm{u}_{\alpha,\beta},\qquad\sigma_{1}\widetilde{\bm{u}}_{\alpha,\beta}\sigma_{1}^{-1}=\widetilde{\bm{q}}^{\beta}\widetilde{\bm{u}}_{\alpha,\beta}.
\ee
$\bm{V}$ in \eqref{equivRepw} is given by 
\be  
\bm{V}=\sig_1\mathscr{T}^{-1}.
\ee
Then one can check that \eqref{equivRepw} holds for all $\suquqt$ generators.

\subsection{Invariant bilinear form}\label{Invariant bilinear form}

We define a bilinear functional on $\overline{\Fd}\otimes\overline{\Fd}$
\begin{eqnarray}
\Psi_{\l}\lt[\phi\otimes\psi\rt]:=\langle \overline{\bm{V}^{-1}\phi} \mid \psi\rangle	
\end{eqnarray}
for any $\phi,\psi\in\overline{\Fd}$. If we uses the Fourier transformation $\psi(\nu,n)=\frac{1}{N}\sum_{m\in\Z/N\Z}\int \rmd \mu\, e^{-\frac{2\pi i}{N}\left(\mu\nu-mn\right)}\psi\left(\mu,m\right)$, the operators $\bm{u},\bmy$ are represented by $\bm{u}\psi\left(\nu,n\right)=e^{-\frac{2\pi i}{N}\left[-ib\nu-n\right]}\psi\left(\nu,n\right)$, $\bm{y}\psi\left(\nu,n\right)=\psi\left(\nu+ib,n-1\right)$. Then the bilinear form $\Psi_\l$ has the following integral expression:
\be 
\Psi_{\l}\lt[\phi\otimes\psi\rt]=\sum_{n\in\mathbb{Z}/N\mathbb{Z}}\int d\nu\left[(-1)^{n^{2}}e^{\frac{\pi i}{N}\left(-\nu^{2}+n^{2}\right)}e^{\frac{\pi}{N}\left(b+b^{-1}\right)\nu}\right]\phi\left(\nu,n\right)\psi\left(-\nu,-n\right).\label{psilambdaintegral}
\ee

\begin{lemma}\label{existinvbilinear}
$\Psi_\l$ is an invariant bilinear form:
\begin{eqnarray}
	\Psi_{\lambda}\lt[\lt(\pi^\l\otimes\pi^{\l}\rt)\lt(\Delta\xi\rt)f\rt]=\eps(\xi)\Psi_{\l}\lt[f\rt]\label{invPsilambda2}
	\end{eqnarray}
for any $f\in\overline{\Fd}\otimes\overline{\Fd}$.
\end{lemma}

\begin{proof}
Given $\phi, \psi\in\overline{\Fd}$, 
	\begin{eqnarray}
		&&\sum_{\xi}\lag \overline{ \bm{w}^{-1}{\pi}^{\l}\left(\xi^{(1)}\right)  \phi}\mid \pi^{\lambda}\left(\xi^{(2)}\right)\psi\rag
		=\sum_{\xi}\lag \overline{ \overline{\pi}^{\l}\left(\xi^{(1)}\right)\bm{w}^{-1} {\phi}}\mid \pi^{\lambda}\left(\xi^{(2)}\right)\psi\rag\nonumber\\
		&=&\sum_{\xi}\lag  \pi^{\l}\left(S(\xi^{(1)})^*\right)\overline{\bm{w}^{-1} {\phi}}\mid \pi^{\lambda}\left(\xi^{(2)}\right)\psi\rag
		=\varepsilon\left(\xi\right)\lag \overline{\bm{w}^{-1} {\phi}}\mid\psi\rag.
	\end{eqnarray}
\end{proof}

We define $\overline{\mathfrak{D}}_{2}\subset{\cal H}\otimes{\cal H}$ the common domain of the operator $\bm{u}_{\alpha,\beta}^{(1)},\bm{u}_{\alpha,\beta}^{(2)},\widetilde{\bm{u}}_{\alpha,\beta}^{(1)},\widetilde{\bm{u}}_{\alpha,\beta}^{(2)}$. By the same argument in Lemma \ref{Frechet}, $\overline{\mathfrak{D}}_2$ is Fr\'echet with the semi-norms
\be
\left\Vert f\right\Vert _{\vec{\alpha},\vec{\beta}}=\left\Vert \bm{u}_{\alpha_{1},\beta_{1}}^{(1)}\bm{u}_{\alpha_{2},\beta_{2}}^{(2)}f\right\Vert ,\qquad \forall f\in\overline{\Fd}_2,\quad \a,\b\in\Z.\label{seminormsD2}
\ee
We denote by $\overline{\mathfrak{D}}_2'$ the space of continuous linear functionals on $\overline{\mathfrak{D}}_2$. We also define the subspace $\Fd_2$ as before by requiring additional  $|f(\vec{\mu}+i\vec{\eta},\vec{m})|\leq C_{\vec \eta} e^{-a|\vec{\mu}|}$ for any $a>0$ and $\vec{\eta}\in\R^2$.

\begin{lemma}\label{bilinearextend}
$\Psi_\lambda$ extends to a continuous invariant linear functional on $\overline{\mathfrak{D}}_2$.
\end{lemma}

\begin{proof}
For any $f(\nu_1,\nu_2,n_1,n_2)\in \overline{\Fd}_2$, the action of $\Psi_\l$ is defined by
\be 
\Psi_{\l}\lt[f\rt]=\sum_{n_1\in\mathbb{Z}/N\mathbb{Z}}\int d\nu_1\bm{V}_1^{-1} f(\nu_1,-\nu_1,n_1,-n_1),\qquad \bm{V}_1^{-1}=(-1)^{n_1}e^{\frac{\pi i}{N}\left(-\nu^{2}_1+n^{2}_1\right)}e^{\frac{\pi}{N}\left(b+b^{-1}\right)\nu_1}.\label{PsiandV1}
\ee
which is an extension of \eqref{psilambdaintegral}. A useful relation is $\left\langle \bm{V}_1^{-1}f\mid\bm{V}_1^{-1}f\right\rangle =\left\langle \bm{u}_1^{-1}f\mid\bm{u}_1^{-1}f\right\rangle $. We define ${g\left(\nu_{1},\mu_2,n_{1},m_2\right)}=[e^{\frac{2\pi i}{N'}(-ib^{-1}\nu_{1}+n_{1})}+e^{-\frac{2\pi i}{N'}(-ib^{-1}\nu_{1}+n_{1})}]^{-1}[e^{\frac{2\pi i}{N'}(-ib^{-1}\mu_2+m_2)}+e^{-\frac{2\pi i}{N'}(-ib^{-1}\mu_2+m_2)}]^{-1}\in{\cal H}\otimes{\cal H}$ where $N'=N/l$ is odd ($l\in\Z_+$) so that $g$ has no pole.
\be
\lt|\Psi_{\l}\lt[f\rt]\rt|&=&\lt|\sum_{n_{1},m_{2}\in\mathbb{Z}/N\mathbb{Z}}\frac{1}{N}\int d\nu_{1}d\mu_{2}\bm{V}_{1}^{-1}e^{\frac{2\pi i}{N}\left(\mu_{2}\nu_{1}-m_{2}n_{1}\right)}f(\nu_{1},\mu_{2},n_{1},m_{2})\rt|\nonumber\\
&=&\lt|\lag g\mid\bm{V}_{1}^{-1}e^{\frac{2\pi i}{k}\left(\mu_{2}\nu_{1}-m_{2}n_{1}\right)}f/\overline{g}\rag\rt|
\leq \Vert g \Vert\,\Vert(\bm{u}_1^l+\bm{u}^{-l}_1)(\bm{y}^l_2+\bm{y}^{-l}_2)\bm{V}_1^{-1} f\Vert\nonumber\\
&=&\Vert g \Vert\,\Vert(\bm{u}_1^{l-1}+\bm{u}_1^{-l-1})(\bm{y}^l_2+\bm{y}^{-l}_2) f\Vert\nonumber\\
&\leq& \Vert g \Vert\lt(\Vert\bm{u}_1^{l-1}\bm{y}_2^l f\Vert+\Vert\bm{u}_1^{l-1}\bm{y}_2^{-l} f\Vert+\Vert\bm{u}_1^{-l-1}\bm{y}_2^l f\Vert+\Vert\bm{u}_1^{-l-1}\bm{y}_2^{-l}f\Vert\rt).\label{trickg}
\ee 
This implies $\Psi_\l$ is continuous on $\overline{\Fd}_2$.

The invariance \eqref{invPsilambda2} can be extended for all $f\in\overline{\Fd}_2$: ${\Fd}\otimes{\Fd}$ is dense in $\overline{\Fd}_2$ by the Fr\'echet topology, consider any sequence $f_n\to f$ converges by the semi-norms, where $f_n\in {\Fd}\otimes{\Fd}$ and $f\in\overline{\Fd}_2$. Since $\lt(\pi^\l\otimes\pi^{\l}\rt)\lt(\Delta\xi\rt)$ is a polynomial of $\bm{u},\bm{y}$ and leaves $\overline{\Fd}_2$ invariant, we have $\Psi_{\lambda}\lt[f_n\rt]\to \Psi_{\lambda}\lt[f\rt]$ and $\Psi_{\lambda}\lt[\lt(\pi^\l\otimes\pi^{\l}\rt)\lt(\Delta\xi\rt)f_n\rt]\to \Psi_{\lambda}\lt[\lt(\pi^\l\otimes\pi^{\l}\rt)\lt(\Delta\xi\rt)f\rt]$ by the above continuity. Then $\Psi_{\lambda}\lt[\lt(\pi^\l\otimes\pi^{\l}\rt)\lt(\Delta\xi\rt)f\rt]=\eps(\xi)\Psi_{\lambda}\lt[f\rt]$ is implied by $\Psi_{\lambda}\lt[\lt(\pi^\l\otimes\pi^{\l}\rt)\lt(\Delta\xi\rt)f_n\rt]=\eps(\xi)\Psi_{\lambda}\lt[f_n\rt]$.

\end{proof}

\begin{theorem}\label{uniquebilinear}

The invariant bilinear form only exists for with $\l_1=\l_2^{\pm1}$. When  $\l_1=\l_2^{\pm1}$, the invariant bilinear form is unique up to complex rescaling in the space of linear functionals on $\Fd_2$. 

\end{theorem}

\begin{proof}
Consider an arbitrary tensor product representation $\pi^{\l_1}\otimes\pi^{\l_2}$ on $\overline{\Fd}_2$. For any $\xi\in\suquqt$, an invariant bilinear form $\Psi$ satisfies 
	\begin{eqnarray}
	\Psi\lt[\lt(\pi^{\l_1}\otimes\pi^{\l_2}\rt)\lt(\Delta\xi\rt)f\rt]=\eps(\xi)\Psi\lt[f\rt],\qquad \forall f\in \overline{\Fd}_2\label{invbilinearf}
	\end{eqnarray}
We are going to explicitly solve \eqref{invbilinearf} for $\Psi$ with $\xi=K$, $KF$, $E$ and their tilded partners, and the result determines $\Psi$: For any $f\left(\nu_{1},\nu_2,n_{1},n_2\right)\in\ch\otimes\ch $, we define a unitary transformation $Uf\left(\nu_{1},\nu,n_{1},n\right)=f\left(\nu_{1},\nu-\nu_{1},n_{1},n-n_{1}\right)$. We have the following relations on $\overline{\Fd}_2$
\be
	U\bm{u}_{\alpha,\beta}^{(2)}U^{-1}=\bm{u}_{-\alpha,0}^{(1)}\bm{u}_{\alpha,\beta},\qquad U\bm{u}_{\alpha,\beta}^{(1)}U^{-1}=\bm{u}_{\alpha,\beta}^{(1)}\bm{u}_{0,\beta}
\ee
where $\bm{u}_{\a,\b}^{(1)}$ and $\bm{u}_{\a,\b}$ operate on $(\nu_1,n_1)$ and $(\nu,n)$ respectively. We define $\Psi^U$ by $\Psi^U[Uf]=\Psi[f]$. 

Consider $Uf(\nu_1,n_1;\nu,n)=f_1(\nu_1,n_1)\phi(\nu,n)$ with arbitrary $f_1,\phi\in\Fd$. Eq.\eqref{invbilinearf} with $\xi=K,\widetilde{K}$ gives $\Psi^{U}\left[f_{1}\otimes\left(\bm{u}^{-1}-1\right)\phi\right]=\Psi^{U}\left[f_{1}\otimes\left(\wt{\bm{u}}^{-1}-1\right)\phi\right]=0$. By Lemma \ref{distributiondelta}, we obtain $\Psi^{U}\left[f_{1}\otimes\phi\right]=w\left[f_{1}\right]\phi\left(0,0\right)$, where $w$ is a linear functional on $\Fd$.
	
Applying $\xi=KF$ and $Uf=\bm{V}_1\bm{u}_1^{-1}f_1\otimes\phi$ to \eqref{invbilinearf} results in $w\left[\bm{V}_1\left(\bm{y}_1-1\right)f_{1}\right]=0$, for all $ f_1,\phi\in \Fd$. Similarly, applying $\xi=\wt{K}\wt{F}$ and $Uf=(\bm{V}_1\wt{\bm{u}}_1^{-1}f_1\otimes\phi)$ to \eqref{invbilinearf} results in $w\left[\bm{V}_1\left(-\widetilde{\bm{y}}_1+1\right)f_{1}\right]=0$ for all $ f_1,\phi\in \Fd$. 
	As a result, we obtain $w\left[\bm{V}_1f_1\right]=C\sum_{n_1\in\Z/N\Z}\int d\nu_1 f_1\left(\nu_1,n_1\right)$. We determine
	\be 
	\Psi^{U}\left[f_{1}\otimes\phi\right]=C\sum_{n_1\in\mathbb{Z}/N\mathbb{Z}}\int_\R \rmd \nu_1\left(\bm{V}_{1}^{-1}f_{1}\right)\left(\nu_{1},n_{1}\right)\phi\left(0,0\right)
	\ee
For $f=U^{-1}(f_1\otimes \phi)$,  $\Psi[f]=\Psi^U[Uf]$ is identical to $\Psi_\l[f]$ in \eqref{PsiandV1} up to a rescaling. Although we only considered a few special cases of \eqref{invbilinearf}, the expression of $\Psi$ has already uniquely determined up to an overall constant.

Finally, applying $\xi=E$ to \eqref{invbilinearf} results in the constraint in the representation
	\be 
	\lambda_{1}+\lambda_{1}^{-1}=\lambda_{2}+\lambda_{2}^{-1}
	\ee
See Appendix \ref{invbilinearwithE}. The invariant bilinear form only exists for $\pi^{\l_1}\otimes\pi^{\l_2}$ with $\l_1=\l_2^{\pm1}$. We only need to consider $\l_1=\l_2$ since $\pi^\l$ and $\pi^{\l^{-1}}$ are manifestly equivalent. The existence of invariant bilinear form has been proven in Lemma \ref{existinvbilinear} and \ref{bilinearextend}.
	
\end{proof}

Note that in the above proof, the invariant bilinear form is determined without using the continuity on $\Fd_2$. Therefore $\Psi_\l$ is unique in the space of linear functionals, which is larger than $\overline{\Fd}_2'$.

\section{Complex Chern-Simons theory and discrete connections}\label{Complex Chern-Simons theory and discrete connections}

The classical $\Slc$ Chern-Simons theory on an oriented 3-manifold $M_3$ is given by the action
\be
S(\ca,\overline{\ca})=\frac{k+is}{8 \pi} \int_M \operatorname{Tr}\left(\ca \wedge d \ca+\frac{2}{3} \ca \wedge \ca \wedge \ca\right)+\frac{k-is}{8 \pi} \int_M \operatorname{Tr}\left(\overline{\ca} \wedge d \overline{\ca}+\frac{2}{3} \overline{\ca} \wedge \overline{\ca} \wedge \overline{\ca}\right),
\ee
where $k\in\mathbb{Z}_+$ and $s\in\R$. The holomorphic connection $\ca=\ca^a t^a$ where the Lie algebra basis $t^a$ is arbitrary and gives $\mathrm{Tr}(t^a t^b)=g^{ab}$, and the anti-holomorphic connection $\overline{\ca}$ is given by the complex conjugate. 

When we take $M\simeq \Sigma\times \R$ where $\Sigma$ is a 2-surface, the Chern-Simons theory induces the following nontrivial Poisson bracket to the connections on $\Sigma$
\be 
\lt\{\ca_i(z_1)\stackrel{\otimes}{,} \ca_j(z_2)\rt\}=\frac{4\pi}{k+is} \epsilon_{ij}\delta^{(2)}(z_1,z_2)\mathbf{t},\label{poisson1}\\
\lt\{\overline{\ca}_i(z_1)\stackrel{\otimes}{,} \overline{\ca}_j(z_2)\rt\}=\frac{4\pi}{k-is} \epsilon_{ij}\delta^{(2)}(z_1,z_2)\overline{\mathbf t},\label{poisson2}
\ee
where $\mathbf{t}=\sum_{a,b=1}^3(g^{-1})_{ab}t^a\otimes t^b\in sl_2\otimes sl_2$.

Let $\Sigma_{g,m}$ be a 2-surface of genus $g$ and $m$ marked points. We denote by $\pi_{g,m}=\pi_1(\Sigma_{g,m})$ the fundamental group of the 2-surface. $\pi_{g,m}$ is generated by $2g+m$ invertible generators $a_i,b_i,l_\nu$, $i=1,\cdots,g$, $\nu=1,\cdots,m$ subject to the relation,
\be 
[b_g,a_g^{-1}]\cdots[b_1,a_1^{-1}]l_{m}\cdots l_{1}=\rm{id}.
\ee
where $[x,y]=xyx^{-1}y^{-1}$. The holomorphic moduli space $\cm_{g,m}$ of $\Slc$ flat connection on $\Sigma$ is defined by 
\be 
\cm_{g,m}=\mathrm{Hom}\lt(\pi_{g,m},\Slc\rt)/\Slc
\ee 
The moduli space of $\Slc$ flat connection is given by including the anti-holomorphic part: $\cm_{g,m}\times \overline{\cm}_{g,m}$.

An $\Slc$ flat connection $(\ca,\overline{\ca})$ defines $\rho_{\ca,\overline{\ca}}\in \mathrm{Hom}\lt(\pi_{g,m},\Slc\rt)\times \mathrm{Hom}\lt(\pi_{g,m},\overline{\Slc}\rt) $ a representation of fundamental group. The representative of the generators are given by a pair of holomorphic and anti-holomorphic holonomies in 2-dimensional representations $\bf{2}$ and $\overline{\bf 2}$
\be
\lt(M_\nu,\overline{M}_\nu\rt)&=&\rho_{\ca,\overline{\ca}}(l_\nu)=\lt(\mathrm{Hol}_\ca(l_\nu),\overline{\mathrm{Hol}_\ca(l_\nu)}\rt),\label{MMbar}\\
\lt(A_i,\overline{A}_i\rt)&=&\rho_{\ca,\overline{\ca}}(l_\nu)=\lt(\mathrm{Hol}_\ca(a_i),\overline{\mathrm{Hol}_\ca(a_i)}\rt),\\
\lt(B_i,\overline{B}_i\rt)&=&\rho_{\ca,\overline{\ca}}(l_\nu)=\lt(\mathrm{Hol}_\ca(b_i),\overline{\mathrm{Hol}_\ca(b_i)}\rt).\label{BBbar}
\ee
where $\mathrm{Hol}_\ca$ as a $2\times 2$ matrix is the holonomy of the holomorphic $\ca$:
\be 
\mathrm{Hol}_\ca(c)=\calp\exp\int_c\ca\ ,
\ee
for any (oriented) curve $c$, and $\overline{\mathrm{Hol}_\ca(c)}=\mathrm{Hol}_{\overline{\ca}} (c)$. The gauge transformation acts on the holonomies by 
\be 
\mathrm{Hol}_\ca(c)\mapsto g_{t(c)}^{-1}\mathrm{Hol}_\ca(c)g_{s(c)},\qquad g_{t(c)},g_{s(c)}\in\Slc,\label{gaugetrans}
\ee 
where $s(c)$ and $t(c)$ are the source and target of the curve $c$.

Instead of the complex conjugate $\overline{\mathrm{Hol}_\ca(c)}$, it turns out to be convenient to consider
\be 
\widetilde{\mathrm{Hol}_\ca(c)}:=\lt(\mathrm{Hol}_\ca(c)^\dagger\rt)^{-1}.
\ee
Correspondingly, the gauge transformation of $\widetilde{\mathrm{Hol}_\ca(c)}$ is 
\be 
\widetilde{\mathrm{Hol}_\ca(c)}\mapsto \widetilde{g}_{t(c)}^{-1}\widetilde{\mathrm{Hol}_\ca(c)}\widetilde{g}_{s(c)},\qquad \widetilde{g}_{s(c)}=\lt({g}_{s(c)}^\dagger\rt)^{-1},\qquad \widetilde{g}_{t(c)}=\lt({g}_{t(c)}^\dagger\rt)^{-1}.
\ee 
If $t^a$ is anti-Hermitian, we have $\widetilde{\ca}^a=\overline \ca{}^a$, which gives $\widetilde{\mathrm{Hol}_\ca(c)}=\mathrm{Hol}_{\widetilde{\ca}}(c)$ with $\widetilde{\ca}=-{\ca}^\dagger$. We replace the pair of holomorphic and anti-holomorphic holonomies in \eqref{MMbar} - \eqref{BBbar} by $(M_\nu,\widetilde{M}_\nu)$, $(A_i,\widetilde{A}_i)$, $(B_i,\widetilde{B}_i)$, satisfying 
\be 
M_\nu^\dagger =\widetilde{M}_\nu^{-1},\qquad A_i^\dagger =\widetilde{A}_i^{-1},\qquad B_i^\dagger =\widetilde{B}_i^{-1},\label{daggerandtilde}
\ee 
as $2\times 2$ matrices.

The Lie algebra $sl_2$ is generated by $H,\cx,\cy$ satisfy the commutation relation $[H,\cx]=2\cx,\ [H,\cy]=-2\cy,\ [\cx,\cy]=H$. The anti-holomorphic counter-part is denoted by $\widetilde{sl_2}$, whose generators are $\widetilde H,\widetilde\cx,\widetilde\cy$ satisfying the same commutation relation. The $*$-structure on $sl_2\oplus\widetilde{sl_2}$ is given by
\be 
\cx^*=\widetilde{\cx},\qquad \cy^*=\widetilde{\cy},\qquad H^*=-\widetilde H
\ee
The representation on $V^I\otimes \widetilde{V}^J$ is given by the $\bfq,\widetilde{\bfq}\to 1$ limit of \eqref{rhorep1} - \eqref{rhorep3}:
\be 
&\rho^{I}\left(H\right)e_{m}^{I} =2m\, e_{m}^{I},
&\quad\qquad\qquad\widetilde{\rho}^{J}\left({\widetilde{H}}\right)\widetilde{e}_{m}^{J}={-2m}\, \widetilde{e}_{m}^{J}, \label{rhorepc1}\\
&\rho^{I}\lt(\cx\rt)e_{m}^{I} =\sqrt{(I-m)(I+m+1)}\, e_{m+1}^{I},&\quad \widetilde{\rho}^{J}\lt(\widetilde{{\cal Y}}\rt)\widetilde{e}_{m}^{J}=\sqrt{(J-m)(J+m+1)}\, \widetilde{e}_{m+1}^{J}, \label{rhorepc2}\\
&\rho^{I}\lt(\cy\rt)e_{m}^{I} =\sqrt{(I+m)(I-m+1)}\, e_{m-1}^{I}, &\quad \widetilde{\rho}^{J}\lt(\widetilde{{\cal X}}\rt)\widetilde{e}_{m}^{J}=\sqrt{(J+m)(J-m+1)}\, \widetilde{e}_{m-1}^{J}. \label{rhorepc3}
\ee 
Similar to Lemma \ref{rhodagger0}, we have $\rho^{I}\left(\xi\right)^T=\widetilde\rho^{I}\left(\xi^*\right)$ as the representation matrices. We consider the following linear map $W^I:\,V^I\to V^I$, whose matrix elements in $\{e_m^I\}$ are given by
\be 
(W^I)_{\ \ n}^{m}=\left[\rho^{I}(-i\sigma^2)\right]_{\ n}^{m}=\delta_{m,-n}(-1)^{j-n}
\ee
where $\sigma^{a}$, $a=1,2,3$ are Pauli matrices. $W$ gives the following transformation to the representation matrix w.r.t. $\{e_m^I\}$ (or $\wt{e}_m^I$)
\be 
W^I\rho^{I}\left(\xi\right)(W^I)^{-1}=-{\rho}^{I}\left(\xi\right)^T,\qquad \forall\, \xi \in sl_2.
\ee
We still denote by $\rho^I$ the representation of the Lie group $\Slc$ on $V^I$, and we have
\be 
W^I\rho^{I}\left(g\right)(W^I)^{-1}={\rho}^{I}\left(g^{-1}\right)^T,\qquad g\in\Slc.\label{WrhogW}
\ee
Apply this relation to the holonomies $M_\nu,A_i,B_i$, we obtain 
\be 
W^I \overline{M_\nu^I} (W^I)^{-1}=\widetilde{M}_\nu^I,\qquad W^I \overline{A_i^I} (W^I)^{-1}=\widetilde{A}_i^I,\qquad W^I \overline{B_i^I} (W^I)^{-1}=\widetilde{B}_i^I.\label{WMWM}
\ee
where ${M}_\nu^I=\rho^I({M}_\nu)$ and the same for ${A}_i^I,{B}_i^I$.

A ciliated fat graph $\G$ on $\Sig_{g,m}$ divide $\Sig_{g,m}$ into plaquettes. Each plaquette is either contractible or contains a unique hole. The orientation of the surface induces at each vertex a cyclic order on the adjacent edges. The ciliated fat graph implements an additional linear order at each vertex. When drawing the underlying fat graph on a sheet of paper in such a way that the cyclic order is counterclockwise everywhere, The ciliated fat graph has at each vertex a small cilium separating the minimal and the maximal end incident to that vertex.

\begin{figure}[h]
	\centering
	\includegraphics[width=0.5\textwidth]{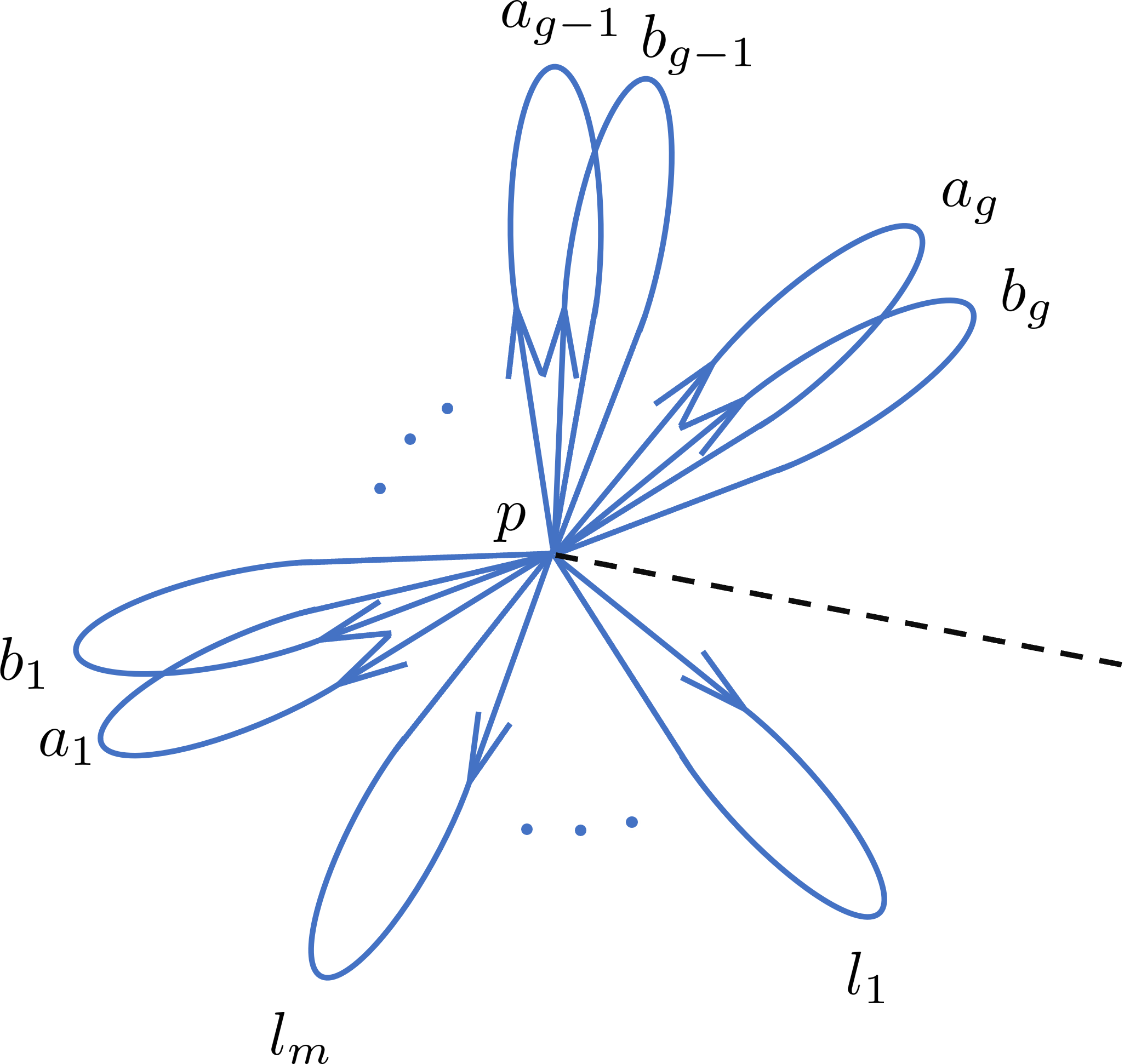}
	\caption{The standard graph $\G_{g,m}$: the cyclic order at the vertex $p$ is counterclockwise. The dashed line is the cilium.}
	\label{graph}
\end{figure}

The standard graph $\G_{g,m}$ consists a single vertex $p$ and generators $a_i,b_i,l_\nu$ of $\pi_{g,m}$ as loops based at $p$ (see FIG.\ref{graph}). In the following, we focus on the standard graph and the associated holonomies $M_\nu,A_i,B_i,\widetilde{M}_\nu,\widetilde{A}_i,\widetilde{B}_i$ to proceed the analysis. The results about operator algebra of physical observables and physical Hilbert space from the combinatorial quantization are expected to not depend on the choice of ciliated fat graph \cite{Alekseev:1994au}.

We define the space of discrete holomorphic connections $\Fa(\G_{g,m})\simeq\Slc^{\times (2g+m)}$ to be the space of holomorphic holonomies $M_\nu,A_i,B_i$. The anti-holomorphic counterpart $\widetilde{\Fa(\G_{g,m})}$ is the space of $\widetilde{M}_\nu,\widetilde{A}_i,\widetilde{B}_i$. The matrix elements of $M_\nu,A_i,B_i,\widetilde{M}_\nu,\widetilde{A}_i,\widetilde{B}_i$ and identity generate the algebra of polynomial functions on the space of discrete connections. We denote this algebra by $\Ff_{g,m}$. Any $f\in\Ff_{g,m}$ invariant under the gauge transformation \eqref{gaugetrans} defines a function on the moduli space $\cm_{g,m}\times \overline{\cm}_{g,m}$, by restricting the variables to satisfy 
\be 
[B_g,A_g^{-1}]\cdots[B_1,A_1^{-1}]M_{m}\cdots M_{1}= [\widetilde B_g,\widetilde A_g^{-1}]\cdots[\widetilde B_1,\widetilde A_1^{-1}]\widetilde M_{m}\cdots\widetilde M_{1}=1.
\ee

Given $X\in sl_2$ and $f(g)$ a holomorphic function on $\Slc$, we define the left and right invariant holomorphic differentials
\be 
<X,\nabla_L f(g)>=\left.\frac{d}{d t} f\left(e^{-t X} g\right)\right|_{t=0}, \qquad 
<X,\nabla_R f(g)>=\left.\frac{d}{d t} f\left(g e^{t X}\right)\right|_{t=0} .
\ee 
similarly, the left and right invariant anti-holomorphic differentials are defined by 
\be  
<\overline{X},\overline{\nabla}_L \wt f(\overline{g})>=\left.\frac{d}{d t} \wt f\left(e^{-t \overline{X}} \overline{g}\right)\right|_{t=0}, \qquad 
<\overline{X},\overline{\nabla}_R \wt f(\overline{g})>=\left.\frac{d}{d t}\wt f\left(\overline{g} e^{t \overline{X}}\right)\right|_{t=0} ,
\ee
We use the following notations for differentials acting on $M_\nu, A_i,B_i$:
\be 
	\nabla_{2 \nu-1}&=&\nabla_R^{M_\nu} \quad, \quad \nabla_{2 \nu}=\nabla_L^{M_\nu} \quad \text { for } \quad \nu=1, \ldots, m ; \\
	\nabla_{m+4 i-3}&=&\nabla_R^{A_i} \quad, \quad \nabla_{m+4 i-1}=\nabla_L^{A_i} \quad \text { for } \quad i=1, \ldots, g ; \\
	\nabla_{m+4 i-2}&=&\nabla_R^{B_i} \quad, \quad \nabla_{m+4 i}=\nabla_L^{B_i} \quad \text { for } i=1, \ldots, g ,
\ee 
similar for the anti-holomorphic counterpart.

The classical $r$-matrix $r\in sl_2\otimes sl_2$ satisfies the classical Yang-Baxter equation 
\be 
\lt[r_{12},r_{13}\rt]+\lt[r_{12},r_{23}\rt]+\lt[r_{13},r_{23}\rt]=0.
\ee 
where e.g. $r_{12}=r\otimes 1$ and
\be
r=\frac{1}{2}\left(H\otimes H\right)+2\left(\cx\otimes \cy\right).
\ee
Given a solution $r$ of the Yang-Baxter equation, $r'$ defined by permuting two copies of $sl_2$ also satisfies the Yang-Baxter equation.

The Fock-Rosly poisson bracket \cite{Fock:1998nu} of holomorphic functions $f,f'$ on $\Fa(\G_{g,m})$ is defined by the following formula
\be 
\lt\{f\stackrel{\otimes}{,} f'\rt\}_{\rm FR}=\frac{1}{2} \sum_i<r, \nabla_i f \wedge \nabla_i f'>+\sum_{i<j}<r, \nabla_i f \wedge \nabla_j f'>.
\ee 
For the anti-holomorphic functions $\widetilde{f},\widetilde{f}'$ on $\overline{\Fa(\G_{g,m})}$, 
\be 
\lt\{\widetilde{f}\stackrel{\otimes}{,} \widetilde{f}'\rt\}_{\overline{\rm FR}}=\frac{1}{2} \sum_i<\overline{r}, \overline{\nabla}_i \widetilde{f} \wedge \overline{\nabla}_i \widetilde{f}'>+\sum_{i<j}<\overline{r}, \overline{\nabla}_i \widetilde{f} \wedge \overline{\nabla}_j \widetilde{f}'>.
\ee 
The poisson bracket of smooth functions $F,F'$ on $\Fa(\G_{g,m})\times \overline{\Fa(\G_{g,m})}$ is defined by
\be	
\{F\stackrel{\otimes}{,} F'\}&=&\frac{4\pi}{k+is}\{F\stackrel{\otimes}{,} F'\}_{\rm FR}
+\frac{4\pi}{k-is}\{F\stackrel{\otimes}{,} F'\}_{\overline{\rm FR}}.
\ee
The poisson bracket between a holomorphic and an anti-holomorphic functions vanishes. This Poisson bracket restricted to the moduli space of flat connections coincides with the poisson bracket computed from \eqref{poisson1} and \eqref{poisson2} \cite{Fock:1998nu}.

The Fock-Rosly poisson brackets among $A_i,B_i,M_\nu$ can be computed, and some examples are 
\be 
\left\{ M_{\nu}^{I}\stackrel{\otimes}{,} M_{\nu}^{J}\right\}_{\rm FR}&=&r^{IJ}\left(M_{\nu}^{I}\otimes M_{\nu}^{J}\right)-\left(M_{\nu}^{I}\otimes M_{\nu}^{J}\right)(r^{\prime})^{IJ}\nonumber\\
&&-\left(M_{\nu}^{I}\otimes1\right)r^{IJ}\left(1\otimes M_{\nu}^{J}\right)+\left(1\otimes M_{\nu}^{J}\right)(r^{\prime} )^{IJ}\left(M_{\nu}^{I}\otimes1\right),\label{poissonFR1}\\
\left\{ M_{\nu}^{I}\stackrel{\otimes}{,} M_{\mu}^{J}\right\}_{\rm FR} &=& \left(M_{\nu}^{I}\otimes M_{\mu}^{J}\right)r^{IJ}+r^{IJ}\left(M_{\nu}^{I}\otimes M_{\mu}^{J}\right)\nonumber\\
&&-\left(M_{\nu}^{I}\otimes1\right)r^{IJ}\left(1\otimes M_{\mu}^{J}\right)-\left(1\otimes M_{\nu}^{J}\right)r^{IJ}\left(M_{\nu}^{I}\otimes1\right),\qquad \nu<\mu,\\
\left\{ A_{i}^{I}\stackrel{\otimes}{,} B_{i}^{J}\right\}_{\rm FR} &=& \left(A_{i}^{I}\otimes B_{i}^{J}\right)r^{IJ}+r^{IJ}\left(A_{i}^{I}\otimes B_{i}^{J}\right)\nonumber\\
&&-\left(A_{i}^{I}\otimes1\right)r^{IJ}\left(1\otimes B_{i}^{J}\right)+\left(1\otimes B_{i}^{J}\right)(r^{\prime})^{IJ}\left(A_{i}^{I}\otimes1\right).\label{poissonFR3}
\ee 
where $I,J$ label the finite-dimensional irreducible representation of the holonomies.

We define the anti-holomorphic $r$-matrix $\widetilde r\in \widetilde{sl_2}\otimes \widetilde{sl_2}$ is by 
\be
\widetilde r=\frac{1}{2}\left(\widetilde H\otimes \widetilde H\right)+2\left(\widetilde\cx\otimes \widetilde\cy\right).
\ee
We have the following relations as the representation matrices
\be 
W^{IJ}r^{IJ}(W^{-1})^{IJ}=\widetilde{r}^{IJ},\qquad W^{IJ}(r')^{IJ}(W^{-1})^{IJ}=(\widetilde{r}')^{IJ}.\label{WrWr}
\ee
where $W^{IJ}=W^I\otimes W^J$ and $(W^{-1})^{IJ}=(W^{-1})^I\otimes (W^{-1})^J$.

The poisson brackets $\{\cdot \stackrel{\otimes}{,}  \cdot\}_{\overline{\mathrm{FR}}}$ of the complex conjugates $\overline{M_\nu^I},\overline{A_i^I},\overline{B_i^I}$ are given by the complex conjugates of \eqref{poissonFR1} - \eqref{poissonFR3}. Considering the matrices are with respect to the basis $\{e^I_m\}$ and $\{\wt{e}^J_m\}$, we have $\overline{r^{IJ}}=r^{IJ}$ and $\overline{(r')^{ IJ}}=(r')^{IJ}$. We sandwich the equations by $W^{IJ}$ and $(W^{-1})^{IJ}$ and use \eqref{WrWr} and \eqref{WMWM}, then we obtain the tilded version of \eqref{poissonFR1} - \eqref{poissonFR3}, e.g.
\be 
\left\{\widetilde M_{\nu}^{I}\stackrel{\otimes}{,} \widetilde M_{\nu}^{J}\right\}_{\overline{\rm FR}}&=&\widetilde r^{IJ}\left(\widetilde M_{\nu}^{I}\otimes \widetilde M_{\nu}^{J}\right)-\left(\widetilde M_{\nu}^{I}\otimes \widetilde M_{\nu}^{J}\right)(\widetilde r^{\prime})^{IJ}\nonumber\\
&&-\left(\widetilde M_{\nu}^{I}\otimes1\right)\widetilde r^{IJ}\left(1\otimes \widetilde M_{\nu}^{J}\right)+\left(1\otimes \widetilde M_{\nu}^{J}\right)(\widetilde r^{\prime} )^{IJ}\left(\widetilde M_{\nu}^{I}\otimes1\right),
\ee
and similar for other two poisson brackets. The algebra $\Ff_{g,m}$ of polynomial functions becomes a poisson $*$-algebra by the implementation of the poisson bracket $\{\cdot,\cdot\}$ to the generators $M_\nu,A_i,B_i,\widetilde{M}_\nu,\widetilde{A}_i,\widetilde{B}_i$, while the $*$-structure is given by \eqref{daggerandtilde}.

We implement a poisson bracket to the group of gauge transformation at $p$:
\be 
\lt\{g^I\stackrel{\otimes}{,} g^J\rt\}&=&-r^{IJ}\left(g^I\otimes g^J\right)+\left(g^I\otimes g^J\right)r^{IJ},\qquad g\in \Slc,\label{poissongg1}\\
\lt\{\widetilde{g}^I\stackrel{\otimes}{,} \widetilde{g}^J\rt\}&=&-\widetilde{r}^{IJ}\left(\widetilde{g}^I\otimes \widetilde{g}^J\right)+\left(\widetilde{g}^I\otimes \widetilde{g}^J\right)\widetilde{r}^{IJ},\qquad \widetilde{g}\in \widetilde{\Slc},\label{poissongg2}\\
\lt\{g^I\stackrel{\otimes}{,} \widetilde{g}^J\rt\}&=&0\label{poissongg3}.
\ee 
so the gauge group becomes a poisson Lie group. The holomorphic and anti-holomorphic gauge transformations
\be 
(\Slc \times \widetilde{\Slc})\times \Ff_{g,m}\to \Ff_{g,m}
\ee 
are poisson maps.

The matrix elements of $M_\nu,A_i,B_i,\widetilde{M}_\nu,\widetilde{A}_i,\widetilde{B}_i$ that are the generators of $\Ff_{g,m}$ are constraint by e.g. $\det(M_\nu)=1$, since they are $\Slc$ holonomies. However, for the convenience of quantization, we formulate $\Ff_{g,m}$ as follows:

$\Ff_{g,m}$ is the poisson algebra spanned by polynomials of the matrix elements of $M^I_\nu,A^I_i,B^I_i,\widetilde{M}^I_\nu,\widetilde{A}^I_i,\widetilde{B}^I_i$, $I\in \mathbb{Z}_+/2$ and their inverses, subject to the poisson brackets and the following product rules
\be
M_{\nu}^{I}\otimes M_{\nu}^{J} &=&\sum_{K}(c_{1}^{IJ})_K {M}_\nu^{K}(c_{2}^{IJ})^{K},\qquad
A_i^{I}\otimes A_i^{J} =\sum_{K}(c_{1}^{IJ})_K A_i^{K}(c_{2}^{IJ})^{K},\label{prodr1}\\
B_i^{I}\otimes B_i^{J}& =&\sum_{K}(c_{1}^{IJ})_K B_i^{K}(c_{2}^{IJ})^{K},\qquad \widetilde M_{\nu}^{I}\otimes\widetilde M_{\nu}^{J} =\sum_{K}(\widetilde c_{1}^{IJ})_K\widetilde {M}_\nu^{K}(\widetilde c_{2}^{IJ})^{K},\label{prodr2}\\
\widetilde A_i^{I}\otimes\widetilde A_i^{J} &=&\sum_{K}(\widetilde c_{1}^{IJ})_K \widetilde A_i^{K}(\widetilde c_{2}^{IJ})^{K},\qquad \widetilde B_i^{I}\otimes\widetilde B_i^{J} =\sum_{K}(\widetilde c_{1}^{IJ})_K \widetilde B_i^{K}(\widetilde c_{2}^{IJ})^{K},\label{prodr3}
\ee
where $ c_{1,2}^{IJ}$ and $\widetilde c_{1,2}^{IJ}$ are the CG maps for $sl_2$ and $\widetilde{sl_2}$ and are the classical analog of $C_{1,2}^{IJ}$ and $\widetilde C_{1,2}^{IJ}$. That $M_\nu,A_i,B_i,\widetilde{M}_\nu,\widetilde{A}_i,\widetilde{B}_i$ have unit determinants is implied by taking $I,J=1/2$ in \eqref{prodr1} - \eqref{prodr3} and using $(c_1^{1/2,1/2})_0=(c_2^{1/2,1/2})^0=\frac{i}{\sqrt{2}}\sigma^2$.


\section{Quantization of discrete connections}\label{Quantization of discrete connections}

\subsection{Graph algebra}

The quantum deformation parameters used in Section \ref{Quantum Lorentz group at level-k} relates to the complex Chern-Simons level by
\be
&&q=e^{h/2},\qquad \bfq=q^2=e^h,\qquad h =\frac{4\pi i}{k+is}=\frac{2\pi i}{k}(1+b^2), \\
&&\widetilde{q}=e^{\widetilde{h}/2},\qquad \widetilde{\bfq}=\widetilde{q}^2=e^{\widetilde h},\qquad \widetilde{h} =\frac{4\pi i}{k-is}=\frac{2\pi i}{k}(1+b^{-2}).
\ee
The parameter $b$ satisfying $|b|=1$, $\mathrm{Re}(b)>0$ and $\im(b)>0$ relates to the ratio $\gamma=s/k$ by
\be
i\gamma =\frac{1-b^2}{1+b^2}.
\ee

We apply the technique of combinatorial quantization \cite{Alekseev:1994pa,Alekseev:1994au,Alekseev:1995rn} to the complex Chern-Simons theory with nonzero level $k$: As the quantization of (the holomorphic part of) $\Ff_{g,m}$, the holomorphic graph algebra $\cl_{g,m}$ is an associative algebra spanned by {polynomials} of the generators that include the identity $\bm{1}$ and the matrix elements of ${\bm A}^I_i, {\bm B}^I_i, {\bm M}^I_\nu\in \mathrm{End}(V^I)\otimes\cl_{g,m} $,  $\nu=1,\cdots,m$, $i=1,\cdots,g$. The vector space $ V^I$, $I\in \mathbb{Z}_+/2$, carries the finite-dimensional irreducible representation of $\uq$. We denote by $R=\sum_i R^{(1)}_i\otimes R^{(2)}_i$ the $R$-element in $\uq$ and by $R^\prime=\sum_i R^{(2)}_i\otimes R^{(1)}_i$. The generators in $\cl_{g,m}$ are subject to the following product rule and exchange relations. 

\begin{itemize}

\item The product rules involving only a single loop:
\be 
\left(R^{-1}\right){\!}^{I J}\stackrel{1}{\bm A}{\!}_i^I R^{I J}\stackrel{2}{\bm A}{\!}_i^J&=&\sum_{K}(C_{1}^{IJ})_K\bm{A}^{K}_i(C_{2}^{IJ})^K, \label{prodrquantum1}\\ 
\left(R^{-1}\right){\!}^{I J}\stackrel{1}{\bm B}{\!}_i^I R^{I J}\stackrel{2}{\bm B}{\!}_i^J&=&\sum_{K}(C_{1}^{IJ})_K\bm{B}^{K}_i(C_{2}^{IJ})^K, \label{prodrquantum2} \\  
\left(R^{-1}\right){\!}^{I J}\stackrel{1}{\bm M}{\!}_\nu^I R^{I J}\stackrel{2}{\bm M}{\!}_\nu^J&=&\sum_{K}(C_{1}^{IJ})_K\bm{M}^{K}_\nu (C_{2}^{IJ})^K, \label{prodrquantum3}
\ee 
where we use the notations $\stackrel{1}{\bm \co}={\bm \co}\otimes 1$ and $\stackrel{2}{\bm \co}=1\otimes {\bm \co}$ for ${\bm \co}\in\cl_{g,m}$. These relations quantize the corresponding product rules in \eqref{prodr1} - \eqref{prodr3} of $\Ff_{g,m}$.

\item The exchange relations involving two loops:
\be 
\left(R^{-1}\right){\!}^{I J}\stackrel{1}{\bm A}{\!}_i^I R^{I J}\stackrel{2}{\bm B}{\!}_i^J&=&\stackrel{2}{\bm B}{\!}_i^J\left(R^{\prime}\right){\!}^{I J}\stackrel{1}{\bm A}{\!}^I_i R^{I J},\label{loopeqnAB}\\
\left(R^{-1}\right){\!}^{I J}\stackrel{1}{\bm M}{\!}_\nu^I R^{I J}\stackrel{2}{\bm M}{\!}_\mu^J&=&\stackrel{2}{\bm M}{\!}_\mu ^J\left(R^{-1}\right){\!}^{I J}\stackrel{1}{\bm M}{\!}^I_\nu R^{I J},\qquad \nu <\mu\ ,\label{loopeqnMM}\\
\left(R^{-1}\right){\!}^{I J}\stackrel{1}{\bm M}{\!}_\nu^I R^{I J}\stackrel{2}{\bm A}{\!}_j^J&=&\stackrel{2}{\bm A}{\!}_j ^J\left(R^{-1}\right){\!}^{I J}\stackrel{1}{\bm M}{\!}^I_\nu R^{I J},\\
\left(R^{-1}\right){\!}^{I J}\stackrel{1}{\bm M}{\!}_\nu^I R^{I J}\stackrel{2}{\bm B}{\!}_j^J&=&\stackrel{2}{\bm B}{\!}_j ^J\left(R^{-1}\right){\!}^{I J}\stackrel{1}{\bm M}{\!}^I_\nu R^{I J},\\
\left(R^{-1}\right){\!}^{I J}\stackrel{1}{\bm A}{\!}_i^I R^{I J}\stackrel{2}{\bm A}{\!}_j^J&=&\stackrel{2}{\bm A}{\!}_j ^J\left(R^{-1}\right){\!}^{I J}\stackrel{1}{\bm A}{\!}^I_i R^{I J},\qquad i<j\ ,\\
\left(R^{-1}\right){\!}^{I J}\stackrel{1}{\bm B}{\!}_i^I R^{I J}\stackrel{2}{\bm B}{\!}_j^J&=&\stackrel{2}{\bm B}{\!}_j ^J\left(R^{-1}\right){\!}^{I J}\stackrel{1}{\bm B}{\!}^I_i R^{I J},\qquad i<j\ ,\\
\left(R^{-1}\right){\!}^{I J}\stackrel{1}{\bm A}{\!}_i^I R^{I J}\stackrel{2}{\bm B}{\!}_j^J&=&\stackrel{2}{\bm B}{\!}_j ^J\left(R^{-1}\right){\!}^{I J}\stackrel{1}{\bm A}{\!}^I_i R^{I J},\qquad i<j\ ,\\
\left(R^{-1}\right){\!}^{I J}\stackrel{1}{\bm B}{\!}_i^I R^{I J}\stackrel{2}{\bm A}{\!}_j^J&=&\stackrel{2}{\bm A}{\!}_j ^J\left(R^{-1}\right){\!}^{I J}\stackrel{1}{\bm B}{\!}^I_i R^{I J},\qquad i<j\ \label{loopeqn10}.
\ee

\end{itemize}

The relations \eqref{prodrquantum1} - \eqref{prodrquantum3} imply the following quadratic relations \cite{Alekseev:1994pa,Alekseev:1994au,Alekseev:1995rn} (see also Appendix \ref{Loop equation} for a derivation)
\be 
\left(R^{-1}\right){\!}^{I J}\stackrel{1}{\bm A}{\!}_i^I R^{I J}\stackrel{2}{\bm A}{\!}_i^J&=&\stackrel{2}{\bm A}{\!}_i^J\left(R^{\prime}\right){\!}^{I J}\stackrel{1}{\bm A}{\!}^I_i\left(R^{\prime -1}\right){\!}^{I J}, \label{lpeqn1}\\ 
\left(R^{-1}\right){\!}^{I J}\stackrel{1}{\bm B}{\!}_i^I R^{I J}\stackrel{2}{\bm B}{\!}_i^J&=&\stackrel{2}{\bm B}{\!}_i^J\left(R^{\prime}\right){\!}^{I J}\stackrel{1}{\bm B}{\!}^I_i\left(R^{\prime -1}\right){\!}^{I J}, \label{lpeqn2} \\  
\left(R^{-1}\right){\!}^{I J}\stackrel{1}{\bm M}{\!}_\nu^I R^{I J}\stackrel{2}{\bm M}{\!}_\nu^J&=&\stackrel{2}{\bm M}{\!}_\nu^J\left(R^{\prime}\right){\!}^{I J}\stackrel{1}{\bm M}{\!}^I_\nu\left(R^{\prime -1}\right){\!}^{I J},\label{lpeqn3}
\ee 
These quadratic relations are called loop equations.

If we relate the $R$-matrix to the classical $r$-matrix by the expansion
\be 
R^{IJ}=1\otimes 1+h\, r^{IJ}+O(h ^2).
\ee
Then the product rules for holomorphic holonomies in \eqref{prodr1} - \eqref{prodr3} of $\Ff_{g,m}$ are recovered by $h\to0$, and the Fock-Rosly poisson brackets of any pairs in $A_i,B_i,M_\nu$ are recovered from \eqref{loopeqnAB} - \eqref{lpeqn3} by 
\be
\lt[\bm{\co},\,\bm{\co}'\rt]=h\widehat{ \{\co ,\, \co' \}}_{\rm FR}+O(h^2),\qquad \co,\co'\in \Ff_{g,m}, \quad \bm{\co},\bm{\co}'\in\cl_{g,m},
\ee 
where $\widehat{\cdot}:\co\mapsto \bm{\co}$ is the quantization map.

Classically, the holonomies with $I=1/2$ have unit determinant, at the quantum level, the graph algebra implies the following result.

\begin{lemma}\label{qdetM}

The relations \eqref{prodrquantum1} - \eqref{prodrquantum3} imply $\mathrm{Det}_{\bfq^2}({\bm A}^{1/2}_i)=\mathrm{Det}_{\bfq^2}({\bm B}^{1/2}_i)=\mathrm{Det}_{\bfq^2}({\bm M}^{1/2}_\nu)=1$, where $\mathrm{Det}_{\bfq^2}({\bm M})={\bm M}^1_{\ 1}{\bm M}^2_{\ 2}-\bfq^2 {\bm M}^2_{\ 1}{\bm M}^1_{\ 2}$ is a quantum deformation of the determinant for $2\times 2$ matrix.

\end{lemma}

\begin{proof}
Consider \eqref{prodrquantum3} with $I=J=1/2$, $K$ on the right-hand side can only take values $K=0,1$. Contracting the equation with $(C_2^{\half\half})^0$ and $(C_1^{\half\half})_0$ picks up only the $K=0$ term on the right-hand side. We obtain
\be 
(C_2^{\half\half})^0_{ab}(R^{-1})_{\ ij}^{ab}({\bm{M}})_{\ k}^{i}R_{\ ml}^{kj}({\bm{M}})_{\ n}^{l}(C_1^{\half\half})_0^{mn}=1
\ee
We check that the left-hand side equals $\mathrm{Det}_{\bfq^2}({\bm M})$, by using the $R$-matrix in \eqref{Rmatrixhalfrep}, $(C_2^{\half\half})^0_{mn}=(C_1^{\half\half})_0^{mn}=\delta_{m,-n}\frac{(-1)^{1/2-m}}{\sqrt{[2]_{\bm{q}}}}\bm{q}^{m}$, and the relations from \eqref{lpeqn3}: $\left(\bfq-{\bfq}^{-1}\right) (\bm{M}^1_{\ 1})^2+\bfq \bm{M}^1_{\ 2}\bm{M}^2_{\ 1}-\bfq \bm{M}^2_{\ 1}\bm{M}^1_{\ 2}+\left({\bfq}^{-1}-\bfq\right) \bm{M}^2_{\ 2}\bm{M}^1_{\ 1}=0$ and $\bm{M}^2_{\ 2}\bm{M}^1_{\ 1}=\bm{M}^1_{\ 1}\bm{M}^2_{\ 2}$.

\end{proof}

The anti-holomorphic graph algebra $\widetilde{\cl}_{g,m}$ is defined by the following replacements applied to $\cl_{g,m}$:
\be 
\bm{A}^I_i\mapsto \widetilde{\bm{A}}^I_i, \qquad \bm{B}^I_i\mapsto \widetilde{\bm{B}}^I_i,\qquad \bm{M}^I_\nu\mapsto \widetilde{\bm{M}}^I_\nu,\qquad R\mapsto \widetilde{R}, \qquad C_{1,2}\mapsto \widetilde{C}_{1,2}\label{replacementtilde}
\ee
where $\widetilde{R}=\sum_i \widetilde{R}^{(1)}_i\otimes \widetilde{R}^{(2)}_i$ the $R$-element in $U_{\widetilde{h}}(sl_2)$. In analogy with Lemma \ref{qdetM}, we derive $\mathrm{Det}_{\wt{\bfq}^2}(\wt{\bm A}^{1/2}_i)=\mathrm{Det}_{\wt{\bfq}^2}(\wt{\bm B}^{1/2}_i)=\mathrm{Det}_{\wt{\bfq}^2}(\wt{\bm M}^{1/2}_\nu)=1$, where $\mathrm{Det}_{\wt{\bfq}^2}(\wt{\bm M})=\wt{\bm M}^1_{\ 1}\wt{\bm M}^2_{\ 2}-\wt{\bfq}^2 \wt{\bm M}^1_{\ 2}\wt{\bm M}^2_{\ 1}$.

The quantum graph algebra for the complex CS theory at level-$k$ is given by the tensor product between the holomorphic and anti-holomorphic parts: ${\cl}_{g,m}\otimes \widetilde{\cl}_{g,m}$.

\subsection{$*$-structure}

The classical poisson algebra $\Ff_{g,m}$ has the natural $*$-structure given by the complex conjugate. The quantization gives the $*$-structure to the graph algebra ${\cl}_{g,m}\otimes \widetilde{\cl}_{g,m}$. It is convenient to formulate the $*$-structure as the following
\be
(\bm{A}^I_i)^m_{\ \ n}{}^{ *}= [(\widetilde{\bm{A}}^I_i)^{-1}]^n_{\ m},\qquad 
(\bm{B}^I_i)^m_{\ \ n}{}^{ *}= [(\widetilde{\bm{B}}^I_i)^{-1}]^n_{\ m},\qquad 
(\bm{M}^I_\nu)^m_{\ \ n}{}^{ *}= [(\widetilde{\bm{M}}^I_\nu)^{-1}]^n_{\ m},\label{star structure}
\ee
as the quantum analog to \eqref{daggerandtilde}. $[(\bm{M}^I_\nu)^{-1}]^{m}_{\ n}$ and $[(\wt{\bm{M}}^I_\nu)^{-1}]^{m}_{\ n}$ are polynomial functions of $[\bm{M}^I_\nu]^i_{\ j}$ and $[\wt{\bm{M}}^I_\nu]^i_{\ j}$ repectively, and same for other generators. For $I=1/2$, we have
\be
(\bm{M}^{1/2}_\nu)^{-1}&=&\left(\begin{array}{cc}
\bm{q}^{2}\bm{M}_{\ 2}^{2}+\left(1-\bm{q}^{2}\right)\bm{M}_{\ 1}^{1} & -\bm{q}^{2}\bm{M}_{\ 2}^{1}\\
-\bm{q}^{2}\bm{M}_{\ 1}^{2} & \qquad\bm{M}_{\ 1}^{1}
\end{array}\right),\label{inverseM1}\\
(\widetilde{\bm{M}}^{1/2}_\nu)^{-1}&=&\left(\begin{array}{cc}
\widetilde{\bm{M}}_{\ 2}^{2} & -\widetilde{\bm{q}}^{2}\widetilde{\bm{M}}_{\ 2}^{1}\\
-\widetilde{\bm{q}}^{2}\widetilde{\bm{M}}_{\ 1}^{2} & \qquad\widetilde{\bm{q}}^{2}\widetilde{\bm{M}}_{\ 1}^{1}+\left(1-\widetilde{\bm{q}}^{2}\right)\widetilde{\bm{M}}_{\ 2}^{2}
\end{array}\right),\label{inverseM2}
\ee
where $\bm{M}^i_{\ j}\equiv [\bm{M}^{1/2}_\nu]^i_{\ j}$ and $\wt{\bm{M}}^i_{\ j}\equiv [\wt{\bm{M}}^{1/2}_\nu]^i_{\ j}$. For $I>1/2$, $(\bm{M}^{I}_\nu)^{-1}$ and $(\widetilde{\bm{M}}^{I}_\nu)^{-1}$ are obtained by 
\be
(\stackrel{2}{\bm{M}}{\!}^J_\nu)^{-1}({R}^{-1})^{IJ}(\stackrel{1}{\bm{M}}{\!}^I_\nu)^{-1}{R}^{IJ}=\sum_{K}({C}_{1}^{IJ})_{K}({\bm{M}}_\nu^{K})^{-1}({C}_{2}^{IJ})^{K},\\
(\stackrel{2}{\widetilde{\bm{M}}}{\!}^J_\nu)^{-1}(\widetilde{R}^{-1})^{IJ}(\stackrel{1}{\widetilde{\bm{M}}}{\!}^I_\nu)^{-1}\widetilde{R}^{IJ}=\sum_{K}(\widetilde{C}_{1}^{IJ})_{K}(\widetilde{\bm{M}}^{K}_\nu)^{-1}(\widetilde{C}_{2}^{IJ})^{K},\label{inverseLoopt}
\ee
which come from taking inverses of the left and right hand sides in the relation \eqref{prodrquantum3} and its tilded partner.

To show that the $*$-structure leave ${\cl}_{g,m}\otimes \widetilde{\cl}_{g,m}$ invariant, it is sufficient to show that the action of $*$-structure transforms the algebraic relations \eqref{prodrquantum1} - \eqref{lpeqn3} to their tilded partners: First, if we recover all matrix indices in e.g. \eqref{prodrquantum3},
\be 
(R^{-1})_{\ bj}^{ai}({\bm{M}})_{\ c}^{b}R_{\ dk}^{cj}({\bm{M}})_{\ l}^{k}=\sum_{K}(C_{1}^{IJ})_{\ Kj}^{ai}(\bm{M}^{K})_{\ k}^{j}(C_{2}^{IJ})_{\ \ dl}^{Kk}.
\ee
Apply the action of $*$ and use \eqref{star structure}, \eqref{Rdagger}, and \eqref{C1conjugate}, we obtain the relation \eqref{inverseLoopt}, which give the tilded partner of \eqref{prodrquantum3}:
\be 
(\widetilde R^{-1})_{\ bj}^{ai}(\widetilde{\bm{M}})_{\ c}^{b}\widetilde R_{\ dk}^{cj}(\widetilde{\bm{M}})_{\ l}^{k}=\sum_{K}(\widetilde C_{1}^{IJ})_{\ Kj}^{ai}(\widetilde{\bm{M}}^{K})_{\ k}^{j}(\widetilde C_{2}^{IJ})_{\ \ dl}^{Kk},
\ee
by taking the inverses of the left and right hand sides. The similar computations can be applied to \eqref{prodrquantum1} and \eqref{prodrquantum2} and obtain their tilded partners. In the above computation, the action of $*$ on the left-hand side followed by the inverse is equivalent to the replacements \eqref{replacementtilde}. The same action on the left and right hand sides of \eqref{loopeqnAB} - \eqref{lpeqn3} has the same effect and thus results in their tilded partners.



\subsection{Quantum group gauge invariance}

Recall that the group of gauge transformations $\Slc\times \wt{\Slc}$ is a poisson Lie group by the poisson brackets \eqref{poissongg3} - \eqref{poissongg3}. The combinatorial quantization deforms the gauge group to the quantum group. Indeed, the quantization of the poisson brackets gives the commutation relation
\be
R^{IJ}\stackrel{1}{g}{\!}^I\stackrel{2}{g}{\!}^J=\stackrel{2}{g}{\!}^{J}\stackrel{1}{g}{\!}^{I}R^{IJ},\qquad \widetilde{R}^{IJ}\stackrel{1}{\widetilde g}{\!}^I\stackrel{2}{\widetilde g}{\!}^J=\stackrel{2}{\widetilde g}{\!}^{J}\stackrel{1}{\widetilde g}{\!}^{I}\widetilde{R}^{IJ}.\label{Rgg=ggR0}
\ee
This is exactly the same as the commutation relation for the multiplication of $SL_{\bfq}(2)\otimes SL_{\widetilde \bfq}(2) $ (see \eqref{Rgg=ggR} in Appendix \ref{duality and star-Hopf algebra strucutre}), so we understand $g^I\in \mathrm{End}(V^I)\otimes SL_{\bfq}(2)$ and $\widetilde{g}^I\in \mathrm{End}(V^I)\otimes SL_{\widetilde \bfq}(2) $. 



It is straight-forward to check that \eqref{prodrquantum1} - \eqref{lpeqn3} are invariant under the following holomorphic gauge transformation
\be
\bm{M}^I_\nu\mapsto (g^I)^{-1}\bm{M}^I_\nu g^I,\qquad \bm{A}^I_i\mapsto (g^I)^{-1}\bm{A}^I_i g^I,\qquad \bm{B}^I_i\mapsto (g^I)^{-1}\bm{B}^I_i g^I,\label{quantumgauge0}
\ee 
where $(g^I)^{-1}=\cs(g^I)$. The results from the gauge transformation belong to $\mathrm{End}(V^I)\otimes\cl_{g,m}\otimes SL_{\bfq}(2)$. The action of $ *\otimes \star$ relates \eqref{quantumgauge0} to the anti-holomorphic gauge transformation 
\be
\widetilde{\bm{M}}^I_\nu\mapsto (\widetilde g^I)^{-1}\widetilde{\bm{M}}^I_\nu \widetilde g^I,\qquad \widetilde{\bm{A}}^I_i\mapsto (\widetilde g^I)^{-1}\widetilde{\bm{A}}^I_i \widetilde g^I,\qquad \widetilde{\bm{B}}^I_i\mapsto (\widetilde g^I)^{-1}\widetilde{\bm{B}}^I_i \widetilde{g}^I.
\ee 
which leaves the algebraic relations in $\widetilde \cl_{g,m}$ invariant.

We can let \eqref{quantumgauge0} act on $\xi\in\suq$: For instance
\be 
\xi\lt((\bm{M}^I_\nu)^i_{\ l}\rt)&:=&\lag \cs(g^I)^{i}_{\ j}(\bm{M}^I_\nu)^{j}_{\ k} (g^I)^{k}_{\ l},\xi\rag=\lag \cs(g^I)^{i}_{\ j} \otimes (g^I)^{k}_{\ l},\Delta\xi\rag(\bm{M}^I_\nu)^{j}_{\ k}\nonumber\\
&=&\sum_{\a}\rho^{I}\left(S(\xi_{\a}^{(1)})\right)^{i}_{\ j}(\bm{M}^{I}_\nu)^{j}_{\ k}\rho^{I}\left(\xi_{\a}^{(2)}\right)^{k}_{\ l}\label{xiM}
\ee 
and $\xi(\bm{A}^I_i),\xi(\bm{B}^I_i)$ are defined similarly. As a result, any $\xi\in \suq$ gives the holomorphic gauge transformation as an automorphism $\xi(\cdot):\ \cl_{g,m}\to \cl_{g,m}$.

For any $\widetilde{\zeta}\in \suqt$, the anti-holomorphic gauge transformation defines the automorphism $\widetilde{\zeta}(\cdot):\ \widetilde{\cl}_{g,m}\to \widetilde{\cl}_{g,m}$ by 
\be 
\widetilde{\zeta}\lt((\widetilde{\bm{M}}^I_\nu)^i_{\ l}\rt):=\sum_{\sigma}\widetilde \rho^{I}\left(S(\widetilde \zeta_{\sigma}^{(1)})\right)^{i}_{\ j}(\widetilde{\bm{M}}^{I}_\nu)^{j}_{\ k}\widetilde \rho^{I}\left(\widetilde \zeta_{\sigma}^{(2)}\right)^{k}_{\ l},\label{xiMtilde}
\ee
and similar for $\widetilde \zeta(\widetilde{\bm{A}}^I_i),\widetilde \zeta(\widetilde{\bm{B}}^I_i)$.

For the identity $\bm{1}$ in ${\cl}_{g,m}\otimes \widetilde{\cl}_{g,m}$, we define
\be 
\xi(\bm{1}):=\eps(\xi)\bm{1},\qquad \widetilde\xi(\bm{1}):=\eps(\widetilde \xi)\bm{1}.\label{modelalg1}
\ee
These can also be obtained by taking $I$ to be trivial representation in \eqref{xiM} and \eqref{xiMtilde} and using $\rho^I(\xi)=\eps(\xi)$ for the trivial representation. 

For any $\xi,\widetilde{\zeta}\in  \suquqt$ and $\bm{\co},\bm{\co}'\in {\cl}_{g,m}\otimes \widetilde{\cl}_{g,m}$, the gauge transformation on $\bm{\co}\bm{\co}'$ is given by
\be
\xi\lt(\bm{\co}\bm{\co}'\rt)=\sum_\alpha\xi^{(1)}_\a\lt(\bm{\co}\rt)\xi^{(2)}_\a\lt(\bm{\co}'\rt),\qquad \widetilde{\zeta}\lt(\bm{\co}\bm{\co}'\rt)=\sum_\alpha\widetilde\zeta^{(1)}_\a\lt(\bm{\co}\rt)\widetilde\zeta^{(2)}_\a\lt(\bm{\co}'\rt).\label{modelalg2}
\ee
It is sufficient to check for $\bm{\co},\bm{\co}'$ being $\bm{M}^I_\nu,\bm{A}^I_i,\bm{B}^I_i$ or their tilded counterparts: For instance, 
\be 
\xi\left(\bm{M}^{I}_\nu \otimes \bm{A}^J_i\right)
&=&\langle (g^I)^{-1}\bm{M}^I_\nu g^I\otimes(g^J)^{-1}\bm{A}^J_i g^J,\xi\rangle\nonumber\\
&=&\sum_{\alpha}\langle(g^I)^{-1}\bm{M}^I_\nu g^I,\xi_{\alpha}^{(1)}\rangle\otimes \langle(g^J)^{-1}\bm{A}^J_i g^J,\xi_{\alpha}^{(2)}\rangle\nonumber\\
&=&\sum_{\alpha}\xi_{\alpha}^{(1)}\left(\bm{M}^{I}_\nu\right)\otimes \xi_{\alpha}^{(2)}\left(\bm{A}^J_i\right)
\ee
where $\otimes$ is the matrix tensor product of $\mathrm{End}(V^I)\otimes \mathrm{End}(V^J)$.

The compatibility conditions \eqref{modelalg1} and \eqref{modelalg2} imply that ${\cl}_{g,m}\otimes \widetilde{\cl}_{g,m}$ is a module-algebra over the quantum gauge group $\suquqt$.


\subsection{Semi-direct product}\label{Semi-direct product}

Since the quantum group $\suquqt$ acts on the graph algebra ${\cl}_{g,m}\otimes \widetilde{\cl}_{g,m}$, we may extend the graph algebra to the semi-direct product 
\be 
\lt[{\cl}_{g,m}\otimes \widetilde{\cl}_{g,m}\rt]\rtimes \lt[\suquqt\rt]\simeq \lt[{\cl}_{g,m}\rtimes \suq\rt]\otimes\lt[ \widetilde{\cl}_{g,m} \rtimes \suqt\rt]\equiv \Fs_{g,m}\otimes\widetilde{\Fs}_{g,m}
\ee
where $\Fs_{g,m}={\cl}_{g,m}\rtimes \suq$ and $\widetilde{\cl}_{g,m} \rtimes \suqt$. 


Let us clarify the algebraic structure of the semi-direct product: Given a Hopf algebra $\cg$ (such as $\suquqt$) and a $\cg$-module algebra $\Fa$ (such as $\cl_{g,m}\otimes \widetilde{\cl}_{g,m}$), the semi-direct product $\Fa\rtimes \cg$ is an associative algebra made by the vector space $\Fa\otimes \cg$ and the following bilinear multiplication 
\be 
(\bm{a}\otimes\xi)(\bm{b}\otimes \zeta )=\sum_\alpha\bm{a}(\xi^{(1)}_\a\act \bm{b})\otimes\xi^{(2)}_\alpha \zeta,\qquad \bm{a},\bm{b}\in\Fa,\quad \xi,\zeta\in \cg,
\ee 
where $\act$ denotes the action $\cg$ on $\Fa$. In the following, we often use the notation $\bm{a}\otimes\xi\equiv \bm{a}\xi$. Note that $\xi$ and $\bm{a}$ are generally non-commutative: $\xi \bm{a}=(\bm{1}\otimes\xi)(\bm{a}\otimes 1 )=\sum_\alpha (\xi^{(1)}_\a\act \bm{a})\otimes\xi^{(2)}_\alpha \neq \bm{a}\xi$. The above semi-direct product is also called the smash product \cite{MOLNAR197729}.

A generic element in $\Fs_{g,m}\otimes \widetilde{\Fs}_{g,m}$ is written as $\bm{\co}\widetilde{\bm \co}'\xi\widetilde{\zeta}$, where $\bm{\co}\widetilde{\bm \co}'\in {\cl}_{g,m}\otimes \widetilde{\cl}_{g,m}$ and $\xi\widetilde{\zeta}\in \suquqt$. The left multiplication of $\xi$ gives the holomorphic gauge transformation: e.g.
\be 
\xi\bm{M}_\nu^I=\sum_\alpha \xi^{(1)}_\a(\bm{M}_\nu^I)\ \xi^{(2)}_\a,\qquad \xi\bm{A}_i^I=\sum_\alpha \xi^{(1)}_\a(\bm{A}_i^I)\ \xi^{(2)}_\a,\qquad \xi\bm{B}_i^I=\sum_\alpha \xi^{(1)}_\a(\bm{B}_i^I)\ \xi^{(2)}_\a,\label{xileftmultipli1}
\ee
Similarly, the left multiplication of $\widetilde{\xi}$ gives the anti-holomorphic gauge transformation
\be 
\widetilde\xi\widetilde{\bm{M}}_\nu^I=\sum_\alpha \widetilde\xi^{(1)}_\a(\widetilde{\bm{M}}_\nu^I)\ \widetilde\xi^{(2)}_\a,\qquad \widetilde\xi\widetilde{\bm{A}}_i^I=\sum_\alpha \widetilde\xi^{(1)}_\a(\widetilde{\bm{A}}_i^I)\ \widetilde\xi^{(2)}_\a,\qquad \widetilde\xi\widetilde{\bm{B}}_i^I=\sum_\alpha \widetilde\xi^{(1)}_\a(\widetilde{\bm{B}}_i^I)\ \widetilde\xi^{(2)}_\a,\label{xileftmultipli2}
\ee
These relations are the quantum analogs of the gauge transformation with classical Lie group: $U_g\co U_g^{-1}=\co^{(g)}$, where $U_g$ is the unitary operator representing the gauge transformation.

We define $\mu^I(\xi)\in \mathrm{End}(V^I)\otimes \suq$ and $\mu^I(\widetilde\xi)\in \mathrm{End}(\widetilde V^I)\otimes\suqt$ by
\be
\mu^I(\xi)=\sum_{\alpha}\rho^{I}\left(\xi_{\alpha}^{(1)}\right)\xi_{\alpha}^{(2)},\qquad \widetilde\mu^I(\widetilde\xi)=\sum_{\alpha}\widetilde\rho^{I}\left(\widetilde\xi_{\alpha}^{(1)}\right)\widetilde\xi_{\alpha}^{(2)}.
\ee

\begin{lemma}\label{muMMmurelations}
The left multiplication rules \eqref{xileftmultipli1} and \eqref{xileftmultipli2} are equivalent to the following relations
\be 
&&\mu^{I}(\xi)\bm{M}_{\nu}^{I}=\bm{M}_{\nu}^{I}\mu^{I}(\xi),\qquad \mu^{I}(\xi)\bm{A}_{i}^{I}=\bm{A}_{i}^{I}\mu^{I}(\xi),\qquad \mu^{I}(\xi)\bm{B}_{i}^{I}=\bm{B}_{i}^{I}\mu^{I}(\xi),\\
&&\widetilde\mu^{I}(\widetilde\xi)\widetilde{\bm{M}}_{\nu}^{I}=\widetilde{\bm{M}}_{\nu}^{I}\widetilde\mu^{I}(\widetilde \xi),\qquad \widetilde\mu^{I}(\widetilde\xi)\widetilde{\bm{A}}_{i}^{I}=\widetilde{\bm{A}}_{i}^{I}\widetilde\mu^{I}(\widetilde\xi),\qquad \widetilde\mu^{I}(\widetilde\xi)\widetilde{\bm{B}}_{i}^{I}=\widetilde{\bm{B}}_{i}^{I}\widetilde\mu^{I}(\widetilde\xi).
\ee
\end{lemma}

\begin{proof}
We take the relation for $\bm{M}_{\nu}^{I}$ as an example,
\be 
\mu^{I}(\xi)\bm{M}_{\nu}^{I}&=&\sum_{\alpha}\rho^{I}\left(\xi_{\alpha}^{(1)}\right)\xi_{\alpha}^{(2)}\bm{M}_{\nu}^{I}=\sum_{\alpha,\beta}\rho^{I}\left(\xi_{\alpha}^{(1)}\right)\xi_{\alpha,\beta}^{(2,1)}\left(\bm{M}_{\nu}^{I}\right)\xi_{\alpha,\beta}^{(2,2)}\nonumber\\
&=&\sum_{\alpha,\beta,\sigma}\rho^{I}\left(\xi_{\alpha}^{(1)}\right)\rho^{I}\left(S\left(\xi_{\alpha,\beta,\sigma}^{(2,1,1)}\right)\right)\bm{M}_\nu^{I}\rho^{I}\left(\xi_{\alpha,\beta,\sigma}^{(2,1,2)}\right)\xi_{\alpha,\beta}^{(2,2)}\nonumber\\
&=&\sum_{\alpha,\beta,\sigma}\rho^{I}\left(\xi_{\alpha,\beta,\sigma}^{(1,1,1)}\right)\rho^{I}\left(S\left(\xi_{\alpha,\beta,\sigma}^{(1,1,2)}\right)\right)\bm{M}_\nu^{I}\rho^{I}\left(\xi_{\alpha,\beta}^{(1,2)}\right)\xi_{\alpha}^{(2)}\nonumber\\
&=&\sum_{\alpha,\beta}\varepsilon\left(\xi_{\alpha,\beta}^{(1,1)}\right)\bm{M}_\nu^{I}\rho^{I}\left(\xi_{\alpha,\beta}^{(1,2)}\right)\xi_{\alpha}^{(2)}\nonumber\\
&=&\sum_{\alpha}\bm{M}_\nu^{I}\rho^{I}\left(\xi_{\alpha}^{(1)}\right)\xi_{\alpha}^{(2)}=\bm{M}_{\nu}^{I}\mu^{I}(\xi).
\ee
The 4th step uses $(\mathrm{id}\otimes\Delta\otimes\mathrm{id})(\mathrm{id}\otimes\Delta)\Delta=(\Delta\otimes\mathrm{id}\otimes\mathrm{id})(\Delta\otimes\mathrm{id})\Delta$ by the co-associativity of $\uq$. The 5th step uses $\sum_\alpha \xi_{\alpha}^{(1)}S(\xi_{\alpha}^{(2)})=\eps(\xi)1$ for all $\xi\in \uq$. The 6th step uses $\left(\varepsilon\otimes1\right)\Delta=\mathrm{id} $. 

Conversely, $\mu^{I}(\xi)\bm{M}_{\nu}^{I}=\bm{M}_{\nu}^{I}\mu^{I}(\xi)$ implies
\be
\sum_{\a,\b}\rho^{I}\left(S(\xi_{\a}^{(1)})\right)\rho^{I}\left(\xi_{\a,\b}^{(2,1)}\right)\xi_{\a,\b}^{(2,2)}\bm{M}_{\nu}^{I}=\sum_{\a,\b}\rho^{I}\left(S(\xi_{\a}^{(1)})\right)\bm{M}_\nu^{I}\rho^{I}\left(\xi_{\alpha,\b}^{(2,1)}\right)\xi_{\alpha,\b}^{(2,2)}
\ee
By the co-associativity, the right-hand side equals $\sum_{\a,\b}\rho^{I}(S(\xi_{\a,\b}^{(1,1)}))\bm{M}_\nu^{I}\rho^{I}(\xi_{\alpha,\b}^{(1,2)})\xi_{\alpha}^{(2)}=\sum_\alpha \xi^{(1)}_\a(\bm{M}_\nu^I)\ \xi^{(2)}_\a$. Therefore, by using again the co-associativity
\be
\sum_\alpha \xi^{(1)}_\a(\bm{M}_\nu^I)\ \xi^{(2)}_\a
=\sum_{\a,\b}\rho^{I}\left(S(\xi_{\a,\b}^{(1,1)})\right)\rho^{I}\left(\xi_{\a,\b}^{(1,2)}\right)\xi_{\a}^{(2)}\bm{M}_{\nu}^{I}=\sum_{\a}\eps\lt(\xi_\a^{(1)}\rt)\xi_{\a}^{(2)}\bm{M}_{\nu}^{I}=\xi \bm{M}_{\nu}^{I}.
\ee
we have used $\sum_\alpha \xi_{\alpha}^{(1)}S(\xi_{\alpha}^{(2)})=\eps(\xi)1$, and $\left(\varepsilon\otimes1\right)\Delta=\mathrm{id} $ in the last two steps.

The equivalence for other generators can be derived analogously.

\end{proof}

\section{Representation of the graph algebra for a $m$-holed sphere}\label{Representation of the graph algebra}

\subsection{Representation of loop algebra}\label{Representation of loop algebra}

The rest of this paper focuses on the quantization on $m$-holed sphere. We first consider the simplest graph algebra with $m=1$: $\cl_{0,1}\otimes\widetilde\cl_{0,1} $ on a sphere with a single hole. This algebra is generated by the identity $\bm{1}$, $\bm{M}^I$, $\widetilde{\bm{M}}^I$, subject to the following relations
\be 
\left(R^{-1}\right){\!}^{I J}\stackrel{1}{\bm M}{\!}^I R^{I J}\stackrel{2}{\bm M}{\!}^J&=&\sum_{K}(C_{1}^{IJ})_K\bm{M}^{K} (C_{2}^{IJ})^K, \label{RMRMCMC1}\\
\left(\widetilde R^{-1}\right){\!}^{I J}\stackrel{1}{\widetilde{\bm M}}{\!}^I \widetilde R^{I J}\stackrel{2}{\widetilde{\bm M}}{\!}^J&=&\sum_{K}(\widetilde C_{1}^{IJ})_K\widetilde{\bm{M}}^{K} (\widetilde C_{2}^{IJ})^K, \label{RMRMCMC2}
\ee 


We construct a family of representations $\cd^{\l}$ labelled by $\l\in\C^\times$. The representation is carried by the dense domain $\overline{\Fd}$ in the Hilbert space $\ch\simeq L^{2}(\R)\otimes \C^N$, the same as the one carrying the infinite-dimensional $*$-representation $\pi^{\l}$ of $\suquqt$. Indeed, the representation of the generators of $\cl_{0,1}$ is defined by using the representations $\rho^I$ and $\pi^\l$ of the quantum group
\be 
\cd^{\l}\lt(\bm{M}^I\rt)&=& \lt(\rho^I\otimes \pi^\l\rt)\lt(R'R\rt)\equiv  \lt(R'R\rt)^{I\l},\label{looprep1}\\
\cd^{\l}\lt(\widetilde{\bm{M}}^I\rt)&=& \lt(\widetilde\rho^I\otimes \pi^\l\rt)\lt(\widetilde{R}'\widetilde R\rt)\equiv \lt(\widetilde{R}'\widetilde R\rt)^{I\l}.\label{looprep2}
\ee 
The representation $(R'R)^{I\l}$ leaves $\ch$ invariant, although $R^{I\l}$ and $(R^\prime)^{I\l}$ are defined on the larger Hilbert space $\ch_0$. It is sufficient to show this at $I=1/2$, since all $(\bm{M}^I)^m_{\ n}$ for $I>1/2$ are the polynomials of $(\bm{M}^{1/2})^m_{\ n}$ by the relation \eqref{RMRMCMC1}.  $\cd^\l(\bm{M}^{1/2})$ can be written explicitly as (see Appendix \ref{rep of Rs}):
\be
\cd^\l(\bm{M}^{1/2})&=&\left(R^{\prime}R\right)^{1/2,\lambda}=\left(
	\begin{array}{cc}
		{\ck}_\l^2  &\quad q^{-1}\lt(q^2-q^{-2}\rt){\ck}_\l{ \cy}_\l\\
		q^{-1}\lt(q^2-q^{-2}\rt){ \cx}_\l{  \ck}_\l  &\quad q^{-2}\lt(q^2-q^{-2}\rt)^2{ \cx}_\l{ \cy}_\l+{\ck}^{-2}_\l
	\end{array}
	\right),\label{MandUq}
\ee
All matrix elements are the linear combinations of quadratic monomials of $\ck_\l,\cx_\l,\cy_\l$, so they only depend on $\bmx^2=\bm{u}$ and $\bm{y}$. 
\be 
{\cal D}^{\lambda}\left(\bm{M}^{1/2}\right)&=&\begin{pmatrix}\bm{u}_{-1,0} & -iq^{-2}\bm{u}_{-1,1}\\
-iq^{-2}\left[\left(\lambda+\lambda^{-1}\right)\bm{u}_{0,-1}+\bm{u}_{1,-1}+\bm{u}_{-1,-1}\right] & \quad-q^{-2}\left[\left(\lambda+\lambda^{-1}\right)\bm{1}+q^{-2}\bm{u}_{-1,0}\right]
\end{pmatrix}\nonumber
\ee
In the same way, we have
\be
 \cd^\l(\widetilde{\bm{M}}^{1/2})&=&\left(\widetilde R^{\prime}\widetilde R\right)^{1/2,\lambda}=\begin{pmatrix}\widetilde{{\cal K}}_{\lambda}^{-2}+\left(\widetilde{q}^{2}-\widetilde{q}^{-2}\right)^{2}\widetilde{q}^{-2}\widetilde{{\cal X}}_{\lambda}\widetilde{{\cal Y}}_{\lambda} & \quad\left(\widetilde{q}^{2}-\widetilde{q}^{-2}\right)\widetilde{q}^{-1}\widetilde{{\cal X}}_{\lambda}\widetilde{{\cal K}}_{\lambda}\\
\left(\widetilde{q}^{2}-\widetilde{q}^{-2}\right)\widetilde{q}^{-1}\widetilde{{\cal K}}_{\lambda}\widetilde{{\cal Y}}_{\lambda} & \quad\widetilde{{\cal K}}_{\lambda}^{2}
\end{pmatrix}.\label{MandUqtilde}\\
&=&\begin{pmatrix}-\widetilde{q}^{-2}\left[\left(\overline{\lambda}+\overline{\lambda}^{-1}\right)\bm{1}+\widetilde{q}^{-2}\widetilde{\bm{u}}_{-1,0}\right] & \quad-i\widetilde{q}^{-2}\left[\left(\overline{\lambda}+\overline{\lambda}^{-1}\right)\widetilde{\bm{u}}_{0,-1}+\widetilde{\bm{u}}_{1,-1}+\widetilde{\bm{u}}_{-1,-1}\right]\\
	-i\widetilde{q}^{-2}\widetilde{\bm{u}}_{-1,1} & \quad\widetilde{\bm{u}}_{-1,0}
	\end{pmatrix}.\nonumber
\ee
It is also easy to see that $\cd^\l({\bm{M}}^{1/2})^m_{\ n}$ and $\cd^\l(\widetilde{\bm{M}}^{1/2})^m_{\ n}$ are polynomials of the generators of $\suquqt$. Therefore, the representation $\cd^\l$ leave $\ch$ invariant. This clarifies the statement below \eqref{Et359} that $U_{\bm q}(sl_2)\otimes U_{\wt{\bm q}}(sl_2)$ and $R^{I\l}$ defined on $\ch_0$ are only auxiliary and introduced for some convenience of derivations. We could have define $\cd^\l(\bm{M}^{1/2})$ in terms of $\bm{u}_{\a,\b}$ directly on $\ch$. But defining the representation using $R'$ and $R$ allows us to use the Yang-Baxter equation to simplify several derivations (e.g. see below). The combinatorial quantization only relates to the representation of $\suquqt$. In addition, following the same proof as Lemma \ref{irreducible1} (see \cite{Han:2024nkf}), one can show that the representation $\cd^\l$ of the loop algebra is irreducible on $\overline{\Fd}$. 

It might not difficult to see that the quantum traces of ${\cal D}^{\lambda}\left(\bm{M}^{1/2}\right)$ and ${\cal D}^{\lambda}\left(\wt{\bm{M}}^{1/2}\right)$ relate to the quadratic Casimir operators of the quantum group. We will come back to this point in Section \ref{Wilson loop operators}.



We check that \eqref{looprep1} and \eqref{looprep2} indeed give the representation of the loop algebra: We enlarge the representation $\cd^\l$ to $\ch_0$ and show that 
\be 
\stackrel{1}{\bm{M}}_{3}R_{12}R_{23}^{\prime}= R_{13}^{\prime}R_{13}R_{12}R_{23}^{\prime}=R_{13}^{\prime}R_{23}^{\prime}R_{12}R_{13}=R_{12}R_{23}^{\prime}R_{13}^{\prime}R_{13}=R_{12}R_{23}^{\prime}\stackrel{1}{\bm{M}}_{3},\label{RRRR111}
\ee
where we skip all representation labels. In the second step, we have used $R_{13}R_{12}R_{23}^{\prime}=R_{23}^{\prime}R_{12}R_{13}$ obtained from the Yang-Baxter equation $R_{12}R_{13}R_{23}=R_{23}R_{13}R_{12}$ by switching $2\leftrightarrow 3$. In the third step, we have used $R_{13}^{\prime}R_{23}^{\prime}R_{12}=R_{12}R_{23}^{\prime}R_{13}^{\prime}$ obtained from the primed Yang-Baxter equation $R_{12}^{\prime}R_{13}^{\prime}R_{23}^{\prime}=R_{23}^{\prime}R_{13}^{\prime}R_{12}^{\prime}$ by switching $1\leftrightarrow 2$. Using \eqref{RRRR111},
we check \eqref{RMRMCMC1} is indeed satisfied by the representation
\be
&&\left(R^{-1}\right){}^{IJ}\cd^\l\Big(\stackrel{1}{\bm{M}}{\!\!}^{I}\Big)R^{IJ}\cd^\l\Big(\stackrel{2}{\bm{M}}{\!\!}^{J}	\Big)\nonumber\\
&=&\left(R_{12}^{-1}R_{13}^{\prime}R_{13}R_{12}R_{23}^{\prime}R_{23}\right)^{IJ\lambda}
=\left(R_{12}^{-1}R_{12}R_{23}^{\prime}R_{13}^{\prime}R_{13}R_{23}\right)^{IJ\lambda}=\left(R_{23}^{\prime}R_{13}^{\prime}R_{13}R_{23}\right)^{IJ\lambda}\nonumber\\
&=&\left(\Delta\otimes1\left(R^{\prime}R\right)\right)^{IJ\lambda}
=\sum_{K}(C_{1}^{IJ})_K\left(R^{\prime}R\right)^{K\lambda}(C_{2}^{IJ})^ K=\sum_{K}(C_{1}^{IJ})_K\cd^\l(\bm{M}^{K})(C_{2}^{IJ})^ K.
\ee
This relation holds on $\ch$ since both left and right hand sides leave $\ch$ invariant. A similar computation in the tilded sector shows that \eqref{RMRMCMC2} is also satisfied by the representation.

We check that $\cd^{\l}$ is a $*$-representation by using Lemma \ref{rhodagger0} and the property that $\pi^\l$ is a $*$-representation
\be
\cd^{\lambda}\left((\bm{M}^{J})_{\ n}^{m}\right)^{\dagger}	&=&\sum_{\alpha,\beta}\overline{\rho^{J}(R_{\alpha}^{(2)})_{\ j}^{m}}\overline{\rho^{J}(R_{\beta}^{(1)})_{\ n}^{j}}\pi^{\lambda}(R_{\alpha}^{(1)}R_{\beta}^{(2)})^{\dagger}\nonumber\\
&=&\sum_{\alpha,\beta}\widetilde \rho^{J}(R_{\alpha}^{(2)*})_{\ m}^{j}\widetilde\rho^{J}(R_{\beta}^{(1)*})_{\ j}^{n}\pi^{\lambda}(R_{\beta}^{(2)*}R_{\alpha}^{(1)*})\nonumber\\
&=&\sum_{\alpha,\beta}\widetilde\rho^{J}(R_{\beta}^{(1)*}R_{\alpha}^{(2)*})_{\ m}^{n}\pi^{\lambda}(R_{\beta}^{(2)*}R_{\alpha}^{(1)*})\nonumber\\
&=&(\widetilde\rho^{J})_{\ m}^{n}\otimes\pi^{\lambda}\left(R^{*\otimes*}R^{\prime*\otimes*}\right)=(\widetilde\rho^{J})_{\ m}^{n}\otimes\pi^{\lambda}\left(\widetilde{R}^{-1}\widetilde{R}^{\prime-1}\right)\nonumber\\
&=& \cd^{\lambda}\left([(\widetilde{\bm{M}}^{J})^{-1}]_{\ m}^{n}\right)
\ee
where we have used the expression $R=\sum_\a R^{(1)}_\a\otimes R^{(2)}_\a$.

\subsection{Representation of graph algebra}

The graph algebra $\cl_{0,m}\otimes\widetilde{\cl}_{0,m}$ on a sphere with $m$ holes is generated by the identity $\bm{1}$ and the matrix elements of $\bm{M}^I_\nu, \widetilde{\bm{M}}^I_\nu$, $\nu=1,\cdots,m$, subject to the relations 
\be 
\left(R^{-1}\right){\!}^{I J}\stackrel{1}{\bm M}{\!}^I_\nu R^{I J}\stackrel{2}{\bm M}{\!}^J_\nu&=&\sum_{K}(C_{1}^{IJ})_K\bm{M}^{K}_\nu (C_{2}^{IJ})^K, \label{RMRMCMC31}\\
\left(\widetilde R^{-1}\right){\!}^{I J}\stackrel{1}{\widetilde{\bm M}}{\!}^I_\nu \widetilde R^{I J}\stackrel{2}{\widetilde{\bm M}}{\!}^J_\nu&=&\sum_{K}(\widetilde C_{1}^{IJ})_K\widetilde{\bm{M}}^{K}_\nu (\widetilde C_{2}^{IJ})^K, \label{RMRMCMC32}\\
\left(R^{-1}\right){\!}^{I J}\stackrel{1}{\bm M}{\!}_\nu^I R^{I J}\stackrel{2}{\bm M}{\!}_\mu^J&=&\stackrel{2}{\bm M}{\!}_\mu ^J\left(R^{-1}\right){\!}^{I J}\stackrel{1}{\bm M}{\!}^I_\nu R^{I J},\qquad \nu <\mu\ ,\\
\left(\widetilde R^{-1}\right){\!}^{I J}\stackrel{1}{\widetilde{\bm M}}{\!}_\nu^I \widetilde R^{I J}\stackrel{2}{\widetilde{\bm M}}{\!}_\mu^J&=&\stackrel{2}{\widetilde{\bm M}}{\!}_\mu ^J\left(\widetilde R^{-1}\right){\!}^{I J}\stackrel{1}{\widetilde{\bm M}}{\!}^I_\nu \widetilde R^{I J},\qquad \nu <\mu\ ,
\ee 

The representation of the graph algebra $\cd^{\l_1,\cdots,\l_m}$ is carried by the dense domain $\overline{\Fd}_m$ in 
\be
\ch_{\l_1,\cdots,\l_m}:=\ch_{\l_1}\otimes \cdots\otimes\ch_{\l_m}, \label{chlambda1m}
\ee
$\cd^{\l_1,\cdots,\l_m}$ closely relates to the tensor product $\pi^{\l_1}\otimes\cdots\otimes\pi^{\l_m}$, and $\ch_{\l_\nu}\simeq\ch\simeq L^2(\R)\otimes\C^N$ carries the representation $\pi^{\l_\nu}$ associated to the $\nu$-th hole. $\overline{\Fd}_m$ is the maximal common domain of $\bm{u}_{\a,\b}^{(\nu)}$, $\nu=1,\cdots,m$, $\a,\b\in\Z$, where $\bm{u}_{\a,\b}^{(\nu)}$ acts on $\ch_{\l_{\nu}}$. $\overline{\Fd}_m$ is a Fr\'echet space whose semi-norms generalizes from the ones of $\overline{\Fd}_2$ in \eqref{seminormsD2}.

We denote by $\cl_{\nu}\otimes \widetilde{\cl}_\nu\subset \cl_{0,m}\otimes\widetilde{\cl}_{0,n}$ the loop algebra generated by $\bm{1},\ \bm{M}^I_\nu, \widetilde{\bm{M}}^I_\nu$.  The action of $\cl_{\nu}\otimes \widetilde{\cl}_\nu$ on the $\nu$-th copy of the tensor product \eqref{chlambda1m} is given by
\be
\cd_{\nu}^{\lambda_{\nu}}(\bm{a}):=\mathrm{id}_{\l_1}\otimes \cdots \mathrm{id}_{\l_{\nu-1}}\otimes \cd^{\lambda_{\nu}}(\bm{a})\otimes \mathrm{id}_{\l_{\nu+1}}\otimes \cdots \mathrm{id}_{\l_{m}},
\ee
for any $\bm{a}\in \cl_{\nu}\otimes \widetilde{\cl}_\nu$. For any $\xi\in U_{\bm q}(sl_2)\otimes U_{\wt{\bm q}}(sl_2)$, we denote by
\be 
\iota^{\l_1,\cdots,\l_{\nu-1}}_\nu(\xi)=\lt(\pi^{\l_1}\otimes\cdots\otimes\pi^{\l_{\nu-1}}\rt)\lt(\Delta^{(\nu-1)}\xi\rt)\otimes \mathrm{id}_{\l_{\nu}}\otimes \cdots \mathrm{id}_{\l_{m}},\qquad \iota_1(\xi)=\eps(\xi).
\ee
With these notations, the representation $\cd^{\l_1,\cdots,\l_m}\equiv \cd^{\vec \l}$ is defined by
\be 
\cd^{\vec \l}\left(\bm{M}_{\nu}^{J}\right)&=&\left(\rho^{J}\otimes\iota_\nu^{\l_1,\cdots,\l_{\nu-1}}\right)\left(R^{\prime}\right)D_{\nu}^{\lambda_{\nu}}\left(\bm{M}_{\nu}^{J}\right)\left(\rho^{J}\otimes\iota_\nu^{\l_1,\cdots,\l_{\nu-1}}\right)\left(R^{\prime-1}\right),\\
\cd^{\vec \l}\left(\widetilde{\bm{M}}_{\nu}^{J}\right)&=&\left(\rho^{J}\otimes\iota_\nu^{\l_1,\cdots,\l_{\nu-1}}\right)\left(\widetilde R^{\prime}\right)D_{\nu}^{\lambda_{\nu}}\left(\widetilde{\bm{M}}_{\nu}^{J}\right)\left(\rho^{J}\otimes\iota_\nu^{\l_1,\cdots,\l_{\nu-1}}\right)\left(\widetilde R^{\prime-1}\right).
\ee
The relation $(\Delta\otimes1)(R)=R_{13}R_{23}$ and cyclic permuting $(123)$ give $(1\otimes\Delta)(R^{\prime})=R_{12}^{\prime}R_{13}^{\prime}$. In the following, we give the explicit expressions of $\cd^{\lambda_{1},\cdots,\lambda_{m}}\left(\bm{M}_{\nu}^{J}\right)$ for $\nu=1,\cdots,4$:
\be 
\cd^{\vec \l}\left(\bm{M}_{1}^{J}\right)
&=&\left(R_{01}^{\prime}R_{01}\right)^{J\lambda_{1}}\otimes{\rm id}_{\lambda_{2}}\otimes\cdots\otimes{\rm id}_{\lambda_{m}},\\
\cd^{\vec \l}\left(\bm{M}_{2}^{J}\right)
&=&\left[R_{01}^{\prime}\left(R_{02}^{\prime}R_{02}\right)R_{01}^{\prime-1}\right]^{J\lambda_{1}\lambda_{2}}\otimes{\rm id}_{\lambda_{3}}\otimes\cdots\otimes{\rm id}_{\lambda_{m}},\\
\cd^{\vec \l}\left(\bm{M}_{3}^{J}\right)
&=&\left[R_{01}^{\prime}R_{02}^{\prime}\left(R_{03}^{\prime}R_{03}\right)R_{02}^{\prime-1}R_{01}^{\prime-1}\right]^{J\lambda_{1}\lambda_{2}\lambda_{3}}\otimes{\rm id}_{\lambda_{4}}\otimes\cdots\otimes {\rm id}_{\lambda_{m}},\\
\cd^{\vec \l}\left(\bm{M}_{4}^{J}\right)
&=&\left[R_{01}^{\prime}R_{02}^{\prime}R_{03}^{\prime}\left(R_{04}^{\prime}R_{04}\right)R_{03}^{\prime-1}R_{02}^{\prime-1}R_{01}^{\prime-1}\right]^{J\lambda_{1}\lambda_{2}\lambda_{3}\lambda_{4}}\otimes{\rm id}_{\lambda_{5}}\otimes\cdots\otimes {\rm id}_{\lambda_{m}},
\ee
where the slots with the numbers $0$ and $\nu=1,\cdots,m$ correspond to the representations $\rho^J$ and $\pi^{\l_\nu}$ respectively. The computation in Appendix \ref{App:Rep of graph algebra} checks $\cd^{\vec \l}$ is indeed a representation of the graph algebra. Additionally, $\cd^{\vec \l}$ is a $*$-representation: 
\be 
&&\cd^{\vec \l}\left((\bm{M}_{\nu}^{J})^m_{\ n} \right)^{\dagger}\nonumber\\
&=&\sum_{\a,\b}\overline{\rho^{J}\lt(R_\a^{(2)-1}\rt)^l_{\ n}}\iota_{\nu}^{\l_1,\cdots,\l_{\nu-1}}\left(R^{(1)-1}_\a\right)^{\dagger}\cd_{\nu}^{\lambda_{\nu}}\left((\bm{M}_{\nu}^{J})^k_{\ l}\right)^{\dagger}\overline{\rho^{J}\lt(R_\b^{(2)}\rt)^m_{\ k}}\iota_{\nu}^{\l_1,\cdots,\l_{\nu-1}}\left(R^{(1)}_\b\right)^{\dagger}\nonumber\\
&=&\left((\widetilde{\rho}^{J})^n_{\ l}\otimes\iota_{\nu}^{\l_1,\cdots,\l_{\nu-1}}\right)\left(\widetilde{R}^{\prime}\right)\cd_{\nu}^{\lambda_{\nu}}\left(((\widetilde{\bm{M}}_{\nu}^{J})^{-1})^l_{\ k}\right)\left((\widetilde\rho^{J})^k_{\ m}\otimes\iota_{\nu}^{\l_1,\cdots,\l_{\nu-1}}\right)\left(\widetilde R^{\prime-1}\right)\nonumber\\
&=&\cd^{\vec \l}\left(\big((\widetilde{\bm{M}}_{\nu}^{J})^{-1}\big)^n_{\ m}\right).
\ee

The relation \eqref{prodrquantum3} implies $\cd^{\vec \l}((\bm{M}^{J}_\nu)^m_{\ n})$ is a polynomial of $\cd^{\vec \l}((\bm{M}^{1/2}_\nu)^a_{\ b})$, and the situation is similar for the tilded partners. The explicit expressions of $\cd^{\vec \l}(\bm{M}^{1/2}_\nu)$ and $\cd^{\vec \l}(\widetilde{\bm{M}}^{1/2}_\nu)$ are in terms the $U_{\bm q}(sl_2)\otimes U_{\wt{\bm q}}(sl_2)$ generators $\ck,\ck^{-1},\cx,\cy,\widetilde\ck,\widetilde\ck^{-1},\widetilde\cx,\widetilde\cy$ defined on $\ch_0^{\otimes m}$ larger than $\ch_{\l_1,\cdots,\l_m}$. However, $\cd^{\vec \l}(\bm{M}^{1/2}_\nu)$ and $\cd^{\vec \l}(\widetilde{\bm{M}}^{1/2}_\nu)$ leave $\ch_{\l_1,\cdots,\l_m}$ invariant: First, the factors $(R_{0\mu}R_{0\mu}^\prime)^{1/2,\l}$ and $(\widetilde{R}_{0\mu}\widetilde{R}_{0\mu}^\prime)^{1/2,\l}$ have been given by \eqref{MandUq} and \eqref{MandUqtilde} and only contained the quadratic monomials of the generators. Given any operator $(\mathbf{X}_{0,\mu_1,\cdots,\mu_n})^{1/2,\l_{\mu_1},\cdots,\l_{\mu_n}}$ as a $2\times 2$ matrix whose elements are some linear combinations of order-even monomials of $\ck,\ck^{-1},\cx,\cy$ (including order-0), for any $\nu< \mu_1 <\cdots<\mu_n$, the matrix elements of $[R'_{0\nu}\mathbf{X}_{0,\mu_1,\cdots,\mu_n}R^{\prime -1}_{0\nu}]^{1/2,\l_\nu,\l_{\mu_1},\cdots,\l_{\mu_n}}$:
\be
&&\mathbf{X}^1_{\ 1} + (\bfq^{-1/2}-\bfq^{3/2})\mathcal{K}_\nu \mathcal{X}_\nu \mathbf{X}^1_{\ 2}, \qquad  \mathcal{K}^2_\nu \mathbf{X}^1_{\ 2}\nonumber \\
&&\mathcal{K}^{-2}_\nu{\mathbf{X}^2_{\ 1}}+ ({\bfq}^{1/2} -  \bfq^{-3/2}) \mathcal{X}_\nu\mathcal{K}^{-1}_\nu\mathbf{X}^1_{\ 1} +   ({\bfq}^{-1/2} -  \bfq^{3/2}) \mathcal{K}^{-1}_\nu\mathcal{X}_\nu \mathbf{X}^2_{\ 2} -(\bfq -\bfq^{-1})^2 \mathcal{X}_\nu^2\mathbf{X}^1_{\ 2}, \nonumber\\
&& \mathbf{X}^2_{\ 2}+(\bfq^{1/2}-\bfq^{-3/2})  \mathcal{X}_\nu\mathcal{K}_\nu\mathbf{X}^1_{\ 2}
\ee
are still the linear combinations of order-even monomials. Then by recursion, all ${\cal D}^{\vec{\lambda}}(\bm{M}_\nu^{1/2})$ can be shown to only contain order-even monomials of the generators, and the same for ${\cal D}^{\vec{\lambda}}(\wt{\bm{M}}_\nu^{1/2})$. As a result, all these operators are made by $\bm{u}^{(\nu)}_{\a,\b}$ and $\wt{\bm{u}}^{(\nu)}_{\a,\b}$ and thus leaves $\ch_{\l_1,\cdots,\l_m}$ invariant.

\section{Gauge invariance and flatness constraint}\label{Gauge invariance and flatness constraint}

By the discussion in Section \ref{Infinite-dimensional representation of suquqt}, we define the $*$-representation of gauge transformations $\cd^{\vec\l}(\xi)$ acting on states in $\ch_{\vec \l}$ by the tensor product representation of $\suquqt$:
\be 
\cd^{\vec\l}(\xi)=\lt(\pi^{\l_1}\otimes\cdots\otimes\pi^{\l_m}\rt)\Delta^{(m)}(\xi),\qquad \forall\, \xi\in\suquqt.
\ee
Now $\cd^{\vec\l}$ gives $*$-representations to both the graph algebra $\cl_{0,m}\otimes\widetilde{\cl}_{0,m}$ and gauge transformations $\suquqt$. It turns out that this two representations are coherently combined, and $\cd^{\vec \l}$ can be upgraded to the representation of the extended graph algebra $\Fs_{0,m}\otimes\widetilde{\Fs}_{0,m}$ defined in Section \ref{Semi-direct product}. 

\begin{lemma} 
$\cd^{\vec \l}$ extended to $\Fs_{0,m}\otimes\widetilde{\Fs}_{0,m}$ is consistent with the multiplication rules \eqref{xileftmultipli1} and \eqref{xileftmultipli2}.
\end{lemma}

\begin{proof} 
It is sufficient to check the consistency with the relations in Lemma \ref{muMMmurelations}, since they are equivalent to the multiplication rules of $\Fs_{0,m}\otimes\widetilde{\Fs}_{0,m}$. Let us focus on the relation for $\bm{M}_{\nu}^{I}$:
\be 
\cd^{\vec\l}\lt(\mu^{I}(\xi)\rt)\cd^{\vec\l}\lt(\bm{M}_{\nu}^{I}\rt)&=&\lt(\rho^{I}\otimes \cd^{\vec{\l}}\rt)\lt(\Delta\xi\rt)\cd^{\vec{\l}}\lt(\bm{M}_{\nu}^{I}\rt)
=\lt(\Delta^{m+1}\xi\rt)^{I\vec{\l}}\lt[\Delta_{1}^{(\nu-1)}\left(R_{x1}^{\prime}\bm{M}_{x\nu}R_{x1}^{\prime-1}\right)\rt]^{I\vec{\l}}\nonumber\\
&=&\lt[\Delta_{1}^{(\nu-1)}\left(\Delta^{(m-\nu+1)}(\xi)R_{x1}^{\prime}\bm{M}_{x\nu}R_{x1}^{\prime-1}\right)\rt]^{I\vec{\l}}
\ee
We have labelled the tensor slots corresponding to $I,\l_1,\cdots,\l_m$ by $x,1,\cdots, m$, so we write $\cd^{\vec{\l}}\left(\bm{M}_{\nu}^{I}\right)=[\Delta_{1}^{(\nu-1)}\left(R_{x1}^{\prime}\bm{M}_{x\nu}R_{x1}^{\prime-1}\right)]^{I\vec{\l}}$ where $\Delta_{1}^{(\nu-1)}=\mathrm{id}_{x}\otimes\Delta^{(\nu-1)}\otimes\mathrm{id}_\nu$. Skip all representation labels and use the following relations:
\be 
\Delta^{(m-\nu+1)}(\xi)&=&\sum_{\xi}\xi_{x}^{(0)}\otimes\xi_{1}^{(1)}\otimes\xi_{\nu}^{(\nu)}\otimes\xi_{\nu+1}^{(\nu+1)}\otimes\cdots\otimes\xi_{m}^{(m)}\nonumber\\
&=&\sum_{\xi}\Delta\left(\xi_{1}^{(1)}\right)\otimes\xi_{\nu}^{(\nu)}\otimes\xi_{\nu+1}^{(\nu+1)}\otimes\cdots\otimes\xi_{m}^{(m)},\qquad\Delta\left(\xi^{(1)}\right)=\sum_{\xi}\xi^{(0)}\otimes\xi^{(1)}\\
&=&\sum_{\xi}\xi_{x}^{(0)}\otimes\Delta\left(\xi_{\nu}^{(\nu)}\right)\otimes\xi_{\nu+1}^{(\nu+1)}\otimes\cdots\otimes\xi_{m}^{(m)},\qquad\Delta\left(\xi^{(\nu)}\right)=\sum_{\xi}\xi^{(1)}\otimes\xi^{(\nu)}
\ee
and $R^{\prime}\left(\Delta^{\prime}\xi\right)=\left(\Delta\xi\right)R^{\prime}$, 
we derive 
\be
\cd^{\vec\l}\lt(\mu^{I}(\xi)\rt)\cd^{\vec\l}\lt(\bm{M}_{\nu}^{I}\rt)&=&\sum_{\xi}\Delta_{1}^{(\nu-1)}\left(R_{x1}^{\prime}\left(\xi_{x}^{(1)}\xi_{\nu}^{(\nu)}\right)\bm{M}_{x\nu}\left(\xi_{1}^{(0)}\xi_{\nu+1}^{(\nu+1)}\cdots\xi_{m}^{(m)}\right)R_{x1}^{\prime-1}\right)\nonumber\\
&=&\sum_{\xi}\Delta_{1}^{(\nu-1)}\left(R_{x1}^{\prime}\Delta\left(\xi^{(\nu)}\right)\bm{M}_{x\nu}\left(\xi_{1}^{(0)}\xi_{\nu+1}^{(\nu+1)}\cdots\xi_{m}^{(m)}\right)R_{x1}^{\prime-1}\right).
\ee
Moreover, we derive the relation $\Delta\left(\xi^{(\nu)}\right)\bm{M}_{x\nu}=\Delta\left(\xi^{(\nu)}\right)R_{x\nu}^{\prime}R_{x\nu}=R_{x\nu}^{\prime}\Delta^{\prime}\left(\xi^{(\nu)}\right)R_{x\nu}=R_{x\nu}^{\prime}R_{x\nu}\Delta\left(\xi^{(\nu)}\right)$ (this means that $\cd^\l$ is a representation of the extended loop algebra $\Fs_{0,1}\otimes\Fs_{0,1}$). Applying this relation, we obtain 
\be 
\cd^{\vec\l}\lt(\mu^{I}(\xi)\rt)\cd^{\vec\l}\lt(\bm{M}_{\nu}^{I}\rt)&=&\sum_{\xi}\Delta_{1}^{(\nu-1)}\left(R_{x1}^{\prime}\bm{M}_{x\nu}\left(\xi_{x}^{(1)}\xi_{1}^{(0)}\xi_{\nu}^{(\nu)}\xi_{\nu+1}^{(\nu+1)}\cdots\xi_{m}^{(m)}\right)R_{x1}^{\prime-1}\right)\nonumber\\
&=&\sum_{\xi}\Delta_{1}^{(\nu-1)}\left(R_{x1}^{\prime}\bm{M}_{x\nu}R_{x1}^{\prime-1}\left(\xi_{x}^{(0)}\xi_{1}^{(1)}\xi_{\nu}^{(\nu)}\xi_{\nu+1}^{(\nu+1)}\cdots\xi_{m}^{(m)}\right)\right)\nonumber\\
&=&\Delta_{1}^{(\nu-1)}\left(R_{x1}^{\prime}\bm{M}_{x\nu}R_{x1}^{\prime-1}\right)\left(\Delta^{(m+1)}\xi\right).
\ee
Recovering the representation labels, we obtain the result 
\be 
\cd^{\vec\l}\lt(\mu^{I}(\xi)\rt)\cd^{\vec\l}\lt(\bm{M}_{\nu}^{I}\rt)=\cd^{\vec\l}\lt(\bm{M}_{\nu}^{I}\rt)\cd^{\vec\l}\lt(\mu^{I}(\xi)\rt).
\ee
The consistency with other relations in Lemma \ref{muMMmurelations} can be derived analogously. 

\end{proof}

The physical states of the theory are gauge invariant, i.e. $\cd^{\vec \l}(\xi)$ acting on the physical states should be equivalent to the trivial representation:
\begin{eqnarray}
	\cd^{\vec \l}(\xi)\Psi = \eps(\xi)\Psi\ ,\label{gaugeinv}
\end{eqnarray}
for all $\xi$ belonging to the quantum gauge group. It turns out that $\Psi$ satisfying the gauge invariance does not belong to $\ch_{\vec \l}$ but rather belongs to the space of linear functionals on a dense domain in $\ch_{\vec \l}$. The proper formulation of the gauge invariant state $\Psi$ is given by
\begin{eqnarray}
	\Psi\lt[\cd^{\vec \l}(\xi)f\rt] = \eps(\xi)\Psi\lt[f\rt],\label{gaugeinv2}
\end{eqnarray}
for any $f$ belonging to the dense domain.

On an $m$-holed sphere, the defining relation of the moduli space $\cm_{0,m}\times\overline{\cm}_{0,m}$ is the flatness constraint
\begin{eqnarray}
M^I_0=M^I_m M^I_{m-1}\cdots M^I_{2} M^I_{1}=1,\qquad \widetilde{M}^I_0=\widetilde{M}^I_m \widetilde{M}^I_{m-1}\cdots \widetilde{M}^I_{2} \widetilde{M}^I_{1}=1\ .
\end{eqnarray} 
$M_0^I,\widetilde{M}_0^I$ are the holonomies along the largest circle enclosed all holes. They have the quantization as elements in $\Fs_{0,m}\otimes\widetilde{\Fs}_{0,m}$ by
\begin{eqnarray}
	\bm M^I_0=\bm M^I_m \bm M^I_{m-1}\cdots \bm M^I_{2} \bm M^I_{1},\qquad \widetilde{\bm M}^I_0=\widetilde{\bm M}^I_m \widetilde{\bm M}^I_{m-1}\cdots \widetilde{\bm M}^I_{2} \widetilde{\bm M}^I_{1}\ .
\end{eqnarray} 

Assigning $\l_1,\cdots,\l_m$ to $m$ holes, the representation $\cd^{\vec \l}(\bm M^I_0)$ on $\ch_{\vec \l}$ is given by
\begin{eqnarray}
\cd^{\vec \l}\lt(\bm M^I_0\rt)
&=&{\cal D}^{\vec{\lambda}}\left(\bm{M}_{m}^{I}\right)\cdots{\cal D}^{\vec{\lambda}}\left(\bm{M}_{2}^{I}\right){\cal D}^{\vec{\lambda}}\left(\bm{M}_{1}^{I}\right)\nonumber\\
&=&\lt(R_{x1}^{\prime}R_{x2}^{\prime}\cdots R_{xm}^{\prime}R_{xm}\cdots R_{x2}R_{x1}\rt)^{I\vec{\l}}\nonumber\\
&=&\lt(\rho^I\otimes\pi^{\l_1}\otimes\cdots\otimes\pi^{\l_m}\rt)\Delta_{1}^{(m)}\left(R_{x1}^{\prime}R_{x1}\right),
\end{eqnarray}
where we label the tensor slots corresponding to $I,\l_1,\cdots,\l_m$ by $x,1,\cdots, m$. Similarly,
\begin{eqnarray}
\cd^{\vec \l}\lt(\widetilde{\bm M}^I_0\rt)=\lt(\rho^I\otimes\pi^{\l_1}\otimes\cdots\otimes\pi^{\l_m}\rt)\Delta_{1}^{(m)}\left(\widetilde{R}_{x1}^{\prime}\widetilde{R}_{x1}\right). 
\end{eqnarray}

The physical state of the theory needs to solve the quantum flatness constraint: $\cd^{\vec \l}(\bm M^I_0)\Psi=\cd^{\vec \l}(\widetilde{\bm M}^I_0)\Psi=\Psi$. But this is implied by
the gauge invariance \eqref{gaugeinv2}: For any gauge invariant state $\Psi$, we have
\begin{eqnarray}
	\Psi\lt[\cd^{\vec \l}(\bm M^I_0)f\rt] &=& \lt(\rho^I\otimes\eps\rt)\lt(R_{x1}^{\prime}R_{x1}\rt)\Psi\lt[f\rt]=\rho^I(\mathbf{1})\Psi\lt[f\rt],\\
	\Psi\lt[\cd^{\vec \l}(\widetilde{\bm M}^I_0)f\rt] &=& \lt(\rho^I\otimes\eps\rt)\lt(\widetilde{R}_{x1}^{\prime}\widetilde{R}_{x1}\rt)\Psi\lt[f\rt]=\rho^I(\mathbf{1})\Psi\lt[f\rt].
\end{eqnarray}
Therefore, the physical states are selected only by imposing the gauge invariance.

\section{Physical Hilbert space}\label{Physical Hilbert space}

In this section, we derive the physical states $\Psi$ as general solutions to \eqref{gaugeinv2}. We show that the physical states form a Hilbert space, which relate to a fiber Hilbert space in a direct integral decomposition of $\ch_{\vec \l}$. As we see below, the direct integral decomposition of $\ch_{\vec \l}$ results from the Clebsch-Gordan decomposition of the tensor product representations of $\suquqt$.

\subsection{Clebsch-Gordan decomposition of $\pi^{\lambda_1}\otimes\pi^{\l_2}$}

The quadratic Casimirs in $\suquqt$ is given by
\be 
Q=-\left(\boldsymbol{q}-\boldsymbol{q}^{-1}\right)^{2}EF-\boldsymbol{q}^{-1}K-\boldsymbol{q}K^{-1},\\
\widetilde{Q}=-\left(\widetilde{\boldsymbol{q}}-\widetilde{\boldsymbol{q}}^{-1}\right)^{2}\widetilde{E}\widetilde{F}-\widetilde{\boldsymbol{q}}^{-1}\widetilde{K}-\widetilde{\boldsymbol{q}}\widetilde{K}^{-1},
\ee
satisfying $Q^* =\widetilde{Q}$. Given the representation $\pi^\l$ carried by $\ch=L^2(\R)\otimes \C^N$, we obtain 
\be 
Q_{\lambda}=\left(\lambda+\lambda^{-1}\right)\mathrm{id}_{\lambda},\qquad
\widetilde{Q}_{\lambda}=\left(\overline{\lambda}+\overline{\lambda}^{-1}\right)\mathrm{id}_{\lambda}.
\ee

We denote by $Q_{12},\widetilde{Q}_{12}$ the representation of the quadratic Casimir in the tensor product representation: $Q_{12}=(\pi^{\lambda_1}\otimes\pi^{\l_2})\Delta Q$ and $\widetilde{Q}_{12}={Q}_{12}^\dagger$:
\be 
Q_{12}
&=&K_{1}^{-1}Q_{2}+Q_{1}K_{2}+\left(\boldsymbol{q}+\boldsymbol{q}^{-1}\right)K_{1}^{-1}K_{2}+\left(Q_{1}+\boldsymbol{q}^{-1}K_{1}+\boldsymbol{q}K_{1}^{-1}\right)F_{1}^{-1}K_{1}^{-1}K_{2}F_{2}\nonumber\\
&&+\,F_{1}\left(Q_{2}+\boldsymbol{q}^{-1}K_{2}+\boldsymbol{q}K_{2}^{-1}\right)F_{2}^{-1},
\ee
where we use the notations e.g. $\ck_i=\ck_{\l_i}=\pi^{\l_i}(\ck)$ and $Q_i=(\l_i+\l_i^{-1})\mathrm{id}$. The tensor product representation is carried by $\ch\otimes\ch$.

The tensor product representation $\pi^{\lambda_1}\otimes\pi^{\l_2}$ is carried by $\ch\otimes\ch$. We define a unitary transformation $\cu_{12}=C_2\cs_{2}^{-1} t_{12}^{-1}:\ \ch\otimes\ch\to \ch\otimes\ch$, where
\be\label{U12}
t_{12}&=&\varphi\left(q^{2}\bm{u}_{1}\bm{y}_{1}^{-1}\bm{y}_{2},\widetilde{q}^{2}\widetilde{\bm{u}}_{1}\widetilde{\bm{y}}_{1}^{-1}\widetilde{\bm{y}}_{2}\right),\qquad {\cal S}_{2}=e^{\frac{2\pi i}{N}\left(\bm{\mu}_{1}\bm{\nu}_{2}-\bm{m}_{1}\bm{n}_{2}\right)},\nonumber\\
C_{2}&=&\left[\varphi\left(\lambda_{1}\bm{y}_{2},\overline{\lambda}_{1}\widetilde{\bm{y}}_{2}\right)\frac{\varphi\left(\lambda_{2}\bm{u}_{2}^{-1},\overline{\lambda}_{2}\widetilde{\bm{u}}_{2}^{-1}\right)}{\varphi\left(\lambda_{2}\bm{u}_{2},\overline{\lambda}_{2}\widetilde{\bm{u}}_{2}\right)}\right]^{-1}
\ee
are unitary operators on $\ch\otimes\ch$. $\varphi(y,\tilde{y})$ is the quantum dilogarithm function (see Appendix A in \cite{Han:2024nkf} for some details)
\[
\varphi(y,\wt{y})=\left[\frac{\prod_{j=0}^{\infty}\left(1+{\bfq}^{2j+1}y\right)}{\prod_{j=0}^{\infty}\left(1+\wt{{\bfq}}^{-2j-1}\wt{y}\right)}\right]^{-1}.
\]
By the unitary transformation, 
\be 
Q_{2}^{\prime\prime}&=&\mathcal{U}_{12}Q_{12}\mathcal{U}_{12}^{-1}=\lambda_{1}\bm{u}_{2}^{-1}+\lambda_{1}^{-1}\bm{u}_{2}+\lambda_{2}^{-1}\bm{y}_{2}^{-1},\label{Q2pp1}\\
\widetilde{Q}_{2}^{\prime\prime}&=&\mathcal{U}_{12}\widetilde{Q}_{12}\mathcal{U}_{12}^{-1}=\overline{\lambda}_{1}\widetilde{\bm{u}}_{2}^{-1}+\overline{\lambda}_{1}^{-1}\widetilde{\bm{u}}_{2}+\overline{\lambda}_{2}^{-1}\widetilde{\bm{y}}_{2}^{-1} \label{Q2pp2}
\ee
only act on the second factor of $\ch\otimes \ch$ \cite{Han:2024nkf}. The tensor product representations of the generators $\left(\Delta{F}\right)_{12}\equiv \lt(\pi^{\lambda_1}\otimes\pi^{\l_2}\rt)\lt(\Delta F\rt)$ and $\left(\Delta{K}^{\pm 1}\right)_{12}\equiv \lt(\pi^{\lambda_1}\otimes\pi^{\l_2}\rt)\lt(\Delta K^{\pm1}\rt)$ are transformed to only act on the first factor of $\ch\otimes \ch$ \cite{Han:2024nkf}
\be 
\cu_{12}\left(\Delta F\right)_{12}\cu_{12}^{-1}=F_1,\qquad 
\cu_{12}\left(\Delta{ K}^{\pm1}\right)_{12}\cu_{12}^{-1}=K^{\pm1}_1, \label{DeltaYK12}
\ee
while $\left(\Delta E\right)_{12}\equiv\lt(\pi^{\lambda_1}\otimes\pi^{\l_2}\rt)\lt(\Delta E\rt)$ transforms to 
\be
\cu_{12}\left(\Delta E\right)_{12}\cu_{12}^{-1}=\frac{\bfq^{-1}K_1+\bfq K^{-1}_1+Q''_2}{(\bfq-\bfq^{-1})^{2}}F_1^{-1},
\ee
which is $E_\l$ in \eqref{E354} with the substitution $(\l+\l^{-1})\mathrm{id}_\l\to Q''_2$. Similar for the tilded sector
\be 
\cu_{12}\left(\Delta\widetilde{F}\right)_{12}\cu_{12}^{-1}&=&\wt{F}_1,\qquad 
\cu_{12}\left(\Delta\widetilde{K}^{\pm1}\right)_{12}\cu_{12}^{-1}=\wt{K}^{\pm1}_1,\\
\cu_{12}\left(\Delta \wt{E}\right)_{12}\cu_{12}^{-1}&=&\frac{\wt{\bfq}^{-1}\wt K_1+\wt{\bfq} \wt K^{-1}_1+\wt Q''_2}{(\wt{\bfq}-\wt{\bfq}^{-1})^{2}}\wt{F}_1^{-1}.\label{DeltaXtilde12}
\ee
We denote the resulting representation by $\pi^{Q_2''}(\xi)\equiv\cu_{12}(\pi^{\l_1}\otimes\pi^{\l_2})(\Delta\xi)\cu_{12}^{-1}$, in which $\bm{u}_{\a,\b},\wt{\bm u}_{\a,\b}$ act on the 1st factor of $\ch\otimes\ch$, while $Q_2'',\wt Q_2''$ act on the 2nd factor. When $Q_2'',\wt Q_2''$ is diagonalized, $\pi^{Q_2''}(\xi)$ reduces to the irreducible representation at each eigenspace. 

$\css_{12}\in\ch\otimes\ch$ denotes the maximal common dense domain of $(\Delta \xi)_{12}$ where $\xi$ is any polynomial of $E,F,K^{\pm1}$ and their tilded partners, and $\css_{12}''\in\ch\otimes\ch$ denotes the maximal common dense domain of polynomials of $E_1,F_1,K_1^{\pm1},Q_2''$ and the tilded partners. $\css_{12}$ and $\css_{12}''$ are invariant by the actions of $(\pi^{\l_1}\otimes\pi^{\l_2})(\Delta \xi)$ and $\pi^{Q''_2}(\xi)$ respectively. The unitary map $\cu_{12}$ is a bijection of these two domains \cite{Han:2024nkf}:
\be
\cu_{12}:\ \css_{12}\to \css_{12}''.
\ee
Both spaces are Fr\'echet with semi-norms defined by $\Vert \bm{B}f\Vert$, where $\bm{B}$ is a basis of the operator algebra.

As to be discussed in Section \ref{Spectral decompositions} (see also \cite{Han:2024nkf}), $Q_2'',\widetilde{Q}_2''$ has a common generalized eigenvector $\a_\chi$ in the space of distributions dual to $\Fd$: $Q_2''\a_\chi=(\chi+\chi^{-1})\a_\chi$, $\widetilde{Q}_2''\a_\chi=(\overline{\chi}+\overline{\chi}^{-1})\a_\chi$. The representation $\mathcal{U}_{12}\left(\Delta\xi\right)_{12}\mathcal{U}_{12}^{-1}$ on $\ch\otimes\a_\chi$ is identical to the irreducible representation $\pi^\chi$ for all $\xi\in\suquqt$. Indeed, $Q_2''$ is a normal operator, and $Q_2''^\dagger=\widetilde{Q}_2''$. Their spectral decomposition on $\ch$ gives 
\be 
Q_2''=\int_\C\lt(\chi+\chi^{-1}\rt)\rmd P_\chi,\qquad \widetilde{Q}_2''=\int_\C\lt(\overline\chi+\overline\chi^{-1}\rt)\rmd P_\chi.
\ee
The projection-valued measure $\rmd P_\chi$ can be written as
\be
\rmd P_\chi=\rmd\varrho_\chi|\a_\chi\rangle \langle\a_\chi|\ .\label{pvm}
\ee
The spectral measure $\rmd\varrho_\chi$ and the generalized eigenstate $\a_\chi$ will be clarified in Section \ref{Spectral decompositions}. Correspondingly, $\ch$ has the direct integral representation
\be 
\ch\simeq\int_{\mathbb{C}}^{\oplus}\rmd\varrho_\chi\,W_\chi,\label{DIDHW}
\ee
where $\mu$ is the spectral measure (note that $d\mu(\chi)$ is not holomorphic). The fiber Hilbert space $W_\chi$ is the generalize eigenspace of $Q_2''$ and $\widetilde{Q}_2''$, where they reduce to $(\chi+\chi^{-1})\mathrm{id}_{W_\chi}$ and $(\overline\chi+\overline\chi^{-1})\mathrm{id}_{W_\chi}$ \footnote{One may worry that $W_\chi$ is not well-define since the integrand is only well-defined up to measure-zero. However, when we specify the Gelfand triple including $\ch$, the dense subspace $\overline{\Fd}\subset\ch$ and the topological dual $\overline{\Fd}'$, the generalized eigenvectors of certain eigenvalue can be found in the topological dual, then $W_\chi$ is the dual space of the space of generalized eigenvectors. See e.g. \cite{GelfandVol4} for details.}. Applying this decomposition, we obtain the following decomposition of the Hilbert space carrying the tensor product representation:
\be 
\ch_{\l_1}\otimes\ch_{\l_2}\simeq \int_{\mathbb{C}}^{\oplus}\rmd\varrho_\chi\,\ch_{\chi}\otimes W_\chi.
\ee 
We have added the representation labels in order to indicate the representation carried by $\ch$. The corresponding decomposition of the tensor product representation is given by
\be
\cu_{12}\lt[\lt(\pi^{\l_1}\otimes\pi^{\l_2}\rt)\lt(\Delta\xi\rt)\rt]\cu_{12}^{-1}=\int_{\mathbb{C}}\rmd\varrho_\chi\,\pi^\chi(\xi)\otimes |\a_\chi\rangle\langle\a_\chi|.\label{CGdecomp}
\ee
This is the Clebsch-Gordan (CG) decomposition of $\pi^{\l_1}\otimes\pi^{\l_2}$, and the unitary transformation $\cu_{12}$ is the CG map.

\subsection{Diagonalizing the quadratic Casimirs} \label{Spectral decompositions}

As discussed in the last section, a key step of the CG decomposition is the spectral decompositions of the following operators on $\ch\simeq L^2(\R)\otimes \C^N$
\be 
Q_{2}^{\prime\prime}&=&\lambda_{1}\bm{u}^{-1}+\lambda_{1}^{-1}\bm{u}+\lambda_{2}^{-1}\bm{y}^{-1},\\
\widetilde{Q}_{2}^{\prime\prime}&=&\overline{\lambda}_{1}\widetilde{\bm{u}}^{-1}+\overline{\lambda}_{1}^{-1}\widetilde{\bm{u}}+\overline{\lambda}_{2}^{-1}\widetilde{\bm{y}}^{-1}.
\ee
The representation label $\l\in\C^\times$ is arbitrary in the above discussion. From now on, we further assume\footnote{$\l+\l^{-1}$ is the quantum trace of $\cd^{\l}(\bm{M}^{1/2})$. When considering the Chern-Simons theory on a graph complement 3-manifold (see e.g. \cite{Han:2021tzw}), the hole on the $m$-holed sphere is connected to an annulus where $\l$ is one of the canonical variable of the phase space, then \eqref{assumption2N} is the quantization of $\l$.}
\be  
\l_a&=&\exp\lt[\frac{2\pi i}{N}(-ib\mu_a-m_a)\rt],\quad \mu_a\in\R,\ m_a\in\Z/N\Z,\label{assumption2N}
\ee 
so that the following operators are unitary
\be 
\mathcal{S}_{\lambda_{2}}f\left(\mu,m\right)=f\left(\mu-\mu_{2},m-m_{2}\right),\qquad {\mathcal{D}}_{\lambda_{1}}f\left(\mu,m\right)=e^{-\frac{2\pi i}{N}(\mu_1\mu-m_1m)}f\left(\mu,m\right).
\ee
These operators satisfy $\mathcal{S}_{\lambda_{2}}\mathcal{D}_{\lambda_{1}}=e^{\frac{2\pi i}{N}(\mu_{1}\mu_{2}-m_{1}m_{2})}\mathcal{D}_{\lambda_{1}}\mathcal{S}_{\lambda_{2}}$ and 
\be 
&&\mathcal{S}_{\lambda_{2}}\bm{y}^{-1}\mathcal{S}_{\lambda_{2}}^{-1}=\lambda_{2}\bm{y}^{-1},\quad \mathcal{S}_{\lambda_{2}}\bm{u}\mathcal{S}_{\lambda_{2}}^{-1}=\bm{u},\quad \mathcal{S}_{\lambda_{2}}\widetilde{\bm{y}}^{-1}\mathcal{S}_{\lambda_{2}}^{-1}=\overline{\lambda}_{2}\widetilde{\bm{y}}^{-1},\quad \mathcal{S}_{\lambda_{2}}\widetilde{\bm{u}}\mathcal{S}_{\lambda_{2}}^{-1}=\widetilde{\bm{u}},\\
&&\mathcal{D}_{\lambda_{1}}^{-1}\bm{y}\mathcal{D}_{\lambda_{1}}=\bm{y},\qquad \mathcal{D}_{\lambda_{1}}^{-1}\bm{u}\mathcal{D}_{\lambda_{1}}=\lambda_{1}\bm{u},\qquad \mathcal{D}_{\lambda_{1}}^{-1}\widetilde{\bm{y}}\mathcal{D}_{\lambda_{1}}=\widetilde{\bm{y}},\qquad \mathcal{D}_{\lambda_{1}}^{-1}\widetilde{\bm{u}}\mathcal{D}_{\lambda_{1}}=\overline{\lambda}_{1}\widetilde{\bm{u}}.
\ee
The Fourier transformation
\be  
\mathcal{F}f(\mu,m)=\frac{1}{N}\sum_{m^{\prime}\in\mathbb{Z}/N\mathbb{Z}}\int d\mu^{\prime}e^{\frac{2\pi i}{N}\left(\mu\mu^{\prime}-mm^{\prime}\right)}f\left(\mu^{\prime},m^{\prime}\right)
\ee
satisfies
\be 
\cf\bm{u}\cf^{-1} =\bmy,\quad \cf\bm{y}^{-1}\cf^{-1} =\bm{u},\quad \cf\widetilde{\bm{u}}\cf^{-1} =\widetilde{\bmy},\quad \cf\widetilde{\bm{y}}^{-1}\cf^{-1} =\widetilde{\bm{u}}.
\ee
Applying these unitary transformations to $Q''_2$ and $\widetilde{Q}_2''$ leads to the simplification
\be 
\bm{L}\equiv{\cal F}\mathcal{D}_{\lambda_{1}}^{-1}\mathcal{S}_{\lambda_{2}}Q_{2}^{\prime\prime}\mathcal{S}_{\lambda_{2}}^{-1}\mathcal{D}_{\lambda_{1}}{\cal F}^{-1}&=&\bm{y}^{-1}+\bm{y}+\bm{u},\\
\widetilde{\bm{L}}\equiv{\cal F}\mathcal{D}_{\lambda_{1}}^{-1}\mathcal{S}_{\lambda_{2}}\widetilde Q_{2}^{\prime\prime}\mathcal{S}_{\lambda_{2}}^{-1}\mathcal{D}_{\lambda_{1}}{\cal F}^{-1}&=&\widetilde{\bm{y}}^{-1}+\widetilde{\bm{y}}+\widetilde{\bm{u}}.
\ee
These pair of operators are generalizations to the Dehn-twist operators in quantum Teichm\"uller theory \cite{Kashaev:2000ku}, and they recover the Dehn-twist operators when $N=1$. The Dehn-twist operator appeared to be unitary equivalent to the Casimir operator of the tensor product representations for the modular double of $U_q(sl(2,\R))$ \cite{Ponsot:2000mt,Derkachov:2013cqa,Nidaiev:2013bda}.

The eigenfunctions of $\bmL,\wt{\bmL}$ relate to the quantum dilogarithm function
\be
\g(-x,n)&=&\prod_{j=0}^{\infty}\frac{1-{\bfq}^{2j+1}\exp\left[\frac{2\pi i}{N}\left(-ibx\sqrt{N}+n\right)\right]}{1-\widetilde{{\bfq}}^{-2j-1}\exp\left[\frac{2\pi i}{N}\left(-ib^{-1}x\sqrt{N}-n\right)\right]},
\ee
where ${\bfq}=e^{\frac{\pi i}{N}(1+b^{2})}$ and $\widetilde{{\bfq}}=e^{\frac{\pi i}{N}(1+b^{-2})}$. $\g(-x,n)=\mathrm{D}_b(x,n)$ is the quantum dilogarithm over $\R\times \Z/N\Z$ defined by Andersen and Kashaev in \cite{Andersen2014} (see also \cite{andersen2016level}). For $N=1$, $\g(x,n)=\g(x,0)\equiv\g(x)$ is Faddeev's quantum dilogarithm in e.g. \cite{Derkachov:2013cqa,Faddeev:1995nb}:
\be
\g(x)=\exp\lt[\frac{1}{4}\int_{\R+i0^+}\frac{\rmd w}{w}\frac{e^{2 i w x} }{\sinh \left(b^{-1}{w}\right) \sinh (b w)}\rt]
\ee
The relation with the quantum dilogarithm $\varphi$ used above is 
\be
\varphi\lt(-y,-\tilde{y}\rt)=\g\lt(-\frac{\mu}{\sqrt{N}},-m\rt)^{-1}.
\ee
We introduce some short-hand notations that are useful below
\be
c_b=\frac{i}{2}(b+b^{-1}),\quad \omega = \frac{i}{2 b \sqrt{N}},\quad
\o'= \frac{i b}{2 \sqrt{N}},\quad
\o''= \frac{c_b}{\sqrt{N}},\quad
x = \frac{\mu }{\sqrt{N}},\quad
\lambda =\frac{\mu_\chi}{\sqrt{N}}.
\ee
Some useful properties of $\gamma(x,n)$ including the inverse relation, recursion relation, and integral identity can be found in \cite{Han:2024nkf,andersen2016level}.

The following function
\be
\psi_\chi(\mu ,m)&=&\exp \left(\frac{i \pi  m^2}{N}+\frac{i \pi  m_\chi^2}{N}+i \pi  m_\chi^2-i \pi  m-i \pi  \left(x-\omega ''+i \epsilon \right)^2\right)\nonumber\\
&&\gamma \left(-\lambda +x-\omega ''+i \epsilon ,m-m_\chi\right) \gamma \left(\lambda +x-\omega ''+i \epsilon ,m_\chi+m\right),\label{psichimumexpandgamma}
\ee
satisfies the eigen-equations as $\epsilon\to 0$ \cite{Han:2024nkf}
\be 
&&\bmL \psi_\chi(\mu ,m)=\lt(\chi+\chi^{-1}\rt)\psi_\chi(\mu ,m),\qquad \widetilde{\bmL}\psi_\chi(\mu ,m)=\lt(\overline{\chi}+\overline{\chi}^{-1}\rt)\psi_r(\mu ,m),\\
&& \quad \chi=\exp\lt[\frac{2\pi i}{N}(-ib\mu_\chi-m_\chi)\rt],\qquad\quad \overline{\chi}=\exp\lt[\frac{2\pi i}{N}(-ib^{-1}\mu_\chi+m_\chi)\rt].\label{chichibar}
\ee
$\psi_\chi(\mu ,m)$ is invariant manifestly under $(\mu_\chi,m_\chi)\to(-\mu_\chi,-m_\chi)$ and periodic under $m\to m+N$, $m_\chi\to m_\chi+N$. When $N=1$, $m=m_\chi=0$, $\psi_r$ reduces to Kashaev's eigenfunctions \cite{Kashaev:2000ku} ($\phi(x,\l)$ in \cite{Derkachov:2013cqa}) for the Dehn-twist operator. 

A regulator $\epsilon$ has been implemented in $\psi_\chi$. For $\epsilon>0$, $\psi_\chi$ is a meromorphic function with all simple poles in lower-half plane (a pair of poles approach to the real line when $\epsilon\to0$):
\begin{eqnarray}
	\frac{\mu}{\sqrt{N}}=\pm\frac{\mu_{\chi}}{\sqrt{N}}-i\epsilon-i\frac{b^{-1}}{\sqrt{N}}l_{p}-i\frac{b}{\sqrt{N}}m_{p},\qquad m_{p}-l_{p}=m\mp m_{\chi}+N\mathbb{Z},\\
	l_{p},m_{p}\in\mathbb{Z},\qquad l_{p},m_{p}\geq0.\nonumber
\end{eqnarray}
Additionally, $\psi_\chi$ has the following asymptotic behavior as $\mathrm{Re}(\mu)\to\pm\infty$: 
\begin{eqnarray}
\psi_{\chi}\left(\mu,m\right)\sim \begin{cases}
		e^{-\frac{i\pi}{N}\mu^{2}}e^{\frac{2\pi}{\sqrt{N}}\left(-\frac{\mathrm{Re}(b)}{\sqrt{N}}+\epsilon\right)\mu} , & \mathrm{Re}(\mu)\to\infty,\\
		e^{\frac{i\pi}{N}\mu^{2}}e^{\frac{2\pi}{\sqrt{N}}\left(\frac{\mathrm{Re}(b)}{\sqrt{N}}-\epsilon\right)\mu} , & \mathrm{Re}(\mu)\to-\infty,
		\end{cases}
\end{eqnarray}

\begin{lemma}\label{psichiisdistribution}

$\psi_\chi$ is a continuous linear functional in $\overline{\Fd}'$. 

\end{lemma}

\begin{proof}

Given $f\in \Fd$, the action of $\psi_\chi$ is given by 
\begin{eqnarray}
\langle\psi_\chi\mid f\rangle&=&\lim_{\epsilon\to0}\sum_{m\in\mathbb{Z}/N\mathbb{Z}}\int^\infty_{-\infty} d\mu\,\overline{\psi_{\chi}\left(\mu,m\right)}f\left(\mu,m\right)\nonumber\\
&=&\lim_{\epsilon\to0}\sum_{m\in\mathbb{Z}/N\mathbb{Z}}\int^\infty_{-\infty}d\mu\,\overline{\psi_{\chi}\left(\mu+ib^{-1},m+1\right)}f\left(\mu-i b,m+1\right)\nonumber\\
&=&\lim_{\epsilon\to0}\sum_{m\in\mathbb{Z}/N\mathbb{Z}}\int^\infty_{-\infty} d\mu\,\omega_\chi(\mu,\epsilon,m)\lt(\bm y^l +{\bm y}^{-l}\rt)\bm{u}^{-1}f\left(\mu,m\right).\label{psichif}\\
\omega_\chi(\mu,\epsilon,m)&\equiv&\frac{\overline{\psi_{\chi}\left(\mu+ib^{-1},m+1\right)}}{e^{\frac{2\pi i}{N/l}(-ib\mu-m)}+e^{-\frac{2\pi i}{N/l}(-ib\mu-m)}},
\end{eqnarray}
where $l\in\Z_+$ and $N/l$ is odd so that there is no pole. In the 2nd step, the integration contour can be shifted from $\R$ to $\R-ib$ ($\re(b)>0$), because $\psi_\chi$ is analytic in the upper-half plane and the exponentially-decaying behavior of the integrand as $\mathrm{Re}(\mu)\to \pm \infty$. For $\psi_\chi$, the contour $\R+i\re(b)$ is far away from the poles. 

We find the following limits
\be
\lim_{\mu\to\infty}\lt|\frac{\omega_\chi(\mu,\epsilon,m)}{e^{\frac{2\pi}{\sqrt{N}}\left(-\frac{\mathrm{Re}(b)}{\sqrt{N}}+\epsilon\right)\mu}}\rt|=\delta_{l,1}\exp\left(\frac{2\pi\epsilon}{\sqrt{N}}\mathrm{Im}(b)\right),\qquad
\lim_{\mu\to-\infty}\left|\frac{\omega_\chi(\mu,\epsilon,m)}{e^{\frac{2\pi}{\sqrt{N}}\left(\frac{\mathrm{Re}(b)}{\sqrt{N}}-\epsilon\right)\mu}}\right|=\delta_{l,1}\exp\left(-\frac{2\pi\epsilon}{\sqrt{N}}\mathrm{Im}(b)\right)\nonumber
\ee
Assuming that $\epsilon$ took values in the interval $[0,\epsilon_0]$ where $0<\epsilon_0<\frac{\mathrm{Re}(b)}{\sqrt{N}}$, there exists $\mu_0>0$, so that for all $\mu\in(-\infty,-\mu_0)\cup(\mu_0,\infty)$, $\o_\chi$ is uniformly bounded: $\lt|\omega_\chi(\mu,\epsilon,m)\rt|\leq C$ for some constant $C>0$. 
In addition, for any $m\in\Z/N\Z$, $|\omega_\chi(\mu,\epsilon,m)|$ is smooth and thus is bounded on the closed set $[-\mu_0,\mu_0]\times[0,\epsilon_0]$, so this function is bounded on $(-\infty,\infty)\times[0,\epsilon_0]$ by a constant $M$ (independent of $\mu,\epsilon$). Therefore, the dominated convergence can be applied to \eqref{psichif} for $\epsilon\to 0$, with the dominated function $M[\bm y^l +{\bm y}^{-l}]f\left(\mu-i b,m+1\right)$. Moreover, $\omega_\chi|_{\epsilon=0}\in\ch$ and $|\langle\psi_\chi\mid f\rangle|\leq \Vert\omega_\chi|_{\epsilon=0}\Vert\, \Vert\lt(\bm y^l +{\bm y}^{-l}\rt)\bm{u}^{-1}f\Vert$ imply $\psi_\chi$ is continuous on $\Fd$ by the semi-norms. Since $\Fd$ is dense in $\overline{\Fd}$ \footnote{A subspace $\sw\subset\Fd$ is dense in $\overline{\Fd}$ \cite{2008InMat.175..223F}.}, $\psi_\chi$ extends uniquely as a continuous linear functional on $\overline{\Fd}$.

\end{proof}

Moreover, the following results are proven in \cite{Han:2024nkf}:

\begin{theorem}

The eigenstates $\psi_\chi$ satisfy the orthogonality
\be
\lag\psi_\chi\mid \psi_{\chi'}\rag=\varrho(\mu_\chi,m_\chi)\lt[\delta(\mu_\chi-\mu_{\chi'})\delta_{m_\chi,m_{\chi'}}+\delta(\mu_\chi+\mu_{\chi'})\delta_{m_\chi,-m_{\chi'}}\rt]\ ,
\ee
where
\be
\varrho(\mu_\chi,m_\chi)&=&\frac{N^2}{4}\lt[\sin\lt(\frac{2\pi}{N}(ib\mu_\chi+m_\chi)\rt)\sin\lt(\frac{2\pi}{N}(-ib^{-1}\mu_\chi+m_\chi)\rt)\rt]^{-1}\ .
\ee
The eigenstates $\psi_\chi$ satisfy the resolution of identity on $\ch$
\be
\sum_{m_\chi\in\Z/N\Z}\int_0^\infty\rmd\mu_\chi\, \varrho(\mu_\chi,m_\chi)^{-1}\,\mid\psi_\chi\rangle\langle\psi_\chi\mid=\mathrm{id}_{\ch}\ ,\label{resolutioniden}
\ee
for $\mu,\mu'\in\R$ and $m,m'\in \Z/N\Z$.

\end{theorem}

For any $f\in\Fd$, we define
\begin{eqnarray}
	\mathscr{V}_{\psi}:\, f\left(\mu,m\right)\mapsto\mathscr{V}_{\psi}f\left(\mu_{\chi},m_{\chi}\right)=\langle\psi_{\chi}\mid f\rangle.\label{sUf}
\end{eqnarray}
which can be extend to an unitary map from $\ch$ to the direct integral representation $L^2(\C,\rmd\varrho_\chi)$. $\bmL$ and $\wt{\bmL}$ is represented as multiplication operators on $\mathscr{V}_{\psi}f\left(\mu_{\chi},m_{\chi}\right)$
\begin{eqnarray}
\mathscr{V}_{\psi}\bm{L}f\left(\mu_{\chi},m_{\chi}\right)&=\langle\psi_{\chi}\mid\bm{L}f\rangle=\langle\widetilde{\bm{L}}\psi_{\chi}\mid f\rangle=\left(\chi+\chi^{-1}\right)\mathscr{V}_{\psi}f\left(\mu_{\chi},m_{\chi}\right)\\
\mathscr{V}_{\psi}\wt{\bm{L}}f\left(\mu_{\chi},m_{\chi}\right)&=\langle\psi_{\chi}\mid\wt{\bm{L}}f\rangle=\langle{\bm{L}}\psi_{\chi}\mid f\rangle=\left(\overline\chi+\overline\chi^{-1}\right)\mathscr{V}_{\psi}f\left(\mu_{\chi},m_{\chi}\right).
\end{eqnarray}
The resolution of identity \eqref{resolutioniden} implies 
\begin{eqnarray}
\langle f\mid f' \rangle= \sum_{m_\chi\in\Z/N\Z}\int_0^\infty\rmd\mu_\chi\, \varrho(\mu_\chi,m_\chi)^{-1}\, \overline{\mathscr{V}_{\psi}f\left(\mu_{\chi},m_{\chi}\right)}\mathscr{V}_{\psi}f'\left(\mu_{\chi},m_{\chi}\right).\label{sUfsUf}
\end{eqnarray}
Some details about \eqref{sUf} - \eqref{sUfsUf} are discussed in Appendix \ref{App:spectral representation}. Up to a completion, the result \eqref{sUfsUf} gives the direct integral representation of $\ch$ in \eqref{DIDHW}, with the spectral measure $\rmd \varrho_\chi$ given by
\be  
\rmd \varrho_\chi&=&\sum_{m_\chi'\in\Z/N\Z}\rmd\mu_\chi\rmd m_\chi\,\varrho(\mu_\chi,m_\chi)^{-1}\Theta(\mu_\chi)\delta(m_\chi -m_\chi').
\ee
Note that $\mathscr{V}_{\psi}f\left(\mu_{\chi},m_{\chi}\right)$ is invariant under $(\mu_\chi,m_\chi)\to(-\mu_\chi,-m_\chi)$. So \eqref{sUfsUf} can be extended to the integral along the entire $\R$ 
\be
\langle f\mid f' \rangle= \sum_{m_\chi\in\Z/N\Z}\int_{-\infty}^\infty\rmd\mu_\chi\, \varsigma (\mu_\chi,m_\chi)\, \overline{\mathscr{V}_{\psi}f\left(\mu_{\chi},m_{\chi}\right)}\mathscr{V}_{\psi}f'\left(\mu_{\chi},m_{\chi}\right).
\ee
where $	\varrho (\mu _{\chi } ,m_{\chi } )^{-1} = \varsigma (\mu _{\chi } ,m_{\chi } )+ \varsigma (-\mu _{\chi } ,-m_{\chi } )$ and
\begin{equation*}
	\varsigma (\mu ,m )=\frac{1}{N^2}\lt[{e^{\frac{2\pi b\mu }{N} +\frac{2\pi b^{-1} \mu }{N}}} -{e^{-\frac{2\pi b\mu }{N} +\frac{2\pi b^{-1} \mu }{N} +\frac{4i\pi m}{N}}}\rt].
\end{equation*}

The eigenstates $\a_\chi$ of $Q_2'',\widetilde{Q}_2''$ are given by 
\be  
\alpha_\chi&=&\mathcal{S}_{\lambda_{2}}^{-1}\mathcal{D}_{\lambda_{1}}{\cal F}^{-1}\psi_\chi
\ee
where $\chi$ and $\overline{\chi}$ are given by \eqref{chichibar}. 
$\a_\chi$ belongs to $\overline{\Fd}'$. Its action on $f\in\Fd$ is given by 
\begin{eqnarray}
\langle \a_\chi\mid f\rangle=\langle\psi_\chi\mid \cf\cd_{\l_1}^{-1}\cs_{\l_2}f\rangle.
\end{eqnarray}
The unitary transformations $\cf,\cd_{\l_1},\cs_{\l_2}$ leave $\overline{\Fd}$ invariant. The above results imply the orthogonality
\begin{eqnarray}
\lag\a_\chi\mid \a_{\chi'}\rag=\varrho(\mu_\chi,m_\chi)\lt[\delta(\mu_\chi-\mu_{\chi'})\delta_{m_\chi,m_{\chi'}}+\delta(\mu_\chi+\mu_{\chi'})\delta_{m_\chi,-m_{\chi'}}\rt]\ ,
\end{eqnarray}
and the resolution of identity
\begin{eqnarray}
\int_\C\rmd\varrho_\chi \mid\a_\chi\rangle\langle\a_\chi\mid=\mathrm{id}_{\ch}\ .
\end{eqnarray}

\subsection{Two-holed and three-holed spheres}

The representation of the extended graph algebra on 2-holed sphere is carried by $\Fd_2\subset\ch\otimes\ch$, where the representation of gauge transformation is $\pi^{\l_1}\otimes\pi^{\l_2}$. Recall the discussion in Section \ref{Invariant bilinear form}. The gauge invariant in $\Fd_2'$ exists only when $\l_1=\l_2\equiv\l$. When it exists, the gauge invariant is uniquely given by the invariant bilinear form $\Psi_\l$ in \eqref{PsiandV1} up to complex rescaling. Therefore, the Hilbert space of physical states on a 2-hole sphere, $\ch_{phys}^{\l_1,\l_2}$, is 1-dimensional and spanned by $\Psi_\l$ when $\l_1=\l_2$, while being 0-dimensional if $\l_1\neq \l_2$.

For 3-hole sphere, the representation of extended graph algebra is carried by $\ch\otimes\ch\otimes\ch$. We label the representations at the holes by $\l_1,\l_2,\l_3$ and make the assumption \eqref{assumption2N} to these labels. The representation of gauge transformation $\xi\in\suquqt$ is given by $(\pi^{\l_1}\otimes\pi^{\l_2}\otimes\pi^{\l_3})(\Delta^{(3)}\xi)=\sum\pi^{\l_1}(\xi^{(1)})\otimes(\pi^{\l_2}\otimes\pi^{\l_3})(\Delta\xi^{(2)}))$ by the co-associativity. We label $\ch$ by the representation that it carries and let the unitary transformation $\cu_{23}$ ($(1,2)\to(2,3)$ in \eqref{U12}) act on the last two factors in $\ch_{\l_1}\otimes\ch_{\l_2}\otimes\ch_{\l_3}$:
\be 
\cu_{23}:\ \ch_{\l_1}\otimes\ch_{\l_2}\otimes\ch_{\l_3}&\to& \ch_{\l_1}\otimes\ch\otimes\ch,\label{cu231}\\
&\simeq&\int^\oplus\rmd\varrho_\chi\, \ch_{\l_1}\otimes\ch_{\chi}\otimes W_{\chi}. \label{cu232}
\ee
In particular, $\cu_{23}$ is a bijection from $\css_{23}$ (the common domain of $(\pi^{\l_1}\otimes\pi^{\l_2})(\Delta \xi)$) to $\css_{23}''$ (the common domain of $\pi^{Q_3''}(\xi)$). Recall \eqref{Q2pp1} - \eqref{DeltaXtilde12}. For any $\xi\in\suquqt$, the representation of gauge transformation is transformed to $(\pi^\l\otimes\pi^{Q_3''})(\Delta\xi)$, in which $\bm{u}_{\a,\b},\wt{\bm{u}}_{\a,\b}$ act on the first 2 factors of $\ch_{\l_1}\otimes\ch\otimes\ch$, while $Q_3'',\wt Q_3''$ act on the third factor. $Q_3''$ and $\widetilde{Q}_3''$ take the form as \eqref{Q2pp1} and \eqref{Q2pp2} with $2\to 3$. The spectral decomposition of $Q_3''$ and $\widetilde{Q}_3''$ results in the direct integral representation \eqref{cu232}.

The state $f\in \ch_{\l_1}\otimes\ch\otimes\ch$ can be represented by functions
\be
f(\mu_1,m_1;\mu_2,m_2;\mu_3,m_3),\label{frepmum}
\ee
and the inner product is given by
\be
\langle f_1 \mid f_2 \rangle   =\sum _{m_{1} ,m_{2},m_3 ,\in \mathbb{Z} /N\mathbb{Z}}\int _{-\infty }^{\infty } d\mu _{1} d\mu _{2}d\mu_3 \, \overline{f_1 \lt( \mu_1,m_1;\mu_2,m_2;\mu_3,m_3\rt)} \, f_2 \lt( \mu_1,m_1;\mu_2,m_2;\mu_3,m_3\rt).\nonumber
\ee
Note that in our notation, although $(\mu_1,m_1)$ associates to $\ch_{\l_1}$, $(\mu_2,m_2)$ and $(\mu_3,m_3)$ do not associate to $\ch_{\l_2}$ and $\ch_{\l_3}$ respectively but associate to the last 2 factors in $\ch_{\l_1}\otimes\ch\otimes\ch$. The action of $(\pi^{\l_1}\otimes\pi^{Q_3''})(\Delta\xi)$ contains $\bm{u}_{\a,\b},\wt{\bm{u}}_{\a,\b}$'s operating on the first two slots with $(\mu_1,m_1)$ and $(\mu_2,m_2)$ and $Q_3'',\wt{Q}''_3$ operating on the third slot with $(\mu_3,m_3)$. If we define $\overline{\Fd}_3\subset \ch_{\l_1}\otimes\ch\otimes\ch$ to be the common domain of $\bm{u}_{\a,\b}^{(i)}$ ($i=1,2,3$) that are defined in terms of $\mu_i,m_i$ as in \eqref{repuandy} - \eqref{eq:reptor11}, $(\pi^{\l_1}\otimes\pi^{Q_3''})(\Delta\xi)$ leaves $\overline{\Fd}_3$ invariant. $\overline{\Fd}_3$ is Fr\'echet with the semi-norms $\Vert f\Vert_{\vec{\a},\vec{\b}}=\Vert\bm{u}_{\a_1,\b_1}^{(1)}\bm{u}_{\a_2,\b_2}^{(2)}\bm{u}_{\a_3,\b_3}^{(3)}f\Vert$.


Given the direct-integral representation \eqref{cu232}, the state $f$ can be equivalently represented by
\be
f(\mu_1,m_1;\mu_2,m_2;\mu_\chi,m_\chi),\label{frepchi}
\ee
where $(\mu_1,m_1)$ and $(\mu_2,m_2)$ associate to $\ch_{\l_1}$, $\ch_{\chi}$ respectively. In this representation, $Q_3''$ and $\widetilde{Q}_3''$ are multiplication operators
\be
Q_3''f(\mu_1,m_1;\mu_2,m_2;\mu_\chi,m_\chi)=e^{\frac{2\pi i}{N}(-ib\mu_\chi-m_\chi)} f(\mu_1,m_1;\mu_2,m_2;\mu_\chi,m_\chi),\\
\wt Q_3''f(\mu_1,m_1;\mu_2,m_2;\mu_\chi,m_\chi)=e^{\frac{2\pi i}{N}(-ib^{-1}\mu_\chi+m_\chi)} f(\mu_1,m_1;\mu_2,m_2;\mu_\chi,m_\chi)
\ee
The inner product of the Hilbert space is represented by 
\be
\langle f_1 \mid f_2 \rangle   =\sum _{m_{1} ,m_{2} ,\in \mathbb{Z} /N\mathbb{Z}}\int _{-\infty }^{\infty } d\mu _{1} d\mu _{2}\int d\varrho _{\chi} \, \overline{f_1 \lt( \mu _{1} ,m_{1} ; \mu _{2} ,m_{2}  ; \mu_\chi, m_\chi\rt)} \, f_2 \lt( \mu _{1} ,m_{1} ; \mu _{2} ,m_{2}  ; \mu_\chi , m_\chi\rt).\nonumber
\ee
On the direct integral representation, $(\pi^{\l_1}\otimes\pi^{Q_3''})(\Delta\xi)$ reduces to $(\pi^{\l_1}\otimes\pi^{\chi})(\Delta\xi)$, which acts on $\ch_{\l_1}\otimes\ch_\chi$ while leaving $W_\chi$ invariant for each $\chi$. The gauge invariant on $\Fd_2\subset\ch_{\l_1}\otimes\ch_\chi$ is 1-dimensional when $\chi=\l_1$ and zero-dimensional otherwise. Therefore, the gauge invariants for 3-holed sphere should be contained in $W_{\chi=\l_1}$, which is also an 1-dimensional Hilbert space.  


The direct integral representation \eqref{sUfsUf} gives the isomorphism $\ch\simeq L^2(\C,\rmd\varrho_\chi)$. We define $\mathscr{K}\in L^2(\C,\rmd\varrho_\chi)$ of functions $F\left(\mu_{\chi},m_{\chi}\right)$ differentiable for $\mu_{\chi}\in(0,\infty)$ and satisfying $\Vert(\chi+\chi^{-1})^n F\Vert<\infty$ for all $n\in \Z$. 
The domain $\Fd_2\otimes\sk$ is invariant by the action of $(\pi^{\l_{1}}\otimes\pi^{Q''_3})(\Delta\xi)$. We define the union $\mathscr{N}_3=(\Fd_2\otimes\sk)\cup\overline{\Fd}_3$. It is clear that $\mathscr{N}_3$ is invariant by $(\pi^{\l_{1}}\otimes\pi^{Q''_3})(\Delta\xi)$ since both $\Fd_2\otimes\sk$ and $\overline{\Fd}_3$ does.

\begin{theorem}\label{3intertwiner}
The gauge invariants of 3-hole sphere, as linear functionals on $\sn_3$ is 1-dimensional and given by
\be
\Psi_{\l_1,\l_2,\l_3}=c \Psi_{\l_1}\otimes\a_{\l_1}
\ee
where $c\in\C$ is an arbitrary constant. These gauge invariants form the physical Hilbert space $\ch_{phys}^{\l_1,\l_2,\l_3}\simeq W_{\l_1}$ and are continuous linear functionals on $\overline{\Fd}_3$. 

\end{theorem}

\begin{proof}
Since $\a_{\l_1}$ is the generalized eigenfunction of $Q_3'',\wt Q_3''$,
\be
\Psi_{\l_1,\l_2,\l_3}[\cd^{\l_1,\l_2,\l_3}(\xi)f]=c \Psi_{\l_1}\otimes\a_{\l_1}\lt[\lt((\pi^{\l_1}\otimes\pi^{\l_1})(\Delta\xi)\otimes1\rt)f\rt]=\eps(\xi)\Psi_{\l_1,\l_2,\l_3}[f]
\ee
shows that $\Psi_{\l_1,\l_2,\l_3}$ is gauge invariant. The proof of uniqueness and continuity of $\Psi_{\l_1,\l_2,\l_3}$ is postponed and contained in the proof of Theorems \ref{generalPsi} and \ref{PsiFcontinuous}.
\end{proof}

\subsection{$m$-holed sphere}

The extended graph algebra of $m$-holed sphere is carried by $\ch_{\vec \l}=\ch_{\l_1}\otimes\ch_{\l_2}\otimes\cdots\ch_{\l_m-1}\otimes\ch_{\l_m}$, where $\vec{\l}=(\l_1,\cdots,\l_m)$ labels the holes. For any $\xi\in\suquqt$, the gauge transformation is represented by $(\pi^{\l_1}\otimes\cdots\otimes \pi^{\l_m})(\Delta^{(m)}\xi)$, and $\Delta^{(m)}$ can be decomposed into a sequence of $\Delta$'s by  $\Delta^{(m)}\xi=\Delta_{m-1}\Delta_{m-2}\cdots\Delta_1\xi$, where $\Delta_i$ is the co-mulitplication acting on the $i$-th tensor factor. We apply the unitary map $\mathbb{U}_{m-1,m}=\iota_m\circ\cu_{m-1,m}$ to the last two copies of $\ch$, where $\iota_m$ denotes the unitary transformation diagonalizing $Q''_m$ on the last copy of $\ch$:
\be
\mathbb{U}_{m-1,m}:&& \ch_{\l_1}\otimes\cdots\otimes\ch_{\l_m}\to\int ^{\oplus } d\varrho _{\chi _{m-1}}\mathcal{H}_{\lambda _{1}} \otimes \cdots \otimes \mathcal{H}_{\lambda _{m-2}} \otimes \mathcal{H}_{\chi _{m-1}} \otimes W_{\chi _{m-1}},\label{UHW0} 
\ee
On each fiber Hilbert space of the direct integral, the gauge transformation is represented by  $(\pi^{\l_1}\otimes\cdots\otimes\pi^{\l_{m-2}}\otimes \pi^{\chi_{m-1}})(\Delta^{(m-1)}\xi)$, where $\Delta^{(m-1)}\xi=\Delta_{m-2}\cdots\Delta_1\xi$, and leaves $W_{\chi_{m-1}}$ invariant. Then we apply the unitary map $\mathbb{U}_{m-2,m-1}$, which decomposes into $\iota_{m-1}\circ\cu_{m-2,m-1}$ acting on $ \mathcal{H}_{\lambda _{m-2}} \otimes \mathcal{H}_{\chi _{m-1}}$ of each fiber Hilbert space:
\be
\mathbb{U}_{m-2,m-1}\mathbb{U}_{m-1,m}:&& \ch_{\l_1}\otimes\cdots\otimes\ch_{\l_m}\to\int ^{\oplus } d\varrho _{\chi _{m-1}} d\varrho _{\chi _{m-2}}\mathcal{H}_{\lambda _{1}} \otimes \cdots \otimes \mathcal{H}_{\chi _{m-2}} \otimes W_{\chi _{m-2}} \otimes W_{\chi _{m-1}}\nonumber
\ee
The gauge transformation is represented by  $(\pi^{\l_1}\otimes\cdots\otimes\pi^{\l_{m-3}}\otimes \pi^{\chi_{m-2}})(\Delta^{(m-2)}\xi)$, where $\Delta^{(m-2)}\xi=\Delta_{m-3}\cdots\Delta_1\xi$, and leaves $W_{\chi_{m-2}}\otimes W_{\chi_{m-1}}$ invariant. By iteration, we finally obtain
\be
\mathbb{U}:&& \ch_{\l_1}\otimes\cdots\otimes\ch_{\l_m}\to \int ^{\oplus } d\varrho _{\chi _{2}}(\mathcal{H}_{\lambda _{1}} \otimes \mathcal{H}_{\chi _{2}}) \otimes \mathcal{W}_{\chi _{2}},\label{UHW}
\ee
where $\mathbb{U}$ and $\cw_{\chi_2}$ are given by
\be 
\mathbb{U}&=&\mathbb{U}_{2,3}\cdots\mathbb{U}_{m-2,m-1}\mathbb{U}_{m-1,m},\\
\cw_{\chi_2}&=&\int ^{\oplus } d\varrho _{\chi _{m-1}} d\varrho _{\chi _{m-2}} \cdots d\varrho _{\chi _{3}}\, W_{\chi _{2}} \otimes \cdots \otimes W_{\chi _{m-1}} .
\ee
The representation of gauge transformation $\cd^{\vec{\l}}(\xi)$ acting on $\ch_{\l_1}\otimes\cdots\otimes\ch_{\l_m}$ is transformed to $(\pi^\l\otimes\pi^{\chi_2})(\Delta\xi)$ acting on $\ch_{\l_1}\otimes\ch_{\chi_2}$ while leaving $\cw_{\chi_2}$ invariant for each $\chi$.

At each step of iteration, the unitary operator $\mathbb{U}_{i-1,i}$ is a bijection from $\css_{i-1,i}$ (the common domain of $(\pi^{\l_{i-1}}\otimes\pi^{Q''_{i+1}})(\Delta \xi)$) to $\css_{i-1,i}''$ (the common domain of $\pi^{Q_{i}''}(\xi)$ and $Q''_{i+1}$) \footnote{$\mathbb{U}_{i-1,i}$ replaces $\l_i$ in $\cu_{i-1,i}$ by the operator relating to $Q_{i+1}''$, which acts on another Hilbert space different from the domain of $\cu_{i-1,i}$. The domain analysis in \cite{Han:2024nkf} turns out to be still valid. }. The unitary operator $\mathbb{U}$ is a bijection from the domain $\css_m$ of $(\pi^{\l_1}\otimes\cdots\otimes \pi^{\l_m})(\Delta^{(m)}\xi)$ to the domain $\css''_m$ of $(\pi^{\l_{1}}\otimes\pi^{Q''_3})(\Delta\xi)$ and $\{Q_j'', \widetilde{Q}_j''\}_{j=4}^m$.

The gauge invariant on $\Fd_2\subset\ch_{\l_1}\otimes\ch_{\chi_2}$ is 1-dimensional when $\chi_2=\l_1$ and zero-dimensional otherwise. Therefore, the gauge invariants for $m$-holed sphere should be contained in the Hilbert space $\cw_{\chi_2}$. The states in $\cw_{\chi_2}$ are represented by functions $F(\vec{\mu}_\chi,\vec{m}_\chi)$ where $\vec{\mu}_\chi=(\mu_{\chi_3},\cdots,\mu_{\chi_{m-1}})$ and $\vec{m}_\chi=(m_{\chi_3},\cdots,m_{\chi_{m-1}})$. The inner product is given by
\be 
( F\mid G)=\int \prod_{i=3}^{m-1} d\varrho _{\chi _{i}} \, \overline{F\lt(\vec{\mu}_\chi,\vec{m}_\chi\rt)}\, G\lt(\vec{\mu}_\chi,\vec{m}_\chi\rt).
\ee
The Hilbert spaces $\cw_\chi$ for different $\chi$ are isomorphic, so we can simply denote by $\cw$.

Given the direct integral representation in \eqref{UHW}, any state $\psi\in \ch_{\vec \l}$ can be represented by functions 
\be
\psi(\mu_1,m_1;\mu_2,m_2;\mu_{\chi_2},m_{\chi_2};\vec{\mu}_\chi,\vec{m}_\chi)\label{frepchim}
\ee
with the inner product
\begin{eqnarray}
\lag\psi\mid\psi'\rag&=&\sum_{m_1,m_2\in\Z/N\Z}\int\rmd\mu_1\rmd\mu_2\int\rmd\varrho _{\chi _{2}}\prod_{i=3}^{m-1} \rmd\varrho _{\chi _{i}}\nonumber\\
&&\qquad \overline{\psi(\mu_1,m_1;\mu_2,m_2;\mu_{\chi_2},m_{\chi_2};\vec{\mu}_\chi,\vec{m}_\chi)}\, \psi'(\mu_1,m_1;\mu_2,m_2;\mu_{\chi_2},m_{\chi_2};\vec{\mu}_\chi,\vec{m}_\chi)
\end{eqnarray}
In this representation, the operators $Q_3'', \widetilde{Q}_3''$ and $\{Q_j'', \widetilde{Q}_j''\}_{j=4}^m$ are all diagonalized (The eigenvalue of $Q_j''$ relates to $\mu_{\chi_{j-1}},m_{\chi_{j-1}}$).


\begin{theorem}\label{generalPsi}

There is an 1-to-1 correspondence between $F\in\cw$ and a gauge invariant linear functional $\Psi_F$ on $\sn_3\otimes\cw$ that is continuous on $\cw$.
    
\end{theorem}

\begin{proof}
The gauge invariant $\Psi$ satisfies
\be 
\Psi[\cd^{\vec\l}(\xi)\psi]=\Psi[(\pi^{\l_{1}}\otimes\pi^{Q''_3})(\Delta\xi)\psi]=\eps(\xi)\Psi[\psi]
\ee
for all $\xi\in\suquqt$ and $\psi\in\sn_3\otimes\cw$.

Let us choose $\psi =f( \mu _{1} ,m_{1} ; \mu _{2} ,m_{2} )\ \phi (\mu_\chi,m_\chi) \ G(\vec{\mu}_\chi,\vec{m}_\chi)$ for any $f\in\Fd_2$, $G\in\cw$, and $\phi\in\sk$. We may write $\Psi[\psi]=\Psi[f\otimes\phi\otimes G]\equiv\Psi_{\phi,G}[f]$. $\Psi_{\phi,G}$ is a linear functional on $\Fd_2$.

Although generally $\pi^{Q''_3}$ contains $Q_3'',\wt{Q}_3''$ that operate on $\mu_{\chi_2},m_{\chi_2}$, for some special $\xi$, e.g. $\xi=K,KF$, $(\pi^{\l_{1}}\otimes\pi^{Q''_3})(\Delta\xi)$ does not contain  $Q_3'',\wt{Q}_3''$ and only operates on $\mu_1,m_1;\mu_2,m_2$. Let us choose $\xi=K,\wt{K},KF,\wt{K}\wt{F}$ and recall the proof for Theorem \ref{uniquebilinear}. Solving for invariant with only $\xi=K,\wt{K},KF,\wt{K}\wt{F}$ already determine the invariant bilinear $\Psi_{\l_1}^{(2)}$ up to complex rescaling (we have added a superscript for emphasizing its rank), while considering other generators of $\xi$ only amounts to restrict the representations $\l_1=\l_2^{\pm1}$. This implies that $\Psi_{\phi,G}$ is an invariant bilinear, where different $\phi,G$ correspond to rescaling. If we denote $\Psi_{\l_1}^{(2)}$ to be a fiducial invariant bilinear, then
\be
\Psi[f\otimes\phi\otimes G]=C_G[ \phi] \Psi ^{(2)}_{\l_1}[ f].
\ee
Similar to in the proof of Theorem \ref{uniquebilinear}, when we apply $\xi=E,\widetilde{E}$, the invariance results in the constraint
\be 
C_{G}\left[\left( \lambda _{1} +\lambda _{1}^{-1} -Q_3''\right) \phi \right] =C_{G}\left[ \left(\overline{\lambda }_{1} +\overline{\lambda }_{1}^{-1} -\widetilde{Q}_3''\right) \phi \right]=0.
\ee
It implies $C_G=F[G]\alpha_{\l_1}$ (see Appendix \ref{App:spectral representation}), where $F$ is a linear functional on $\cw$. Moreover, the continuity implies $F\in\cw$ and $F[G]=\lag F\mid G\rag$ by Riesz's Theorem. As a result, we obtain
\be 
\Psi=\Psi_{\l_1}^{(2)}\otimes\alpha_{\l_1}\otimes F\equiv \Psi_F,\qquad \forall F\in\cw,
\ee
which manifests the 1-to-1 correspondence between $F\in\cw$ and $\Psi_F$.

\end{proof}

We also consider the representation 
\be
\psi(\mu_1,m_1;\mu_2,m_2;\mu_{3},m_{3};\vec{\mu}_\chi,\vec{m}_\chi)\label{frepmumchim}
\ee
generalizing \eqref{frepmum}, where $Q_3'', \widetilde{Q}_3''$ are not diagonalized and take the form as \eqref{Q2pp1} and \eqref{Q2pp2}. The inner product is represented by
\begin{eqnarray}
\lag\psi\mid\psi'\rag&=&\sum_{m_1,m_2,m_3\in\Z/N\Z}\int\rmd\mu_1\rmd\mu_2\rmd\mu_3\int\prod_{i=3}^{m-1} \rmd\varrho _{\chi _{i}}\nonumber\\
&&\qquad \overline{\psi(\mu_1,m_1;\mu_2,m_2;\mu_{3},m_{3};\vec{\mu}_\chi,\vec{m}_\chi)}\, \psi'(\mu_1,m_1;\mu_2,m_2;\mu_{3},m_{3};\vec{\mu}_\chi,\vec{m}_\chi).
\end{eqnarray}
We define the Fr\'echet space $\overline{\Fd}_m\subset \ch_{\vec \l}$ with the semi-norms $\Vert f\Vert_{\vec{\a},\vec{\b}}=\Vert\bm{u}_{\a_1,\b_1}^{(1)}\bm{u}_{\a_2,\b_2}^{(2)}\bm{u}_{\a_3,\b_3}^{(3)}f\Vert$. Here we only use the operators for the first three slots to define the semi-norm, because only these operators are involved in the representation $\cd^{\vec\l}(\xi)\simeq (\pi^{\l_{1}}\otimes\pi^{Q''_3})(\Delta\xi)$. $\overline{\Fd}_m$ is inside the domain of $(\pi^{\l_{1}}\otimes\pi^{Q''_3})(\Delta\xi)$ and is invariant by the actions of $(\pi^{\l_{1}}\otimes\pi^{Q''_3})(\Delta\xi)$.

The action of $\Psi_F$ on $f\in\overline{\Fd}_m$ can be written explicitly by using \eqref{PsiandV1}
\be
\Psi_F[f]=\sum _{n_1,{m} \in \mathbb{Z} /N\mathbb{Z}}\int d\nu _{1} d\mu \int\prod_{i=3}^{m-1}d\varrho _{\chi_i}  \overline{\psi _{\lambda _{1}}( \mu ,m)} \, \overline{F(\vec{\mu }_\chi,\vec{m}_\chi)} \, \bm{V}_{1}^{-1}f'\lt( \nu _{1} ,n_{1} ;-\nu _{1} -n_{1} ;\mu ,m ;\vec{\mu }_\chi,\vec{m}_\chi\rt),\nonumber\\
\label{PsiFfexplicit}
\ee
where $f'=\cf\cd_{\l_1}^{-1}\cs_{\l_2}f$ with $\cf\cd_{\l_1}^{-1}\cs_{\l_2}$ operating at the third slot.

\begin{theorem}\label{PsiFcontinuous}
For any $F\in\cw$, the linear functional $\Psi_F$ is continuous and invariant on $\overline{\Fd}_m$.
\end{theorem}

\begin{proof}
For the continuity, we employ the above explicit expression of $\Psi_F$ and use a trick similar to \eqref{trickg}
\be
&&\lt|\Psi_F[f]\rt|\nonumber\\
&=&\lt|\sum _{n_1,{m} \in \mathbb{Z} /N\mathbb{Z}}\int d\nu _{1} d\mu_2d\mu \int\prod_{i=3}^{m-1}d\varrho _{\chi_i}\, \overline{g( \nu _{1} ,n_{1} ,\mu _{2} ,m_{2})} \ e^{\frac{2\pi i}{N}( \mu _{2} \nu _{1} -m_{2} n_{1})} \ \rt.\nonumber\\
&&\qquad \lt.\overline{\psi _{\lambda _{1}}\left( \mu +ib^{-1} ,m+1\right)} \, \overline{F(\vec{\mu }_\chi,\vec{m}_\chi)} \frac{\bm{V}_1^{-1}f'\lt( \nu _{1} ,n_{1} ;-\nu _{1} -n_{1} ;\mu-ib ,m+1 ;\vec{\mu }_\chi,\vec{m}_\chi\rt)}{\overline{g( \nu _{1} ,n_{1} ,\mu _{2} ,m_{2})}}\rt|\nonumber\\
&\leq &\Vert g\Vert\, \Vert\o_{\l_1}|_{\epsilon=0}\Vert\, \Vert F\Vert \, \left\Vert \left( \bm{u}_{1}^{l-1} +\bm{u}_1^{-l-1}\right)\left( \bm{y}_{2}^l +\bm{y}_{2}^{-l}\right) \left( \bm{y}^l +\bm{y}^{-l}\right)\bm{u}^{-1} f'\right\Vert.
\ee
The unitary transformation $\cf\cd_{\l_1}^{-1}\cs_{\l_2}$ maps any polynomial of $\bm{u},\bmy$ to a polynomial of $\bm{u}_3,\bmy_3$. For the invariance, one can check that the invariance $\Psi_F[(\pi^{\l_{1}}\otimes\pi^{Q''_3})(\Delta\xi)\psi]=\eps(\xi)\Psi_F[\psi]$ holds for all $\psi$ in $ \overline{\Fd}_2\otimes \overline{\Fd}\otimes\cw$, which is dense in $\overline{\Fd}_m$. Then one can use the limit argument the same as in the proof of Lemma \ref{bilinearextend}.

\end{proof}

We denote by $\mathscr{D}_m$ the maximal domain of $\psi$ such that the invariance $ \Psi_F[(\pi^{\l_{1}}\otimes\pi^{Q''_3})(\Delta\xi)\psi]=\eps(\xi)\Psi_F[\psi]$ holds for all $F\in\cw$. $\mathscr{D}_m$ is clearly dense in $\ch_{\vec\l}$ since it contains both $\overline{\Fd}_m$ and $\sn_3\otimes \cw$.



The physical Hilbert spaces are the space of gauge invariant linear functionals. The above discussion motivates us to define the physical Hilbert space by $\ch_{phys}^{\vec \l}\simeq \cw$. As a special case, when $m=3$, the Hilbert space $\cw$ is one dimensional. $\Psi_F$ reduces to $\Psi_{\l_1,\l_2,\l_3}$, and our result proves Theorem \ref{3intertwiner}.



\subsection{``3j-symbol'' notation}

The gauge invariants can be formulated formally in terms of a ``3j-symbol'' notation: We define the ``3j-symbol'' by 
\begin{eqnarray}
	\Phi_{\lambda_{1}\lambda_{2}}^{\chi}=\langle\alpha_{\chi}\mid{\cal U}_{12}\ . 
	\label{Phinotation}
\end{eqnarray}
Use the relations \eqref{DeltaYK12} - \eqref{DeltaXtilde12} and the fact that $\alpha_\chi$ is the eigenstate of $Q_2''$ and $\widetilde{Q}_2''$, we obtain the intertwining relation
\begin{eqnarray}
\Phi_{\lambda_{1}\lambda_{2}}^{\chi}\left(\pi^{\lambda_{1}}\otimes\pi^{\lambda_{2}}\right)\left(\Delta\xi\right)=\pi^{\chi}\left(\xi\right)\Phi_{\lambda_{1}\lambda_{2}}^{\chi} .\label{intertwiningrelation}
\end{eqnarray}
Following the iterative procedure in \eqref{UHW0} - \eqref{UHW}, we define 
\begin{eqnarray}
\Phi^{\chi_2,\chi_3,\cdots,\chi_{m-1}}_{\l_2,\l_3,\cdots,\l_{m-1},\l_{m}}:=\Phi^{\chi_2}_{\l_2\chi_3}\Phi^{\chi_3}_{\l_3\chi_4}\cdots \Phi^{\chi_{m-2}}_{\l_{m-2}\chi_{m-1}} \Phi^{\chi_{m-1}}_{\l_{m-1}\l_{m}},\label{gluingPhis}
\end{eqnarray}
which satisfies 
\begin{eqnarray}
	\Phi^{\chi_2,\chi_3,\cdots,\chi_{m-1}}_{\l_2,\l_3,\cdots,\l_{m-1},\l_{m}}\left(\pi^{\lambda_{2}}\otimes\pi^{\lambda_{3}}\otimes\cdots\otimes\pi^{\lambda_{m}}\right)\left(\Delta^{(m-1)}\xi\right)=\pi^{\chi_{2}}\left(\xi\right)\Phi^{\chi_2,\chi_3,\cdots,\chi_{m-1}}_{\l_2,\l_3,\cdots,\l_{m-1},\l_{m}}.
\end{eqnarray}
The gauge invariants are given by 
\begin{eqnarray}
\Psi_F[f]&=&\int\prod_{i=3}^{m-1}d\varrho _{\chi_i}\overline{F\lt(\vec{\mu}_\chi,\vec{m}_\chi\rt)}\Psi_{\vec{\mu}_\chi,\vec{m}_\chi}\lt[f\rt],
\ee
where
\be
\Psi_{\vec{\mu}_\chi,\vec{m}_\chi}\lt[f\rt]&=&\Psi_{\l_1}^{(2)}\lt[\lt(\mathrm{id}_{\l_1}\otimes\Phi^{\l_1,\chi_3,\cdots,\chi_{m-1}}_{\l_2,\l_3,\cdots,\l_{m-1},\l_{m}}\rt)f\rt].
\end{eqnarray}

\subsection{Refined algebraic quantization and physical observables}\label{RAQ}

The above discussion of the physical Hilbert space is mainly based on the direct integral representation. We can relate our result to the language of refined algebraic quantization \cite{thiemann2008modern,Giulini:1998rk}: We may write 
\be
\Psi_F[f]=\lt(F\mid {\eta_0(f)}\rt)
\ee
where $(\cdot\mid\cdot)$ is the physical inner product on $\cw$ and
\be
{\eta_0(f)\lt(\vec{\mu }_\chi,\vec{m}_\chi\rt)}=\sum _{n_1\in \mathbb{Z} /N\mathbb{Z}}\int d\nu _{1}\bm{V}_{1}^{-1}f\lt( \nu _{1} ,n_{1} ;-\nu _{1} -n_{1} ;\mu_{\l_1} ,m_{\l_1} ;\vec{\mu }_\chi,\vec{m}_\chi\rt).
\ee
The anti-linear rigging map $\eta:\overline{\Fd}_m\to \overline{\Fd}_m^*$ (the space of linear functionals) is defined by
\be
\eta(f')[f]:=\Psi_{\eta_0(f')}[f].\label{riggingmap}
\ee
It may also be written formally as
\be
\eta(f')[f]&=&\sum_I\overline{\Psi_{F_I}[f']}\Psi_{F_I}[f]=\lag f'\mid\Psi_{\l_1}^{(2)}\otimes\a_{\l_1}\rangle\langle\Psi_{\l_1}^{(2)}\otimes\a_{\l_1}\mid f\rag
\ee
where $\{F_I\}$ is a basis in $\cw$. The relation between physical inner product and rigging map is given by
\be
\lt(\eta(f)\mid\eta(f')\rt)_{phys}:=\eta(f')[f]=\lt(\eta_0(f')\mid\eta_0(f)\rt).
\ee
The physical Hilbert space is spanned by $\eta(\overline{\Fd}_m)/\mathrm{ker}(\eta)$ (the kernel is with respect to $(\cdot\mid\cdot)_{phys}$) and is isomorphic to $\cw$ \footnote{$\eta(\overline{\Fd}_m)$ maps to $\cw$ by definition. The map is surjective, because for any $F\in\cw$, we can find $\eta(f)$ where $f=f_1\lt( \nu _{1} ,n_{1} ;-\nu _{1} -n_{1} ;\mu_{\l_1} ,m_{\l_1}\rt)F\lt(\vec{\mu }_\chi,\vec{m}_\chi\rt)$ for some $f_1$, such that $\eta_0(f)= cF$ for some constant $c$. }.

Recall that in the extended graph algebra $\Fs_{g,m}\otimes\wt{\Fs}_{g,m}$, the commutation relation between $\xi\in\suquqt$ and $\bm{\co}\in \cl_{g,m}\otimes\wt{\cl}_{g,m}$ is given by $\xi\bm{\co}=\sum_{\a}\xi^{(1)}_\a(\bm{\co})\xi_\a^{(2)}$. If $\bm{\co}$ is gauge invariant, $\xi(\bm{\co})=\eps(\xi)\bm{\co}$, we obtain the commutativity $\xi\bm{\co}=\bm{\co}\xi$.

The gauge invariant operator $\bm{\co}$ acts naturally on the physical Hilbert space $\cw$: For any $f\in \mathscr{D}_m$ and any gauge invariant bounded operator $\bm{\co}^\dagger$ leaving $\mathscr{D}_m$ invariant, 
\be
\eps(\xi)\Psi_F\lt[\bm{\co}^\dagger f\rt]=\Psi_F\lt[\cd^{\vec{\l}}(\xi)\bm{\co}^\dagger f\rt]=\Psi_F\lt[\bm{\co}^\dagger\cd^{\vec{\l}}(\xi)f\rt],\qquad \forall F\in\cw,
\ee
implies $\Psi_F\lt[\bm{\co}^\dagger\cdot\rt]$ is a gauge invariant linear functional on $\mathscr{D}_m$ and continuous on $\cw$. 
By Theorem \ref{generalPsi}, there exists $F'\in\cw$ such that $\Psi_F\lt[\bm{\co}^\dagger f\rt]=\Psi_{F'}\lt[f\rt]$, and therefore we obtain an bounded operator $\mathsf{O}:F\mapsto F'$ on $\cw$. In the notation of refined algebraic quantization, the relation can be written as
\be
\eta(f')[\bm{\co}^\dagger f]\equiv \mathsf{O}'\eta(f')[ f]=\lt(\mathsf{O}\eta_0(f')\mid\eta_0(f)\rt).
\ee
$\mathsf{O}$ is an operator only acting on $\vec\mu_\chi,\vec{m}_\chi$. $\mathsf{O}$ commute with $\eta_0$ and can be defined on $\ch_{\vec\l}$: $\mathsf{O}\eta_0(f)=\eta_0(\mathsf{O}f)$ holds on a dense domain in $\ch_{\vec\l}$ \footnote{This is clearly valid at least for factorized $f$, i.e. $f=f_1\lt( \nu _{1} ,n_{1} ;-\nu _{1} -n_{1} ;\mu_{\l_1} ,m_{\l_1}\rt)F\lt(\vec{\mu }_\chi,\vec{m}_\chi\rt)$. Then one can extend the domain of $\mathsf{O}$ to entire $\ch_{\vec\l}$ by continuity. Given the factorization $\ch_{\vec\l}\simeq \ch^{\otimes 3}\otimes\cw$, as the operator on $\ch_{\vec\l}$, $\mathsf{O}$ may be written as $\bm{1}_{\ch^{\otimes 3}}\otimes\mathsf{O}$. Similary $\mathsf{O}^\dagger$ may be written as $\bm{1}_{\ch^{\otimes 3}}\otimes\mathsf{O}^\dagger$, where the second dagger is on $\cw$.}. As a result, the gauge invariant operator $\bm{\co}$ results in $\mathsf{O}'$ acting on the physical state $\eta(f)$ and $\mathsf{O}$ acting on $\ch_{\vec\l}$ satisfying \footnote{The factorized $f$ form a basis in $\mathscr{D}_m$, for $f,f'$ belong to this basis, we have $(\eta_0(\mathsf{O}f')\mid\eta_0(f)) = (\eta_0(f')\mid\eta_0(\mathsf{O}^\dagger f)) $.}
\be
\mathsf{O}'\eta(f')[ f]=\eta(\mathsf{O}f')[f]=\eta(f')[\mathsf{O}^\dagger f],
\ee
for $f,f'$ in a dense domain. We obtain the equivalence $\bm{\co}^\dagger\sim \mathsf{O}^\dagger$ in the sense of $\eta(f')[\bm{\co}^\dagger f]=\eta(f')[\mathsf{O}^\dagger f]$.

The moduli algebra of physical observables is defined by the $*$-algebra of gauge invariant operators modulo the ideal generated by operators $\tr^I[(\bm{M}_0^I-\bm{1})\bm{X}]$ and $\tr^I[(\wt{\bm{M}}_0^I)-\bm{1})\bm{X}]$, where $\bm{X}$ are some operators inserted in order to make the traces gauge invariant. The representation of moduli algebra is often employed in the literature of combinatorial quantization to derive the physical Hilbert space \cite{Alekseev:1995rn,BNR}. However, the above argument indicates that when acting on physical states,  it is sufficient to focus on operators such as $\mathsf{O},\mathsf{O}^\dagger$, which only act on $\vec\mu_\chi,\vec{m}_\chi$, in order to formulate the gauge invariant observables on the physical Hilbert space, due to the above equivalence $\bm{\co}^\dagger\sim \mathsf{O}^\dagger$. Therefore, instead of the moduli algebra, we consider $\sm$ the $*$-algebra of bounded operators on $\ch_{\vec \l}$ that only act on $\vec{\mu}_\chi,\vec{m}_\chi$. Given any pair of operators $\mathsf{O},\mathsf{O}^\dagger\in\sm$, they satisfy
\be
\mathsf{O}'\eta(f')[ f]:=\eta(f')[\mathsf{O}^\dagger f]=\eta(\mathsf{O}f')[f],\qquad \mathsf{O}^{\prime\star}\eta(f')[ f]:=\eta(f')[\mathsf{O} f]=\eta(\mathsf{O}^\dagger f')[f].
\ee
As a standard derivation in refined algebraic quantization,
\be
\lt(\eta(f)\mid\mathsf{O}'\eta(f')\rt)_{phys}=\eta(f')[\mathsf{O}^\dagger f]=\lt(\eta(\mathsf{O}^{\dagger}f)\mid\eta(f')\rt)_{phys}=\lt(\mathsf{O}^{\prime\star}\eta(f)\mid\eta(f')\rt)_{phys},
\ee
shows that the physical Hilbert space $\cw$ carries a $*$-representation of $\sm$. In the next section, we provide useful examples for operators in $\sm$, which are a collection of Wilson loop operators relating to the pants decomposition.

\section{Wilson loops and Fenchel-Nielsen representation}\label{Wilson loop operators}

\subsection{Pants decomposition and Wilson loops}

This subsection relate the procedure in \eqref{UHW} to a pants decomposition of the $m$-holed sphere. Let us first illustrate the main idea by FIG.\ref{pants} as a pants decomposition of 5-holed sphere. The representations labels $\l_1,\cdots,\l_5$ associates to five holes. $\chi_4$ resulting from the CG decomposition of $\pi^{\l_4}\otimes\pi^{\l_5}$ associates to the cut (dashed circle associated with $\chi_4$) that makes a pair of pants together with the holes of $\l_4$ and $\l_5$. The next step of CG decomposition of $\pi^{\l_3}\otimes\pi^{\chi_4}$ associates to the cut associated with $\chi_3$, and the cut makes a pair of pants together with the hole of $\l_3$ and the new hole of $\chi_4$ from the last cut.

\begin{figure}[h]
	\centering
	\includegraphics[width=0.6\textwidth]{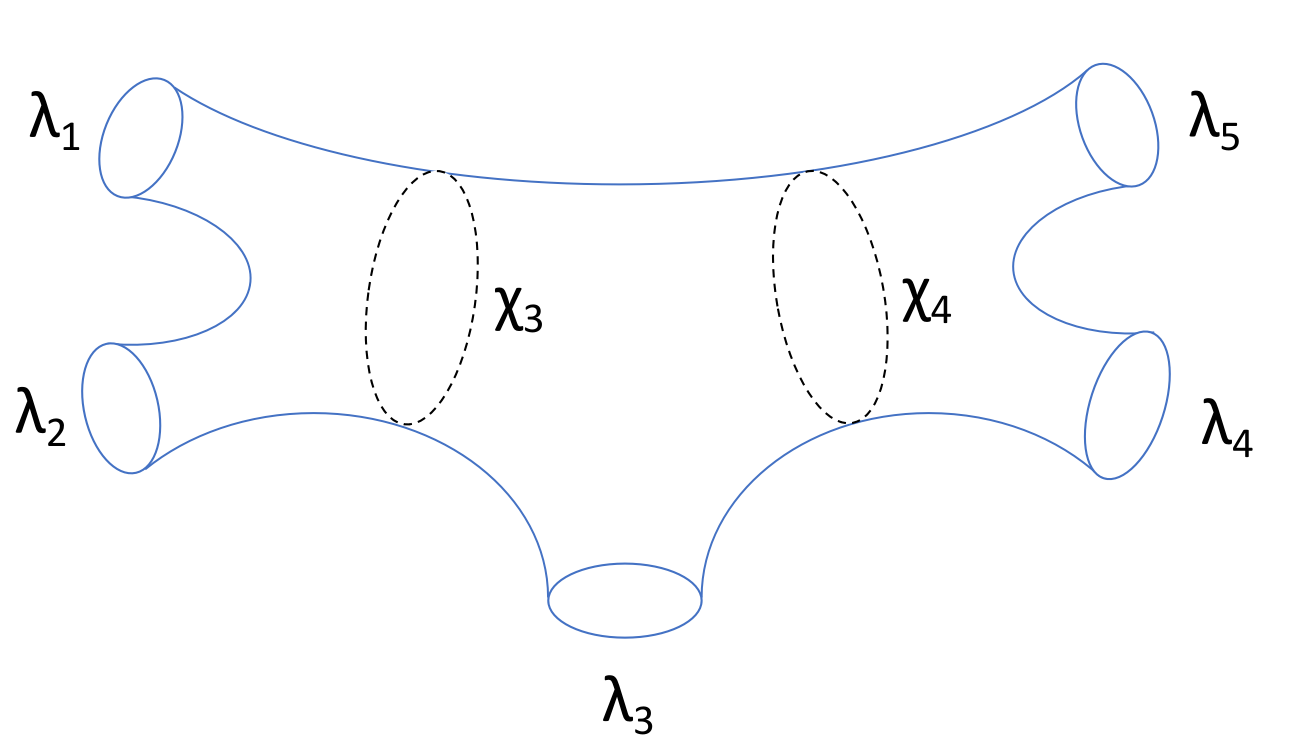}
	\caption{A decomposition of a 5-holed sphere into 3 pairs of pants: The holes are colored by $\l_1,\cdots,\l_5$. 2 cuts (the dashed circle) are colored by $\chi_3$ and $\chi_4$.}
	\label{pants}
\end{figure}

The ``3j-symbol'' $\Phi_{\lambda_{1}\lambda_{2}}^{\chi}$ associates to a pair of pants whose three holes are colored by $\l_1,\l_2$, and $\chi$ (see FIG.\ref{apairofpants}). $\Phi^{\l_1,\chi_3,\cdots,\chi_{m-1}}_{\l_2,\l_3,\cdots,\l_{m-1},\l_{m}}$ in \eqref{gluingPhis} corresponds to gluing a sequence of pants to a $m$-holed sphere, where the holes are colored by $\l_1,\cdots,\l_m$, while the cuts are colored by $\chi_3,\cdots,\chi_{m-1}$.

\begin{figure}[h]
	\centering
	\includegraphics[width=0.3\textwidth]{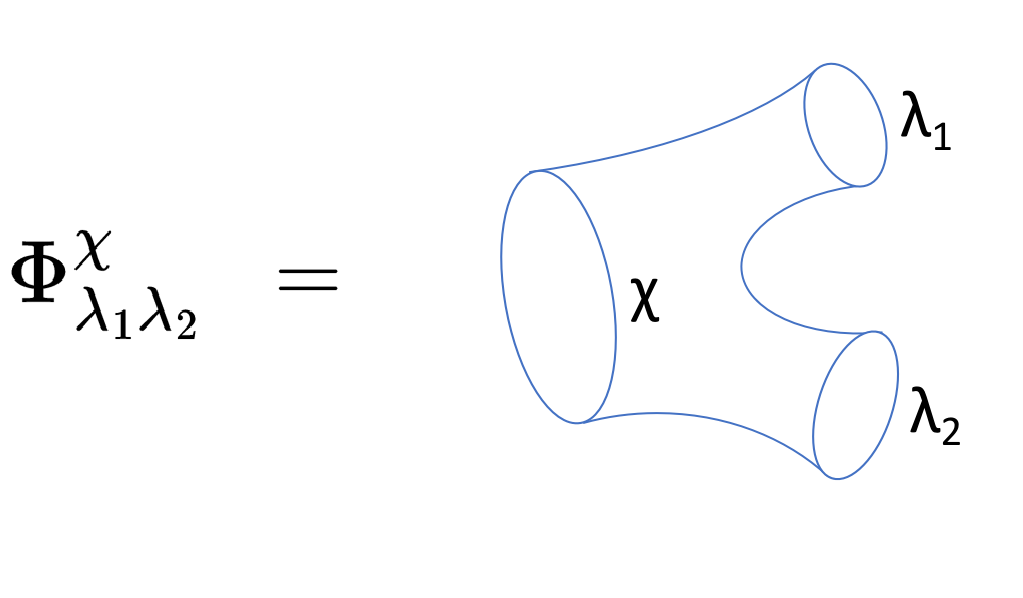}
	\caption{A pair of pants associates to the ``3j-symbol'' $\Phi_{\lambda_{1}\lambda_{2}}^{\chi}$. }
	\label{apairofpants}
\end{figure}

The correspondence with pants decomposition can be formulated more precisely in terms of the Wilson loop operator along the cut. Given a loop that enclose two consecutive holes labelled by $\l_m,\l_{m-1}$, the holomorphic quantum holonomy with the representation $I$ along the loop is given by $\bm{M}^I_{m,m-1}\equiv\bm{M}_{m}^{I}\bm{M}_{m-1}^{I}$. Its representation on $\ch_{\vec \l}$ is 
\be 
{\cal D}^{\vec{\lambda}}\left(\bm{M}_{m}^{I}\right){\cal D}^{\vec{\lambda}}\left(\bm{M}_{m-1}^{I}\right)	
&=&\Delta_{1}^{(m-1)}\left(R_{01}^{\prime}\right)\bm{M}_{0m}\Delta_{1}^{(m-1)}\left(R_{01}^{\prime-1}\right)\Delta_{1}^{(m-2)}\left(R_{01}^{\prime}\right)\bm{M}_{0m-1}\Delta_{1}^{(m-2)}\left(R_{01}^{\prime-1}\right)\nonumber\\
&=&R_{01}^{\prime}R_{02}^{\prime}\cdots R_{0m-1}^{\prime}\left(R_{0m}^{\prime}R_{0m}\right)R_{0m-1}^{\prime-1}\left(R_{0m-1}^{\prime}R_{0m-1}\right)R_{0m-2}^{\prime-1}\cdots R_{01}^{\prime-1}\nonumber\\
&=&R_{01}^{\prime}R_{02}^{\prime}\cdots R_{0m-2}^{\prime} R_{0m-1}^{\prime}R_{0m}^{\prime}R_{0m}R_{0m-1}R_{0m-2}^{\prime-1}\cdots R_{01}^{\prime-1}\nonumber\\
&=&\Delta_{1}^{(m-2)}\left(R_{01}^{\prime}\right)\Delta_{m-1}\left(R_{0m-1}^{\prime}R_{0m-1}\right)\Delta_{1}^{(m-2)}\left(R_{01}^{\prime-1}\right)
\ee
where we have neglect the representation labels $I,\vec{\l}$ in the derivation. The 0-th tensor slot carries the representation $\rho^I$, and the $i$-th tensor ($i>0$) slot carries $\pi^{\l_i}$. 

At the algebraic level, the identity $\bm{1}$ and the matrix elements of $\bm M^I_{m,m-1}$ for all $I\in\Z_+/2$ form a holomorphic loop algebra \cite{Hsiao:2024phb}: 
\be
\left(R^{-1}\right)^{IJ}\bm{M}_{m,m-1}^{I}R^{IJ}\bm{M}_{m,m-1}^{J}
&=&	\sum_{K}(C_{1}^{IJ})_{K}\bm{M}_{m,m-1}^{I}(C_{2}^{IJ})^{K}.
\ee
Similarly, $\wt{\bm M}^I_{m,m-1}=\wt{\bm{M}}_{m}^{I}\wt{\bm{M}}_{m-1}^{I}$ give the anti-holomorphic loop algebra. The $*$-structure implies
\be
(\bm{M}_{m,m-1}^{I})_{\ k}^{i}{}^{*}&=&[(\wt{\bm{M}}_{m,m-1}^{I})^{-1}]_{\ i}^{k}.
\ee

The holomorphic quantum Wilson loop is the quantum trace of the quantum holonomy $\bm{W}^I_{m,m-1}:=\tr_{\bm q}^I\lt[\bm{M}_{m}^{I}\bm{M}_{m-1}^{I}	\rt]$, where $\tr_{\bm q}^I$ denotes the quantum trace of the representation $I$ (recall Section \ref{Finite-dimensional irreducible representations}). The Wilson loop operator is the representation of $\bm{W}^I_{m,m-1}$ on $\ch_{\vec \l}$:
\be 
\cd^{\vec \l}\lt(\bm{W}^I_{m,m-1}\rt)=\tr_{\bm q}^I\lt[{\cal D}^{\vec{\lambda}}\left(\bm{M}_{m}^{I}\right){\cal D}^{\vec{\lambda}}\left(\bm{M}_{m-1}^{I}\right)	\rt]=\lt(\pi^{\l_{m-1}}\otimes\pi^{\l_m}\rt)\Delta_{m-1}\tr_{\bm q}^I\lt[R_{0m-1}^{\prime}R_{0m-1}	\rt],\label{wl1}
\ee
where we have used Lemma \ref{quantumtracelemma}. The anti-holomorphic Wilson loop is defined similarly by $\wt{\bm{W}}^I_{m,m-1}:=\tr_{\wt{\bm q}}^I[\wt{\bm{M}}_{m}^{I}\wt{\bm{M}}_{m-1}^{I}	]$, and
\be 
\cd^{\vec \l}\lt(\wt{\bm{W}}^I_{m,m-1}\rt)=\lt(\pi^{\l_{m-1}}\otimes\pi^{\l_m}\rt)\Delta_{m-1}\tr_{\wt{\bm q}}^I\lt[\wt R_{0m-1}^{\prime}\wt R_{0m-1}	\rt]. ,\label{wlt1}
\ee
Both operator are guage invariant: $\xi(\bm{W}^I_{m,m-1})=\eps(\xi)\bm{W}^I_{m,m-1}$ and $\xi(\wt{\bm{W}}^I_{m,m-1})=\eps(\xi)\wt{\bm{W}}^I_{m,m-1}$, which imply both operators commute with gauge transformations. 

Following the same derivation as in \cite{Alekseev:1994au}, the Wilson loop opeartors satisfy the following (abelian) fusion algebra
\be
\bm{W}^I_{m,m-1}\bm{W}^J_{m,m-1}=\sum_{K=|I-J|}^{I+J}\bm{W}^K_{m,m-1},\qquad \wt{\bm{W}}^I_{m,m-1}\wt{\bm{W}}^J_{m,m-1}=\sum_{K=|I-J|}^{I+J}\wt{\bm{W}}^K_{m,m-1}.\label{fusionalg}
\ee
It implies that each of $\bm{W}^I_{m,m-1},\wt{\bm{W}}^I_{m,m-1}$ for $I>1/2$ is a polynomial $\fp^I$ of $\bm{W}^{1/2}_{m,m-1},\wt{\bm{W}}^{1/2}_{m,m-1}$ with integer coefficients. The $*$-structure and \eqref{inverseM1} - \eqref{inverseM2} implies that 
\be
(\bm{W}^I_{m,m-1})^*=\wt{\bm{W}}^I_{m,m-1}.
\ee
In the case of $I=1/2$, the quantum trace of $R'R$ for relates to the quadratic Casimir by
\be
\tr_{\bm q}^{1/2}\lt[R^{\prime}R\rt]&=&\bm{q}^{-1}K+\bm{q}K^{-1}+\left(\bm{q}-\bm{q}^{-1}\right)^{2}EF=-Q,\\
\tr_{\wt{\bm q}}^{1/2}\lt[\wt{R}^{\prime}\wt{R}\rt]&=&\widetilde{\bm{q}}^{-1}\widetilde{K}+\widetilde{\bm{q}}\widetilde{K}^{-1}+\left(\widetilde{\bm q}-\widetilde{\bm q}^{-1}\right)^{2}\widetilde{E}\widetilde{F}=-\widetilde{Q}.
\ee
They relate the Wilson loop operator of the Chern-Simons theory to the quadratic Casimir of the quantum group. By \eqref{wl1} and \eqref{wlt1}, the Wilson loop operators $\cd^{\vec \l}\lt({\bm{W}}^I_{m,m-1}\rt)$ and $\cd^{\vec \l}\lt(\wt{\bm{W}}^I_{m,m-1}\rt)$ with $I=1/2$ equal $-\Delta Q$ and $-\Delta \wt Q$ represented by $\pi^{\l_{m-1}}\otimes\pi^{\l_m}$, while $\cd^{\vec \l}\lt({\bm{W}}^I_{m,m-1}\rt)$ with $I>1/2$ are the representations of the polynomials of $-\Delta Q$ and $-\Delta \wt Q$. Recall \eqref{Q2pp1} amd \eqref{Q2pp2}, we have
\be
\cd^{\vec \l}\lt({\bm{W}}^{I}_{m,m-1}\rt)&=&\cu_{m-1,m}^{-1}Q_m''\cu_{m-1,m}=\mathbb{U}_{m-1,m}^{-1}\lt[\int \fp^I\lt(\chi_{m-1}+\chi_{m-1}^{-1}\rt)\rmd P_{\chi_{m-1}}\rt]\mathbb{U}_{m-1,m}\nonumber\\
&=&\mathbb{U}^{-1}\lt[\int \fp^I\lt(\chi_{m-1}+\chi_{m-1}^{-1}\rt)\rmd P_{\chi_m}\rt]\mathbb{U}\\
\cd^{\vec \l}\lt(\wt{\bm{W}}^{I}_{m,m-1}\rt)&=&\cu_{m-1,m}^{-1}\wt Q_m''\cu_{m-1,m}=\mathbb{U}_{m-1,m}^{-1}\lt[\int \fp^I\lt(\overline\chi_{m-1}+\overline\chi_{m-1}^{-1}\rt)\rmd P_{\chi_{m-1}}\rt]\mathbb{U}_{m-1,m}\nonumber\\
&=&\mathbb{U}^{-1}\lt[\int \fp^I\lt(\overline\chi_{m-1}+\overline\chi_{m-1}^{-1}\rt)\rmd P_{\chi_m}\rt]\mathbb{U}
\ee
Therefore, the Wilson loop operators $\cd^{\vec \l}\lt({\bm{W}}^{I}_{m,m-1}\rt)$, $\cd^{\vec \l}\lt(\wt{\bm{W}}^{I}_{m,m-1}\rt)$ are diagonalized in the representation \eqref{frepchim}, i.e. they are multiplication operators $\fp^I\lt(\chi_{m-1}+\chi_{m-1}^{-1}\rt)$ and $\fp^I\lt(\overline\chi_{m-1}+\overline\chi_{m-1}^{-1}\rt)$. The polynomial $\fp^I(x)$ satisfies $\fp^0(x)=1$, $\fp^{1/2}(x)=x$, $\overline{\fp^{I}(x)}=\fp^{I}(\overline{x})$ and the fusion algebra
\be
\fp^I\fp^J=\sum_{K=|I-J|}^{I+J}\fp^K.
\ee

The quantum holonomies enclosing three consecutive holes 
\be
\bm{M}^I_{m,m-1,m-2}\equiv\bm{M}_{m}^{I}\bm{M}_{m-1}^{I}\bm{M}_{m-2}^{I},\qquad \wt{\bm{M}}^I_{m,m-1,m-2}\equiv\wt{\bm{M}}_{m}^{I}\wt{\bm{M}}_{m-1}^{I}\wt{\bm{M}}_{m-2}^{I}
\ee
give similar results: Their matrix elements and the identity $\bm{1}$ form a loop algebra. Their Wilson loops $\bm{W}^I_{m,m-1,m-2}:=\tr_{\bm q}^I[\bm{M}^I_{m,m-1,m-2}]$ and $\wt{\bm{W}}^I_{m,m-1,m-2}:=\tr_{\bm q}^I[\wt{\bm{M}}^I_{m,m-1,m-2}]$ are gauge invariant and form the fusion algebra the same as in \eqref{fusionalg}. The Wilson loop operators are given by
\be
\cd^{\vec \l}\lt(\bm{W}^I_{m,m-1,m-2}\rt)&=&\lt(\pi^{\l_{m-2}}\otimes\pi^{\l_{m-1}}\otimes\pi^{\l_m}\rt)\Delta_{m-1}^{(3)}\tr_{\bm q}^I\lt[R_{0m-1}^{\prime}R_{0m-1}	\rt]\\
\cd^{\vec \l}\lt(\wt{\bm{W}}^I_{m,m-1,m-2}\rt)&=&\lt(\pi^{\l_{m-2}}\otimes\pi^{\l_{m-1}}\otimes\pi^{\l_m}\rt)\Delta^{(3)}_{m-1}\tr_{\wt{\bm q}}^I\lt[\wt R_{0m-1}^{\prime}\wt R_{0m-1}	\rt]
\ee
where $\Delta^{(3)}=(\mathrm{id}\otimes\Delta)\Delta=(\Delta\otimes\mathrm{id})\Delta$. It is sufficient to only consider the Wilson loop operators with $I=1/2$ due to the fusion algebra. 
\be 
&&\cd^{\vec \l}\lt(\bm{W}^{1/2}_{m,m-1,m-2}\rt)
\nonumber\\
&=&-\sum_\alpha \pi^{\l_{m-2}}\lt(Q^{(1)}_\alpha\rt)\lt(\pi^{\l_{m-1}}\otimes\pi^{\l_m}\rt)\lt(\Delta Q^{(2)}_\a\rt)\nonumber\\
&=&- \mathbb{U}_{m-1,m}^{-1}\lt[\int \sum_\alpha\pi^{\l_{m-2}}\lt(Q^{(1)}_\alpha\rt) \pi^{\chi_{m-1}}\lt( Q^{(2)}_\a\rt)\otimes \rmd P_{\chi_{m-1}}\rt]\mathbb{U}_{m-1,m}\nonumber\\
&=&- \mathbb{U}_{m-1,m}^{-1}\mathbb{U}_{m-2,m-1}^{-1}\lt[\int Q''_{m-1}\otimes \rmd P_{\chi_{m-1}}\rt]\mathbb{U}_{m-2,m-1}\mathbb{U}_{m-1,m}\nonumber\\
&=&- \mathbb{U}^{-1}\lt[\int (\chi_{m-2}+\chi^{-1}_{m-2})\rmd P_{\chi_{m-2}}\otimes \rmd P_{\chi_{m-1}}\rt]\mathbb{U}.
\ee
Then by the fusion algebra
\be
\cd^{\vec \l}\lt(\bm{W}^{I}_{m,m-1,m-2}\rt)=- \mathbb{U}^{-1}\lt[\int \fp^I\lt(\chi_{m-2}+\chi^{-1}_{m-2}\rt)\rmd P_{\chi_{m-2}}\otimes \rmd P_{\chi_{m-1}}\rt]\mathbb{U}
\ee
A similar diagonalization is obtained for $\cd^{\vec \l}(\wt{\bm{W}}^{I}_{m,m-1,m-2})$:
\be 
\cd^{\vec \l}(\wt{\bm{W}}^{I}_{m,m-1,m-2})=- \mathbb{U}^{-1}\lt[\int \fp^I\lt(\overline\chi_{m-2}+\overline\chi^{-1}_{m-2}\rt)\rmd P_{\chi_{m-2}}\otimes \rmd P_{\chi_{m-1}}\rt]\mathbb{U}.
\ee
Therefore, the Wilson loop operators $\cd^{\vec \l}\lt({\bm{W}}^{I}_{m,m-1,m-2}\rt)$, $\cd^{\vec \l}\lt(\wt{\bm{W}}^{I}_{m,m-1,m-2}\rt)$ are multiplication operators $\fp^I(\chi_{m-2}+\chi_{m-2}^{-1})$ and $\fp^I(\overline{\chi}_{m-2}+\overline{\chi}_{m-2}^{-1})$ in the representation \eqref{frepchim}. 

This generalizes straight-forwardly to the quantum holonomies enclosing $h+1$ consecutive holes 
\be
\bm{M}^I_{m,\cdots,m-h}\equiv\bm{M}_{m}^{I}\bm{M}_{m-1}^{I}\cdots\bm{M}_{m-h}^{I},\qquad \wt{\bm{M}}^I_{m,\cdots,m-h}\equiv\wt{\bm{M}}_{m}^{I}\wt{\bm{M}}_{m-1}^{I}\cdots\wt{\bm{M}}_{m-h}^{I},
\ee
for $h\leq m-3$. The holonomies satisfy the loop algebra. Their Wilson loops $\bm{W}^I_{m,\cdots,m-h}:=\tr_{\bm q}^I[\bm{M}^I_{m,\cdots,m-h}]$ and $\wt{\bm{W}}^I_{m,\cdots,m-h}:=\tr_{\bm q}^I[\wt{\bm{M}}^I_{m,\cdots,m-h}]$ are gauge invariant and satisfy the fusion algebra. The $*$-structure implies $({\bm{W}}^I_{m,\cdots,m-h})^*=\wt{\bm{W}}^I_{m,\cdots,m-h}$. $\cd^{\vec \l}\lt({\bm{W}}^{I}_{m,\cdots,m-h}\rt)$, $\cd^{\vec \l}\lt(\wt{\bm{W}}^{I}_{m,\cdots,m-h}\rt)$ are multiplication operators $\fp^I(\chi_{m-h}+\chi_{m-h}^{-1})$ and $\fp^I(\overline{\chi}_{m-h}+\overline{\chi}_{m-h}^{-1})$ in the representation \eqref{frepchim}.

The above discussion shows that the Wilson loop operators $\cd^{\vec \l}\lt({\bm{W}}^{I}_{m,\cdots,m-h}\rt)$, $\cd^{\vec \l}\lt(\wt{\bm{W}}^{I}_{m,\cdots,m-h}\rt)$ for $h=1,\cdots, m-3$ are simultaneously diagonalized in the representation \eqref{frepchim}, and clearly they are mutually commutative. Moreover, these Wilson loops are along a set of circles that are precisely the cuts making a pants decomposition of the $m$-holed sphere. See FIG.\ref{pants} for an example of $m=5$. The diagonalizations of the Wilson loop operators are equivalent to the direct-integral representation in \eqref{UHW}, so this clarifies the correspondence between the procedure in \eqref{UHW} and a pants decomposition of the $m$-holed sphere.

\subsection{Fenchel-Nielsen representation of the physical Hilbert space}

The Wilson loop operators discussed above belong to the algebra $\sm$ defined in Section \ref{RAQ}. The duals of the Wilson loop operators act on the gauge invariants in the physical Hilbert space $\ch^{\vec \l}_{phys}\simeq \cw$: Recall \eqref{PsiFfexplicit}, we derive for $h=1,\cdots,m-3$,
\be
&&\Psi_F\lt[\cd^{\vec \l}\lt({\bm{W}}^{I}_{m,\cdots,m-h}\rt)f\rt]\nonumber\\
&=&\sum _{n_1,{m} \in \mathbb{Z} /N\mathbb{Z}}\int d\nu _{1} d\mu \int\prod_{i=3}^{m-1}d\varrho _{\chi_i}  \overline{\psi _{\lambda _{1}}( \mu ,m)} \, \overline{F(\vec{\mu }_\chi,\vec{m}_\chi)} \, \bm{V}_{1}^{-1}\fp^I\lt(\chi_{m-h}+\chi_{m-h}^{-1}\rt)f'\lt( \nu _{1} ,n_{1} ;-\nu _{1} -n_{1} ;\mu ,m ;\vec{\mu }_\chi,\vec{m}_\chi\rt),\nonumber\\
&=&\sum _{n_1,{m} \in \mathbb{Z} /N\mathbb{Z}}\int d\nu _{1} d\mu \int\prod_{i=3}^{m-1}d\varrho _{\chi_i}  \overline{\psi _{\lambda _{1}}( \mu ,m)} \, \overline{\fp^I\lt(\overline\chi_{m-h}+\overline\chi_{m-h}^{-1}\rt)F(\vec{\mu }_\chi,\vec{m}_\chi)} \, \bm{V}_{1}^{-1}f'\lt( \nu _{1} ,n_{1} ;-\nu _{1} -n_{1} ;\mu ,m ;\vec{\mu }_\chi,\vec{m}_\chi\rt),\nonumber\\
&=&\Psi_{\fp^I\lt(\overline\chi_{m-h}+\overline\chi_{m-h}^{-1}\rt)F}\lt[f\rt]
\ee
and similarly
\be
\Psi_F\lt[\cd^{\vec \l}\lt(\wt{\bm{W}}^{I}_{m,\cdots,m-h}\rt)f\rt]=\Psi_{\fp^I\lt(\chi_{m-h}+\chi_{m-h}^{-1}\rt)F}\lt[f\rt].
\ee
Due to the 1-to-1 correspondence $\Psi_F\leftrightarrow F$, these define a $*$-representation Wilson loops denoted by ${\mathsf{W}}^{I}_{m,\cdots,m-h}$ and $\wt{\mathsf{W}}^{I}_{m,\cdots,m-h}$, $h=1,\cdots,m-3$ on $\ch_{phys}^{\vec \l}$: For any function $F(\vec{\mu}_\chi,\vec{m}_\chi)$ in a dense domain in $\ch_{phys}^{\vec\l}$, 
\be 
{\mathsf{W}}^{I}_{m,\cdots,m-h}F\lt(\vec{\mu}_\chi,\vec{m}_\chi\rt)&=&\fp^I\lt(\chi_{m-h}+\chi_{m-h}^{-1}\rt)F\lt(\vec{\mu}_\chi,\vec{m}_\chi\rt),\label{Wdiagonal1}\\
\wt{\mathsf{W}}^{I}_{m,\cdots,m-h}F\lt(\vec{\mu}_\chi,\vec{m}_\chi\rt)&=&\fp^I\lt(\overline\chi_{m-h}+\overline\chi_{m-h}^{-1}\rt)F\lt(\vec{\mu}_\chi,\vec{m}_\chi\rt).\label{Wdiagonal2}
\ee
The Wilson loops diagonalized by this representation are along the cuts for the pants decomposition. Classically, they are complexifications of the Fenchel-Nielsen (FN) coordinates of the Teichm\"uller space. Therefore, it is natural to call the representation of $\ch_{phys}^{\vec \l}$ as the $L^2$-space of $F(\vec{\mu}_\chi,\vec{m}_\chi)$ the FN representation. The FN representation of the complex Chern-Simons theory has been derived from the quantization of Fock-Goncharov coordinates in the case of 4-holed sphere \cite{Han:2024nkf}. The above discussion derives the FN representation on $m$-holed sphere for arbitrary $m\geq 4$ from the combinatorial quantization of the theory.

\subsection{Crossing symmetry}

Different pants decompositions of the $m$-holed sphere should correspond to different FN representations of the physical Hilbert space, and any two different FN representations should be related by a unitary transformation. In the following, we consider the situation of the 4-holed sphere: two different pants decompositions are obtained by choosing to cut the 4-holed sphere along the S-cycle (circling the 1st and 2nd holes) or the T-cycle (circling the 2nd and 3rd holes) gives, see FIG.\ref{STcycle}. The change between these two pants decompositions on the 4-holed sphere is the elementary move whose compositions give general pants decompositions on $m$-holed sphere for $m\geq 4$. It is referred to as the A-move in the Moore-Seiberg groupoid \cite{Moore:1988qv}.

\begin{figure}[h]
	\centering
	\includegraphics[width=0.9\textwidth]{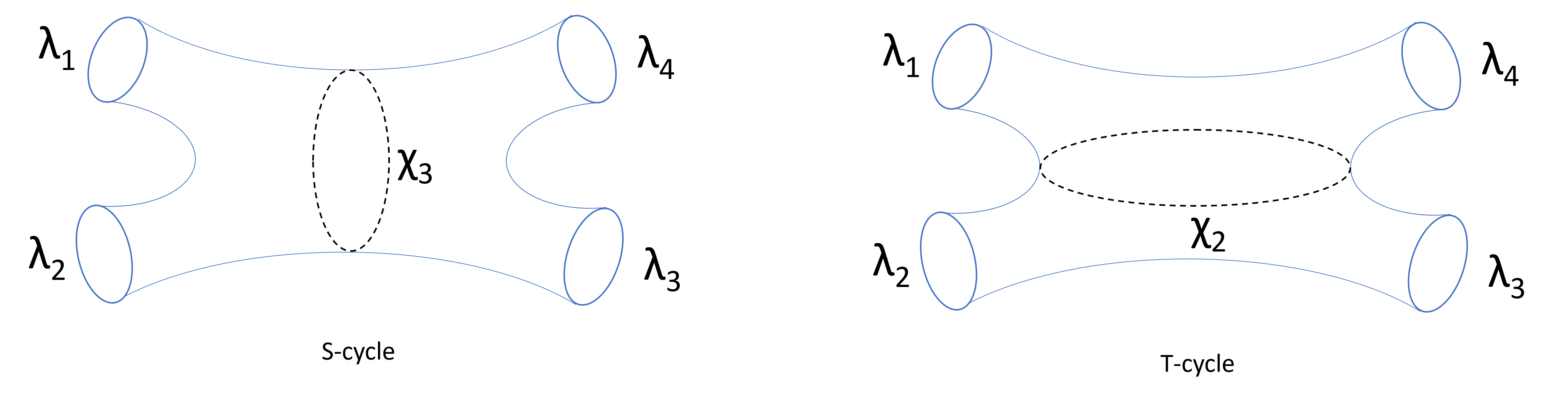}
	\caption{A change of pants decomposition from cutting along the S-cycle to cutting along the T-cycle.}
	\label{STcycle}
\end{figure}

By the coassociativity, $\Delta^{(4)}\xi=(1\otimes 1\otimes\Delta)(1\otimes\Delta)\Delta \xi=(1\otimes\Delta\otimes 1)(1\otimes\Delta)\Delta \xi$. Two equivalent decompositions of $\Delta^{(4)}$ relates to two different pants decompositions: First, based on $\Delta^{(4)}\xi=(1\otimes 1\otimes\Delta)(1\otimes\Delta)\Delta \xi$, we construct the following unitary transformation in the same way as before
\be
\mathbb{U}_S=\mathbb{U}_{2,3}\mathbb{U}_{3,4}: \quad \ch_{\l_1}\otimes\cdots\otimes\ch_{\l_4}\to \int^\oplus\rmd\varrho_{\chi}\lt(\ch_{\l_1}\otimes\ch_{\chi}\rt)\otimes \cw_\chi.
\ee
The physical Hilbert space is
\be
\cw^{(S)}\equiv\cw_{\l_1}\simeq \int^\oplus\rmd\varrho_{\chi_3}W_{\l_1}\otimes W_{\chi_3},\label{WS}
\ee
$\cw^{(S)}$ is isomorphic to an $L^2$-type Hilbert space $L^2(\C,\rmd\varrho_{\chi_3})$, in which the states are represented by functions $F(\mu_{\chi_3},m_{\chi_3})$ ($W_{\l_1}\simeq \C$ is 1-dimensional). The Wilson loop operators along the S-cycle, $\bm{W}_S^I=\bm{W}_{4,3}^I$ and  $\wt{\bm{W}}_S^I=\wt{\bm{W}}_{4,3}^I$ are diagonalized on $\ch_{\vec \l}$ by $\mathbb{U}_S$. $\mathsf{W}^I_{S}=\mathsf{W}^I_{4,3}$ and $\wt{\mathsf{W}}^I_{S}=\wt{\mathsf{W}}^I_{4,3}$ are diagonalized in the above direct integral representation of $\cw^{(S)}$ as in \eqref{Wdiagonal1} and \eqref{Wdiagonal2} (take $m=4$ and $h=1$):
\be 
\mathsf{W}_{S}^IF\lt(\mu_{\chi_3},m_{\chi_3}\rt)&=&\fp^I\lt(\chi_{3}+\chi_{3}^{-1}\rt)F\lt(\mu_{\chi_3},m_{\chi_3}\rt),\\
\wt{\mathsf{W}}_{S}^IF\lt(\mu_{\chi_3},m_{\chi_3}\rt)&=&\fp^I\lt(\overline\chi_{3}+\overline\chi_{3}^{-1}\rt)F\lt(\mu_{\chi_3},m_{\chi_3}\rt).
\ee

Based on $\Delta^{(4)}\xi=(1\otimes\Delta\otimes 1)(1\otimes\Delta)\Delta \xi$, we use
\be
\mathbb{U}_{2,3}:&& \ch_{\l_1}\otimes\cdots\otimes\ch_{\l_4}\to \int^\oplus\rmd\varrho_{\chi_2}\ch_{\l_1}\otimes\ch_{\chi_2}\otimes W_{\chi_2}\otimes \ch_{\l_4},\\
\mathbb{U}_{T}=\mathbb{U}_{2,4}\mathbb{U}_{2,3}:&& \ch_{\l_1}\otimes\cdots\otimes\ch_{\l_4}\to \int^\oplus\rmd\varrho_{\chi}\rmd\varrho_{\chi_2}\ch_{\l_1}\otimes\ch_{\chi}\otimes W_{\chi_2}\otimes W_{\chi}
\ee
The physical Hilbert space is given by
\be
\cw^{(T)}\simeq \int^\oplus\rmd\varrho_{\chi_2} W_{\chi_2}\otimes W_{\l_1}.\label{WT}
\ee
The Wilson loop operators along the T-cycle, $\bm{W}_{3,2}$ and $\wt{\bm{W}}_{3,2}$, relates to $Q_{23}$ and $\wt{Q}_{23}$, which are diagonalized on $\ch_{\vec\l}$ by $\mathbb{U}_{2,3}$. Similar to $\cw^{(S)}$, the states in $\cw^{(T)}$ are represented by functions $F(\mu_{\chi_2},m_{\chi_2})$. For $F(\mu_{\chi_2},m_{\chi_2})$ in a dense domain of $ \cw^{(T)}$, the Wilson loop operators $\mathsf{W}^I_T=\mathsf{W}^I_{3,2}$ and $\wt{\mathsf{W}}^I_{T}=\wt{\mathsf{W}}^I_{3,2}$ are diagonalized
\be 
\mathsf{W}^I_T F\lt(\mu_{\chi_2},m_{\chi_2}\rt)&=&\fp^I\lt(\chi_{2}+\chi_{2}^{-1}\rt)F\lt(\mu_{\chi_2},m_{\chi_2}\rt),\label{Wdiagonal1}\\
\wt{\mathsf{W}}^I_{T} F\lt(\mu_{\chi_2},m_{\chi_2}\rt)&=&\fp^I\lt(\overline\chi_{2}+\overline\chi_{2}^{-1}\rt)F\lt(\mu_{\chi_2},m_{\chi_2}\rt).\label{Wdiagonal2}
\ee

The physical Hilbert spaces from different pants decompositions are unitary equivalent, since both $\cw^{(S)}$ and $\cw^{(T)}$ are isomorphic to $L^2(\C,d\rho_{\chi})$. Two direction integral representations \eqref{WS} and \eqref{WT} are unitary equivalent representations of the same Hilbert space $\cw\simeq \cw^{(S)}\simeq \cw^{(T)}$. 

Since $\mathbb{U}_S$ and $\mathbb{U}_T$ diagonalize $\bm{W}^I_S$ and $\bm{W}^I_T$ respectively, we have the equality
\be
\mathbb{U}_S\cd^{\vec\l}(\bm{W}^I_S)\mathbb{U}_S^{-1}=\mathbb{U}_T\cd^{\vec\l}(\bm{W}^I_T)\mathbb{U}_T^{-1},\qquad \mathbb{U}_S\cd^{\vec\l}(\wt{\bm{W}}^I_S)\mathbb{U}_S^{-1}=\mathbb{U}_T\cd^{\vec\l}(\wt{\bm{W}}^I_T)\mathbb{U}_T^{-1},
\ee
which equal the multiplication operators $\fp^I(\chi+\chi^{-1})$ and $\fp^I(\wt\chi+\wt\chi^{-1})$. The unitary operator $\mathbb{U}_T^{-1}\mathbb{U}_S$ transforms $\bm{W}^I_S$ to $\bm{W}^I_T$ and thus realizes on $\ch_{\vec\l}$ the A-move in Moore-Seiberg groupoid. 


\section*{Acknowledgements}

M.H. acknowledges Chen-Hung Hsiao and Qiaoyin Pan for helpful discussions. M.H. receives supports from the National Science Foundation through grant PHY-2207763, the College of Science Research Fellowship at Florida Atlantic University, and the Blaumann Foundation.  

\appendix

\section{Duality and $\star$-Hopf algebra structure}\label{duality and star-Hopf algebra strucutre}

The duality between ${SL}_{\bfq}(2)\otimes {SL}_{\widetilde \bfq}(2)$ and $\suquqt$ is given by 
\be
\lag (g^I)^i_{\ j}\otimes (\widetilde{g}^J)^k_{\ l},\xi\otimes \widetilde{\zeta}\rag = \rho^{I}(\xi)^i_{\ j}\rho^{J}(\widetilde{\zeta})^k_{\ l}.
\ee
In a more concise notation, one may write $(g^I)^i_{\ j}(\widetilde{g}^J)^k_{\ l}=(g^I)^i_{\ j}\otimes(\widetilde{g}^J)^k_{\ l}\in {SL}_{\bfq}(2)\otimes {SL}_{\widetilde \bfq}(2)$ and $\xi\widetilde{\zeta}=\xi\otimes \widetilde{\zeta}\in \suquqt$.

The identity of ${SL}_{\bfq}(2)$ or $ {SL}_{\widetilde \bfq}(2)$ is dual to the counit of $\suq $ or $  \suqt$, i.e.
\be
\lag (g^I)^i_{\ j},\xi \widetilde{\zeta}\rag = \rho^{I}(\xi)^i_{\ j}\eps(\widetilde{\zeta}),\qquad \lag  (\widetilde{g}^J)^k_{\ l},\xi\otimes \widetilde{\zeta}\rag = \eps(\xi) \rho^{J}(\widetilde{\zeta})^k_{\ l}.
\ee
the counit $\eps$ is the same as the trivial representation.

The multiplication of ${SL}_{\bfq}(2)\otimes {SL}_{\widetilde \bfq}(2)$ is the dual of the comultiplication of $\suquqt$: First of all, $(g^I)^i_{\ j}(\widetilde{g}^J)^k_{\ l}$ can be understood as the multiplication between $(g^I)^i_{\ j}$ and $(\widetilde{g}^J)^k_{\ l}$, consistent with the duality
\be 
&&\langle (g^I)^i_{\ j}(\widetilde{g}^J)^k_{\ l},\xi\widetilde{\zeta}\rangle=\langle (g^I)^i_{\ j}\otimes(\widetilde{g}^J)^k_{\ l},\Delta\left(\xi\widetilde{\zeta}\right)\rangle=\langle (g^I)^i_{\ j}\otimes(\widetilde{g}^J)^k_{\ l},\Delta\xi\Delta\widetilde{\zeta}\rangle\nonumber\\
&=&\sum_{\alpha,\beta}\langle (g^I)^i_{\ j}\otimes(\widetilde{g}^J)^k_{\ l},\xi_{\alpha}^{1}\widetilde{\zeta}_{\beta}^{1}\otimes\xi_{\alpha}^{2}\widetilde{\zeta}_{\beta}^{2}\rangle
=\sum_{\alpha,\beta}\rho^{I}\left(\xi_{\alpha}^{1}\right)^i_{\ j}\varepsilon\left(\xi_{\alpha}^{2}\right)\varepsilon\left(\widetilde{\zeta}_{\beta}^{1}\right)\rho^{J}\left(\widetilde{\zeta}_{\beta}^{2}\right)^k_{\ l}=\rho^{I}(\xi)^i_{\ j}\rho^{J}(\widetilde{\zeta})^k_{\ l}
\ee 
A similar computation shows that $\langle (\widetilde{g}^J)^k_{\ l}(g^I)^i_{\ j},\xi\widetilde{\zeta}\rangle$ gives the same result, so we obtain $[(g^I)^i_{\ j},(\widetilde{g}^J)^k_{\ l}]=0$.

The multiplication between $(g^I)^i_{\ j} $ and $ ({g}^J)^k_{\ l}$ can be studied by the following computation ($\widetilde{\zeta}$ are ignored below since inserting $\widetilde{\zeta}$ only adds a factor $\eps(\widetilde{\zeta})$ to each step)
\be 
\lag (g^I)^i_{\ j}({g}^J)^k_{\ l},\xi\rag&=&\lag (g^I)^i_{\ j}\otimes({g}^J)^k_{\ l},\Delta \xi\rag=\sum_\alpha\rho^{I}(\xi^1_\a)^i_{\ j}\rho^{J}(\xi^2_\a)^k_{\ l}\nonumber\\
\lag ({g}^J)^k_{\ l}(g^I)^i_{\ j},\xi\rag&=&\lag ({g}^J)^k_{\ l}\otimes (g^I)^i_{\ j},\Delta \xi\rag=\sum_\alpha\rho^{I}(\xi^2_\a)^i_{\ j}\rho^{J}(\xi^1_\a)^k_{\ l}\nonumber\\
&=&\lag (g^I)^i_{\ j}\otimes({g}^J)^k_{\ l},\Delta^\prime \xi\rag=\lag (g^I)^i_{\ j}\otimes({g}^J)^k_{\ l},R\Delta (\xi)R^{-1}\rag\nonumber\\
&=&\lt[R^{IJ}\lag g^I\otimes{g}^J,\Delta (\xi)\rag (R^{-1})^{IJ}\rt]^{i,\ k}_{\ j,\ l}=\lt[R^{IJ}\lag g^I{g}^J,\xi\rag (R^{-1})^{IJ}\rt]^{i,\ k}_{\ j,\ l}
\ee
A similar computation can be done in the tilded sector. The results imply
\be
R^{IJ}  g^I{g}^J=g^J g^I R^{IJ},\qquad \widetilde{R}^{IJ} \widetilde{ g}^I\widetilde{g}^J=\widetilde{g}^J \widetilde{g}^I \widetilde{R}^{IJ}.\label{Rgg=ggR}
\ee
where $g^I\in \mathrm{End}(V^I)\otimes SL_{\bfq}(2)$ and $\widetilde{g}^I\in \mathrm{End}(V^I)\otimes SL_{\widetilde \bfq}(2)$. They are the same relations obtained by quantizing the poisson Lie group of gauge transformations \eqref{poissongg1} - \eqref{poissongg3}.

It is straight-forward to check that $\langle \det {}_{\bfq}(g),\xi\widetilde{\zeta}\rangle=\langle \det{}_{\widetilde \bfq}(\widetilde{g}),\xi\widetilde{\zeta}\rangle=\eps(\xi\widetilde{\zeta})$, which implies $\det {}_{\bfq}(g)=\det{}_{\widetilde \bfq}(\widetilde{g})=1$. 

The comultiplication of $(g^I)^i_{\ j}$ is obtained by
\be
\lag \delta (g^I)^i_{\ j},\xi\otimes\zeta \rag=\lag (g^I)^i_{\ j},\xi\zeta\rag =\rho^I(\xi\zeta)^i_{\ j}=\rho^I(\xi)^i_{\ k}\rho^I(\zeta)^k_{\ j}=\lag (g^I)^i_{\ k}\otimes (g^I)^k_{\ j},\xi\otimes\zeta \rag
\ee
The comultiplication of $(\widetilde{g}^I)^i_{\ j}$ can be obtained similarly.

The relations $\cs(g^I)=(g^I)^{-1}$ is implied by $m(\cs\otimes \mathrm{id})\delta g^I=m(\mathrm{id}\otimes \cs)\delta g^I =\epsilon(g^I) 1$ where $m$ denotes the multiplication, and similar for $\cs(\widetilde{g}^I )=(\widetilde{g}^I)^{-1}$. In addition,
\be 
\lag \cs^{2n}(g^I),\xi\rag =\lag g^I,S^{2n}(\xi)\rag =\lag g^I,u^n\xi (u^{-1})^n\rag=\rho^I\lt(u\rt)^n\lag g^I,\xi\rag\rho^I\lt(u^{-1}\rt)^{n},\qquad u=\bfq^{H}
\ee 
A similar relation holds for $\cs^{2n}(\widetilde g^I)$ with $\widetilde{u}=\overline{\bfq}^{\widetilde H}$. They imply 
\be 
\cs^{2n}(g^I)=\rho^I\lt(u\rt)^n g^I \rho^I\lt(u^{-1}\rt)^n,\qquad \cs^{2n}(\widetilde{g}^I)=\rho^I\lt(\widetilde{u}\rt)^n \widetilde{g}^I \rho^I\lt(\widetilde{u}^{-1}\rt)^n
\ee
For example,
\be 
\cs^{2n}(g^{1/2})=\left(
	\begin{array}{cc}
		\bfq^{n} & 0\\
		0 & \bfq^{-n}
	\end{array}
	\rt)g^{1/2}\left(
		\begin{array}{cc}
			\bfq^{-n} & 0\\
			0 & \bfq^{n}
		\end{array}
		\rt).
\ee

The $\star$-structure of the dual Hopf algebra $\cg^\prime$ is induced from the $*$-Hopf algebra $\cg$ by
\be 
\lag \phi^\star, \xi\rag=\overline{\lag \phi, S(\xi)^*\rag},\qquad \phi\in \cg',\quad \xi\in\cg.
\ee
Apply this relation to $\cg^\prime={SL}_{\bfq}(2)\otimes {SL}_{\widetilde \bfq}(2)$ and $\cg=\suquqt$, we first insert the generators $\widetilde{\cx},\widetilde{\cy},\widetilde{\ck}^{\pm1}$:
\be 
\left\langle (g^{I})_{\ j}^{i}{}^{\star},\widetilde{{\cal X}}\right\rangle 
&=&\overline{\left\langle (g^{I})_{\ j}^{i},S\left(\widetilde{{\cal X}}\right)^{*}\right\rangle }=\overline{\left\langle (g^{I})_{\ j}^{i},\left(-\widetilde{\boldsymbol{q}}\widetilde{\mathcal{X}}\right)^{*}\right\rangle }=-\widetilde{\boldsymbol{q}}\overline{\left\langle (g^{I})_{\ j}^{i},\mathcal{X}\right\rangle }=-\widetilde{\boldsymbol{q}}\overline{\rho^{I}(\mathcal{X})_{\ j}^{i}}\nonumber\\
&=&-\widetilde{\boldsymbol{q}}\widetilde{\rho}^{I}(\widetilde{\mathcal{X}})_{\ i}^{j}=\left\langle S\left((\widetilde{g}^{I})_{\ i}^{j}\right),\widetilde{{\cal X}}\right\rangle \\
\left\langle (g^{I})_{\ j}^{i}{}^{\star},\widetilde{{\cal Y}}\right\rangle 	
&=&\overline{\left\langle (g^{I})_{\ j}^{i},S\left(\widetilde{{\cal Y}}\right)^{*}\right\rangle }=\overline{\left\langle (g^{I})_{\ j}^{i},\left(-\widetilde{\boldsymbol{q}}^{-1}\widetilde{\mathcal{Y}}\right)^{*}\right\rangle }=-\widetilde{\boldsymbol{q}}^{-1}\overline{\left\langle (g^{I})_{\ j}^{i},\mathcal{Y}\right\rangle }=-\widetilde{\boldsymbol{q}}^{-1}\overline{\rho^{I}(\mathcal{Y})_{\ j}^{i}}\nonumber\\
&=&-\widetilde{\boldsymbol{q}}^{-1}\widetilde{\rho}^{I}(\widetilde{\mathcal{Y}})_{\ i}^{j}=\left\langle S\left((\widetilde{g}^{I})_{\ i}^{j}\right),\widetilde{{\cal Y}}\right\rangle \\
\left\langle (g^{I})_{\ j}^{i}{}^{\star},\widetilde{{\cal K}}^{\pm1}\right\rangle 	
&=&\overline{\left\langle (g^{I})_{\ j}^{i},S\left(\widetilde{{\cal K}}^{\pm1}\right)^{*}\right\rangle }=\overline{\left\langle (g^{I})_{\ j}^{i},\left(\widetilde{{\cal K}}^{\mp1}\right)^{*}\right\rangle }=\overline{\left\langle (g^{I})_{\ j}^{i},{\cal K}^{\mp1}\right\rangle }=\overline{\rho^{I}({\cal K}^{\mp1})_{\ j}^{i}}\nonumber\\
&=&\widetilde{\rho}^{I}(\widetilde{\mathcal{K}}^{\mp1})_{\ i}^{j}=\left\langle S\left((\widetilde{g}^{I})_{\ i}^{j}\right),\widetilde{{\cal K}}^{\pm1}\right\rangle 
\ee 
For the identity, $\langle (g^{I})_{\ j}^{i}{}^{\star},1\rangle=\overline{\left\langle (g^{I})_{\ j}^{i},S\left(1\right)^{*}\right\rangle }=\delta^i_{\ j}= \langle S\left((\widetilde{g}^{I})_{\ i}^{j}\right),1\rangle$.
Once $\langle (g^I)^i_{\ j}{}^\star,\,\cdot\,\rangle=\langle \cs((\widetilde{g}^I)^j_{\ i}), \,\cdot\,\rangle$ holds for generators, it holds for general monomials of generators, because if $\langle (g^I)^i_{\ j}{}^\star,\widetilde{\xi} \rangle=\langle \cs((\widetilde{g}^I)^j_{\ i}), \widetilde{\xi}\rangle$ and $\langle (g^I)^i_{\ j}{}^\star,\widetilde{\zeta} \rangle=\langle \cs((\widetilde{g}^I)^j_{\ i}), \widetilde{\zeta}\rangle$ hold, then
\be 
\lag (g^I)^i_{\ j}{}^\star,\widetilde{\xi}\widetilde{\zeta}\rag&=&\lag \delta (g^I)^i_{\ j}{}^\star,\widetilde{\xi}\otimes\widetilde{\zeta}\rag=\sum_k\lag  (g^I)^i_{\ k}{}^\star,\widetilde{\xi}\rag\lag (g^I)^k_{\ j}{}^\star,\widetilde{\zeta}\rag\nonumber\\
&=&\sum_k\lag  \cs\Big((\widetilde g^I)^k_{\ i}{}\Big),\widetilde{\xi}\rag\lag \cs\lt((\widetilde g^I)^j_{\ k}{}\rt),\widetilde{\zeta}\rag=\sum_k\lag (\widetilde g^I)^j_{\ k}{},S\lt(\widetilde{\zeta}\rt)\rag\lag  (\widetilde g^I)^k_{\ i}{},S\lt(\widetilde{\xi}\rt)\rag\nonumber\\
&=&\lag  (\widetilde g^I)^j_{\ i},S\lt(\widetilde{\zeta}\rt)S\lt(\widetilde{\xi}\rt)\rag =\lag  \cs\lt((\widetilde g^I)^j_{\ i}\rt),\widetilde{\xi}\widetilde{\zeta}\rag.
\ee 
A similar computation can be done for $(\widetilde g^I)^i_{\ j}{}^\star$. As a result,
\be 
(g^I)^i_{\ j}{}^\star=\cs\lt((\widetilde g^I)^j_{\ i}\rt),\qquad (\widetilde g^I)^i_{\ j}{}^\star=\cs\lt(( g^I)^j_{\ i}\rt). 
\ee 
We check the $\star$-structure is consistent with \eqref{Rgg=ggR}: First, we write the left equation in \eqref{Rgg=ggR} in terms of indices
\be
(R^{IJ})^{i,a}_{\ j,b}(g^{I})^j_{\ k}(g^{J})^b_{\ c}=(g^{J})^a_{\ b}(g^{I})^i_{\ j}(R^{IJ})^{j,b}_{\ k,c}.
\ee
Then we apply $\cs$ and $\star$ to the equation
\be 
\overline{(R^{IJ})^{i,a}_{\ j,b}}\cs\Big((g^{I})^j_{\ k}\Big)^\star \cs\Big((g^{J})^b_{\ c}\Big)^\star=\cs\Big((g^{J})^a_{\ b}\Big)^\star \cs\Big((g^{I})^i_{\ j}\Big)^\star \, \overline{(R^{IJ})^{j,b}_{\ k,c}}.
\ee
Using $(R^{IJ})^{\dagger\otimes \dagger}=(\widetilde{R}^{-1})^{IJ}$ and $\star^2=\mathrm{id}$, we obtain
\be
(\widetilde g^I)^k_{\ j}(\widetilde g^J)^c_{\ b}{(\widetilde{R}^{-1 IJ})^{j,b}_{\ i,a}}=(\widetilde{R}^{-1 IJ})^{k,c}_{\ j,b}(\widetilde g^J)^b_{\ a} (\widetilde g^I)^j_{\ i}.
\ee
Multiplying $R^{IJ}$ from the left and right of the equation. We recover the right equation in \eqref{Rgg=ggR}.

\section{The derivation about $\xi=E$ for Theorem \ref{uniquebilinear}}\label{invbilinearwithE}

\be
0&=&	\left(-\frac{iq^{-1}}{\bm{q-q}^{-1}}\right)^{-1}\Psi\left[U^{-1}U\left(\pi^{\lambda_{1}}\otimes\pi^{\lambda_{2}}\right)\left(\Delta E\right)U^{-1}Uf\right]\nonumber\\
&=&	C\sum_{n_{1},n\in\mathbb{Z}/N\mathbb{Z}}\int d\nu_{1}d\nu\delta\left(\nu\right)\delta_{\exp\left[\frac{2\pi i}{N}n\right],1}\bm{V}_{1}^{-1}\left[U\left(\pi^{\lambda_{1}}\otimes\pi^{\lambda_{2}}\right)\left(\Delta E\right)U^{-1}\right]Uf\left(\nu_{1},\nu,n_{1},n\right)\nonumber\\
&=&	C\sum_{n_{1},n\in\mathbb{Z}/N\mathbb{Z}}\int d\nu_{1}d\nu\delta\left(\nu\right)\delta_{\exp\left[\frac{2\pi i}{N}n\right],1}\bm{V}_{1}^{-1}\left[\left(\lambda_{1}+\lambda_{1}^{-1}\right)\bm{u}_{1,-1}^{(1)}\bm{u}_{-1,-1}+\bm{u}_{2,-1}^{(1)}\bm{u}_{-1,-1}+\bm{u}_{0,-1}^{(1)}\bm{u}_{-1,-1}\rt.\nonumber\\
&&\qquad \lt.+\left(\lambda_{2}+\lambda_{2}^{-1}\right)\bm{u}_{0,-1}+\bm{u}_{-1,0}^{(1)}\bm{u}_{1,-1}+\bm{u}_{1,0}^{(1)}\bm{u}_{-1,-1}\right]Uf\left(\nu_{1},\nu,n_{1},n\right)\nonumber\\
&=&	C\sum_{n_{1},n\in\mathbb{Z}/N\mathbb{Z}}\int d\nu_{1}d\nu\delta\left(\nu\right)\delta_{\exp\left[\frac{2\pi i}{N}n\right],1}\left[-\left(\lambda_{1}+\lambda_{1}^{-1}\right)\bm{q}\bm{u}_{0,-1}^{(1)}\bm{u}_{-1,-1}-\bm{q}\bm{u}_{1,-1}^{(1)}\bm{u}_{-1,-1}-\bm{q}\bm{u}_{-1,-1}^{(1)}\bm{u}_{-1,-1}\rt.\nonumber\\
&&\qquad\lt.+\left(\lambda_{2}+\lambda_{2}^{-1}\right)\bm{u}_{0,-1}+\bm{u}_{-1,0}^{(1)}\bm{u}_{1,-1}+\bm{u}_{1,0}^{(1)}\bm{u}_{-1,-1}\right]\bm{V}_{1}^{-1}Uf\left(\nu_{1},\nu,n_{1},n\right)\nonumber\\
&=&	C\sum_{n_{1},n\in\mathbb{Z}/N\mathbb{Z}}\int d\nu_{1}d\nu\delta\left(\nu\right)\delta_{\exp\left[\frac{2\pi i}{N}n\right],1}\left[-\left(\lambda_{1}+\lambda_{1}^{-1}\right)\bm{u}_{0,-1}^{(1)}-\bm{u}_{1,-1}^{(1)}-\bm{u}_{-1,-1}^{(1)}\rt.\nonumber\\
&&\qquad\lt.+\left(\lambda_{2}+\lambda_{2}^{-1}\right)+\bm{q}\bm{u}_{-1,0}^{(1)}+\bm{q}^{-1}\bm{u}_{1,0}^{(1)}\right]\bm{y}^{-1}\bm{V}_{1}^{-1}Uf\left(\nu_{1},\nu,n_{1},n\right)\nonumber\\
&=&	C\sum_{n_{1},n\in\mathbb{Z}/N\mathbb{Z}}\int d\nu_{1}d\nu\delta\left(\nu\right)\delta_{\exp\left[\frac{2\pi i}{N}n\right],1}\left[-\left(\lambda_{1}+\lambda_{1}^{-1}\right)+\left(\lambda_{2}+\lambda_{2}^{-1}\right)\right]\bm{y}^{-1}\bm{V}_{1}^{-1}Uf\left(\nu_{1},\nu,n_{1},n\right)
\ee

\section{Representations of $R$-elements}\label{rep of Rs}

\be 
R^{1/2,\lambda}&=&\begin{pmatrix}\ck_\l & \left(\bm{q}-\bm{q}^{-1}\right)\bm{q}^{-1/2}\cy_\l\\
0 & \ck_\l^{-1}
\end{pmatrix},\qquad \quad \widetilde{R}^{1/2,\lambda}=\begin{pmatrix}\widetilde{\ck}^{-1}_\l & 0\\
\left(\widetilde{\bm{q}}-\widetilde{\bm{q}}^{-1}\right)\widetilde{\bm{q}}^{-1/2}\widetilde{\cy}_\l & \quad\widetilde{\ck}_\l
\end{pmatrix}\\
R^{\prime 1/2,\lambda}&=&\begin{pmatrix}\ck_\l & 0\\
\left(\bm{q}-\bm{q}^{-1}\right)\bm{q}^{-1/2}{\cal X}_\l\quad & \ck^{-1}_\l
\end{pmatrix},\qquad \widetilde{R}^{\prime1/2,\lambda}=\begin{pmatrix}\widetilde{\ck}^{-1} & \qquad\left(\widetilde{\bm{q}}-\widetilde{\bm{q}}^{-1}\right)\widetilde{\bm{q}}^{-1/2}\widetilde{{\cal X}}_\l\\
0 & \widetilde{\ck}_\l
\end{pmatrix}
\ee
We check that 
\be 
(R^{1/2,\lambda})^{\dagger\otimes\dagger}=(\widetilde{R}^{1/2,\lambda})^{-1}\qquad (R^{\prime1/2,\lambda})^{\dagger\otimes\dagger}=(\widetilde{R}^{\prime1/2,\lambda})^{-1}.
\ee
Eqs.\eqref{MandUq} and \eqref{MandUqtilde} are derived from these representations of $R$-matrices.

\section{Loop equations}\label{Loop equation}

By \eqref{Crho1} and the relation $\Delta \xi=R^{-1}(\Delta'\xi) R$ for $\xi\in \uq$,
	\be 
	(C_{1}^{JI})^{ja}_{\ Kk'}{\rho}^{K}\left(\xi\right)^{k'}_{\ k}&=&\lt[\left(\rho^{J}\otimes\rho^{I}\right)\left(\Delta\xi\right)\rt]^{ja}_{\ j'a'}(C_{1}^{JI})^{j'a'}_{\ \ Kk}\nonumber\\
	&=&[\left(R^{JI}\right)^{-1}]^{ja}_{\ kb}\lt[\left(\rho^{J}\otimes\rho^{I}\right)\left(\Delta^{\prime}\xi\right)\rt]^{kb}_{\ k'b'}(R^{JI})^{k'b'}_{\ \ lc}(C_{1}^{JI})^{lc}_{\ Kk}
	\ee
	If we write $\Delta\xi = \sum_\alpha \xi^{(1)}_\a\otimes\xi^{(2)}_\a$, the above formula can be written as
	\be 
	(R^{JI})^{j a}_{\ kb}(C_{1}^{JI})^{kb}_{\ Kk'}{\rho}^{K}\left(\xi\right)^{k'}_{\ k}=\sum_\a\rho^J(\xi^{(2)}_\a)^k_{\ k'}\rho^I(\xi^{(1)}_\a)^b_{\ b'}(R^{JI})^{k'b'}_{\ \ lc}(C_{1}^{JI})^{lc}_{\ Kk}.
	\ee
	Since $(R'^{IJ})^{ a j}_{\ bk}=(R^{JI})^{j a}_{\ kb}$, we obtain
	\be 
	(R'^{IJ})^{ a j}_{\ bk}(C_{1}^{JI})^{kb}_{\ Kk'}{\rho}^{K}\left(\xi\right)^{k'}_{\ k}=\lt[\rho^I\otimes\rho^J(\Delta\xi)\rt]^{bk}_{\ b'k'}(R'^{IJ})^{b'k'}_{\ \ cl}(C_{1}^{JI})^{lc}_{\ Kk}
	\ee
	Insert $[\lt(\rho^I\otimes\rho^J\rt)(\Delta\xi)]_{\ \ bk}^{b'k'}=(C_1^{IJ})_{\ \ Kl}^{bk}\rho^K(\xi)_{\ l'}^{l}(C_2^{IJ})_{\ \ b'k'}^{ Kl'}$ and using \eqref{CCKK}, we obtain that 
	\be
	\lt[(C_2^{IJ})_{\ \ aj}^{ K'l}(R'^{IJ})^{ a j}_{\ bk}(C_{1}^{JI})^{kb}_{\ Kk'}\rt]{\rho}^{K}\left(\xi\right)^{k'}_{\ k}=\delta^{K'}_K\rho^K(\xi)_{\ l'}^{l}\lt[(C_2^{IJ})_{\ \ b'k'}^{ Kl'}(R'^{IJ})^{b'k'}_{\ \ cj}(C_{1}^{JI})^{jc}_{\ Kk}\rt].
	\ee
	By Schur's lemma, the quantity in the square bracket is a multiple of identity:
	\be 
	&&\lt(C_2^{IJ}\rt)_{\ \ b'k'}^{ Kl'}(R'^{IJ})^{b'k'}_{\ \ cj}(C_{1}^{JI})^{jc}_{\ Kk}=\cc_K\delta^{l'}_k,\\
	\text{or}&& (R'^{IJ})^{bk}_{\ \ cj}(C_{1}^{JI})^{jc}_{\ Kk}=\cc_K(C^{IJ}_1)^{bk}_{\ Kk},\qquad (C_2^{IJ})_{\ \ b'k'}^{ Kl}(R'^{IJ})^{b'k'}_{\ \ cj}=\cc_K(C_2^{JI})_{\ \ jc}^{ Kl},
	\ee
	where $\cc_K$ is a constant.

	Consider \eqref{prodrquantum3} and interchange the labels $I\leftrightarrow J$ and $1\leftrightarrow 2$ on top of $\bm{M}$: 
	\be 
	&&[\left(R^{-1}\right){}^{JI}]_{\ mc}^{ja}(\stackrel{2}{\bm{M}}{}^{J})_{\ n}^{m}(R^{JI})_{\ kd}^{nc}(\stackrel{1}{\bm{M}}{}^{I})_{\ b}^{d}=\sum_{K}(C_{1}^{JI})_{\ K c}^{ja}(\bm{M}^{K})_{\ d}^{c}(C_{2}^{JI})_{\ \ kb}^{Kd}\nonumber\\
	&=&\sum_{K}[(R'^{-1})^{IJ}]^{aj}_{\ \ bk}(C_{1}^{IJ})_{\ K c}^{bk}(\bm{M}^{K})_{\ d}^{c}(C_{2}^{IJ})_{\ \ b'k'}^{Kd}(R'^{IJ})^{b'k'}_{\ \ bk}.
	\ee
	The left-hand side can be rewritten as $[\left(R'^{-1}\right){}^{IJ}]_{\ cm}^{aj}(\stackrel{2}{\bm{M}}{}^{J})_{\ n}^{m}(R'^{IJ})_{\ dk}^{cn}(\stackrel{1}{\bm{M}}{}^{I})_{\ b}^{d}$. As a result, we obtain
	\be 
	&&(\stackrel{2}{\bm{M}}{}^{J})_{\ n}^{m}(R'^{IJ})_{\ dk}^{cn}(\stackrel{1}{\bm{M}}{}^{I})_{\ b}^{d}[(R'^{-1})^{IJ}]^{bk}_{\ \ b'k'}=\sum_{K}(C_{1}^{IJ})_{\ K c}^{cm}(\bm{M}^{K})_{\ d}^{c}(C_{2}^{IJ})_{\ \ b'k'}^{Kd}\nonumber\\
	&=&[(R^{-1})^{IJ}]_{\ aj}^{cm}(\stackrel{1}{\bm{M}}{}^{I})_{\ d}^{a}(R^{IJ})_{\ b'n}^{dj}(\stackrel{2}{\bm{M}}{}^{J})_{\ k'}^{n}.
	\ee
	Eqs.\eqref{lpeqn1} and \eqref{lpeqn2} can be derived similarly.

\section{Some computations of the representation of graph algebra}\label{App:Rep of graph algebra}

In the following, we check the representation satisfy the holomorphic part $\cl_{0,m}$ of the graph algebra 
\be 
\cd^{\vec{\l}}\left(\bm{M}_{\nu}^{J}\right)&=&\left(\rho^{J}\otimes\iota_\nu\right)\left(R^{\prime}\right)D_{\nu}^{\lambda_{\nu}}\left(\bm{M}_{\nu}^{J}\right)\left(\rho^{J}\otimes\iota_\nu\right)\left(R^{\prime-1}\right),\label{DlMnuJ}
\ee
where $\vec{\l}=(\lambda_{1},\cdots,\lambda_{m})$ and we skip the representation label of $\iota_\nu$. We label the tensor slots corresponding to $I,J,\l_1,\cdots,\l_m$ by $x,y,1,\cdots, m$ and skip all representation labels. Eq.\eqref{DlMnuJ} is written as
\be 
\cd^{\vec{\l}}\left(\bm{M}_{\nu}^{J}\right)= \Delta_{1}^{(\nu-1)}\left(R_{x1}^{\prime}\right)\bm{M}_{x\nu}\Delta_{1}^{(\nu-1)}\left(R_{x1}^{\prime-1}\right)=\Delta_{1}^{(\nu-1)}\left(R_{x1}^{\prime}\bm{M}_{x\nu}R_{x1}^{\prime-1}\right)
\ee
where $\Delta_{1}^{(\nu-1)}=\mathrm{id}_{x}\otimes\mathrm{id}_{y}\otimes\Delta^{(\nu-1)}\otimes\mathrm{id}_\nu$. 

Applying the representation to the left-hand side of \eqref{prodrquantum3}:
\be
(R^{-1})^{IJ}{\cal D}^{\vec{\lambda}}\left(\bm{M}_{\nu}^{I}\right)R^{IJ}{\cal D}^{\vec{\lambda}}\left(\bm{M}_{\nu}^{J}\right)&=&\Delta_{1}^{(\nu-1)}\left(R_{xy}^{-1}R_{x1}^{\prime}\bm{M}_{x\nu}R_{x1}^{\prime-1}R_{xy}R_{y1}^{\prime}\bm{M}_{y\nu}R_{y1}^{\prime-1}\right)\nonumber\\
&=&\Delta_{1}^{(\nu-1)}\left(R_{xy}^{-1}R_{x1}^{\prime}R_{y1}^{\prime}\bm{M}_{x\nu}R_{xy}\bm{M}_{y\nu}R_{x1}^{\prime-1}R_{y1}^{\prime-1}\right)\nonumber\\
&=&\Delta_{1}^{(\nu-1)}\left(R_{y1}^{\prime}R_{x1}^{\prime}R_{xy}^{-1}\bm{M}_{x\nu}R_{xy}\bm{M}_{y\nu}R_{x1}^{\prime-1}R_{y1}^{\prime-1}\right)
\ee
where we have used the relations $R_{x1}^{\prime-1}R_{xy}R_{y1}^{\prime}=R_{y1}^{\prime}R_{xy}R_{x1}^{\prime-1}$ and $R_{y1}^{\prime}R_{x1}^{\prime}R_{xy}^{-1}=R_{xy}^{-1}R_{x1}^{\prime}R_{y1}^{\prime}$.
By $(1\otimes\Delta)(R)=R_{13}R_{12}$ and cyclic permutation $(123)$, we obtain $(\Delta\otimes1)(R^{\prime})=R_{23}^{\prime}R_{13}^{\prime}$, and therefore,
\be
(R^{-1})^{IJ}{\cal D}^{\vec{\lambda}}\left(\bm{M}_{\nu}^{I}\right)R^{IJ}{\cal D}^{\vec{\lambda}}\left(\bm{M}_{\nu}^{J}\right)=\sum_{K}\Delta_{1}^{(\nu-1)}\left[\left(\Delta_{x}R_{x1}^{\prime}\right)(C_{1}^{IJ})_{K}\bm{M}_{\nu}^{K}(C_{2}^{IJ})^{K}\left(\Delta_{x}R_{x1}^{\prime-1}\right)\right].
\ee
We have used the result from the representation of loop algebra. By the intertwining property of $C_1$ and $C_2$, we have 
\be 
\left(\Delta_{x}R_{x1}^{\prime}\right)^{IJ\lambda}(C_{1}^{IJ})_{K}=(C_{1}^{IJ})_{K}\left(R^{\prime}\right)^{K\lambda},\qquad(C_{2}^{IJ})^{K}\left(\Delta_{x}R_{x1}^{\prime-1}\right)^{IJ\lambda}=\left(R^{\prime-1}\right)^{K\lambda}(C_{2}^{IJ})^{K}. 
\ee
As a result,
\be 
(R^{-1})^{IJ}{\cal D}^{\vec{\lambda}}\left(\bm{M}_{\nu}^{I}\right)R^{IJ}{\cal D}^{\vec{\lambda}}\left(\bm{M}_{\nu}^{J}\right)&=&\sum_{K}(C_{1}^{IJ})_{K}\left[\Delta_{1}^{(\nu-1)}\left(R^{\prime}\right)\bm{M}_{\nu}^{K}\Delta_{1}^{(\nu-1)}\left(R^{\prime-1}\right)\right](C_{2}^{IJ})^{K}\nonumber\\
&=&\sum_{K}(C_{1}^{IJ})_{K}{\cal D}^{\vec{\lambda}}\left(\bm{M}_{\nu}^{K}\right)(C_{2}^{IJ})^{K}. 
\ee
It shows that the representation is compatible to \eqref{prodrquantum3}.

The representation of the relation \eqref{loopeqnMM} involves both ${\cal D}^{\vec{\lambda}}\left(\bm{M}_{\nu}^{I}\right)$ and ${\cal D}^{\vec{\lambda}}\left(\bm{M}_{\mu}^{I}\right)$ for $\nu<\mu$. ${\cal D}^{\vec{\lambda}}\left(\bm{M}_{\mu}^{I}\right)$ can be written as
\be 
\bm{M}_{y\mu}&=& \Delta_{1}^{(\mu-1)}\left(R_{x1}^{\prime}\right)\bm{M}_{x\mu}\Delta_{1}^{(\mu-1)}\left(R_{x1}^{\prime-1}\right)\nonumber\\
&=&\Delta_{1}^{(\nu-1)}\lt(R_{y1}^{\prime}\rt)\Delta_{\nu}^{(\mu-\nu)}\lt(R_{y\nu}^{\prime}\rt)\bm{M}_{y\mu}\Delta_{\nu}^{(\mu-\nu)}\lt(R_{y\nu}^{\prime-1}\rt)\Delta_{1}^{(\nu-1)}\lt(R_{y1}^{\prime-1}\rt)
\ee
We have used the following relation:
\be
\Delta_{1}^{(\mu-1)}\lt(R_{y1}^{\prime}\rt)&=&R_{y1}^{\prime}R_{y2}^{\prime}R_{y3}^{\prime}\cdots R_{y\nu-1}^{\prime}R_{y\nu}^{\prime}\cdots R_{y\mu-1}^{\prime}
\nonumber\\
&=&\left(\Delta_{1}^{(\nu-1)}R_{y1}^{\prime}\right)R_{y\nu}^{\prime}\left(\Delta_{\nu+1}^{(\mu-\nu-1)}R_{y\nu+1}^{\prime}\right)
\ee
Applying the representation to the left-hand side of \eqref{loopeqnMM}:
\be
&&(R^{-1})^{IJ}{\cal D}^{\vec{\lambda}}\left(\bm{M}_{\nu}^{I}\right)R^{IJ}{\cal D}^{\vec{\lambda}}\left(\bm{M}_{\mu}^{J}\right)\nonumber\\
&=&\Delta_{1}^{(\nu-1)}\otimes\Delta_{\nu+1}^{(\mu-\nu-1)}\left[R_{xy}^{-1}\left(R_{x1}^{\prime}\bm{M}_{x\nu}R_{x1}^{\prime-1}\right)R_{xy}\left(R_{y1}^{\prime}R_{y\nu}^{\prime}R_{y\nu+1}^{\prime}\bm{M}_{y\mu}R_{y\nu+1}^{\prime-1}R_{y\nu}^{\prime-1}R_{y1}^{\prime-1}\right)\right].\label{RDMRDM111}
\ee
By the relations $R_{x1}^{\prime-1}R_{xy}R_{y1}^{\prime}=R_{y1}^{\prime}R_{xy}R_{x1}^{\prime-1}$ and $R_{xy}^{-1}R_{x1}^{\prime}R_{y1}^{\prime}=R_{y1}^{\prime}R_{x1}^{\prime}R_{xy}^{-1}$,
\be 
=\Delta_{1}^{(\nu-1)}\otimes\Delta_{\nu+1}^{(\mu-\nu-1)}\left[R_{y1}^{\prime}R_{x1}^{\prime}R_{xy}^{-1}\bm{M}_{x\nu}R_{xy}\left(R_{y\nu}^{\prime}R_{y\nu+1}^{\prime}\bm{M}_{y\mu}R_{y\nu+1}^{\prime-1}R_{y\nu}^{\prime-1}\right)R_{x1}^{\prime-1}R_{y1}^{\prime-1}\right]
\ee
By the relations $\bm{M}_{x\nu}R_{xy}R_{y\nu}^{\prime}=R_{xy}R_{y\nu}^{\prime}\bm{M}_{x\nu}$, $R_{x1}^{\prime-1}R_{y1}^{\prime-1}=R_{xy}^{-1}R_{y1}^{\prime-1}R_{x1}^{\prime-1}R_{xy}$, and $R_{y\nu}^{\prime-1}R_{xy}^{-1}\bm{M}_{x\nu}=\bm{M}_{x\nu}R_{y\nu}^{\prime-1}R_{xy}^{-1}$,
\be 
=\Delta_{1}^{(\nu-1)}\otimes\Delta_{\nu+1}^{(\mu-\nu-1)}\left(R_{y1}^{\prime}R_{y\nu}^{\prime}R_{y\nu+1}^{\prime}\bm{M}_{y\mu}R_{y\nu+1}^{\prime-1}R_{y\nu}^{\prime-1}R_{x1}^{\prime}R_{xy}^{-1}R_{y1}^{\prime-1}\bm{M}_{x\nu}R_{x1}^{\prime-1}R_{xy}\right)
\ee
By the relation $R_{y1}^{\prime-1}R_{xy}^{-1}R_{x1}^{\prime}=R_{x1}^{\prime}R_{xy}^{-1}R_{y1}^{\prime-1}$,
\be 
&=&\Delta_{1}^{(\nu-1)}\otimes\Delta_{\nu+1}^{(\mu-\nu-1)}\left[\left(R_{y1}^{\prime}R_{y\nu}^{\prime}R_{y\nu+1}^{\prime}\bm{M}_{y\mu}R_{y\nu+1}^{\prime-1}R_{y\nu}^{\prime-1}R_{y1}^{\prime-1}\right)R_{xy}^{-1}\left(R_{x1}^{\prime}\bm{M}_{x\nu}R_{x1}^{\prime-1}\right)R_{xy}\right]\nonumber\\
&=&\left[\left(\Delta_{1}^{(\mu-1)}R_{y1}^{\prime}\right)\bm{M}_{y\mu}\left(\Delta_{1}^{(\mu-1)}R_{y1}^{\prime-1}\right)\right]R_{xy}^{-1}\left[\Delta_{1}^{(\nu-1)}\left(R_{x1}^{\prime}\right)\bm{M}_{x\nu}\Delta_{1}^{(\nu-1)}\left(R_{x1}^{\prime-1}\right)\right]R_{xy}\nonumber\\
&=&{\cal D}^{\vec{\lambda}}\left(\bm{M}_{\mu}^{J}\right)\left(R^{-1}\right)^{IJ}{\cal D}^{\vec{\lambda}}\left(\bm{M}_{\nu}^{I}\right)R^{IJ}.
\ee
This shows that the representation is compatible to the relation \eqref{loopeqnMM}.

\section{Spectral representation}\label{App:spectral representation}

First, we can interchange the limit $\epsilon\to0$ and the integral once the $\mu$-integration contour is deformed $\mu\to\mu-i\alpha$ (for all $\a>0$), similar to the discussion in \eqref{psichif}, for any $f\in\Fd$
\be
\mathscr{V}_{\psi}f\left(\mu_{\chi},m_{\chi}\right)&=&\langle\psi_{\chi}\mid f\rangle=\lim_{\epsilon\to0}\sum_{m\in\mathbb{Z}/N\mathbb{Z}}\int_{-\infty}^{\infty}d\mu\,\overline{\psi_{\chi}^{\epsilon}\left(\mu,m\right)}f\left(\mu,m\right)\nonumber\\
&=&\lim_{\epsilon\to0}\sum_{m\in\mathbb{Z}/N\mathbb{Z}}\int_{-\infty}^{\infty}d\mu\,\overline{\psi_{\chi}^{\epsilon}\left(\mu+ib^{-1},m+1\right)}f\left(\mu-ib,m+1\right)\nonumber\\
&=&\lim_{\epsilon\to0}\sum_{m\in\mathbb{Z}/N\mathbb{Z}}\int_{-\infty}^{\infty}d\mu\,\frac{\overline{\psi_{\chi}^{\epsilon}\left(\mu+ib^{-1},m+1\right)}}{e^{\frac{2\pi i}{N/l}(-ib\mu-m)}+e^{-\frac{2\pi i}{N/l}(-ib\mu-m)}}\left(\bm{y}^l+\bm{y}^{-l}\right)\bm{u}^{-1}f\left(\mu,m\right)\nonumber\\
&=&\sum_{m\in\mathbb{Z}/N\mathbb{Z}}\int_{-\infty}^{\infty}d\mu\,\overline{\psi_{\chi}^{\epsilon=0}\left(\mu+ib^{-1},m+1\right)}f\left(\mu-ib,m+1\right)\nonumber\\
&=&\sum_{m\in\mathbb{Z}/N\mathbb{Z}}\int_{-\infty}^{\infty}d\mu\,\overline{\psi_{\chi}^{\epsilon=0}\left(\mu+i\alpha,m+1\right)}f\left(\mu-i\alpha,m+1\right)
\ee
where we have use $\psi^\epsilon_\chi$ to represent \eqref{psichimumexpandgamma} to make $\epsilon$ explicit. $N/l$ is odd for avoiding any pole.

The following derivation shows $\bmL$ is diagonalized in this representation
\be
&&\mathscr{V}_{\psi}\bm{L}f\left(\mu_{\chi},m_{\chi}\right)
=\lim_{\epsilon\to0}\sum_{m\in\mathbb{Z}/N\mathbb{Z}}\int_{-\infty}^{\infty}d\mu\,\overline{\psi_{\chi}^{\epsilon}\left(\mu,m\right)}\bm{L}f\left(\mu,m\right)
=\lim_{\epsilon\to0}\sum_{m\in\mathbb{Z}/N\mathbb{Z}}\int_{-\infty}^{\infty}d\mu\,\overline{\widetilde{\bm{L}}\psi_{\chi}^{\epsilon}\left(\mu,m\right)}f\left(\mu,m\right)\nonumber\\
&=&\lim_{\epsilon\to0}\sum_{m\in\mathbb{Z}/N\mathbb{Z}}\int_{-\infty}^{\infty}d\mu\,\frac{\chi+\chi^{-1}+R(\mu,m,\epsilon)}{e^{\frac{2\pi i}{N/l}(-ib\mu-m)}+e^{-\frac{2\pi i}{N/l}(-ib\mu-m)}}\frac{\overline{\psi_{\chi}^{\epsilon}\left(\mu+ib^{-1},m+1\right)}}{e^{\frac{2\pi i}{N/l}(-ib\mu-m)}+e^{-\frac{2\pi i}{N/l}(-ib\mu-m)}}\left(\bm{y}^l+\bm{y}^{-l}\right)^{2}\bm{u}^{-1}f\left(\mu,m\right)\nonumber\\
&=&\sum_{m\in\mathbb{Z}/N\mathbb{Z}}\int_{-\infty}^{\infty}d\mu\,\frac{\chi+\chi^{-1}}{e^{\frac{2\pi i}{N/l}(-ib\mu-m)}+e^{-\frac{2\pi i}{N/l}(-ib\mu-m)}}\frac{\overline{\psi_{\chi}^{\epsilon=0}\left(\mu+ib^{-1},m+1\right)}}{e^{\frac{2\pi i}{N/l}(-ib\mu-m)}+e^{-\frac{2\pi i}{N/l}(-ib\mu-m)}}\left(\bm{y}^l+\bm{y}^{-l}\right)^{2}\bm{u}^{-1}f\left(\mu,m\right)\nonumber\\
&=&\left(\chi+\chi^{-1}\right)\langle\psi_{\chi}\mid f\rangle=\left(\chi+\chi^{-1}\right)\mathscr{V}_{\psi}f\left(\mu_{\chi},m_{\chi}\right)
\ee
where 
\be
R(\mu,m,\epsilon)=\bm{q}^{-2}\left(1-e^{-\frac{2i\pi}{\sqrt{N}}b\epsilon}\right)e^{\frac{2\pi i}{N}\left(-ib\mu-m\right)}+\bm{q}^{2}\left(1-e^{\frac{2i\pi}{\sqrt{N}}b\epsilon}\right)e^{-\frac{2\pi i}{N}\left(-ib\mu-m\right)}
\ee
The reason for interchanging limit and integral (in the 4th step) is similar to the one in proving Lemma  \ref{psichiisdistribution}.


We derive the direct integral representation \eqref{sUfsUf}: We define $\varsigma(\mu_{\chi},m_{\chi})=\varsigma_1(\mu_{\chi},m_{\chi})+\varsigma_2(\mu_{\chi},m_{\chi})$ with
\be
\varsigma_1(\mu_{\chi},m_{\chi})&=&\frac{1}{N^{2}}e^{\frac{2\pi\mu_{\chi}b}{N}+\frac{2\pi\mu_{\chi}b^{-1}}{N}}\nonumber\\
\varsigma_2(\mu_{\chi},m_{\chi})&=&-\frac{1}{N^{2}}e^{-\frac{2\pi\mu_{\chi}b}{N}+\frac{2\pi\mu_{\chi}b^{-1}}{N}+\frac{4i\pi m_{\chi}}{N}}
\ee 
satisfies $\varrho^{-1}(\mu_\chi,m_\chi)=\varsigma(\mu_\chi,m_\chi)+\varsigma(-\mu_\chi,-m_\chi)$. We consider the following regularized delta function with $\delta>2\epsilon>0$:
\be
\frac{\delta}{2\pi\left(z-2i\epsilon\right)\left(z+i(\delta-2\epsilon)\right)}=\int_{-\infty}^\infty \frac{\rmd k}{2\pi}e^{ikz} s_{\epsilon,\delta}(k),\qquad s_{\epsilon,\delta}(k)=\begin{cases}
e^{-k\left(\delta-2\epsilon\right)} & k>0,\\
e^{2k\epsilon} & k<0.
\end{cases}
\ee
The inner product between $f.f'\in\Fd$ can be written as
\be  	
\langle f\mid f'\rangle&=&\sum_{m,m'\in\mathbb{Z}/N\mathbb{Z} }\ \lim_{\substack{\delta\to0\\\delta>2\epsilon}}\ \int_{\R^2}dx dy\left[\frac{\sqrt{N}\delta}{2\pi\left(x-y-2i\epsilon\right)\left(i(\delta-2\epsilon)+x-y\right)}\delta_{m.m'}\right]\nonumber\\
&&\qquad \qquad\qquad \qquad\qquad \qquad \overline{f\left(x-i\epsilon,m\right)}f^{\prime}\left(y-i\epsilon,m\right)\label{innerprodinstandard}
\ee
We use the result in \cite{Han:2024nkf} that 
\be
\frac{\delta}{2\pi\sqrt{N}\left(x-y-2i\epsilon\right)\left(i(\delta-2\epsilon)+x-y\right)}\delta_{m.m'}= I^\delta_1+I^\delta_2.\label{regulardeltaI1I2}
\ee  
$I^\delta_1$ and $I^\delta_2$ are integrals given by the same integrand but along different contours $\mathscr{C}_1$ and $\mathscr{C}_2$
\begin{eqnarray}
I^\delta_{1,2}
&=&\sum_{m_{\chi}\in\mathbb{Z}/N\mathbb{Z}}\int_{-\infty}^{\infty}d\mu_{\chi}\,e^{-\frac{2\pi}{\sqrt{N}}\delta\mu_{\chi}}\int_{\mathscr{C}_{1,2}}ds\,\Psi_{m,m'}\left(x-i\epsilon,y+i\epsilon,\mu_{\chi},m_{\chi},s\right)
\end{eqnarray}
The contour $\mathscr{C}_1$ is along $\mathbb{R}+i\im(s)$ with $2\epsilon-\delta<\im(s)<0$, while $\mathscr{C}_2$ is obtained by $\mathscr{C}_1$ by shifting $s\to s+ib/\sqrt{N}$ then reversing the direction. $\Psi_{m,m'}\left(x-i\epsilon,y+i\epsilon,\mu_{\chi},m_{\chi},s\right)$ is given by
\be
\Psi_{m,m'}&=& \frac{\zeta _0}{N^{3/2}}\sum_{c\in\Z/N\Z} e^{-i \pi  \left(\alpha ^2-\beta ^2\right)+\frac{i \pi  \left(a^2+2 a c-p^2\right)}{N}-i \pi  (a-p)-\frac{2 i \pi}{N}  d c}\nonumber\\
&&e^{2 i \pi  s \left(\alpha +\omega ''\right)-2 i \pi  t (s-2\omega '')}\frac{ \gamma (\alpha +t,a+d) \gamma \left(-\alpha +\beta +s-\omega '',-a-c+p\right) }{ \gamma \left(s+\omega '',-c\right) \gamma (\beta +t,d+p) \gamma \left(-\alpha +\beta -\omega '',p-a\right)},\\
-t&=& \lambda=\frac{\mu_\chi}{\sqrt{N}} ,\qquad \alpha = -\omega ''+y+i \epsilon ,\qquad\beta = x+\omega ''-i \epsilon ,\qquad a= m',\qquad p= m,\qquad -d= {m_r}.\nonumber
\ee
The following relations are shown in \cite{Han:2024nkf}:
\be 
I_1^\delta&=&\sum_{m_{\chi}\in\mathbb{Z}/N\mathbb{Z}}\int_{-\infty}^{\infty}d\mu_{\chi}\,e^{-\frac{2\pi}{\sqrt{N}}\delta\mu_{\chi}}\varsigma_{1}(\mu_{\chi},m_{\chi})\psi_{\chi}^{\epsilon}\left(y,m\right)\overline{\psi_{\chi}^{\epsilon}\left(x,m^{\prime}\right)},\label{I1psipsi}\\
I_2^{\delta=0}&=&\sum_{m_{\chi}\in\mathbb{Z}/N\mathbb{Z}}\int_{-\infty}^{\infty}d\mu_{\chi}\,\varsigma_{2}(\mu_{\chi},m_{\chi})\psi_{\chi}^{\epsilon}\left(y,m\right)\overline{\psi_{\chi}^{\epsilon}\left(x,m^{\prime}\right)}. \label{I2psipsi}
\ee
Inserting the relation \eqref{regulardeltaI1I2} in \eqref{innerprodinstandard}, the resulting integrand 
\be
\ci^\delta(x,y,\mu_\chi,s)\equiv\sum_{m,m'\in\Z/N\Z} e^{-\frac{2\pi}{\sqrt{N}}\delta\mu_\chi}\Psi_{m,m'}\left(x-i\epsilon,y+i\epsilon,\mu_{\chi},m_{\chi},s\right)\overline{f\left(y-i\epsilon,m\right)}f^{\prime}\left(x-i\epsilon,m'\right)\nonumber
\ee
exponentially suppresses at infinity for all integration variables $x,y,\mu_\chi \in \R$ and $s\in\mathscr{C}_{1,2}$. When $s\in\mathscr{C}_{2}$, $\ci^\delta(x,y,\mu_\chi,s)$ still exponentially suppresses for all integration variables even for $\delta=0$. We can freely interchange the order of integrations. Moreover, by analyticity  
\be 
\int dxdy\,\Psi_{m,m'}\left(x-i\epsilon,y+i\epsilon,\mu_{\chi},m_{\chi},s\right)\overline{f\left(y-i\epsilon,m\right)}f^{\prime}\left(x-i\epsilon,m\right)
\ee
is independent of $\epsilon>0$. Therefore, the limit in \eqref{innerprodinstandard} becomes only $\delta\to0$. By \eqref{I1psipsi},
\be
&&\langle f\mid f'\rangle=\sum_{m_{\chi}\in\mathbb{Z}/N\mathbb{Z}}\lim_{\delta\to0}\int_{-\infty}^{\infty}d\mu_{\chi}\,e^{-\frac{2\pi}{\sqrt{N}}\delta\mu_{\chi}}\varsigma_{1}(\mu_{\chi},m_{\chi})\langle f\mid\psi_{\chi}\rangle\langle\psi_{\chi}\mid f^{\prime}\rangle+\nonumber\\
&&\sum_{m,m',m_{\chi}\in\mathbb{Z}/N\mathbb{Z}}\lim_{\delta\to0}\int_{-\infty}^{\infty}d\mu_{\chi}e^{-\frac{2\pi}{\sqrt{N}}\delta\mu_{\chi}}\int_{\mathscr{C}_{2}}ds\int dxdy\,\Psi_{m,m'}\left(x-i\epsilon,y+i\epsilon,\mu_{\chi},m_{\chi},s\right)\overline{f\left(y-i\epsilon,m\right)}f^{\prime}\left(x-i\epsilon,m\right)\nonumber
\ee
The second integrand can be bounded by exponentially decaying function independent of $\delta$. Therefore we can interchange the limit and integral. By the relation \eqref{I2psipsi}, we obtain
\be
\langle f\mid f'\rangle&=&\sum_{m_{\chi}\in\mathbb{Z}/N\mathbb{Z}}\lim_{\delta\to0}\int_{-\infty}^{\infty}d\mu_{\chi}\lt[e^{-\frac{2\pi}{\sqrt{N}}\delta\mu_{\chi}}\varsigma_{1}(\mu_{\chi},m_{\chi})+\varsigma_{2}(\mu_{\chi},m_{\chi})\rt]\langle f\mid\psi_{\chi}\rangle\langle\psi_{\chi}\mid f^{\prime}\rangle\label{F13}\\
&=&\sum_{m_{\chi}\in\mathbb{Z}/N\mathbb{Z}}\lim_{\delta\to0}\int_{-\infty}^{\infty}d\mu_{\chi}\lt[e^{-\frac{2\pi}{\sqrt{N}}\delta\mu_{\chi}}\varsigma_{1}(\mu_{\chi},m_{\chi})+\varsigma_{2}(\mu_{\chi},m_{\chi})\rt]\overline{\mathscr{V}_{\psi} f(\mu_\chi,m_\chi)}\mathscr{V}_{\psi} f'(\mu_\chi,m_\chi).\nonumber
\ee

We denote by $\varsigma_\delta(\mu_\chi,m_\chi)=e^{-\frac{2\pi}{\sqrt{N}}\delta\mu_{\chi}}\varsigma_{1}(\mu_{\chi},m_{\chi})+\varsigma_{2}(\mu_{\chi},m_{\chi})$ and $\varrho^{-1}_\delta(\mu_\chi,m_\chi)=\varsigma_\delta(\mu_\chi,m_\chi)+\varsigma_\delta(-\mu_\chi,-m_\chi)$. We define the Hilbert space $L^2(\C,\rmd\varrho_\chi)$ by the resulting inner product
\be
\lag F\mid G\rag_\varrho=\sum_{m_{\chi}\in\mathbb{Z}/N\mathbb{Z}}\lim_{\delta\to0}\int_{0}^{\infty}d\mu_{\chi}\varrho^{-1}_\delta(\mu_\chi,m_\chi)\overline{F(\mu_\chi,m_\chi)}G(\mu_\chi,m_\chi).\label{limitinnerprod}
\ee
The result \eqref{F13} has the following implications:

\begin{itemize}

\item $\Vert \mathscr{V}_\psi f\Vert_\varrho=\Vert f\Vert $ indicates that $\mathscr{V}_\psi$ is norm preserving, of the domain of $\mathscr{V}_\psi$ can be extended to entire $\ch$. $\mathscr{V}_\psi: \ch\to L^2(\C,\rmd\varrho_\chi)$ is injective, since $\mathscr{V}_\psi f=0$ implies $f=0$.

\item We have the resolution of identity as the weak limit
\be
\text{w-}\!\!\lim_{\delta\to0}\sum_{m_{\chi}\in\mathbb{Z}/N\mathbb{Z}}\int_{0}^{\infty}d\mu_{\chi}\varrho^{-1}_\delta(\mu_\chi,m_\chi)\mid\psi_\chi\rangle\langle\psi_\chi\mid =\mathrm{id}_{\ch}.\label{weaklimitid}
\ee

\end{itemize}

\noindent
We define $\mathscr{V}_{\psi,\delta>2\epsilon}^{-1}$ by 
\be
\mathscr{V}_{\psi,\delta>2\epsilon}^{-1}\mid F\rangle_\varrho =\sum_{m_{\chi}\in\mathbb{Z}/N\mathbb{Z}}\int_{0}^{\infty}d\mu_{\chi}\varrho^{-1}_\delta(\mu_\chi,m_\chi)\mid\psi_\chi\rangle F(\mu_\chi,m_\chi).
\ee
Consider the dense domain $C_{c,0}^\infty(\R_{\geq 0})\otimes\C^N$ containing the smooth and compact support functions of $\mu_\chi$ vanishing at $\mu_\chi=0$ (the one-side derivatives exist at $\mu_\chi=0$), For any pair $F(\mu_\chi,m_\chi),G(\mu_\chi,m_\chi)\in C_{c,0}^\infty(\R_{\geq 0})\otimes\C^N $, \footnote{The regularized delta function is given by $\delta^\epsilon(\mu_\chi-\mu_\chi')=\int ds\,e^{2\pi is\left(\mu_{\chi}-\mu_{\chi}^{\prime}\right)}f_{\epsilon}\left(\chi,\chi^{\prime},s\right)$, where $f_{\epsilon}\left(\chi,\chi^{\prime},s\right)$ is analytic and exponentially decays as $s\to\pm\infty$, and $\lim_{\epsilon \to0}f_{\epsilon}=1$ \cite{Han:2024nkf}. Consider $\int_{0}^{\infty}d\mu_{\chi}^{\prime}\rho_{\delta}\left(\mu_{\chi}^{\prime}\right)^{-1}f_{\epsilon,\chi,\chi^{\prime}}\left(k\right)e^{-2\pi ik\mu_{\chi}^{\prime}}G\left(\mu_{\chi}^{\prime}\right)\equiv\tilde{G}_{\epsilon}(\chi,k)$, skipping the dependence on $m_\chi$. For $\mu_{\chi}$ in a compact interval, one can show that $\left|\tilde{G}_{\epsilon}(\chi,k)\right|\leq\frac{C_{2}}{1+\left|2\pi k\right|^{2}}$. Then $\int dke^{2\pi ik\mu_{\chi}}\tilde{G}_{\epsilon}(\chi,k)\to \rho_{\delta}\left(\mu_{\chi}\right)^{-1}G\left(\mu_{\chi}\right)$ as $\epsilon \to 0$ by the dominated convergence.
}
\be
&&\lim_{\substack{\delta\to0\\\delta>2\epsilon}}\ \sum_{m\in\mathbb{Z}/N\mathbb{Z}}\int d\mu\overline{\mathscr{V}_{\psi,\delta>2\epsilon}^{-1}F\left(\mu,m\right)}\mathscr{V}_{\psi,\delta>2\epsilon}^{-1}G\left(\mu,m\right)\nonumber\\
&=&	\lim_{\substack{\delta\to0\\\delta>2\epsilon}}\ \sum_{m_{\chi},m_{\chi}^{\prime}\in\mathbb{Z}/N\mathbb{Z}}\int_{0}^{\infty}d\mu_{\chi}d\mu_{\chi}^{\prime}\,\varrho_{\delta}\left(\mu_{\chi},m_{\chi}\right)\varrho_{\delta}\left(\mu_{\chi}^{\prime},m_{\chi}^{\prime}\right)\overline{F\left(\mu_{\chi},m_{\chi}\right)}\nonumber\\
&&\left[\varrho(\mu_{\chi},m_{\chi})\left(\delta^{\epsilon}(\mu_{\chi}-\mu_{\chi}^{\prime})\delta_{m_{\chi},m_{\chi'}}+\delta^{\epsilon}(\mu_{\chi}+\mu_{\chi}^{\prime})\delta_{m_{\chi},-m_{\chi'}}\right)\right]G\left(\mu_{\chi}^{\prime},m_{\chi}^{\prime}\right)\nonumber\\
&=&	\sum_{m_{\chi}\in\mathbb{Z}/N\mathbb{Z}}\int_{0}^{\infty}d\mu_{\chi}\,\varrho\left(\mu_{\chi},m_{\chi}\right)^{-1}\overline{F\left(\mu_{\chi},m_{\chi}\right)}G\left(\mu_{\chi}^{\prime},m_{\chi}^{\prime}\right).
\ee
The above implies for any $(\delta,\epsilon)$ in a small neighborhood and $\delta>2\epsilon>0$, $\left\Vert \mathscr{V}_{\psi,\delta>2\epsilon}^{-1}F\right\Vert \leq C\left\Vert F\right\Vert _{\varrho}$ so $\mathscr{V}_{\psi,\delta>2\epsilon}^{-1}: L^2(\C,\rmd\varrho_\chi)\to\ch$ is a bounded operator, and
\be
\lim_{\delta\to0}\left\Vert \mathscr{V}_{\psi,\delta>2\epsilon}^{-1}F\right\Vert =\left\Vert F\right\Vert _{\varrho},\label{convergenceofnorm}
\ee 
for any $F\in L^2(\C,\rmd\varrho_\chi)$.

For any pair $f\in\ch$ and $F\in L^2(\C,\rmd\varrho_\chi)$
\be
\lim_{\delta\to0}\langle f\mid \mathscr{V}_{\psi,\delta>2\epsilon}^{-1}\mid F\rangle= \langle \mathscr{V}_\psi f\mid F\rangle_\varrho \label{F17}
\ee
is a continuous linear functional of $f\in\Fd$, so $\lim_{\delta\to0}\mathscr{V}_{\psi,\delta>2\epsilon}\mid F\rangle_\varrho$ weakly converges in $\ch$ by Riesz's Theorem. We also have the convergence of norm \eqref{convergenceofnorm}. Therefore, $\lim_{\delta\to0}\mathscr{V}_{\psi,\delta>2\epsilon}\mid F\rangle_\varrho$ converges strongly to a state in $\ch$. As a result, we have the stong operator limit $\text{s-}\!\!\lim_{\delta\to0}\mathscr{V}_{\psi,\delta>2\epsilon}^{-1}$ converge to a norm preserving injection from $ L^2(\C,\rmd\varrho_\chi)$ to $\ch$.

The resolution of identity implies for any $f\in\ch$ and $F$ in the image of $\mathscr{V}_\psi$ 
\be
\langle f \mid \mathscr{V}_\psi^{-1}\mid F\rangle=\lim_{\delta\to0}\langle f\mid \mathscr{V}_{\psi,\delta>2\epsilon}^{-1}\mid F\rangle.
\ee
Therefore $\text{s-}\!\!\lim_{\delta\to0}\mathscr{V}_{\psi,\delta>2\epsilon}^{-1}=\mathscr{V}_\psi^{-1}$, and it implies that $\mathscr{V}_{\psi}^{-1}$ is defined on entire $L^2(\C,\rmd\varrho_\chi)$. So $\mathscr{V}_{\psi}$ is surjective. As a result, we prove that $\mathscr{V}_{\psi}$ is an isomorphism from $\ch$ to $L^2(\C,\rmd\varrho_\chi)$.

For any $F(\mu_\chi,m_\chi)$ in the dense domain of the smooth and compact support functions of $\mu_\chi$ and any $G\in L^2(\C,\rmd\varrho_\chi)$ continuous in $\mu_\chi$,
\be
\lag F\mid G\rag_\varrho=\sum_{m_{\chi}\in\mathbb{Z}/N\mathbb{Z}}\int_{0}^{\infty}d\mu_{\chi}\varrho^{-1}(\mu_\chi,m_\chi)\overline{F(\mu_\chi,m_\chi)}G(\mu_\chi,m_\chi).\label{F20}
\ee
since the integral and limit $\delta\to0$ can be interchanged by the dominated convergence theorem. The inner products \eqref{F20} and \eqref{limitinnerprod} are identical on the dense domain, so they have identical Cauchy sequences, then the completions give identical Hilbert spaces. It means that one can always interchange the limit and integral in \eqref{limitinnerprod}, as well as in \eqref{F13} and \eqref{weaklimitid}.

We define a subspace $\mathscr{K}\in L^2(\C,\rmd\varrho_\chi)$ of functions $F\left(\mu_{\chi},m_{\chi}\right)$ differentiable for $\mu_{\chi}\in(0,\infty)$ and satisfying $\Vert(\chi+\chi^{-1})^n F\Vert<\infty$ for all $n\in \Z$. We want to determine the linear functional $\xi$ on $\sk$ satisfing 
\be
\xi\lt[\lt((\chi+\chi^{-1})-(\chi_0+\chi_0^{-1})\rt)G\rt]=0, \qquad \forall G\in\sk,
\ee
for a given $\chi_0=\exp[\frac{2\pi i}{N}\left(-ib\mu_{\chi_{0}}-m_{\chi_{0}}\right)]$ with $\mu_{\chi_{0}}\geq 0$ and $m_{\chi_0}\in \Z/N\Z$. For any $F\in \sk$ satisfying $F\left(\mu_{\chi_{0}},m_{\chi_{0}}\right)=0$, we define 
\be
G\left(\mu_{\chi},m_{\chi}\right):=\frac{F\left(\mu_{\chi},m_{\chi}\right)}{\left(\chi+\chi^{-1}\right)-\left(\chi_{0}+\chi_{0}^{-1}\right)}
\ee
If $ \mu_{\chi_{0}}>0$ or $m_{\chi_{0}}\neq0$, $G\left(\mu_{\chi_{0}},m_{\chi_{0}}\right)=\frac{N}{2\pi b}(\chi_0-\chi_0^{-1})^{-1}{\partial_{\mu_{\chi}}F\left(\mu_{\chi_{0}},m_{\chi_{0}}\right)}$ is regular. If $\mu_{\chi_{0}}=0$ and $m_{\chi_{0}}=0$, $G\left(\mu_{\chi},m_{\chi}\right)$ is singular at $(0,0)$, but 
\be
\lim_{\mu_{\chi}\to0}\varrho\left(\mu_{\chi},0\right)^{-1}\left|G\left(\mu_{\chi},0\right)\right|^{2}=\frac{1}{4\pi^{2}}\left|\partial_{\mu_{\chi}}F\left(0,0\right)\right|^{2}
\ee
is regular. Moreover, for any given $ \mu_{\chi_{0}}, m_{\chi_{0}}$, there exists $\mathring{\mu}_{\chi}\gg\mu_{\chi_{0}}$,
\be
&&\sum_{m_{\chi}\in\mathbb{Z}/N\mathbb{Z}}\int_{0}^{\infty}d\mu_{\chi}\varrho\left(\mu_{\chi},m_{\chi}\right)^{-1}\left|\left(\chi+\chi^{-1}\right)^{n}G\left(\mu_{\chi},m_{\chi}\right)\right|^{2}\nonumber\\
&=&	\text{finite}+\sum_{m_{\chi}\in\mathbb{Z}/N\mathbb{Z}}\int_{\mathring{\mu}_{\chi}}^{\infty}d\mu_{\chi}\varrho\left(\mu_{\chi},m_{\chi}\right)^{-1}\left|\frac{X^{n}F\left(\mu_{\chi},m_{\chi}\right)}{X-X_{0}}\right|^{2}\nonumber\\
&\leq&\text{finite}+C \sum_{m_{\chi}\in\mathbb{Z}/N\mathbb{Z}}\int_{\mathring{\mu}_{\chi}}^{\infty}d\mu_{\chi}\varrho\left(\mu_{\chi},m_{\chi}\right)^{-1}\left|X^{n-1}F\left(\mu_{\chi},m_{\chi}\right)\right|^{2}<\infty 
\ee
where $X=\chi+\chi^{-1}$ and $X_0=\chi_0+\chi_0^{-1}$. So $G\in\sk$. In summary, any $F\in \sk$ satisfying $F\left(\mu_{\chi_{0}},m_{\chi_{0}}\right)=0$ can be written as $F\left(\mu_{\chi},m_{\chi}\right)=\lt[\left(\chi+\chi^{-1}\right)-\left(\chi_{0}+\chi_{0}^{-1}\right)\rt]G\left(\mu_{\chi},m_{\chi}\right)$ with $G\in \sk$. Then for any $F\in \sk$ and any $\phi\in\sk$ satisfying $\phi(\mu_{\chi_{0}},m_{\chi_{0}})=1 $, 
\be
F\left(\mu_{\chi},m_{\chi}\right)-F(\mu_{\chi_{0}},m_{\chi_{0}})\phi(\mu_\chi,m_{\chi})=\lt[\left(\chi+\chi^{-1}\right)-\left(\chi_{0}+\chi_{0}^{-1}\right)\rt]G\left(\mu_{\chi},m_{\chi}\right)
\ee
Then $\xi$ is determined as a delta function
\be
\xi[F]=cF(\mu_{\chi_{0}},m_{\chi_{0}}),\label{deltafunctionxi}
\ee
where $c=\xi[\phi]$ is a constant, since $\phi$ is arbitrary and independent of $F$.

In a similar way, 
\be
\xi\lt[\lt((\overline\chi+\overline\chi^{-1})-(\overline\chi_0+\overline\chi_0^{-1})\rt)G\rt]=0, \qquad \forall G\in\sk,
\ee
can also determine $\xi$ to satisfy \eqref{deltafunctionxi}.

\bibliographystyle{jhep}
\bibliography{muxin}

\providecommand{\href}[2]{#2}\begingroup\raggedright\begin{thebibliography}{10}

\bibitem{ChernSimons1974}
S.-S. Chern and J.~H. Simons, {\it Characteristic forms and geometric invariants},  {\em Annals of Mathematics} {\bf 99} (1974), no.~1 48--69.

\bibitem{Witten1989a}
E.~Witten, {\it {Quantum Field Theory and the Jones Polynomial}},  {\em Communications in Mathematical Physics} {\bf 121} (1989) 351--399.

\bibitem{ReshetikhinTuraev1991}
N.~Y. Reshetikhin and V.~G. Turaev, {\it Invariants of 3-manifolds via link polynomials and quantum groups},  {\em Inventiones mathematicae} {\bf 103} (1991), no.~1 547--597.

\bibitem{Elitzur:1989nr}
S.~Elitzur, G.~W. Moore, A.~Schwimmer, and N.~Seiberg, {\it {Remarks on the Canonical Quantization of the Chern-Simons-Witten Theory}},  {\em Nucl. Phys. B} {\bf 326} (1989) 108--134.

\bibitem{Moore:1989yh}
G.~W. Moore and N.~Seiberg, {\it {Taming the Conformal Zoo}},  {\em Phys. Lett. B} {\bf 220} (1989) 422--430.

\bibitem{Witten1988}
E.~Witten, {\it 2+ 1 dimensional gravity as an exactly soluble system},  {\em Nuclear Physics B} (1988).

\bibitem{thiemann2008modern}
T.~Thiemann, {\em Modern Canonical Quantum General Relativity}.
\newblock Cambridge University Press, 2007.

\bibitem{Engle2011}
J.~Engle, K.~Noui, A.~Perez, and D.~Pranzetti, {\it {The SU(2) Black Hole entropy revisited}},  {\em JHEP} {\bf 1105} (2011) 016, [\href{http://arxiv.org/abs/1103.2723}{{\tt arXiv:1103.2723}}].

\bibitem{Witten:2015aoa}
E.~Witten, {\it {Three lectures on topological phases of matter}},  {\em Riv. Nuovo Cim.} {\bf 39} (2016), no.~7 313--370, [\href{http://arxiv.org/abs/1510.07698}{{\tt arXiv:1510.07698}}].

\bibitem{1995AdPhy..44..405W}
X.-G. {Wen}, {\it {Topological orders and edge excitations in fractional quantum Hall states}},  {\em Advances in Physics} {\bf 44} (Sept., 1995) 405--473, [\href{http://arxiv.org/abs/cond-mat/9506066}{{\tt cond-mat/9506066}}].

\bibitem{Freedman:2001eqc}
M.~H. Freedman, A.~Kitaev, M.~J. Larsen, and Z.~Wang, {\it {Topological Quantum Computation}},  \href{http://arxiv.org/abs/quant-ph/0101025}{{\tt quant-ph/0101025}}.

\bibitem{Witten1991}
E.~Witten, {\it Quantization of chern-simons gauge theory with complex gauge group},  {\em Communications in Mathematical Physics} {\bf 137} (1991), no.~1 29--66.

\bibitem{analcs}
E.~Witten, {\it {Analytic Continuation Of Chern-Simons Theory}},  {\em Chern-Simons Gauge Theory: 20 years after} (2010) 347--446, [\href{http://arxiv.org/abs/1001.2933}{{\tt arXiv:1001.2933}}].

\bibitem{QFTvolume}
T.~Dimofte and S.~Gukov, {\it {Quantum field theory and the volume conjecture}},  {\em Contemp.Math.} {\bf 541} (2011) 41--67, [\href{http://arxiv.org/abs/1003.4808}{{\tt arXiv:1003.4808}}].

\bibitem{1997LMaPh..39..269K}
R.~M. {Kashaev}, {\it {The Hyperbolic Volume of Knots from the Quantum Dilogarithm}},  {\em Letters in Mathematical Physics} {\bf 39} (Feb., 1997) 269--275, [\href{http://arxiv.org/abs/q-alg/9601025}{{\tt q-alg/9601025}}].

\bibitem{1999math......5075M}
H.~{Murakami} and J.~{Murakami}, {\it {The colored Jones polynomials and the simplicial volume of a knot}},  {\em arXiv Mathematics e-prints} (May, 1999) math/9905075, [\href{http://arxiv.org/abs/math/9905075}{{\tt math/9905075}}].

\bibitem{HHKR}
H.~M. Haggard, M.~Han, W.~Kaminski, and A.~Riello, {\it {SL(2,C) Chern-Simons Theory, a non-Planar Graph Operator, and 4D Loop Quantum Gravity with a Cosmological Constant: Semiclassical Geometry}},  {\em Nucl. Phys.} {\bf B900} (2015) 1--79, [\href{http://arxiv.org/abs/1412.7546}{{\tt arXiv:1412.7546}}].

\bibitem{Han:2021tzw}
M.~Han, {\it {Four-dimensional spinfoam quantum gravity with a cosmological constant: Finiteness and semiclassical limit}},  {\em Phys. Rev. D} {\bf 104} (2021), no.~10 104035, [\href{http://arxiv.org/abs/2109.00034}{{\tt arXiv:2109.00034}}].

\bibitem{DGG11}
T.~Dimofte, D.~Gaiotto, and S.~Gukov, {\it {Gauge Theories Labelled by Three-Manifolds}},  {\em Commun.Math.Phys.} {\bf 325} (2014) 367--419, [\href{http://arxiv.org/abs/1108.4389}{{\tt arXiv:1108.4389}}].

\bibitem{Cordova:2013cea}
C.~Cordova and D.~L. Jafferis, {\it {Complex Chern-Simons from M5-branes on the Squashed Three-Sphere}},  {\em JHEP} {\bf 11} (2017) 119, [\href{http://arxiv.org/abs/1305.2891}{{\tt arXiv:1305.2891}}].

\bibitem{Dimofte2011}
T.~Dimofte, {\it {Quantum Riemann surfaces in Chern-Simons theory}},  {\em Adv.Theor.Math.Phys.} {\bf 17} (2013) 479--599, [\href{http://arxiv.org/abs/1102.4847}{{\tt arXiv:1102.4847}}].

\bibitem{levelk}
T.~Dimofte, {\it {Complex Chern-Simons Theory at Level k via the 3d-3d Correspondence}},  {\em Commun. Math. Phys.} {\bf 339} (2015), no.~2 619--662, [\href{http://arxiv.org/abs/1409.0857}{{\tt arXiv:1409.0857}}].

\bibitem{2020arXiv201215630E}
J.~{Ellegaard Andersen}, A.~{Malus{\`a}}, and G.~{Rembado}, {\it {Genus-one complex quantum Chern--Simons theory}},  {\em arXiv e-prints} (Dec., 2020) arXiv:2012.15630, [\href{http://arxiv.org/abs/2012.15630}{{\tt arXiv:2012.15630}}].

\bibitem{andersen2016level}
J.~E. Andersen and S.~Marzioni, {\it Level n teichm\"uller tqft and complex chern-simons theory},  \href{http://arxiv.org/abs/1612.06986}{{\tt arXiv:1612.06986}}.

\bibitem{Andersen2014}
J.~{Ellegaard Andersen} and R.~{Kashaev}, {\it {Complex Quantum Chern-Simons}},  {\em ArXiv e-prints} (Sept., 2014) [\href{http://arxiv.org/abs/1409.1208}{{\tt arXiv:1409.1208}}].

\bibitem{Bar-Natan1991}
D.~Bar-Natan and E.~Witten, {\it Perturbative expansion of chern-simons theory with noncompact gauge group},  {\em Communications in Mathematical Physics} {\bf 141} (1991), no.~2 423--440.

\bibitem{DGLZ}
T.~Dimofte, S.~Gukov, J.~Lenells, and D.~Zagier, {\it {Exact results for perturbative Chern-Simons theory with complex gauge group}},  {\em Commun.Num.Theor.Phys.} {\bf 3} (2009) 363--443.

\bibitem{GukovMarinoPutrov2017}
S.~Gukov, M.~Mari{\~n}o, and P.~Putrov, {\it Resurgence in complex chern--simons theory},  {\em Annales Henri Poincar{\'e}} {\bf 18} (2017), no.~4 1117--1169, [\href{http://arxiv.org/abs/1605.07615}{{\tt arXiv:1605.07615}}].

\bibitem{Fock:1998nu}
V.~V. Fock and A.~A. Rosly, {\it {Poisson structure on moduli of flat connections on Riemann surfaces and r matrix}},  {\em Am. Math. Soc. Transl.} {\bf 191} (1999) 67--86, [\href{http://arxiv.org/abs/math/9802054}{{\tt math/9802054}}].

\bibitem{Alekseev:1994pa}
A.~Y. Alekseev, H.~Grosse, and V.~Schomerus, {\it {Combinatorial quantization of the Hamiltonian Chern-Simons theory}},  {\em Commun. Math. Phys.} {\bf 172} (1995) 317--358, [\href{http://arxiv.org/abs/hep-th/9403066}{{\tt hep-th/9403066}}].

\bibitem{Alekseev:1994au}
A.~Y. Alekseev, H.~Grosse, and V.~Schomerus, {\it {Combinatorial quantization of the Hamiltonian Chern-Simons theory. 2.}},  {\em Commun. Math. Phys.} {\bf 174} (1995) 561--604, [\href{http://arxiv.org/abs/hep-th/9408097}{{\tt hep-th/9408097}}].

\bibitem{Alekseev:1995rn}
A.~Y. Alekseev and V.~Schomerus, {\it {Representation theory of Chern-Simons observables}},  \href{http://arxiv.org/abs/q-alg/9503016}{{\tt q-alg/9503016}}.

\bibitem{BNR}
E.~Buffenoir, K.~Noui, and P.~Roche, {\it {Hamiltonian quantization of Chern-Simons theory with SL(2,C) group}},  {\em Class.Quant.Grav.} {\bf 19} (2002) 4953, [\href{http://arxiv.org/abs/hep-th/0202121}{{\tt hep-th/0202121}}].

\bibitem{Podles1990}
P.~Podles and S.~Woronowicz, {\it {Quantum deformation of Lorentz group}},  {\em Commun.Math.Phys.} {\bf 130} (1990) 381--431.

\bibitem{BR}
E.~Buffenoir and P.~Roche, {\it {Harmonic analysis on the quantum Lorentz group}},  {\em Commun.Math.Phys.} {\bf 207} (1999) 499--555, [\href{http://arxiv.org/abs/q-alg/9710022}{{\tt q-alg/9710022}}].

\bibitem{Gaiotto:2024osr}
D.~Gaiotto and J.~Teschner, {\it {Schur Quantization and Complex Chern-Simons theory}},  \href{http://arxiv.org/abs/2406.09171}{{\tt arXiv:2406.09171}}.

\bibitem{Han:2024nkf}
M.~Han, {\it {Representations of a quantum-deformed Lorentz algebra, Clebsch-Gordan map, and Fenchel-Nielsen representation of quantum complex flat connections at level-$k$}},  \href{http://arxiv.org/abs/2402.08176}{{\tt arXiv:2402.08176}}.

\bibitem{Derkachov:2013cqa}
S.~E. Derkachov and L.~D. Faddeev, {\it {3j-symbol for the modular double of $SL_q(2,\mathbb{R})$ revisited}},  {\em J. Phys. Conf. Ser.} {\bf 532} (2014) 012005, [\href{http://arxiv.org/abs/1302.5400}{{\tt arXiv:1302.5400}}].

\bibitem{Kashaev2001}
R.~Kashaev, {\em The Quantum Dilogarithm and Dehn Twists in Quantum Teichm{\"u}ller Theory}, pp.~211--221.
\newblock Springer Netherlands, Dordrecht, 2001.

\bibitem{Nidaiev:2013bda}
I.~Nidaiev and J.~Teschner, {\it {On the relation between the modular double of $U_q(sl(2,R))$ and the quantum Teichmueller theory}},  \href{http://arxiv.org/abs/1302.3454}{{\tt arXiv:1302.3454}}.

\bibitem{MR1059324}
P.~Podle\'{s} and S.~L. Woronowicz, {\it Quantum deformation of {L}orentz group},  {\em Comm. Math. Phys.} {\bf 130} (1990), no.~2 381--431.

\bibitem{majid2000foundations}
S.~Majid, {\em Foundations of Quantum Group Theory}.
\newblock Cambridge University Press, 2000.

\bibitem{2008InMat.175..223F}
V.~V. {Fock} and A.~B. {Goncharov}, {\it {The quantum dilogarithm and representations of quantum cluster varieties}},  {\em Inventiones Mathematicae} {\bf 175} (Sept., 2008) 223--286, [\href{http://arxiv.org/abs/math/0702397}{{\tt math/0702397}}].

\bibitem{MOLNAR197729}
R.~K. Molnar, {\it Semi-direct products of hopf algebras},  {\em Journal of Algebra} {\bf 47} (1977), no.~1 29--51.

\bibitem{GelfandVol4}
I.~M. Gelfand and G.~E. Shilov, {\em {Generalized functions, Vol.4}}.
\newblock Academic Press, New York, NY, 1967.

\bibitem{Kashaev:2000ku}
R.~M. Kashaev, {\it {On the spectrum of Dehn twists in quantum Teichmuller theory}},  \href{http://arxiv.org/abs/math/0008148}{{\tt math/0008148}}.

\bibitem{Ponsot:2000mt}
B.~Ponsot and J.~Teschner, {\it {Clebsch-Gordan and Racah-Wigner coefficients for a continuous series of representations of U(q)(sl(2,R))}},  {\em Commun. Math. Phys.} {\bf 224} (2001) 613--655, [\href{http://arxiv.org/abs/math/0007097}{{\tt math/0007097}}].

\bibitem{Faddeev:1995nb}
L.~Faddeev, {\it {Discrete Heisenberg-Weyl group and modular group}},  {\em Lett.Math.Phys.} {\bf 34} (1995) 249--254, [\href{http://arxiv.org/abs/hep-th/9504111}{{\tt hep-th/9504111}}].

\bibitem{Giulini:1998rk}
D.~Giulini and D.~Marolf, {\it {On the generality of refined algebraic quantization}},  {\em Class. Quant. Grav.} {\bf 16} (1999) 2479--2488, [\href{http://arxiv.org/abs/gr-qc/9812024}{{\tt gr-qc/9812024}}].

\bibitem{Hsiao:2024phb}
C.-H. Hsiao and Q.~Pan, {\it {Quantum Curved Tetrahedron, Quantum Group Intertwiner Space, and Coherent States}},  \href{http://arxiv.org/abs/2407.03242}{{\tt arXiv:2407.03242}}.

\bibitem{Moore:1988qv}
G.~W. Moore and N.~Seiberg, {\it {Classical and Quantum Conformal Field Theory}},  {\em Commun. Math. Phys.} {\bf 123} (1989) 177.

\end{thebibliography}\endgroup

\end{document}